\documentclass[12pt,american,english]{article}
\usepackage{lmodern}
\usepackage{lmodern}
\usepackage[T1]{fontenc}
\usepackage[latin9]{inputenc}
\usepackage{geometry}
\usepackage{babel}
\usepackage{array}
\usepackage{float}
\usepackage{booktabs}
\usepackage{mathrsfs}
\usepackage{mathtools}
\usepackage{enumitem}
\usepackage{amsmath}
\usepackage{amsthm}
\usepackage{amssymb}
\usepackage{graphicx}
\usepackage{setspace}
\usepackage[authoryear,comma]{natbib}
\usepackage[pdfusetitle,
 bookmarks=true,bookmarksnumbered=false,bookmarksopen=false,
 breaklinks=false,pdfborder={0 0 1},backref=false,colorlinks=false]
 {hyperref}
\hypersetup{
 colorlinks,linkcolor={pAlgae},citecolor={pAlgae},urlcolor={pAlgae}}

\makeatletter

\providecommand{\tabularnewline}{\\}
\newcommand{\lyxdot}{.}

\theoremstyle{plain}
\newtheorem{thm}{\protect\theoremname}[section]
\theoremstyle{definition}
\newtheorem{rem}{\protect\remarkname}[section]
\theoremstyle{plain}
\newtheorem{assumption}{\protect\assumptionname}
\theoremstyle{plain}
\newtheorem{cor}{\protect\corollaryname}[section]
\theoremstyle{plain}
\newtheorem{lem}{\protect\lemmaname}[section]
\theoremstyle{definition}
\newtheorem*{defn*}{\protect\definitionname}

\usepackage{babel}
\usepackage{mathrsfs}
\usepackage{bbm}
\usepackage{indentfirst}
\usepackage{xcolor}
\allowdisplaybreaks
\definecolor{pAlgae}{RGB}{0, 0, 102}

\title{Inference in a Stationary/Nonstationary Autoregressive Time-Varying-Parameter
Model}
\author{Donald W. K. Andrews and Ming Li\thanks{Andrews: Department of Economics, Yale University. Email: \protect\href{mailto:donald.andrews@yale.edu}{donald.andrews@yale.edu}. Li: Department of Economics, National University of Singapore. Email:
\protect\href{mailto:mli@nus.edu.sg}{mli@nus.edu.sg}. The authors thank the participants of the Yale econometrics seminar, the SMU econometrics seminar, the Stock and Watson conference, and the CUHK Econometrics Workshop for helpful comments.}}
\date{\today}

\ifdefined\showcaptionsetup
 \PassOptionsToPackage{caption=false}{subfig}
\fi
\usepackage{subfig}
\makeatother

\addto\captionsamerican{\renewcommand{\assumptionname}{Assumption}}
\addto\captionsamerican{\renewcommand{\corollaryname}{Corollary}}
\addto\captionsamerican{\renewcommand{\definitionname}{Definition}}
\addto\captionsamerican{\renewcommand{\lemmaname}{Lemma}}
\addto\captionsamerican{\renewcommand{\remarkname}{Remark}}
\addto\captionsamerican{\renewcommand{\theoremname}{Theorem}}
\addto\captionsenglish{\renewcommand{\assumptionname}{Assumption}}
\addto\captionsenglish{\renewcommand{\corollaryname}{Corollary}}
\addto\captionsenglish{\renewcommand{\definitionname}{Definition}}
\addto\captionsenglish{\renewcommand{\lemmaname}{Lemma}}
\addto\captionsenglish{\renewcommand{\remarkname}{Remark}}
\addto\captionsenglish{\renewcommand{\theoremname}{Theorem}}
\providecommand{\assumptionname}{Assumption}
\providecommand{\corollaryname}{Corollary}
\providecommand{\definitionname}{Definition}
\providecommand{\lemmaname}{Lemma}
\providecommand{\remarkname}{Remark}
\providecommand{\theoremname}{Theorem}

\begin{document}
\global\long\def\labelenumi{\normalfont(\alph{enumi})}%
 
\global\long\def\qed{\ \qedsymbol}%
 \makeatletter \renewenvironment{proof}[1][\proofname]{\par\pushQED{\qed}\normalfont\topsep6\p@\@plus6\p@\relax\trivlist\item[\hskip\labelsep\bfseries#1\@addpunct{.}]\ignorespaces}{\popQED\endtrivlist\@endpefalse}
\makeatother

\global\long\def\a{\alpha}%
\global\long\def\b{\beta}%
\global\long\def\g{\gamma}%
\global\long\def\d{\delta}%
\global\long\def\e{\epsilon}%
\global\long\def\l{\lambda}%
\global\long\def\t{\theta}%
\global\long\def\o{\omega}%
\global\long\def\s{\sigma}%
\global\long\def\G{\Gamma}%
\global\long\def\D{\Delta}%
\global\long\def\L{\Lambda}%
\global\long\def\T{\Theta}%
\global\long\def\O{\Omega}%
\global\long\def\R{\mathbb{R}}%
\global\long\def\N{\mathbb{N}}%
\global\long\def\Q{\mathbb{Q}}%
\global\long\def\I{\mathbb{I}}%
\global\long\def\P{P}%
\global\long\def\E{E}%
\global\long\def\B{\mathbb{\mathbb{B}}}%
\global\long\def\S{\mathbb{\mathbb{S}}}%
\global\long\def\V{\mathbb{\mathbb{V}}\text{ar}}%
\global\long\def\X{{\bf X}}%
\global\long\def\cX{\mathscr{X}}%
\global\long\def\cY{\mathscr{Y}}%
\global\long\def\cA{\mathscr{A}}%
\global\long\def\cB{\mathscr{B}}%
\global\long\def\cM{\mathscr{M}}%
\global\long\def\cN{\mathcal{N}}%
\global\long\def\cG{\mathcal{G}}%
\global\long\def\cC{\mathcal{C}}%
\global\long\def\sp{\,}%
\global\long\def\es{\emptyset}%
\global\long\def\mc#1{\mathscr{#1}}%
\global\long\def\ind{\mathbf{\mathbbm1}}%
\global\long\def\indep{\perp}%
\global\long\def\any{\forall}%
\global\long\def\ex{\exists}%
\global\long\def\p{\partial}%
\global\long\def\cd{\cdot}%
\global\long\def\Dif{\nabla}%
\global\long\def\imp{\Rightarrow}%
\global\long\def\iff{\Leftrightarrow}%
\global\long\def\up{\uparrow}%
\global\long\def\down{\downarrow}%
\global\long\def\arrow{\rightarrow}%
\global\long\def\rlarrow{\leftrightarrow}%
\global\long\def\lrarrow{\leftrightarrow}%
\global\long\def\abs#1{\left|#1\right|}%
\global\long\def\norm#1{\left\Vert #1\right\Vert }%
\global\long\def\rest#1{\left.#1\right|}%
\global\long\def\bracket#1#2{\left\langle #1\middle\vert#2\right\rangle }%
\global\long\def\sandvich#1#2#3{\left\langle #1\middle\vert#2\middle\vert#3\right\rangle }%
\global\long\def\turd#1{\frac{#1}{3}}%
\global\long\def\ellipsis{\textellipsis}%
\global\long\def\sand#1{\left\lceil #1\right\vert }%
\global\long\def\wich#1{\left\vert #1\right\rfloor }%
\global\long\def\sandwich#1#2#3{\left\lceil #1\middle\vert#2\middle\vert#3\right\rfloor }%
\global\long\def\abs#1{\left|#1\right|}%
\global\long\def\norm#1{\left\Vert #1\right\Vert }%
\global\long\def\rest#1{\left.#1\right|}%
\global\long\def\inprod#1{\left\langle #1\right\rangle }%
\global\long\def\ol#1{\overline{#1}}%
\global\long\def\ul#1{\underline{#1}}%
\global\long\def\td#1{\tilde{#1}}%
\global\long\def\bs#1{\boldsymbol{#1}}%
\global\long\def\upto{\nearrow}%
\global\long\def\downto{\searrow}%
\global\long\def\pto{\rightarrow_{p}}%
\global\long\def\dto{\rightarrow_{d}}%
\global\long\def\asto{\rightarrow_{a.s.}}%
\global\long\def\gto{\rightarrow}%
\global\long\def\fto{\Rightarrow}%
\begin{titlepage}

\clearpage\maketitle \thispagestyle{empty}
\begin{abstract}
This paper considers nonparametric estimation and inference in first-order
autoregressive (AR(1)) models with deterministically time-varying
parameters. A key feature of the proposed approach is to allow for
time-varying stationarity in some time periods, time-varying nonstationarity
(i.e., unit root or local-to-unit root behavior) in other periods,
and smooth transitions between the two. The estimation of the AR parameter
at any time point is based on a local least squares regression method,
where the relevant initial condition is endogenous. We obtain limit
distributions for the AR parameter estimator and t-statistic at a
given point $\tau$ in time when the parameter exhibits unit root,
local-to-unity, or stationary/stationary-like behavior at time $\tau.$
These results are used to construct confidence intervals and median-unbiased
interval estimators for the AR parameter at any specified point in
time. The confidence intervals have correct asymptotic coverage probabilities
with the coverage holding uniformly over stationary and nonstationary
behavior of the observations.

\bigskip{}

\end{abstract}
\end{titlepage}

\section{\protect\label{sec:MT-Intro}Introduction}

Autoregressive models --- stationary or nonstationary --- are workhorse
models in econometric time series. In this paper, we consider a deterministically
time-varying parameter (TVP) autoregressive model that allows for
stationary and non-stationary behavior at different points in the
time period of interest. Thus, the level of persistence of the time
series can change over the time period. The motivation for considering
such a model is that the economy is in continual transition due to
technological, institutional, political, and demographic changes.
The model considered allows for the intercept and error variance of
the model also to be time-varying, not just the AR parameter.

The estimation method employed is local least squares, which depends
on a tuning parameter $h$ that determines the local neighborhood
that is considered. We construct a confidence interval (CI) for the
value of the AR parameter at a time point $\tau$ via the inversion
of tests, as is common in the literature for constant parameter AR
models, e.g., see \citet*{stock1991confidence,andrews1993exactly,hansen1999grid,mikusheva2007uniform},
and \citet{andrews2014conditional}. We show that the CI's have correct
uniform asymptotic coverage probability for a parameter space that
allows the time series to be stationary in parts of the time period
and nonstationary in other parts, using the approach in \citet*{andrews2020generic}.
We also construct asymptotically median-unbiased interval estimators
(MUE's) in an analogous fashion.

In the TVP case, the initial condition is endogenous due to the choice
of the local neighborhood and depends on the potentially different
behavior of the time series prior to the local neighborhood. For a
given time point $\tau$ of interest, we find that the asymptotic
distributions of the LS estimator and t-statistic depend on the endogenous
initial condition in the local-to-unity case. This is analogous to
the asymptotic effect of the initial condition--under certain assumptions
on the initial condition--in constant parameter AR models, see \citet*{elliott1999efficient,elliott2001confidence},
and \citet{muller2003tests}. On the other hand, the endogenous initial
condition does not affect the asymptotic distributions when the AR
coefficient at time $\tau$ is more distant from one than local-to-unity.

We note that whether a time series is local-to-unity at time $\tau,$
or not, depends on $\tau$ and the chosen bandwidth $h.$ It is important
to provide asymptotic results that hold uniformly over a parameter
space that does not depend on $h,$ which we do.

We introduce a method for determining the tuning parameter $h$ based
on a forecast-error criterion. We provide conditions under which this
data-dependent choice of $h$ is asymptotically equivalent to an infeasible
choice that minimizes the unobserved ``empirical loss.'' These results
are similar to results for i.i.d. models given in \citet{li1987asymptotic}
and \citet{andrews1991asymptotic}.

We provide Monte Carlo simulation results for the methods introduced
in the paper. We consider true autoregressive functions whose shapes
are sinusoidal, linear, partly flat/partly linear, flat, and kinked
linear. We consider cases where the functions are close, or equal,
to one in some regions, but different from one by varying amounts
in other regions. We find that the proposed CI has reasonably good
coverage probabilities and short average lengths for most of the data-generating
processes considered. For example, nominal 95\% CI's are found to
have finite-sample coverage probabilities ranging from 92.5\% to 96\%
in 88.8\% of the cases, across 205 cases. The lowest coverage probabilities,
in the range from 87.5\% to 90\%, occur only in 2.0\% of the cases.
The MUE's are found to have very small finite-sample median bias across
the different cases considered. The magnitude of the data-dependent
choice of $h$ varies widely depending upon the shape of the autoregressive
function, as desired.

We provide some empirical applications of the methods to monthly inflation
and real exchange rates in several countries using data from the IMF
International Financial Statistics database. We find that the inflation
series exhibit noticeable time variation of the AR parameter across
the time period considered. We discuss how this relates to the literature
on the persistence properties of inflation. On the other hand, we
find that the real exchange rate series have nearly constant AR values
across time that are equal to, or close to, one. Hence, the TVP methods
are capable of producing constant AR values when it is appropriate
to do so. We discuss how these results relate to the literature on
the persistence of exchange rates. In Section \ref{sec:addl-Empirical-Results}
of the Supplemental Material, we also report results for interest
rates for several countries and results for a number of US macroeconomic
time series using the Federal Reserve Economic Database (FRED).

The results of the paper apply to a TVP-AR(1) model. For some time
series, a TVP-AR(p) model with $p>1$ may be more appropriate than
a TVP-AR(1) model. The methods introduced in the paper can be extended
to a TVP-AR(p) model, see Section \ref{sec:Extension-to-TVP-AR(p)}
of the Supplemental Material for details.

Relative to the literature, the contribution of this paper is to develop
methods for a deterministically time-varying parameter (TVP) AR(1)
model that allows time-varying stationarity in some time periods and
time-varying nonstationarity in other periods. The resulting model
is much more flexible than a constant parameter model. No paper in
the existing literature does this. The methods we employ are quite
similar to those employed in constant parameter AR(1) models that
impose stationarity or nonstationarity across the whole time period,
such as those referenced above. In particular, we invert tests of
null hypotheses concerning the AR coefficient (at a particular point
in time) and utilize the nonstandard asymptotic distributions of the
t-statistics for such null hypotheses to obtain critical values. Our
results differ from those obtained for constant parameter AR(1) models
in that (i) we consider a local neighborhood of the time point of
interest $\tau,$ indexed by a bandwidth parameter $h,$ and within
that time period the AR coefficient changes with time, which causes
biases that have to be accounted for, (ii) the initial condition of
the local neighborhood depends on the choice of the local neighborhood
and on the past behavior of the TVP-AR(1) process, which may differ
from its current behavior, which needs to be accounted for, and (iii)
we use the data to select a suitable local neighborhood using a forecast-error
criterion function. None of these features arise in a constant parameter
model.

The literature contains numerous papers that consider time-varying
parameter AR models. Some of these papers consider deterministic TVP's,
as in this paper, but they do not allow for stationarity in some time
periods and nonstationarity in others. References for stationary (i.e.,
short-range dependent) TVP AR models include \citet*{Rao1970,Grenier1983,dahlhaus1996kullback,Dahlhaus1997},
\citet{Dahlhaus1998}, \citet*{dahlhaus1999nonlinear}, \citet*{Moulines2005AoS},
\citet*{PhillipsXu2008Adaptive}, \citet*{ding2017ejs}, \citet{vanDelft2018EJS},
and \citet*{karmakar2022simultaneous}. References for nonstationary
(i.e., long-range dependent) TVP AR models include \citet*{bykhovskaya2018boundary,bykhovskaya2020point},
which focus on tests of a unit root null hypothesis against functional
local-to-unity alternatives, as opposed to CI's for an AR parameter,
which is the focus of this paper. No papers in the literature consider
CI's for an AR parameter in nonstationary TVP AR models.

The literature also includes papers on random coefficient (RC) AR
models and functional coefficient (FC) AR models (in which the AR
coefficient depends on observable variables). References for stationary
RC AR models includes \citet{NichollsQuinn1980}, \citet{NichollsQuinn1981},
\citet*{doan1984forecasting}, and \citet{cogley2005drifts}. References
for papers on RC AR models with random coefficients that follow a
nonstationary process include \citet{follmer1993microeconomic}, \citet*{GIRAITISKapetaniosYates2014rckernel,GiraitisKapetaniosYates2018},
and \citet*{tao2019random}. A reference for stationary FC AR models
is \citet*{CaiFanYao2000}. References for nonstationary FC AR models
includes \citet{Juhl2005}, \citet{Lieberman2012}, and \citet*{lieberman2014norming,lieberman2017multivariate,lieberman2018iv}.

This paper is organized as follows. Section \ref{sec:MT-Model=000020Setup}
introduces the TVP-AR(1) model. Section \ref{sec:MT-CI=000020AR=000020Par}
introduces the CI and MUE for the AR parameter at time $\tau.$ Section
\ref{sec:MT-Choice_of_h} introduces the data-dependent method for
choosing the bandwidth parameter $h$ based on a forecast-error criterion.
Section \ref{sec:MT-Monte-Carlo-Simulations} presents the Monte Carlo
simulation results. Section \ref{sec:MT-Empirical-Applications} presents
the empirical results. Section \ref{sec:MT-Asymptotics} shows that
the CI has correct uniform asymptotic coverage probability, the MUE
is asymptotically median unbiased, and the data-dependent method for
choosing the bandwidth parameter $h$ has some some desirable asymptotic
properties. It also provides the asymptotic behavior of the local
least squares estimator and t-statistic under a variety of drifting
sequences of distributions, which are used in the proof of the uniform
asymptotic coverage probability results and the asymptotic median
unbiasedness results. The Supplemental Material includes the critical
values of the limiting distribution of our t-statistic in Section
\ref{sec:SM-Critical-Values-J_psi}, the proofs of the results of
the paper in Section \ref{sec:Theory}, additional simulation results
in Section \ref{sec:Additional-Simulation-Results}, a description
of how to extend the methods to TVP-AR(p) models for $p>1$ in Section
\ref{sec:Extension-to-TVP-AR(p)}, and information about the empirical
applications and additional empirical results in Section \ref{sec:addl-Empirical-Results}.

All limits in this paper are as $n\rightarrow\infty.$ For notational
simplicity, but with some abuse of notation, we let $\cdot/nh$ denote
$\cdot/(nh)$ throughout the paper.

\section{\protect\label{sec:MT-Model=000020Setup}Model}

\numberwithin{equation}{section}

The TVP-AR(1) model we consider is 
\begin{align}
Y_{t} & =\mu_{t}+Y_{t}^{*}\ \text{and}\nonumber \\
Y_{t}^{*} & =\rho_{t}Y_{t-1}^{*}+\sigma_{t}U_{t},\ \text{for\ }t=1,...,n,\label{eq:MT-tvp-model}
\end{align}
where $\rho_{t}\in\left[-1+\varepsilon_{1},1\right]$ for some $0<\varepsilon_{1}<2$.
The autoregressive parameter $\rho_{t}$ is allowed to vary with time
$t.$ A key feature of the model is that it allows for stationary,
unit root, or local-to-unity behavior at different points in time.
The errors $\left\{ U_{t}:t=0,1,...,n\right\} $ are a stationary
martingale difference sequence under $F$ with $\E_{F}\left(\rest{U_{t}}\mathscr{G}_{t-1}\right)=0\ \text{a.s.}$,
$\E_{F}\left(\rest{U_{t}^{2}}\mathscr{G}_{t-1}\right)=1\ \text{a.s.}$,
and $\E_{F}\left(\rest{U_{t}^{4}}\mathscr{G}_{t-1}\right)<M\ \text{a.s.}$
for some $M\in\left(0,\infty\right)$, where $\mathscr{G}_{t}$ is
some non-decreasing sequence of $\sigma$-fields for which $\sigma\left(U_{0},...,U_{t},Y_{0}^{*}\right)\subseteq\mathscr{G}_{t}$
for $t=1,...,n$.

We assume $\rho_{t},$ $\mu_{t},$ and $\sigma_{t}^{2}$ satisfy 
\begin{equation}
\rho_{t}:=\rho\left(t/n\right),\ \mu_{t}:=\mu\left(t/n\right),\ \text{and }\sigma_{t}^{2}:=\sigma^{2}\left(t/n\right),\label{eq:MT-rhomusigma}
\end{equation}
respectively, where $\rho\left(\cdot\right)\ $is a twice continuously
differentiable function on $\left[0,1\right]$ and $\mu\left(\cdot\right)\ $and
$\sigma^{2}\left(\cdot\right)$ are Lipschitz functions on $\left[0,1\right]$.
Given \eqref{eq:MT-rhomusigma}, $Y_{t}$, $Y_{t}^{*}$, $\rho_{t}$,
$\mu_{t}$, and $\sigma_{t}$ depend implicitly on $n$.

Let $\tau\in\left(0,1\right)$. We consider estimation and inference
concerning 
\begin{equation}
\rho(\tau),\label{eq:MT-paramofinterest}
\end{equation}
which is the value of the autoregressive function $\rho(\cdot)$ at
the $\tau$ fraction of the way through the sample.

For ease of reading, the definition of the parameter space of functions
$\rho\left(\cdot\right)$, $\mu\left(\cdot\right),\ $and $\sigma^{2}\left(\cdot\right)$
that is considered, which includes some structure on the $\rho\left(\cdot\right)$
and $\mu\left(\cdot\right)\ $functions, is given in Section \ref{sec:MT-Parameter=000020Space}
below. 

\section{Confidence Interval for the Autoregressive \protect \protect \\
 Parameter \texorpdfstring{$\boldsymbol{\rho(\tau)}$}{rho(tau)}
\protect\label{sec:MT-CI=000020AR=000020Par}}

\subsection{Local Least Squares Estimator of \texorpdfstring{$\boldsymbol{\rho(\tau)}$}{rho(tau)}}

We employ a LS estimator of $\rho(\tau)$ based on the time periods
$t=T_{1},...,T_{2}$, where 
\begin{equation}
T_{1}=\lfloor n\tau\rfloor-\lfloor nh/2\rfloor\ \text{and\ }T_{2}=\lfloor n\tau\rfloor+\lfloor nh/2\rfloor,\label{eq:MT-timeperiod}
\end{equation}
for a bandwidth parameter $h$. The total number of periods in $\left[T_{1},T_{2}\right]$
is within one of $nh$. For this range of time periods, the initial
condition time period is $T_{0}:=T_{1}-1.$

For the asymptotic results given below the bandwidth $h$ satisfies
$h\gto0$ and $nh\gto\infty$ as $n\gto\infty$. In Section \ref{sec:MT-Choice_of_h}
below, we introduce a data-dependent bandwidth that is smaller or
larger depending on how wiggly or flat the true $\rho\left(\cdot\right)\ $
and $\mu\left(\cdot\right)\ $ functions are. 

Define 
\begin{equation}
\ol Y_{nh}:=\dfrac{1}{nh}\sum_{t=T_{1}}^{T_{2}}Y_{t}\ \text{and}\ \ol Y_{nh,-1}:=\dfrac{1}{nh}\sum_{t=T_{1}}^{T_{2}}Y_{t-1}.\label{eq:MT-Y_dagger}
\end{equation}

To estimate $\rho\left(\tau\right)$, we regress $Y_{t}$ on a constant
and $Y_{t-1}.$ The resulting local LS estimator $\widehat{\rho}_{n\tau}$
is 
\begin{equation}
\widehat{\rho}_{n\tau}=\dfrac{\sum_{t=T_{1}}^{T_{2}}\left(Y_{t-1}-\ol Y_{nh,-1}\right)\left(Y_{t}-\ol Y_{nh}\right)}{\sum_{t=T_{1}}^{T_{2}}\left(Y_{t-1}-\ol Y_{nh,-1}\right)^{2}}.\label{eq:MT-rho_hat}
\end{equation}

\subsection{Confidence Interval for \texorpdfstring{$\boldsymbol{\rho(\tau)}$}{rho(tau)}
\protect\label{subsec:MT-CI=000020AR=000020Par} }

The CI for $\rho\left(\tau\right)$ that we consider is obtained by
inverting tests of null hypotheses of the form $H_{0}:\rho\left(\tau\right)=\rho_{0}$
for different values $\rho_{0}\in\left[-1+\varepsilon_{1},1\right]$.

The estimator of the time-varying variance $\sigma^{2}\left(\cdot\right)$
at $t/n=\tau$ is defined to be 
\begin{equation}
\widehat{\sigma}_{n\tau}^{2}\coloneqq\left(nh\right)^{-1}\sum_{t=T_{1}}^{T_{2}}\left[Y_{t}-\ol Y_{nh}-\widehat{\rho}_{n\tau}\left(Y_{t-1}-\ol Y_{nh,-1}\right)\right]^{2}.\label{eq:MT-sigma=000020ntau=000020def}
\end{equation}

For arbitrary $\rho_{0}\in(-1,1],$ the t-statistic that is used to
construct the CI for $\rho(\tau)$ is 
\begin{equation}
T_{n}\left(\rho_{0}\right)\coloneqq\dfrac{\left(nh\right)^{1/2}\left(\widehat{\rho}_{n\tau}-\rho_{0}\right)}{\widehat{s}_{n\tau}},\ \text{where }\widehat{s}_{n\tau}^{2}\coloneqq\widehat{\sigma}_{n\tau}^{2}/\left(nh\right)^{-1}\sum_{t=T_{1}}^{T_{2}}\left(Y_{t-1}-\ol Y_{nh,-1}\right)^{2}.\label{eq:MT-t-stat=000020def}
\end{equation}

Let $B(\cd)$ denote a standard Brownian motion on $[0,1].$ Let $Z_{1}$
be a standard normal random variable that is independent of $B\left(\cd\right)$.
Define 
\begin{align}
I_{\psi}\left(s\right):= & \int_{0}^{s}\exp\left\{ -\left(s-r\right)\psi\right\} dB\left(r\right),\nonumber \\
I_{\psi}^{*}\left(s\right)\coloneqq & \begin{cases}
I_{\psi}\left(s\right)+\dfrac{1}{\sqrt{2\psi}}\exp\left(-\psi s\right)Z_{1} & \text{for }\psi>0\\
B\left(s\right) & \text{for }\psi=0,\text{ and}
\end{cases}\nonumber \\
I_{D,\psi}^{*}\left(s\right):= & I_{\psi}^{*}\left(s\right)-\int_{0}^{1}I_{\psi}^{*}\left(r\right)dr.\label{eq:MT-I_D,tau(r)=000020def}
\end{align}
The stochastic process $I_{\psi}\left(s\right)$ is an Ornstein-Uhlenbeck
process on $\left[0,1\right]$ with parameter $\psi.$

Below we consider sequences of functions $\left\{ \rho_{n}\left(\cdot\right),\mu_{n}\left(\cdot\right),\sigma_{n}\left(\cdot\right)\right\} _{n\geq1}$
and null hypotheses $H_{0}:\rho_{n}\left(\tau\right)=\rho_{0,n},$
where the null hypotheses values $\rho_{0,n}$ depend on $n$ for
$n\geq1.$ For suitable sequences $\left\{ \rho_{0,n}\right\} _{n\geq1}$,
we show that under $H_{0},$ 
\begin{equation}
T_{n}\left(\rho_{0,n}\right)\dto J_{\psi}\ \text{for}\ \psi\in\left[0,\infty\right],\label{eq:MT-Asy=000020distn=000020T=000020stat}
\end{equation}
where $T_{n}\left(\rho_{0,n}\right)$ is defined with $\rho_{0,n}$
in place of $\rho_{0}$ in \eqref{eq:MT-t-stat=000020def}, $\psi$
depends on the sequence $\left\{ \rho_{0,n}\right\} _{n\geq1},$ and
$J_{\psi}$ is defined as follows. For $\psi=\infty$, which corresponds
to a ``stationary'' sequence $\left\{ \rho_{0,n}\right\} _{n\geq1}$,
$J_{\psi}$ has a $N(0,1)$ distribution. For $\psi\in\left[0,\infty\right)$,
which corresponds to a local-to-unity or unit root sequence $\left\{ \rho_{0,n}\right\} _{n\geq1}$,
\begin{equation}
J_{\psi}:=\left(\int_{0}^{1}I_{D,\psi}^{*2}\left(s\right)ds\right)^{-1/2}\int_{0}^{1}I_{D,\psi}^{*}\left(s\right)dB\left(s\right).\label{eq:MT-Distn=000020J_psi}
\end{equation}

For $\alpha\in(0,1),$ let $c_{\psi}\left(\alpha\right)$ denote the
$\alpha$ quantile of the distribution of $J_{\psi}$. For given $\a$,
we compute $c_{\psi}\left(\alpha\right)$, the $\alpha$-quantile
of $J_{\psi}$ in \eqref{eq:MT-Distn=000020J_psi}, by simulating
the asymptotic distribution $J_{\psi}$. To do so, $B=300,000$ independent
constant coefficient AR(1) sequences are generated with innovations
$U_{t}\sim_{iid}N\left(0,1\right)$, stationary start-up, $n=25,000$,
and $\rho=\exp\left(-\psi/n\right)$. For each sequence, the test
statistic $T_{n}\left(\rho\right)$ defined in \eqref{eq:MT-t-stat=000020def}
is calculated. Then the simulated estimate of $c_{\psi}\left(\a\right)$
is the $\a$-quantile of the empirical distribution of the $B=300,000$
realizations of the test statistic $T_{n}\left(\rho\right)$.

The nominal $1-\alpha$ equal-tailed two-sided CI for $\rho(\tau)$
is 
\begin{align}
CI_{n,\tau} & :=\left\{ \rho_{0}\in\left[-1+\varepsilon_{1},1\right]:c_{\psi_{nh,\rho_{0}}}\left(\alpha/2\right)\leq T_{n}(\rho_{0})\leq c_{\psi_{nh,\rho_{0}}}\left(1-\alpha/2\right)\right\} ,\ \text{where}\ \nonumber \\
\psi_{nh,\rho_{0}} & :=-nh\ln\left(\rho_{0}\right)\ \text{for }\rho_{0}>0\ \text{and }\psi_{nh,\rho_{0}}:=\infty\ \text{for }\rho_{0}\leq0.\label{eq:MT-CI=000020defn}
\end{align}

The CI $CI_{n,\tau}$ can be computed by taking a fine grid of values
$\rho_{0}\in\left[-1+\varepsilon_{1},1\right]$ and comparing $T_{n}\left(\rho_{0}\right)$
to $c_{\psi_{nh,\rho_{0}}}\left(\alpha/2\right)$ and $c_{\psi_{nh,\rho_{0}}}\left(1-\alpha/2\right).$
Table \ref{tab:SM-Critical-Values_psi} provides the critical values
$c_{\psi_{nh,\rho_{0}}}\left(\alpha/2\right)$ and $c_{\psi_{nh,\rho_{0}}}\left(1-\alpha/2\right)$
for $\a=.05$, and .1. Given these critical values, computation of
equal-tailed two-sided 90\% and 95\% CI's for $\rho(\tau)$ is fast.

The correct asymptotic size and asymptotic similarity of the CI $CI_{n,\tau}$
are established in Theorem \ref{thm:MT-Cor=000020Asy=000020Size}
in Section \ref{sec:MT-Asymptotics} below.

\subsection{Median-Unbiased Interval Estimator of \texorpdfstring{$\boldsymbol{\rho(\tau)}$}{rho(tau)}\protect\label{subsec:MT-Median-Unbiased-Interval-Est}}

By definition, an estimator $\widehat{\theta}_{n}$ of a parameter
$\theta$ is median unbiased if $P\left(\widehat{\theta}_{n}\geq\theta\right)\geq1/2$
and $P\left(\widehat{\theta}_{n}\leq\theta\right)\geq1/2.$ Here we
introduce a median-unbiased interval estimator of $\rho\left(\tau\right)$
that satisfies an analogous condition. Also, with probability close
to one, the estimator is a single point.\footnote{If the .5 quantile of the asymptotic null distribution of the t-statistic
was strictly decreasing in $\rho_{0},$ or equivalently, $c_{\psi}\left(.5\right)$
was strictly increasing in $\psi$, then the proposed interval estimator
would be a point estimator with probability one. Because this condition
fails to hold exactly, but almost holds, there is a very small probability
that the estimator is a short interval, rather than a point.} Let $CI_{n,\tau}^{up}\left(.5\right)$ and $CI_{n,\tau}^{low}\left(.5\right)$
denote level $.5$ one-sided upper-bound and lower-bound CIs for $\rho\left(\tau\right)$,
respectively. By definition, 
\begin{align}
CI_{n,\tau}^{up}\left(.5\right) & :=\left\{ \rho_{0}\in\left[-1+\varepsilon_{1},1\right]:c_{\psi_{nh,\rho_{0}}}\left(.5\right)\leq T_{n}\left(\rho_{0}\right)\right\} \ \text{and}\nonumber \\
CI_{n,\tau}^{low}\left(.5\right) & :=\left\{ \rho_{0}\in\left[-1+\varepsilon_{1},1\right]:T_{n}\left(\rho_{0}\right)\leq c_{\psi_{nh,\rho_{0}}}\left(.5\right)\right\} \ \text{for}\ \psi_{nh,\rho_{0}}\ \text{as\ in}\ \text{(\ref{eq:MT-CI=000020defn})}.\label{eq:MT-1-sided=000020CIs}
\end{align}

The median-unbiased interval estimator $\widetilde{\rho}_{n\tau}$
of $\rho\left(\tau\right)$ is defined by 
\begin{align}
\widetilde{\rho}_{n\tau} & =\left[\widetilde{\rho}_{n\tau,low},\widetilde{\rho}_{n\tau,up}\right],\ \text{where}\nonumber \\
\widetilde{\rho}{}_{n\tau,up} & =\max\left\{ \rho_{0}:\rho_{0}\in CI_{n,\tau}^{up}\left(.5\right)\right\} \ \text{and}\nonumber \\
\widetilde{\rho}{}_{n\tau,low} & =\min\left\{ \rho_{0}:\rho_{0}\in CI_{n,\tau}^{low}\left(.5\right)\right\} .\label{eq:MT-MUE_Def}
\end{align}
By construction, $\widetilde{\rho}_{n\tau,low}\leq\widetilde{\rho}_{n\tau,up}.$\footnote{This holds because $\widetilde{\rho}{}_{n\tau,up}\geq\sup\left\{ \rho_{0}\in\left[-1+\varepsilon_{1},1\right]:c_{\psi_{nh,\rho_{0}}}\left(.5\right)=T_{n}\left(\rho_{0}\right)\right\} $
and $\widetilde{\rho}{}_{n\tau,low}$ is less than or equal to the
infimum of the values in the same set. } In addition, $\widetilde{\rho}_{n\tau}$ is a singleton whenever
the set $\left\{ \rho_{0}\in\left[-1+\varepsilon_{1},1\right]:T_{n}\left(\rho_{0}\right)=c_{\psi_{nh,\rho_{0}}}\left(.5\right)\right\} $
contains a single point, in which case $\widetilde{\rho}_{n\tau}$
equals this point. Simulations show that $\widetilde{\rho}_{n\tau}$
is a singleton with probability close to one and a very short interval
when it is not a singleton. Table \ref{tab:SM-Critical-Values_psi}
provides the critical values $c_{\psi_{nh,\rho_{0}}}\left(.5\right)$
for a wide range of $\psi$. Given these critical values, computation
of $\widetilde{\rho}_{n\tau}$ is fast.

The estimator $\widetilde{\rho}_{n\tau}$ has the following median-unbiasedness
property 
\begin{align}
\liminf\limits_{n\rightarrow\infty}P\left(\widetilde{\rho}_{n\tau,up}\geq\rho\left(\tau\right)\right) & \geq1/2\ \textup{and}\nonumber \\
\liminf\limits_{n\rightarrow\infty}P\left(\widetilde{\rho}_{n\tau,low}\leq\rho\left(\tau\right)\right) & \geq1/2.\label{eq:MT-1-sided=000020CIs-1-1}
\end{align}
Furthermore, as shown in Corollary \ref{cor:MT-Med-Unbiased} in Section
\ref{sec:MT-Asymptotics} below, these asymptotic properties hold
in a uniform sense over the parameter space. This property is important
for ensuring that the asymptotic properties of $\widetilde{\rho}_{n\tau}$
reflect its finite-sample properties.

\section{Data-dependent Bandwidth Parameter \texorpdfstring{$\boldsymbol{h}$}{h}\protect\label{sec:MT-Choice_of_h}}

The CI for $\rho\left(\tau\right)$ proposed in Section \ref{subsec:MT-CI=000020AR=000020Par}
requires a tuning parameter $h$ that determines the interval of observations
used to construct the CI. This section introduces a data-dependent
method of selecting $h$ that minimizes a forecast-error criterion.
We present its asymptotic properties in Section \ref{subsec:MT-Asymptotic-Results-for-hhat}
under somewhat high-level conditions.

The forecast-error criterion is the average over $t=1,...,n$ of the
squared errors from forecasting $Y_{t}$ using the predicted value
$\widehat{Y}_{t}$. Specifically, for $t>\lfloor nh\rfloor,$ let
$(\widehat{\mu}_{t-1}(h),\widehat{\rho}_{t-1}(h))$ denote the LS
estimator of $(\mu_{t},\rho_{t})$ from the regression of $Y_{s}$
on a constant and $Y_{s-1}$ using the $nh$ observations $s=t-\lfloor nh\rfloor,...,t-1.$
For $t\leq\lfloor nh\rfloor,$ since there are fewer than $\lfloor nh\rfloor$
observations, $(\mu_{t},\rho_{t})$ is estimated using the LS estimator
$(\widehat{\mu}_{t-1}(h),\widehat{\rho}_{t-1}(h))$ based on the $\lfloor nh\rfloor$
observations $s=1,...,t-1,t+1,...,\lfloor nh\rfloor+1.$ The observation
$Y_{t}$ is predicted by $\widehat{\mu}_{t-1}(h)+Y_{t-1}\widehat{\rho}_{t-1}(h)$
for $t=1,...,n$. The forecast-error criterion is 
\begin{equation}
FE_{n}(h):=n^{-1}\sum_{t=1}^{n}(Y_{t}-\widehat{\mu}_{t-1}(h)-Y_{t-1}\widehat{\rho}_{t-1}(h))^{2}.\label{eq:MT-Forecast=000020error=000020defn}
\end{equation}

We assume that $h$ is chosen from a finite set $\mathcal{H}_{n},$
which depends on $n$ and whose cardinality may depend on $n.$ By
definition, the data-dependent choice of $h,$ $\widehat{h},$ minimizes
$FE_{n}(h)$ over $\mathcal{H}_{n}$: \footnote{If the argmin is not unique, for specificity $\widehat{h},$ is taken
to be the largest argmin. But, non-uniqueness occurs with probability
zero provided the innovations have a continuous component.} 
\begin{equation}
\widehat{h}:=\arg\min_{\mathcal{H}_{n}}FE_{n}\left(h\right).\label{eq:MT-h^hat=000020defn}
\end{equation}

\begin{rem}
The criterion $FE_{n}(h)$ has the desirable property of being an
average of unbiased risk estimators for $t>\lfloor nh\rfloor.$ This
holds because 
\begin{eqnarray}
 &  & \hspace{-0.08in}\E(Y_{t}-\widehat{\mu}_{t-1}(h)-Y_{t-1}\widehat{\rho}_{t-1}(h))^{2}\nonumber \\
 & = & \hspace{-0.08in}\E(U_{t}-(\widehat{\mu}_{t-1}(h)-\mu_{t})+Y_{t-1}(\widehat{\rho}_{t-1}(h)-\rho_{t})))^{2}\nonumber \\
 & = & \hspace{-0.08in}\E U_{t}^{2}-2\E U_{t}(\widehat{\mu}_{t-1}(h)-\mu_{t}+Y_{t-1}(\widehat{\rho}_{t-1}(h)-\rho_{t}))+\E(\widehat{\mu}_{t-1}(h)-\mu_{t}+Y_{t-1}(\widehat{\rho}_{t-1}(h)-\rho_{t}))^{2}\nonumber \\
 & = & \hspace{-0.08in}\E U_{t}^{2}+\E(\widehat{\mu}_{t-1}(h)-\mu_{t}+Y_{t-1}(\widehat{\rho}_{t-1}(h)-\rho_{t}))^{2},
\end{eqnarray}
where the third equality holds for $t>\lfloor nh\rfloor$ since $(\widehat{\mu}_{t-1}(h),\widehat{\rho}_{t-1}(h))$
is a function of the innovations indexed by $s\leq t-1$ and $\E(U_{t}|U_{t-1},...,U_{1},Y_{0}^{\ast})=0.$
For the relatively small number of initial observations with $t\leq\lfloor nh\rfloor,$
the third equality does not hold because $\widehat{\rho}_{t-1}(h)$
depends on observations $s$ with $s\geq t.$

In terms of computation, the criterion $FE_{n}(h)$ has the advantage
that the same $\widehat{h}$ is employed for multiple $\lfloor nh\rfloor$
values of interest. It also avoids introducing additional tuning parameters,
such as the length of a forecasting period. 
\end{rem}
\begin{rem}
In Section \ref{subsec:MT-Asymptotic-Results-for-hhat}, we show that
the data-dependent $\widehat{h}$ value achieves the asymptotically
optimal trade-off between bias and variance obtained by the infeasible
$\widehat{h}_{opt}$ which minimizes the ``empirical loss'' as a
function of $h$ defined in \eqref{eq:MT-Empirical=000020loss=000020defn}
below. In consequence, undersmoothing $\widehat{h}$ should yield
a value of $h$ for which the bias is dominated by the variance asymptotically,
the CI's defined above have correct asymptotic size, and the median-unbiased
interval estimator is asymptotically median unbiased. We employ a
relatively sophisticated undersmoothing method. The objective of undersmoothing
(i.e., making $\widehat{h}_{us}$ smaller than $\widehat{h}$) is
to reduce the bias. The cost of undersmoothing is an increase in variance.
We want to undersmooth more when the cost of doing so in terms of
increasing the variance is relatively small, and we want to undersmooth
less when the cost is larger. This means that we need to take account
of the shape of the $FE_{n}(h)$ function. When it is flatter at its
minimum, we want to undersmooth more. This leads to the following
definition: 
\begin{equation}
\widehat{h}_{us0}=\min\left\{ h\mathcal{\in H}_{n}:FE_{n}(h)\leq Q_{c_{1}}(FE_{n})\right\} ,\label{eq:MT-hhat_us0-def}
\end{equation}
 where $Q_{c_{1}}(FE_{n})$ is the $c_{1}$ quantile of $FE_{n}$
among $h\mathcal{\in H}_{n}$ and we take $c_{1}=.2$ in the simulations
and applications. This definition does not guarantee that $\widehat{h}_{us0}$
is of a smaller order than $\widehat{h}$, which is necessary to ensure
proper asymptotics. In consequence, we modify the definition as follows:
\begin{equation}
\widehat{h}_{us}=\min\left\{ \widehat{h}_{us0},\widehat{h}_{us1}\right\} ,\label{eq:MT-hhat_us-def}
\end{equation}
where $\widehat{h}_{us1}=c_{2}n^{-a}\widehat{h}$ for $c_{2},a>0.$
In the simulations and applications, we use $c_{2}=1.5$ and $a=1/10.$
\end{rem}

\section{Monte Carlo Simulations\protect\label{sec:MT-Monte-Carlo-Simulations}}

In this section, we analyze the finite-sample performance of the methods
introduced above using Monte Carlo simulations. First, we describe
the data generating processes (DGP's) considered and how the CI and
MUE are implemented. Then, we report the results. We find that the
proposed CI has reasonably good coverage probabilities and short average
lengths for most of the DGP's considered. We also find that the MUE
performs quite well both when the true DGP of $\rho_{t}$ is time-varying
and flat.

\subsection{Simulation Setup and Methodology}

We consider 21 DGP's for $\rho_{t}$. The graphs of the $\rho_{t}$
functions are given in Figures \ref{fig:MT-Sim_CP_AL_MAD_1}--\ref{fig:MT-Sim_CP_AL_MAD_3}
and Figures \ref{fig:Sim_CP_AL_MAD_SM1}--\ref{fig:Sim_CP_AL_MAD_SM4},
which appear in the Supplemental Material. There are five categories
of $\rho_{t}$ functions: sinusoidal, linear, flat linear, flat, and
kinked linear. For the sinusoidal functions, we consider sinusoidal
functions for $\rho_{t}$ with 6 different ranges and shapes. For
example, ``sin 1.00-0.90-1.00'' corresponds to the sinusoidal function
where $\rho_{t}$ first achieves its maximum of 1.00 at $t/n=.20$,
then drops to 0.90 at $t/n=.60$ and finally increases to 1.00 at
$t/n=1$, with a frequency of 2.5$\pi$. For linear functions, we
consider 4 different linear DGP's for $\rho_{t}$ corresponding to
different values of the slopes and intercepts. For flat linear functions,
the first half of $\rho_{t}$ is flat and the second half is linear,
and 4 different cases are considered. For example, ``flat-lin 0.90-0.99''
means 
\[
\rho_{t}=\rho(t/n)=\begin{cases}
.90 & \text{for }t/n\in\left[0,1/2\right],\\
.18t/n+.81 & \text{for }t/n\in(1/2,1].
\end{cases}
\]
For the flat functions, we consider 3 values .75, .90, and .99. In
the Supplemental Material, we also report results for kinked linear
functions, where the first and second halves of $\rho_{t}$ are linear,
but with slopes of different signs. For example, ``kinked 1.00-0.80-1.00''
is a DGP for $\rho_{t}$ where the first half of $\rho_{t}$ decreases
linearly from 1 to .80 while the second half increases linearly from
.80 to 1. We consider 4 different cases in this category.

For each of the 21 DGP's for $\rho_{t}$, we consider both constant
and time-varying DGP's for $\mu_{t}$ and $\sigma_{t}$ in \eqref{eq:MT-tvp-model}
except for that of ``sin 1.00-0.90-1.00'' with constant $\mu_{t}$
and $\s_{t}$ to present the legend for all the figures. Define $\mu_{t}^{*}=(1-\rho_{t})\mu_{t}$,
which allows one to rewrite \eqref{eq:MT-tvp-model} as 
\begin{equation}
Y_{t}=\mu_{t}^{*}+\rho_{t}Y_{t-1}+\sigma_{t}U_{t},\text{ for }t=1,2,...,n.\label{eq:MT-Sim_Yt_DGP}
\end{equation}
When $\mu_{t}$ and $\sigma_{t}$ are constant, we take $\mu_{t}=0$
and $\sigma_{t}=1$, which implies $\mu_{t}^{*}=0$. When $\mu_{t}$
and $\sigma_{t}$ are time-varying, we generate $\mu_{t}$ such that
$\mu_{t}^{*}$ increases linearly from -.1 to .1, and $\sigma_{t}$
increases linearly from .95 to 1.05 as $t/n$ goes from 0 to 1. In
consequence, there are a total of 41 ($=21\times2-1$) DGP's for $\left(\rho_{t},\mu_{t}^{*},\s_{t}\right)$.
We take $U_{t}$ to be i.i.d. $N\left(0,1\right)$, initialize $Y_{0}$
by drawing it from a $N\left(0,\ol{\sigma}_{n}^{2}\right)$ distribution,
where $\ol{\sigma}_{n}=1/\left(1-\ol{\rho}_{n}^{2}\right)$ and $\ol{\rho}_{n}=n^{-1}\sum_{t=1}^{n}\rho_{t}$,
and form the $Y_{t}$ sequence as in \eqref{eq:MT-Sim_Yt_DGP}.

We compute a nominal .95 two-sided CI and MUE for $\rho(\tau)$ for
each of 5 time points of interest, indexed by $\tau\in\left\{ .2,.4,.6,.8,1\right\} $.
The CI and MUE are implemented as follows. First, we compute $\widehat{h}$
based on the method described in Section \ref{sec:MT-Choice_of_h}
and take the undersmoothed $\widehat{h}_{us}$ to be as defined in
(\ref{eq:MT-hhat_us-def}) with $c_{1}=.2,$ $c_{2}=1.5,$ and $a=1/10.$
In this step, the set of $nh$ values based on $\mathcal{H}_{n}$
is $\left\{ 140,155,...,500,650,...,1500\right\} $. Next, for each
candidate value $\rho_{0}\in\left\{ -1,-.995,...,.945,.95,.951,.952,...,1\right\} $,
we calculate the t-statistics $T_{n}\left(\rho_{0}\right)$ defined
in \eqref{eq:MT-t-stat=000020def} from the regression of $Y_{t}$
on a constant and $Y_{t-1}$ using the $n\widehat{h}_{us}$ observations
centered around $n\tau$. When $n\tau$ is close to the boundary,
we use $n\widehat{h}_{us}/2$ observations to the side that has abundant
data, and as many observations as available to the other side until
the boundary is hit. In consequence, a different number of observations
for different $\tau$ may be used even though the $n\widehat{h}_{us}$
value is the same. Comparing the $T_{n}\left(\rho_{0}\right)$ values
with corresponding critical values $c_{\psi_{nh,\rho_{0}}}\left(\alpha/2\right)$
and $c_{\psi_{nh,\rho_{0}}}\left(1-\alpha/2\right)$ (for $\psi_{nh,\rho_{0}}$
defined in \eqref{eq:MT-CI=000020defn}) at each candidate $\rho_{0}$
gives the nominal $1-\alpha$ equal-tailed two-sided CI for $\rho(\tau)$
as defined in \eqref{eq:MT-CI=000020defn}. We compute coverage probabilities
and average lengths of the CI's across all simulations for each point
of interest for each DGP. 

To calculate the MUE for $\rho(\tau)$, we follow the procedure described
in Section \ref{subsec:MT-Median-Unbiased-Interval-Est} and use $\widetilde{\rho}_{n\tau}$
as the MUE of $\rho\left(\tau\right)$ when $\widetilde{\rho}_{n\tau}$
is a singleton. When $\widetilde{\rho}{}_{n\tau,up}\neq\widetilde{\rho}{}_{n\tau,low}$,
we take $\widetilde{\rho}{}_{n\tau,up}$ as the MUE for $\rho\left(\tau\right)$
based on two considerations. First, using the upper bound $\widetilde{\rho}{}_{n\tau,up}$
has the theoretical property that it is not median-biased towards
zero for positive $\rho\left(\tau\right)$ values because, by Corollary
\ref{cor:MT-Med-Unbiased}, $\liminf_{n\rightarrow\infty}\inf_{\lambda\in\Lambda_{n}}P_{\lambda}\left(\widetilde{\rho}_{n\tau,up}\geq\rho\left(\tau\right)\right)\geq1/2$.
Second, in the simulations we find that $\widetilde{\rho}_{n\tau}$
is much more likely to be an interval when the true $\rho\left(\tau\right)$
is close to one. We report the absolute median biases of the MUE and
the range of the MUE values across the simulations for each time point
of interest for each DGP in the figures.

For each of the 41 DGP's for $Y_{t}$, we run $M=5,000$ simulations.
The length of the $Y_{t}$ sequence is $n=1,500$. Given the 41 DGP's
and 5 time points of interest for each one, a total of 205 cases are
considered.

\subsection{Discussion of Results}

First, Table \ref{tab:MT-Range_of_CP} summarizes the CP results by
reporting the number and percentage of times that the CI CP's lie
in different ranges across the 205 cases considered. The nominal CP
is .95. Table \ref{tab:MT-Range_of_CP} shows that only in a small
fraction of cases, 2.0\%, is the CP quite low, i.e., in $[.875,.90).$
In the majority of cases, 92.2\%, the CP is $.925$ or larger, which
shows that the proposed method has reasonably good CP regardless of
whether the underlying DGP is curvy or flat and whether the point
of interest is in the middle or at the boundary of the time period.
\begin{table}
\begin{centering}
\caption{\protect\label{tab:MT-Range_of_CP}Distribution of CP's Across 205
Cases}
\par\end{centering}
\centering{}%
\begin{tabular}{cccccc}
\toprule 
CP Range & {[}.875,.90) & {[}.90,.925) & {[}.925,.94) & {[}.94,.96) & {[}.96,.965{]}\tabularnewline
\midrule 
Number & 4 & 12 & 34 & 148 & 7\tabularnewline
Percent & 2.0\% & 5.9\% & 16.6\% & 72.2\% & 3.4\%\tabularnewline
\bottomrule
\end{tabular}
\end{table}

Second, we consider summary statistics for the absolute values of
the median biases (AMB) of the MUE, i.e., $\abs{\text{median}\left(\widetilde{\rho}_{n\tau,up}-\rho\left(\tau\right)\right)}$,
across the 205 cases considered. The mean and median of the absolute
median biases across all cases are .004 and .003, respectively. The
range is $[.000,.023]$. Thus, the magnitudes of the absolute median
biases of the MUE are generally small.

Third, we consider the values of $n\widehat{h}_{us}$ selected by
the data-dependent method. The mean, median, and range are 212, 207,
and {[}56, 433{]}, respectively. The $n\widehat{h}_{us}$ values vary
depending on the shape of the $\rho$ function considered. The flatter
is the true $\rho$ function, the larger are the $n\widehat{h}_{us}$
values.

Now, we describe the simulation results for the different $\rho$
functions considered. Figure \ref{fig:MT-Sim_CP_AL_MAD_1} provides
simulation results for five sinusoidal $\rho$ functions with constant
$\mu$ and $\sigma^{2}$ functions. The amplitudes of the sinusoidal
functions increase as one moves down the rows. In the first column
of graphs, the $\rho$ function's maximum value occurs at $\tau=1,$
which corresponds to $t=n\tau=n.$ In the second column of graphs,
the $\rho$ function's minimum value occurs at $\tau=1.$ Each graph
consists of an ``upper'' and a ``lower'' graph. The ``upper''
graph shows the true $\rho$ function, the average CI lower and upper
bounds at five points of interest $\tau=.2,$ $.4,$ $.6,$ $.8,$
$1.0,$ where $t=n\tau,$ and twice the average lower and upper absolute
deviations of the MUE (i.e., $2\E|\widehat{\rho}_{t}-\rho_{t}|\ind(\widehat{\rho}_{t}<\rho_{t})$
and $2\E|\widehat{\rho}_{t}-\rho_{t}|\ind(\widehat{\rho}_{t}>\rho_{t}),$
respectively, whose sum is somewhat analogous to two standard deviations)
at each of these points of interest. The difference between the average
upper and lower CI bounds gives the average lengths (AL's) of the
CI's. The ``lower'' graphs report the CI coverage probabilities
(CP's) at the five points of interest. The CI's all have nominal CP
.95.

The sinusoidal graphs in Figure \ref{fig:MT-Sim_CP_AL_MAD_1} show
the following: (i) All but three of the CP's range from .925 to .965,
which are close to the nominal CP of .95. (ii) The AL's are shortest
for $\rho(\tau)$ values closest to one, and longest for $\rho(\tau)$
values farthest from one. This reflects the $nh$-consistency of the
LS estimator in the (temporally local) unit root and local-to-unity
scenarios, and the $(nh)^{1/2}$-consistency of the LS estimator in
the (temporally local) stationary scenarios. (iii) The AL's are larger
for the $\tau=1$ than the $\tau<1$ cases because fewer observations
are used to construct the CI in the former cases, which reflects the
need to reduce the boundary bias in order to obtain a good CP. Combining
this fact with point (ii), we find the largest AL's occur when $\tau=1$
{\parfillskip=0pt\par}
\begin{figure}[H]
\captionsetup[subfigure]{position=top,font=scriptsize,singlelinecheck=off,justification=raggedright}\vspace{-2.5em}

\begin{centering}
\subfloat[sin 1.00-0.90-1.00, constant $\mu$ and $\sigma$]{\noindent\includegraphics[scale=0.36,viewport=0bp 0bp 576bp 576bp]{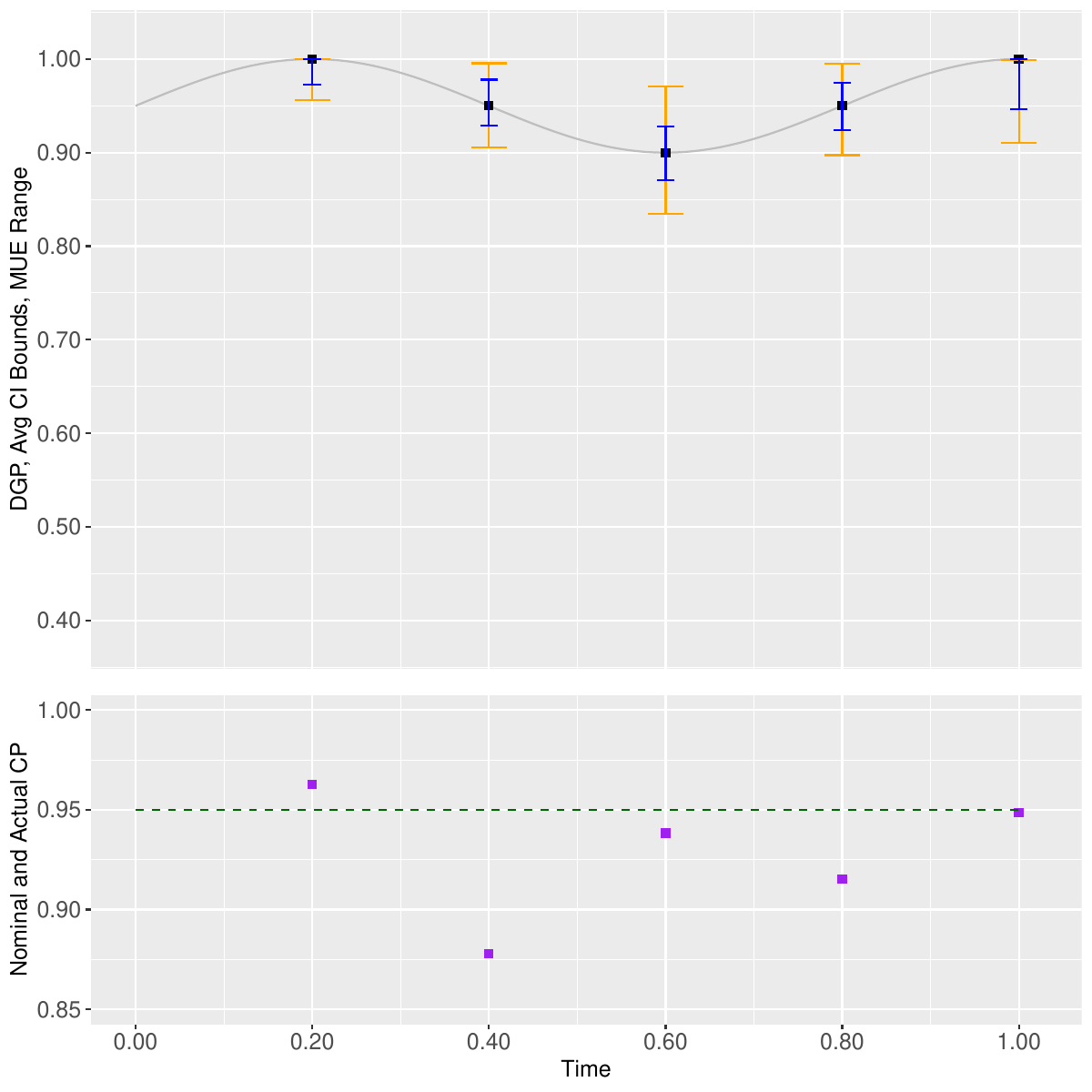}

}\quad{}\subfloat[sin 0.90-1.00-0.90, constant $\mu$ and $\sigma$]{\includegraphics[scale=0.36]{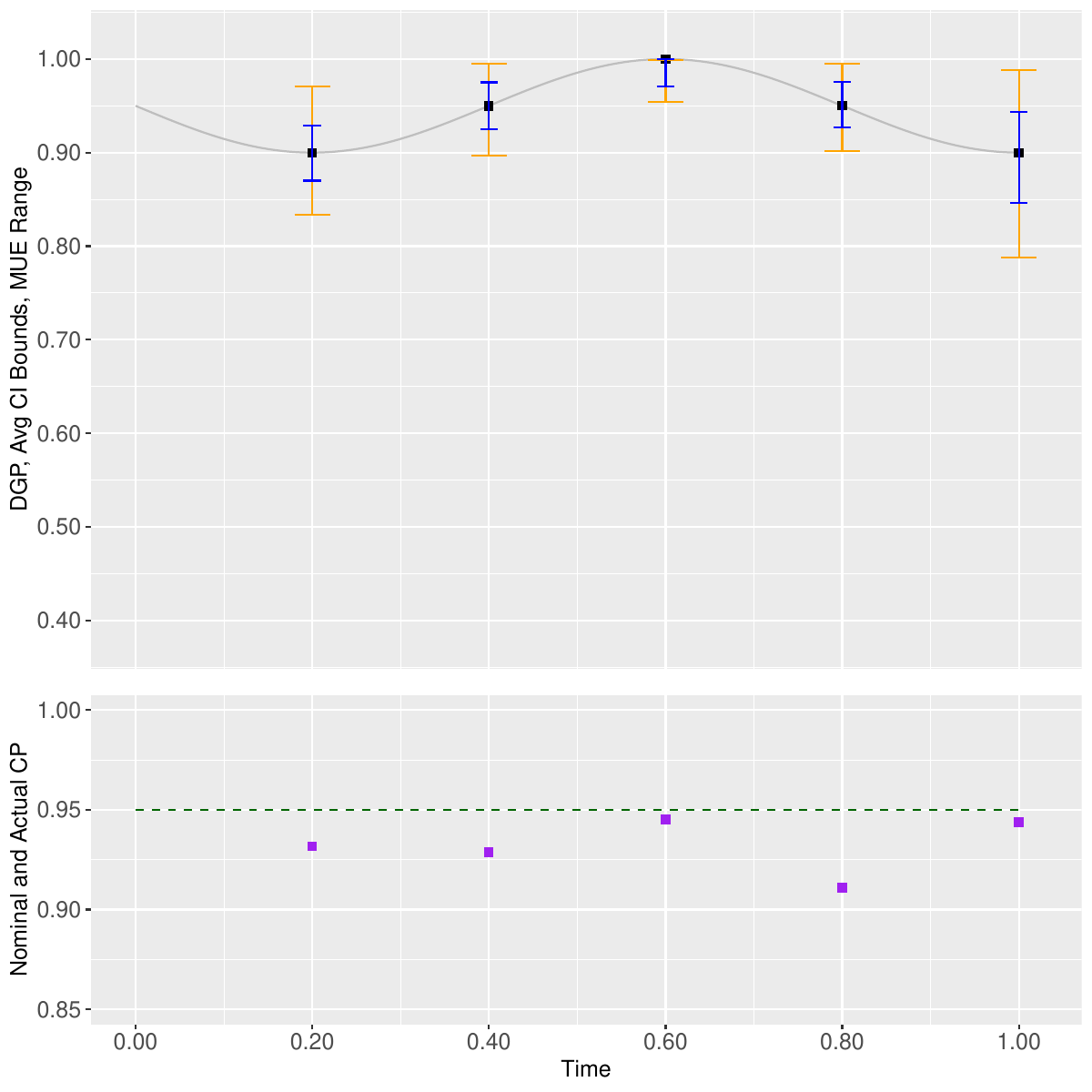}

}
\par\end{centering}
\vspace{-0.65em}

\begin{centering}
\subfloat[sin 1.00-0.80-1.00, constant $\mu$ and $\sigma$]{\includegraphics[scale=0.36]{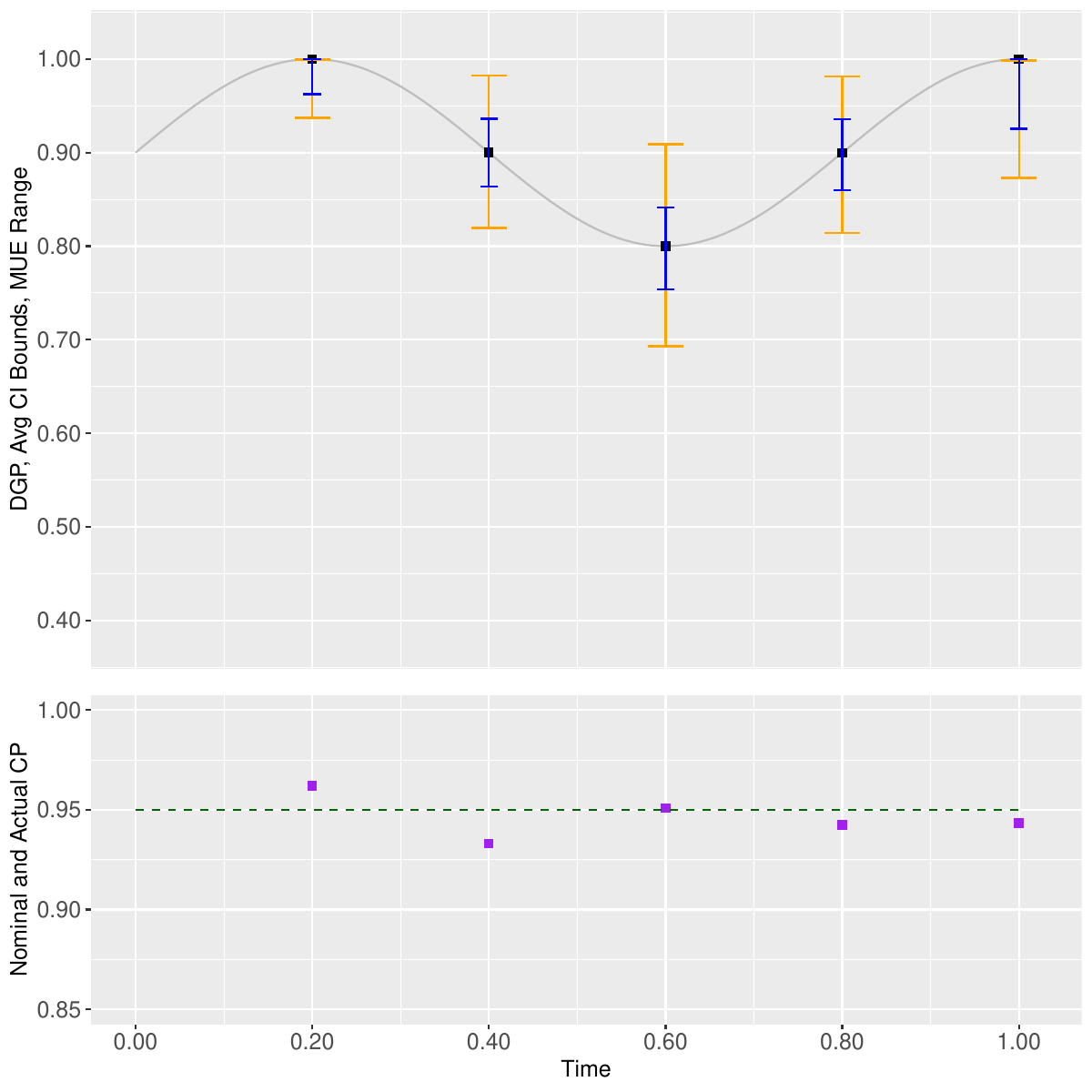}

}\quad{}\subfloat[sin 0.80-1.00-0.80, constant $\mu$ and $\sigma$]{\includegraphics[scale=0.36]{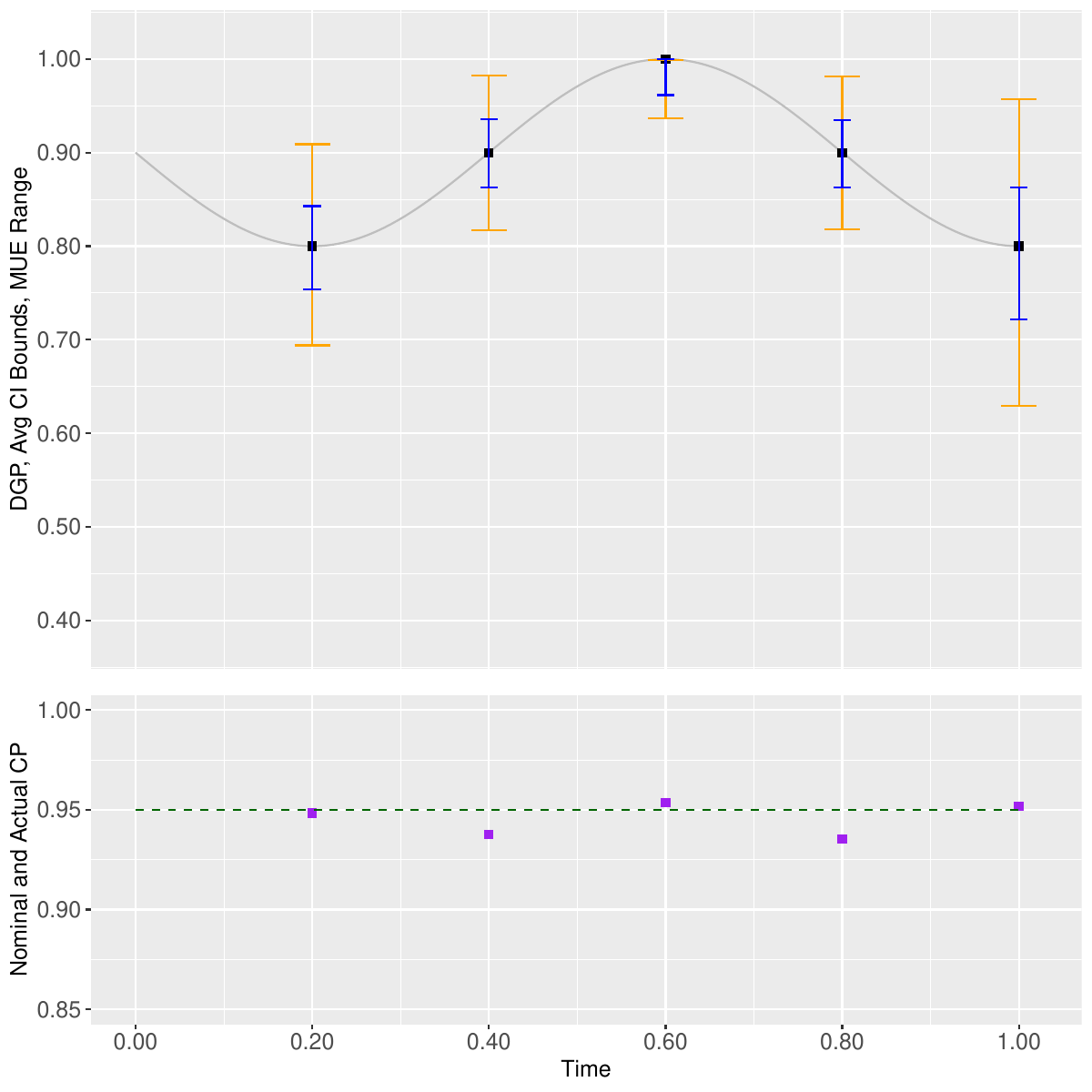}

}
\par\end{centering}
\vspace{-0.65em}

\begin{centering}
\subfloat[sin 1.00-0.60-1.00, constant $\mu$ and $\sigma$]{\includegraphics[scale=0.36]{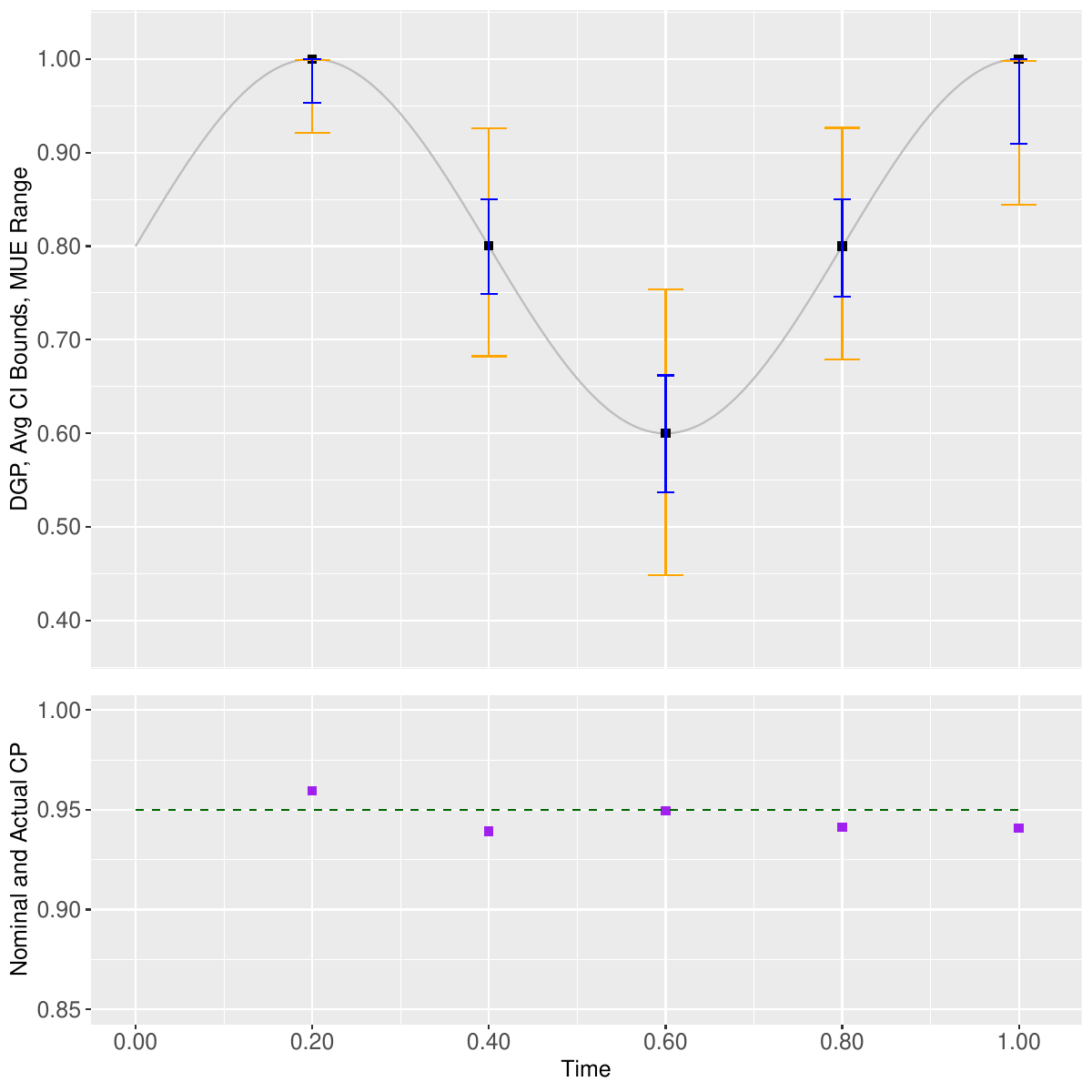}

}\quad{}\captionsetup[subfloat]{position=top,font=scriptsize,singlelinecheck=off,justification=raggedright,labelformat=empty}\subfloat[Legend for the Figures]{\includegraphics[scale=0.36]{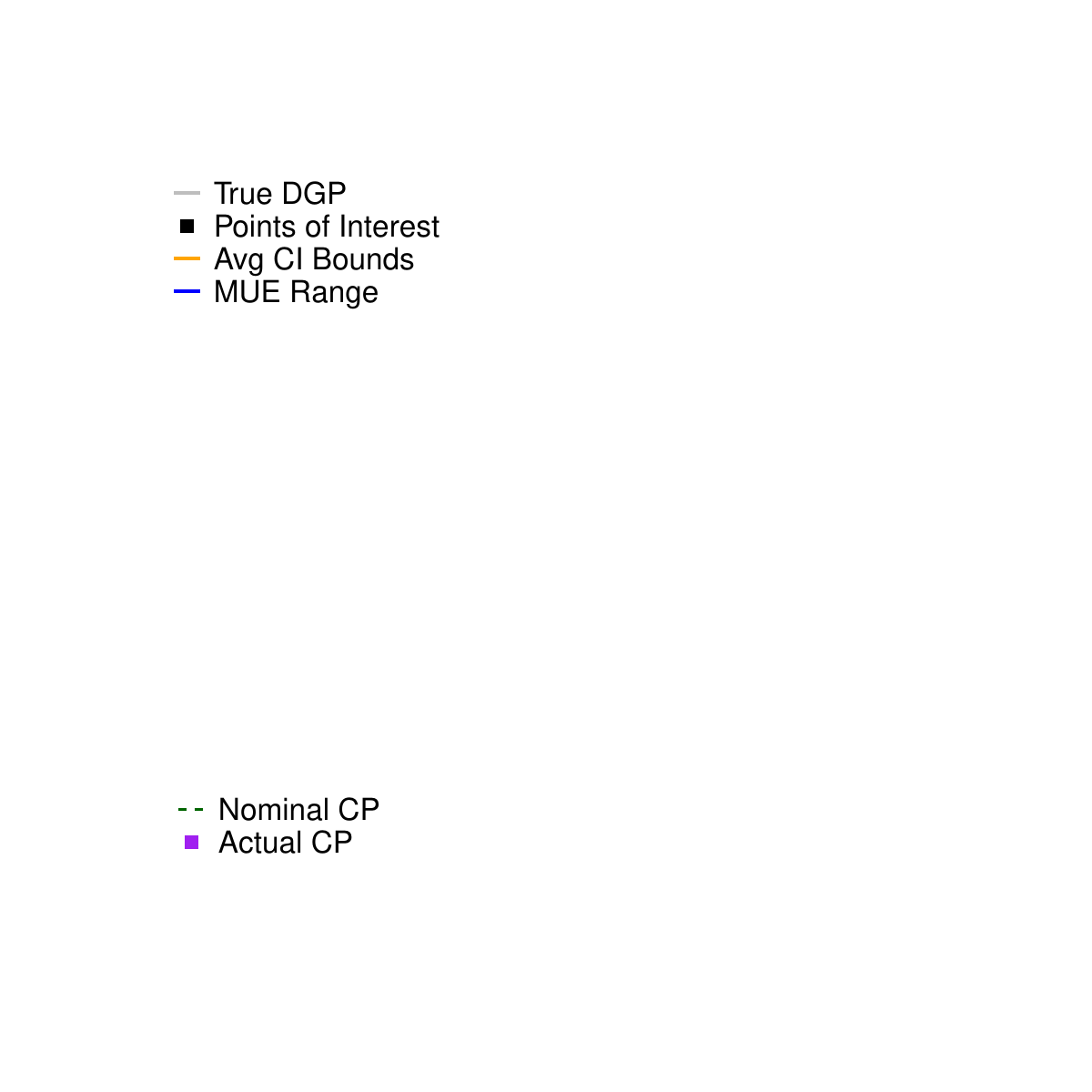}

}
\par\end{centering}
\vspace{-0.75em}

\caption{\protect\label{fig:MT-Sim_CP_AL_MAD_1}CP's and AL's of CI's for $\rho\left(\tau\right)$
and MAD's of the MUE of $\rho\left(\tau\right)$}
\end{figure}
{\noindent}and $\rho(\tau)$ is far from one (which occurs in the
two graphs in the second column). (iv) The magnitudes of the MUE average
absolute deviations are roughly proportional to the AL's of the CI's.
(v) The simulation results for the sinusoidal functions with time-varying
$\mu$ and $\sigma^{2}$ functions, which are given in Figure \ref{fig:Sim_CP_AL_MAD_SM1}
in the Supplemental Material, are quite similar to those in Figure
\ref{fig:MT-Sim_CP_AL_MAD_1}.

Figure \ref{fig:MT-Sim_CP_AL_MAD_2}(a)-(d) provides results for linear
$\rho$ functions with constant $\mu$ and $\sigma^{2}$ functions.
The results are similar to those for the sinusoidal functions, although
all of the CP's lie between .924 and .957. Points (ii)-(v) above also
apply to the linear functions in Figure \ref{fig:MT-Sim_CP_AL_MAD_2}.
Results for time-varying $\mu$ and $\sigma^{2}$ are provided in
Figure \ref{fig:Sim_CP_AL_MAD_SM2} in the Supplemental Material,
and are similar to those in Figure \ref{fig:MT-Sim_CP_AL_MAD_2}.

Figures \ref{fig:MT-Sim_CP_AL_MAD_2}(e)-(f) and \ref{fig:MT-Sim_CP_AL_MAD_3}(a)-(b)
report results for flat-linear $\rho$ functions with $\rho(t)$ varying
between $.90$ and $.99$ in Figure \ref{fig:MT-Sim_CP_AL_MAD_2}(e)-(f)
and between $.80$ and $.99$ in Figure \ref{fig:MT-Sim_CP_AL_MAD_3}(a)-(b).
Figures \ref{fig:MT-Sim_CP_AL_MAD_2}(e) and \ref{fig:MT-Sim_CP_AL_MAD_3}(a)
report results for constant $\mu$ and $\sigma^{2};$ while Figures
\ref{fig:MT-Sim_CP_AL_MAD_2}(f) and \ref{fig:MT-Sim_CP_AL_MAD_3}(b)
report results for time-varying $\mu$ and $\sigma^{2}.$ The CP results
in Figures \ref{fig:MT-Sim_CP_AL_MAD_2}(e)-(f) and \ref{fig:MT-Sim_CP_AL_MAD_3}(a)-(b)
lie between .932 and .956.

Figures \ref{fig:Sim_CP_AL_MAD_SM2}(e)-(f) and \ref{fig:Sim_CP_AL_MAD_SM3}(a)-(b)
provide analogous results to those just discussed for flat-linear
$\rho$ functions, but for the case where the flat part has value
$.99,$ rather than $.80$ or $.90,$ and the linear part has a negative
slope, rather than a positive slope. The CP results for these cases
are lower than those for the flat-linear $\rho$ functions in Figures
\ref{fig:MT-Sim_CP_AL_MAD_2} and \ref{fig:MT-Sim_CP_AL_MAD_3}. For
example, Figure \ref{fig:Sim_CP_AL_MAD_SM2}(e) shows a CP of .905
at $\tau=.60.$ For the other four cases, the CP's are in the range
of .917 to .957.

Figure \ref{fig:MT-Sim_CP_AL_MAD_3}(c)-(f) considers flat $\rho$
functions. In Figure \ref{fig:MT-Sim_CP_AL_MAD_3}(c)-(d), $\rho=.99$
and the $\mu$ and $\sigma^{2}$ functions are constant and time-varying,
respectively. Figure \ref{fig:MT-Sim_CP_AL_MAD_3}(e)-(f) is analogous,
but with $\rho=.90.$ Figure \ref{fig:Sim_CP_AL_MAD_SM3}(c)-(d) also
is analogous, but with $\rho=.75.$ The CP's in the flat $\rho$ function
cases are all close to .95 and lie between .932 and .960. This occurs
because there is no bias accruing due to a time-varying $\rho$ function.
However, in the cases of flat $\rho,$ $\mu,$ and $\sigma^{2}$ functions,
the CI's are not as short as the oracle CI's that rely on knowledge
that the $\rho,$ $\mu,$ and $\sigma^{2}$ functions are all flat
and take $nh=n.$ For the case of $\rho=.99$ and flat $\mu$ and
$\sigma^{2}$ functions, the AL's of the TVP CI and the oracle CI
are $.022$ and $.012,$ respectively, at $\tau=.2$. For $\rho=.90,$
they are $.085$ and $.039$ at the same $\tau$. For $\rho=.75$,
they are $.124$ and $.062$. There is a substantial increase in the
average lengths in the flat cases using the methods of this paper,
in order to ensure good CP's in the near flat cases.

Lastly, Figures \ref{fig:Sim_CP_AL_MAD_SM3}(e)-(f) and \ref{fig:Sim_CP_AL_MAD_SM4}(a)-(f)
in the Supplemental Material provide results for kinked $\rho$ functions
that either linearly increase until $\tau=.5$ and then linearly decrease,
or linearly decrease until $\tau=.5$ and then linearly increase.
Both constant and time-varying $\mu$ and $\sigma^{2}$ functions
are considered. In two cases, the CP's are .904 and .906. In all other
cases, {\parfillskip=0pt\par}
\begin{figure}[H]
\captionsetup[subfigure]{position=top,font=scriptsize,singlelinecheck=off,justification=raggedright}\vspace{-2.5em}

\begin{centering}
\subfloat[linear 0.90-1.00, constant $\mu$ and $\sigma$]{\includegraphics[scale=0.36]{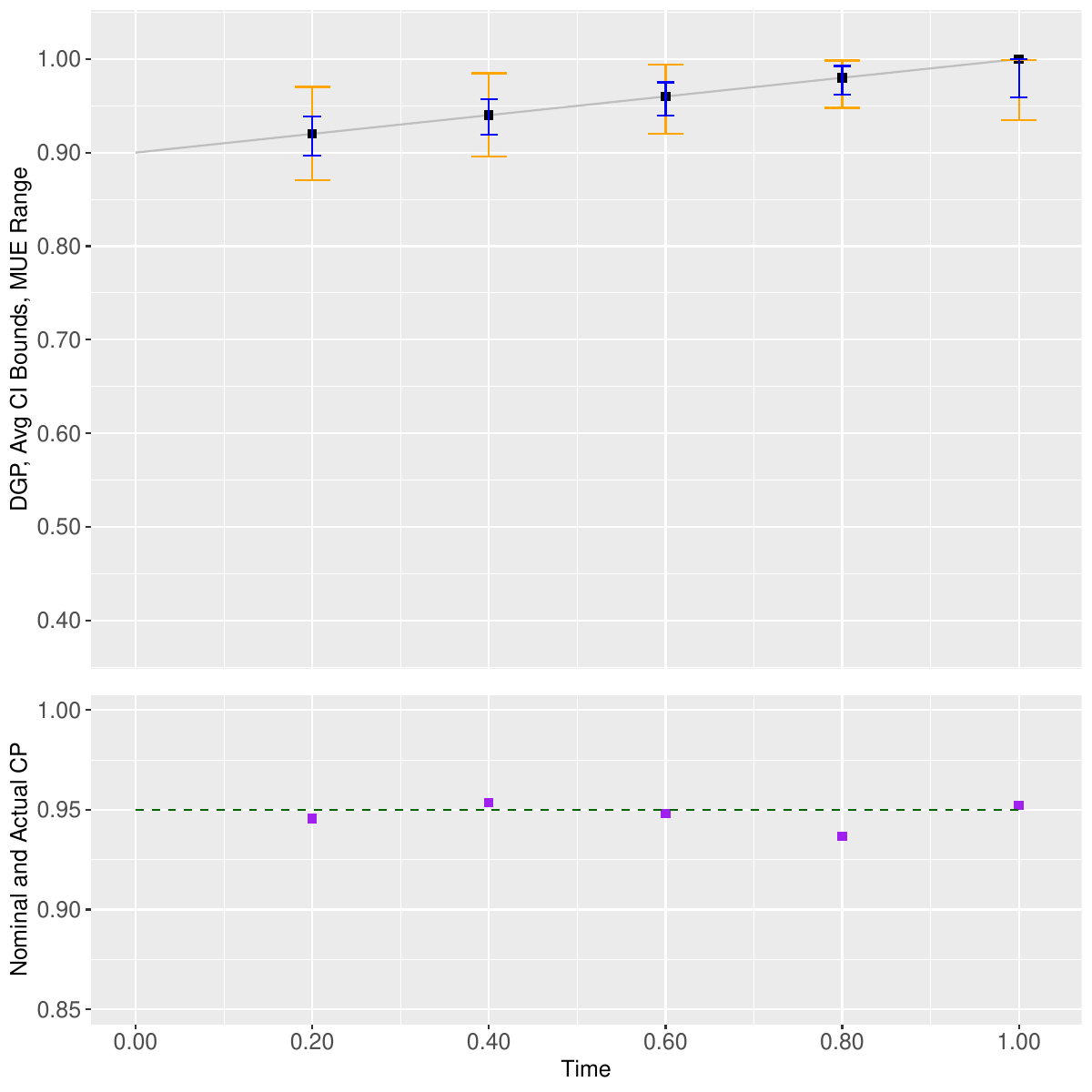}

}\quad{}\subfloat[linear 1.00-0.90, constant $\mu$ and $\sigma$]{\includegraphics[scale=0.36]{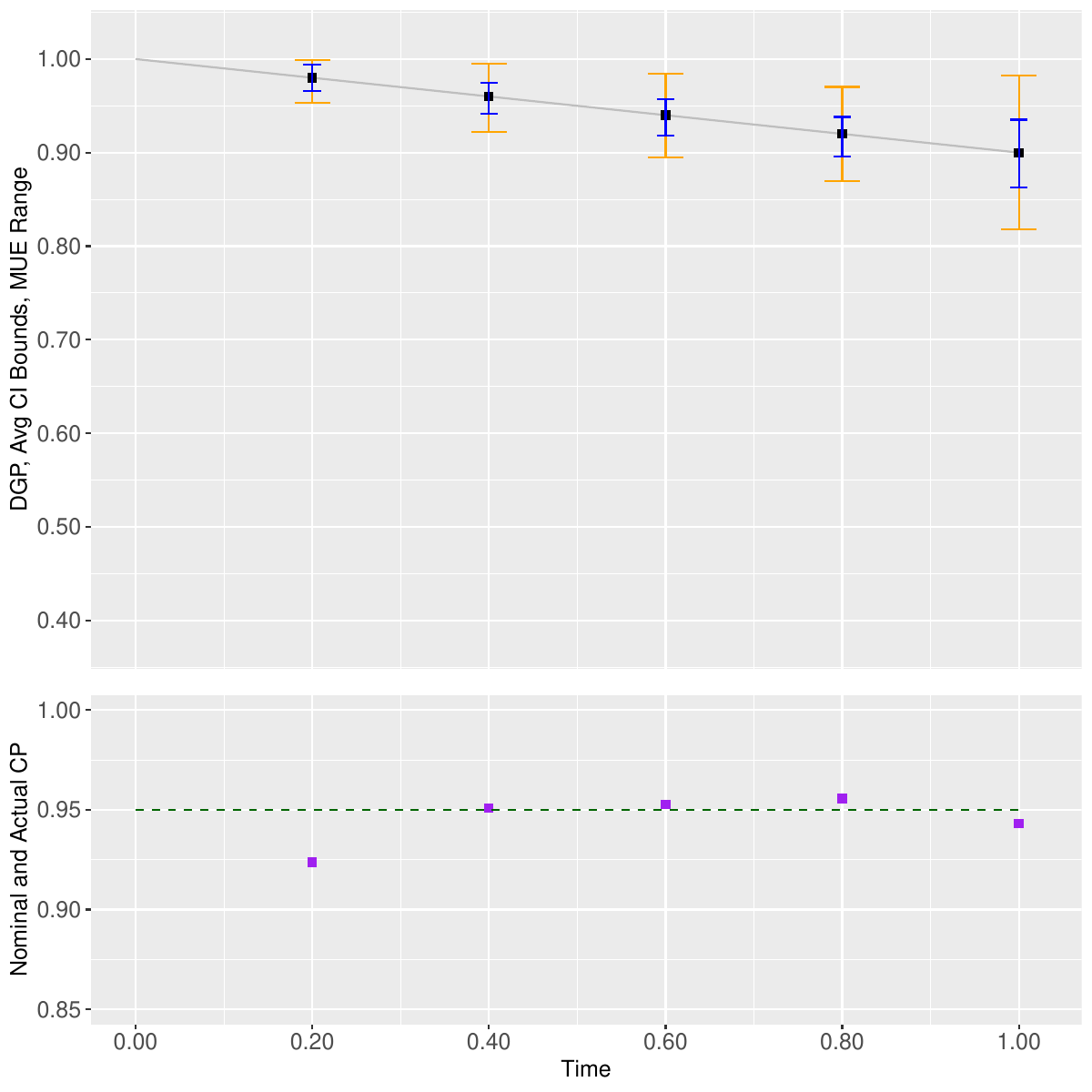}

}
\par\end{centering}
\vspace{-0.65em}

\begin{centering}
\subfloat[linear 0.60-0.90, constant $\mu$ and $\sigma$]{\includegraphics[scale=0.36]{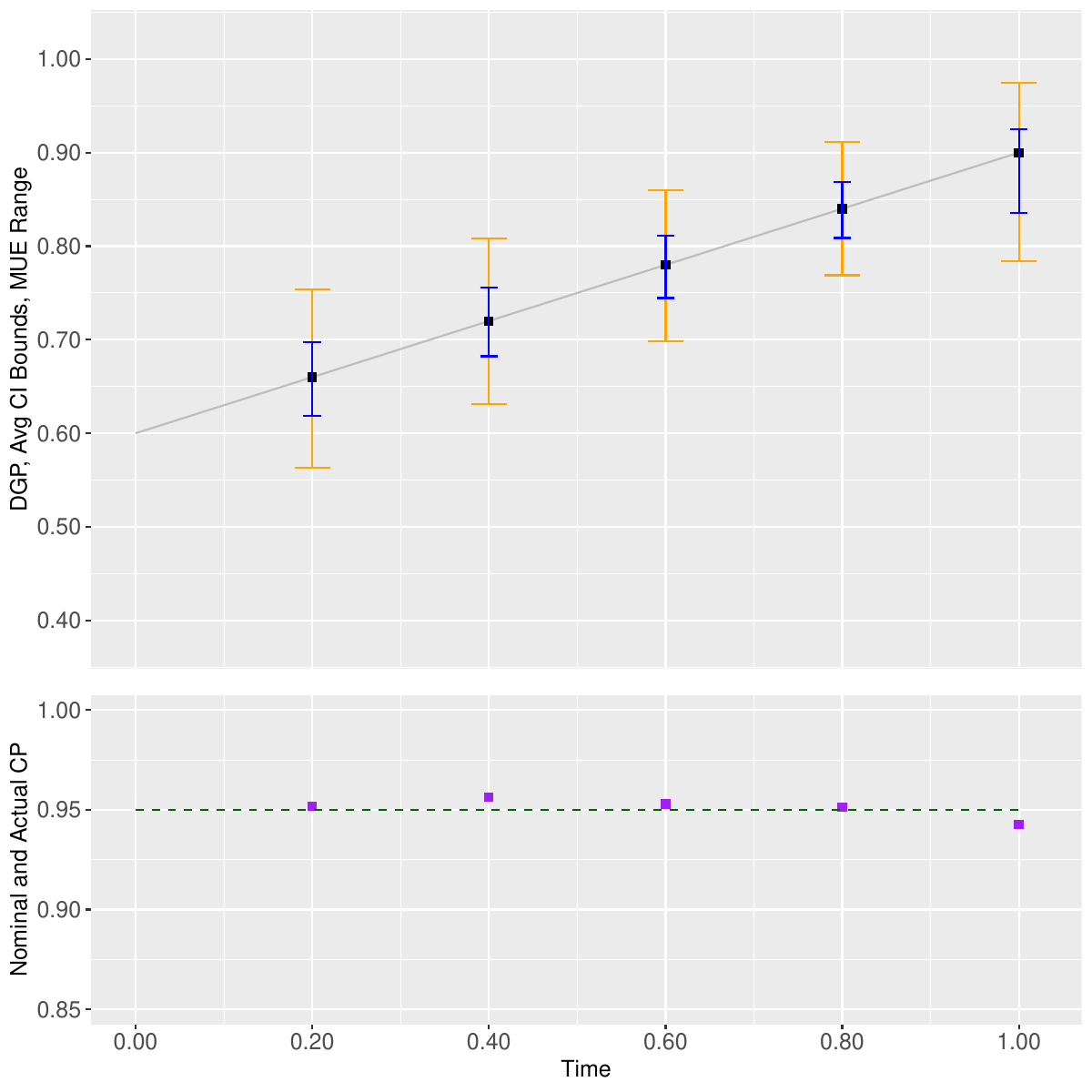}

}\quad{}\subfloat[linear 0.90-0.60, constant $\mu$ and $\sigma$]{\includegraphics[scale=0.36]{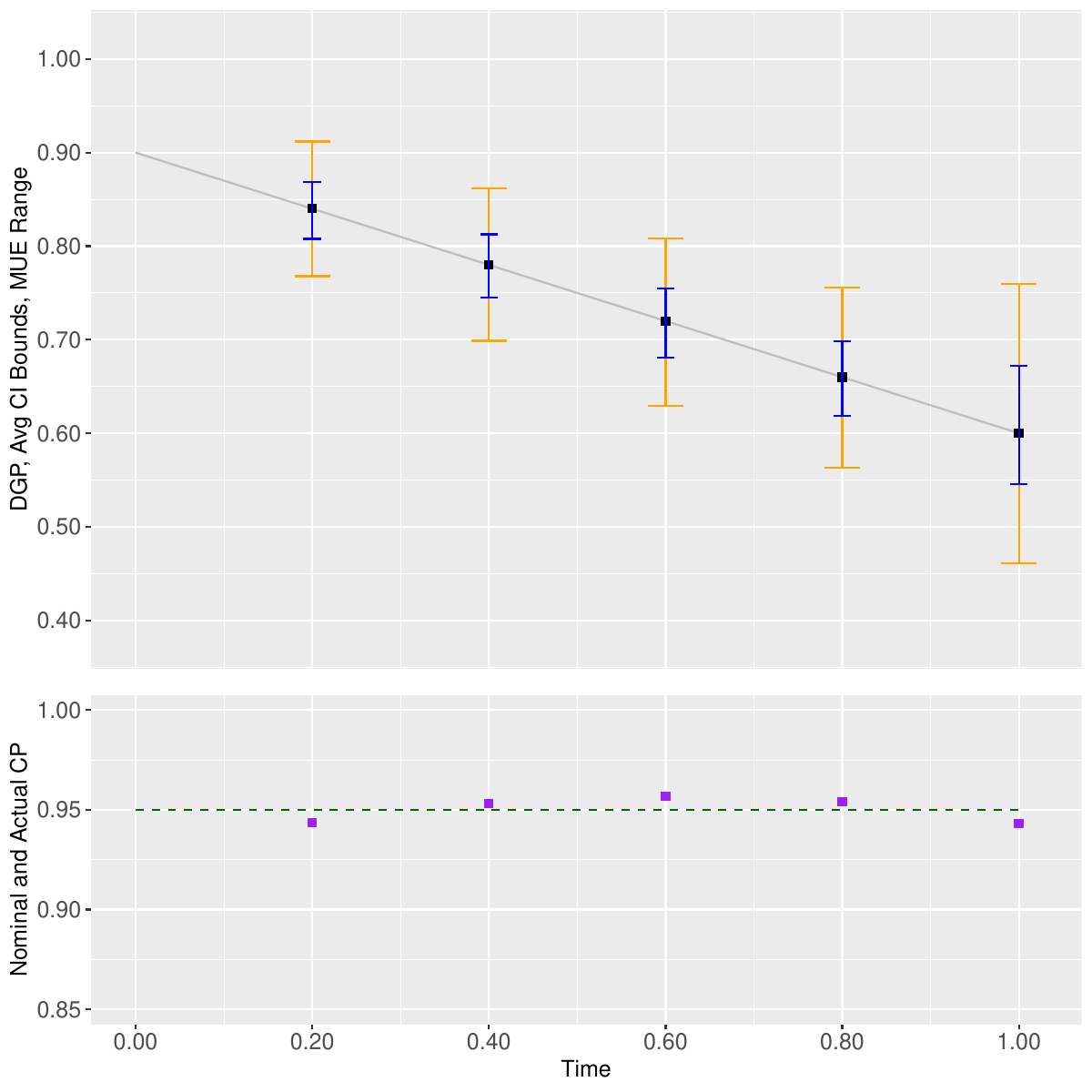}

}
\par\end{centering}
\vspace{-0.65em}

\begin{centering}
\subfloat[flat-lin 0.90-0.99, constant $\mu$ and $\sigma$]{\includegraphics[scale=0.36]{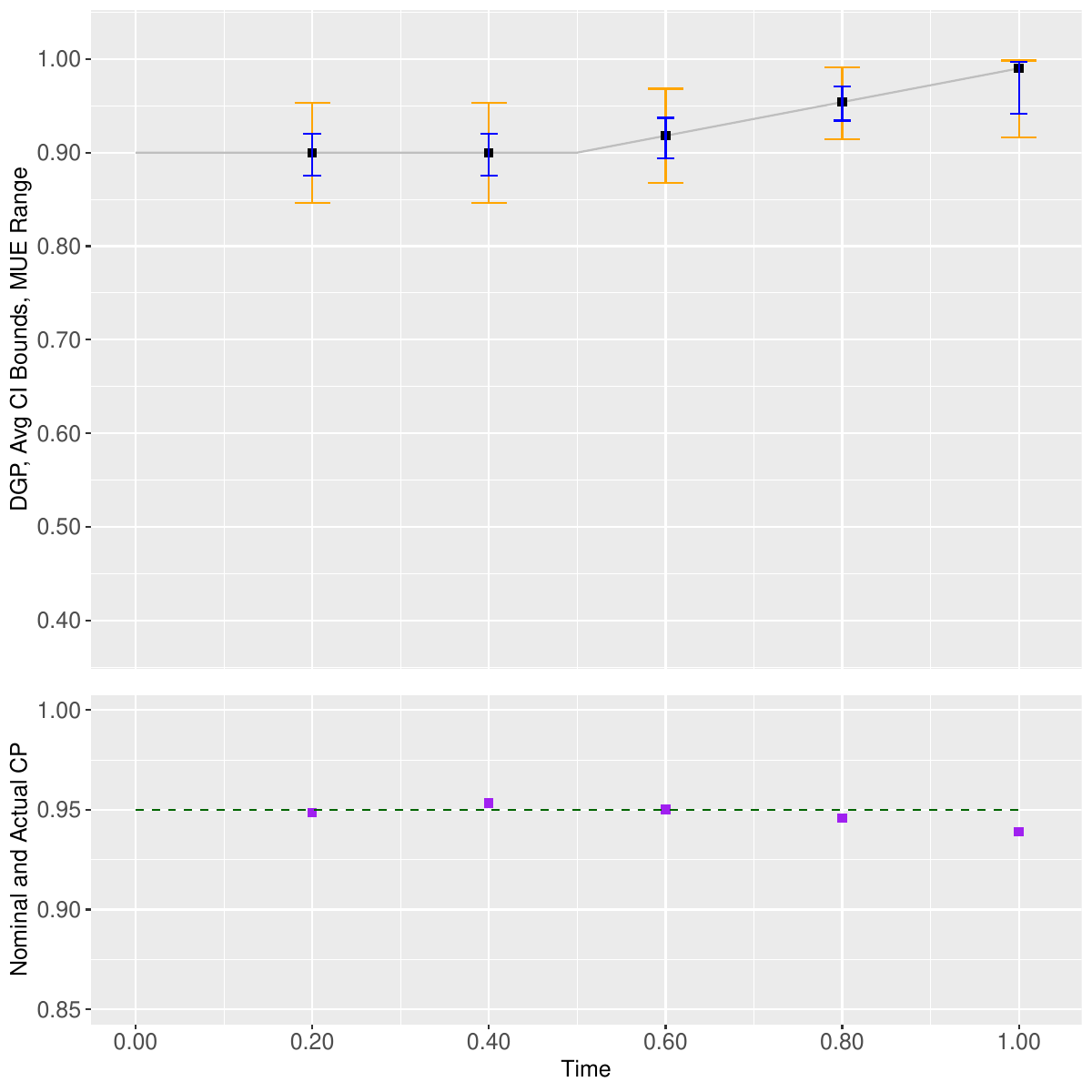}

}\quad{}\subfloat[flat-lin 0.90-0.99, time-varying $\mu$ and $\sigma$]{\includegraphics[scale=0.36]{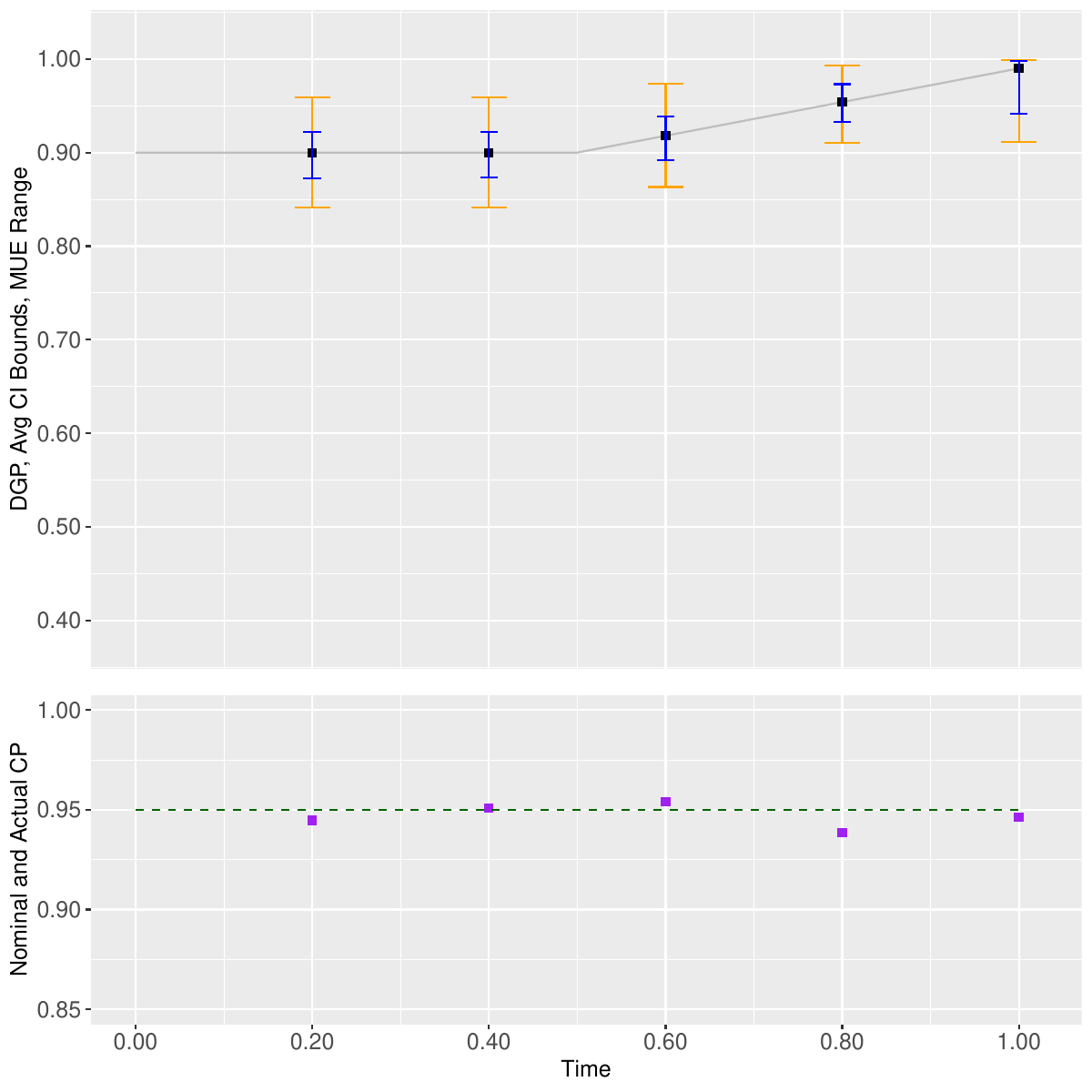}

}
\par\end{centering}
\vspace{-0.75em}

\caption{\protect\label{fig:MT-Sim_CP_AL_MAD_2}CP's and AL's of CI's for $\rho\left(\tau\right)$
and MAD's of the MUE of $\rho\left(\tau\right)$}
\end{figure}
\begin{figure}[H]
\captionsetup[subfigure]{position=top,font=scriptsize,singlelinecheck=off,justification=raggedright}\vspace{-2.5em}

\begin{centering}
\subfloat[flat-lin 0.80-0.99, constant $\mu$ and $\sigma$]{\includegraphics[scale=0.36]{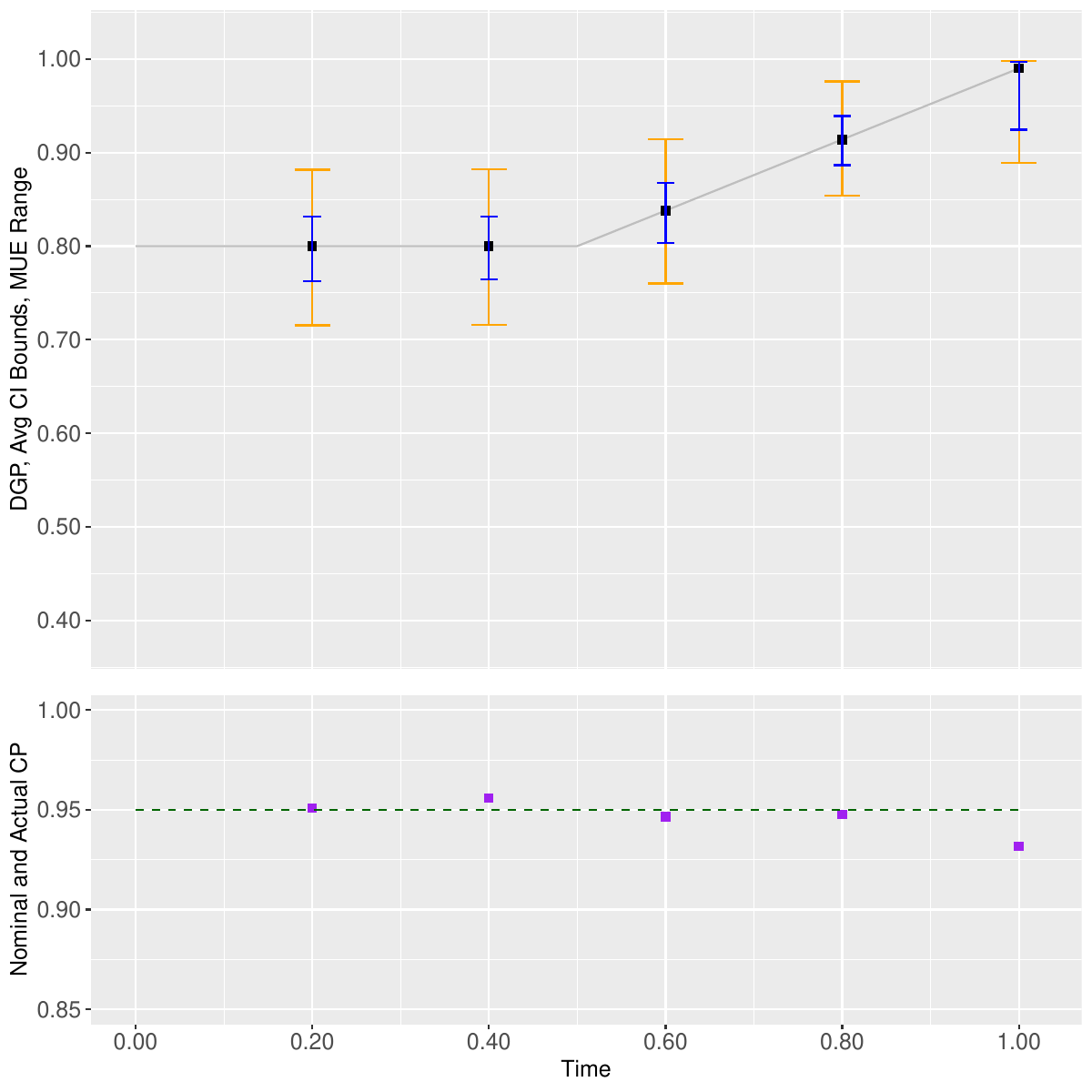}

}\quad{}\subfloat[flat-lin 0.80-0.99, time-varying $\mu$ and $\sigma$]{\includegraphics[scale=0.36]{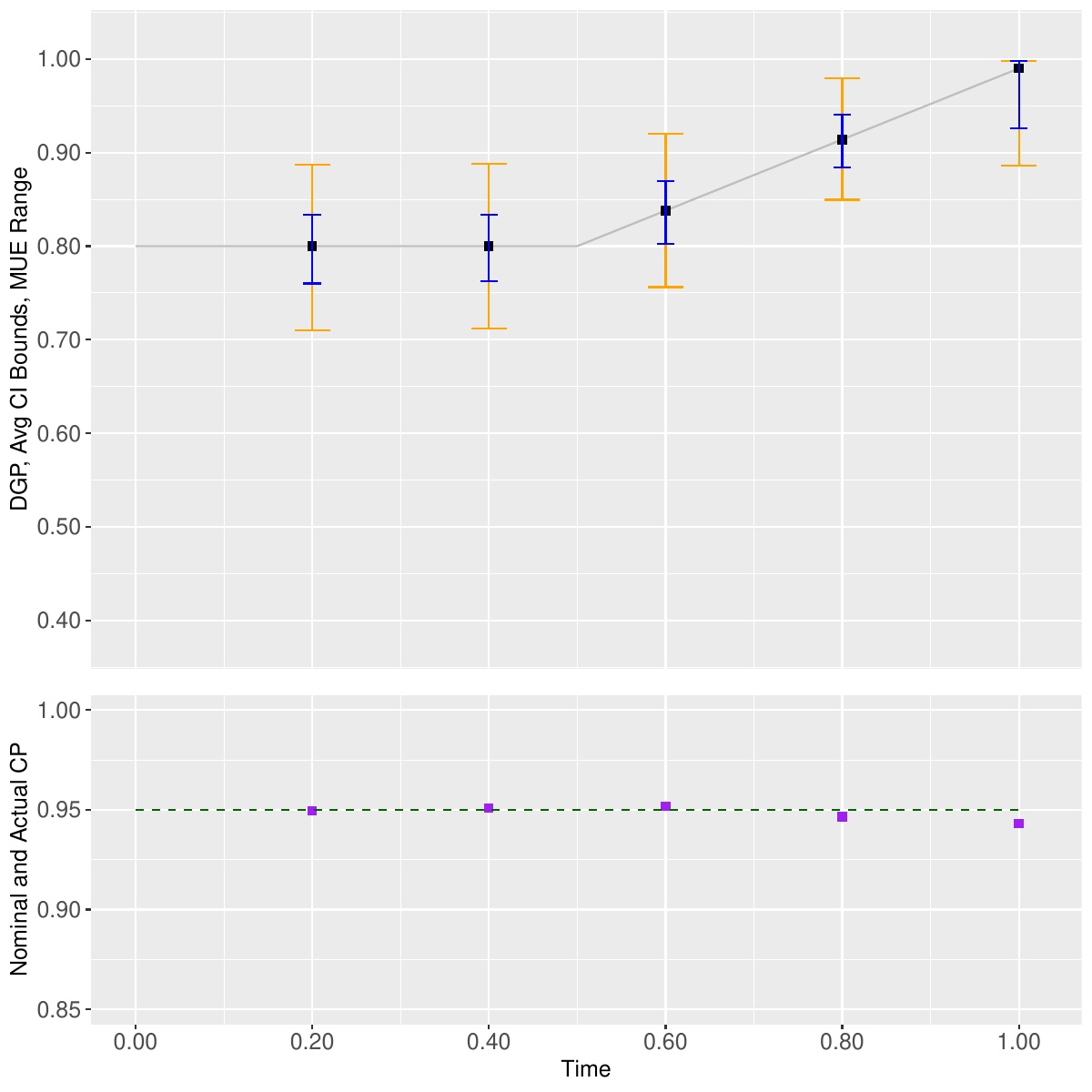}

}
\par\end{centering}
\vspace{-0.65em}

\begin{centering}
\subfloat[flat 0.99, constant $\mu$ and $\sigma$]{\includegraphics[scale=0.36]{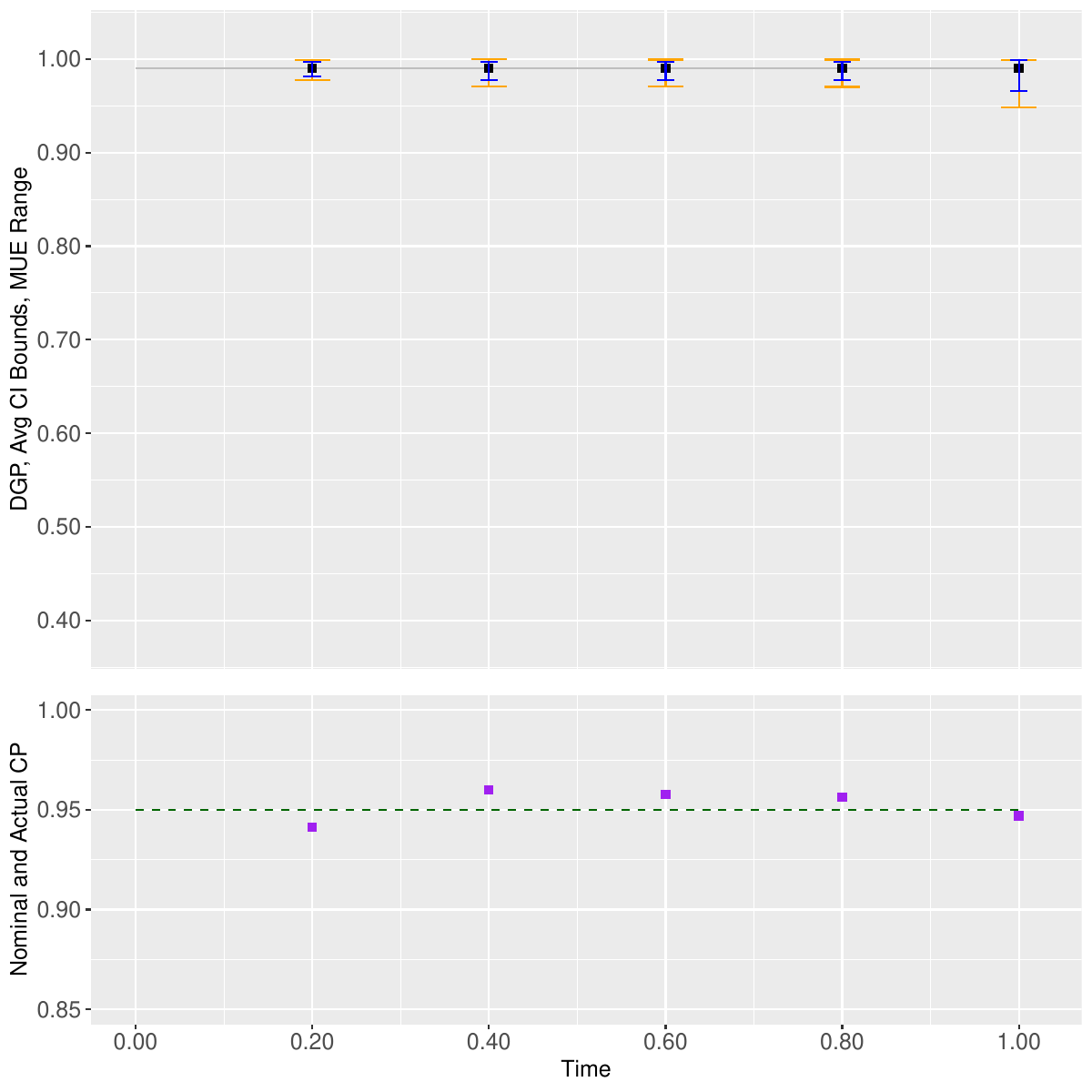}

}\quad{}\subfloat[flat 0.99, time-varying $\mu$ and $\sigma$]{\includegraphics[scale=0.36]{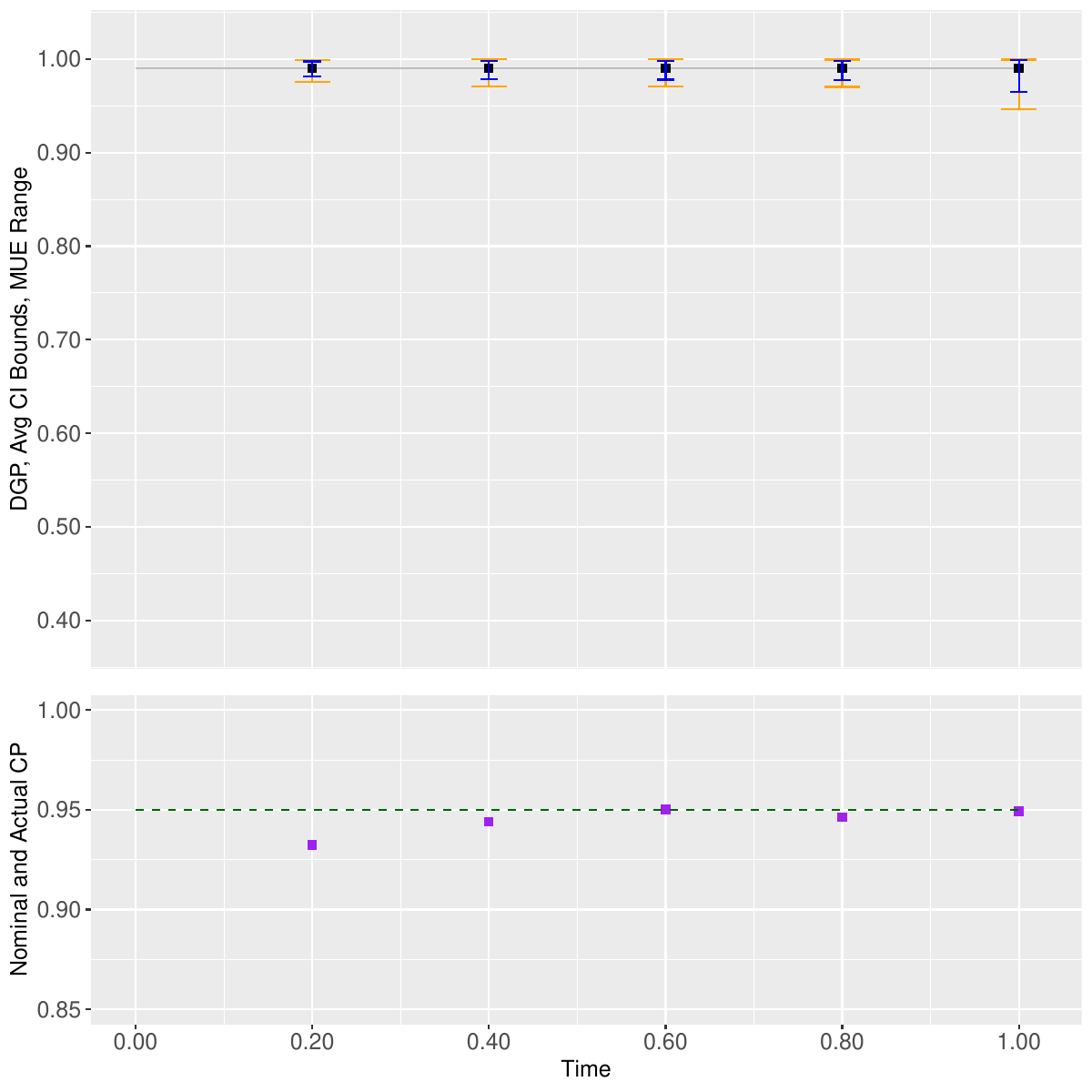}

}
\par\end{centering}
\vspace{-0.65em}

\begin{centering}
\subfloat[flat 0.90, constant $\mu$ and $\sigma$]{\includegraphics[scale=0.36]{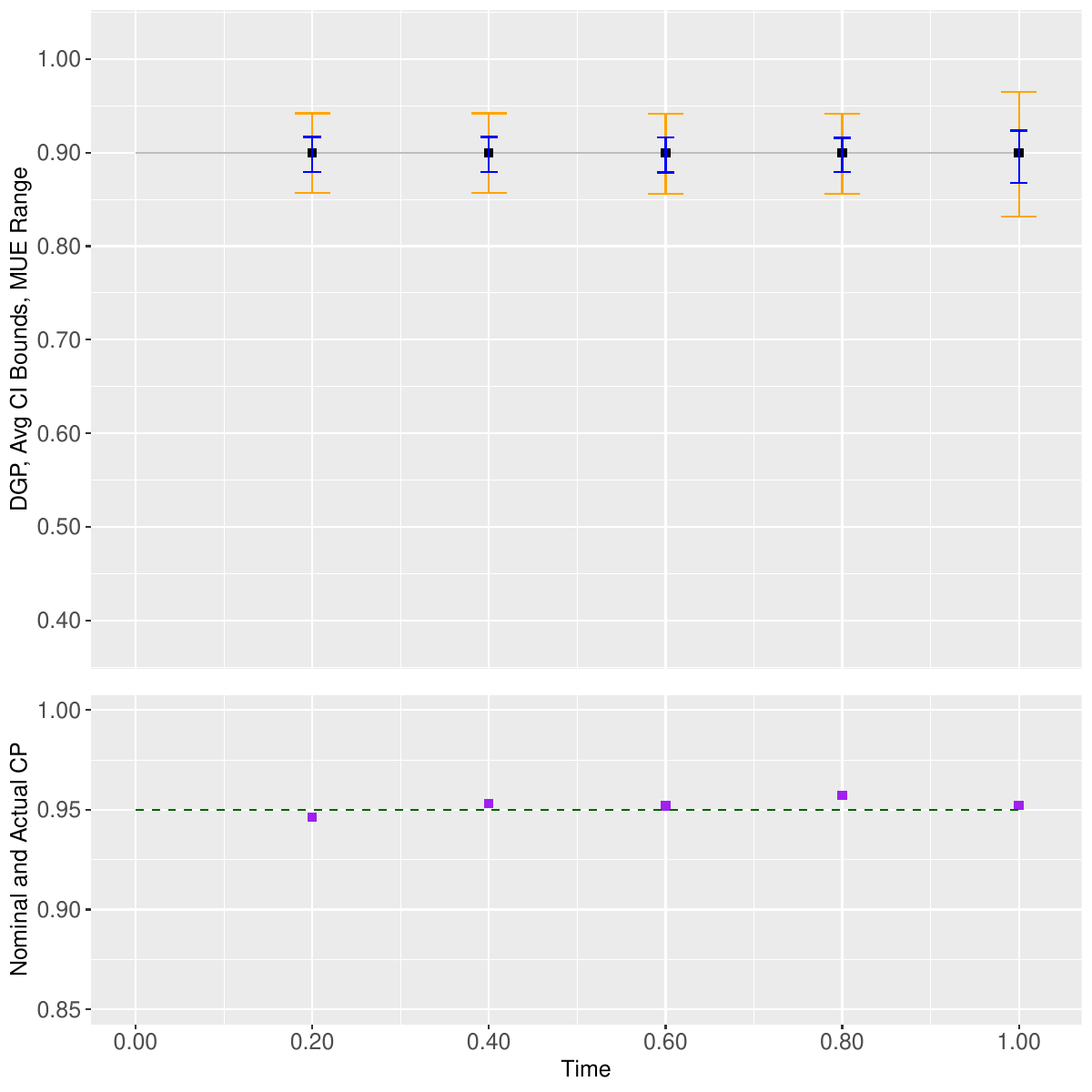}

}\quad{}\subfloat[flat 0.90, time-varying $\mu$ and $\sigma$]{\includegraphics[scale=0.36]{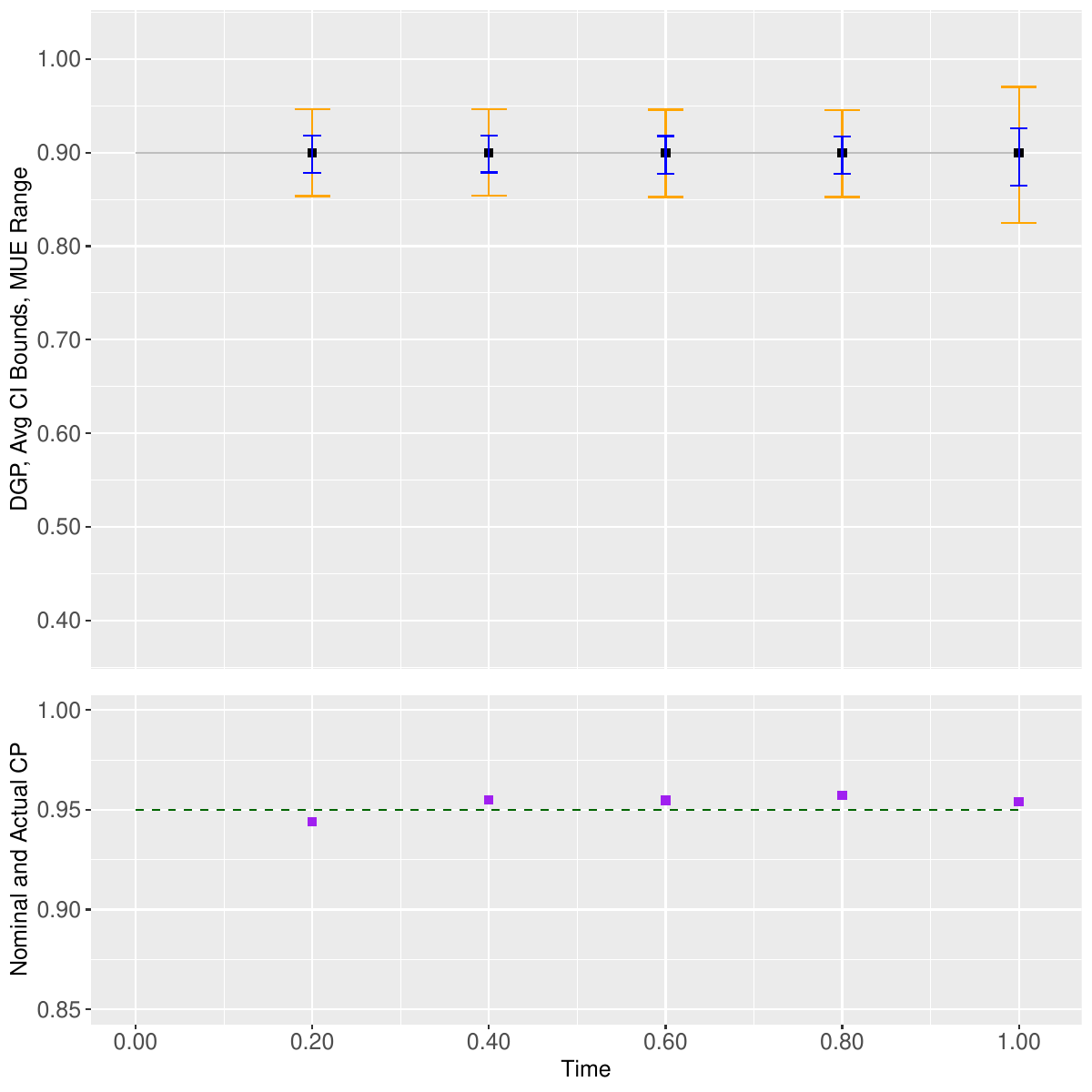}

}
\par\end{centering}
\vspace{-0.75em}

\caption{\protect\label{fig:MT-Sim_CP_AL_MAD_3}CP's and AL's of CI's for $\rho\left(\tau\right)$
and MAD's of the MUE of $\rho\left(\tau\right)$}
\end{figure}
{\noindent}the CP's are in the range of .930 to .953. When the difference
between the minimum and maximum $\rho$ value considered is large,
viz., $.4,$ the AL's of the CI's are large. This occurs because the
data-dependent choice of $h$ is small in order to avoid bias.

Overall, the CP's of the CI's are quite good with only a few being
as low as .88 or .89 out of the 205 cases considered. The AL's vary
substantially across scenarios depending on how close $\rho(\tau)$
is to one and how much the $\rho$ function varies with time, as is
to be expected.

\section{Empirical Applications\protect\label{sec:MT-Empirical-Applications}}

This section presents applications of the proposed methods to time
series of inflation and exchange rates in several countries. The data
comes from the IMF (International Financial Statistics (IFS)) database. 

In the Supplemental Material, results are provided for some additional
countries and for interest rates. The Supplemental Material also provides
applications to eight macroeconomic series for the US, using the Federal
Reserve Economic Data (FRED).

In terms of computation, we set $nh_{\min}=.2n$ and $nh_{\max}=2n$,
where $n$ is the sample size. For $nh$ values between $nh_{\min}$
and $nh_{\text{mid}}$, where $nh_{\text{mid}}\coloneqq.5n$, we use
a grid size of $.02n$, while for $nh$ values between $nh_{\text{mid}}$
and $nh_{\max}$, we use a grid size of $.05n$ because the range
of $nh$ values is wider than between $nh_{\min}$ and $nh_{\text{mid}}$.

\subsection{Inflation\protect\label{subsec:MT-EMP_Inflation}}

First, we apply our method to the monthly inflation data, defined
as the percentage change in CPI over the previous month, see Figure
\ref{fig:MT-EMP_AR1_1}. We consider three countries: the US, Canada,
and Germany. The data span is Feb 1955 to Oct 2022, which contains
$n=813$ observations for each country. We compute the $n\widehat{h}_{us}$
values based on (\ref{eq:MT-hhat_us-def}). Given these, we compute
the MUE's and 90\% CI's for $\rho\left(t\right)$ at each $t$ and
for each country. As a robustness check, we multiply the undersmoothed
$n\widehat{h}_{us}$ by 1.5 and present the results in the right column
of each figure. The final $n\widehat{h}_{us}$ values used for the
computations are listed in the titles of the figures. For all of the
inflation series, Ljung-Box tests based on six lags of the residuals
from the AR(1) model fail to reject the null hypothesis of no autocorrelation
at the 5\% significance level, see Table \ref{tab:ACF_Test_AR1}.

For comparative purposes, we also fit a constant autoregression coefficient
AR(1) model to the data and report the estimated $\widehat{\rho}$
and its 90\% CI in the same graph. To obtain the constant parameter
MUE's denoted by the flat red solid lines and constant parameter CI's
denoted by the flat red dotted lines, we fix $n\widehat{h}_{us}=2n$
and apply our method. Note that the 
\begin{figure}[H]
\captionsetup[subfigure]{position=top,font=scriptsize,singlelinecheck=off,justification=raggedright}\vspace{-1.2cm}

\begin{centering}
\subfloat[US Inflation, $n\widehat{h}_{us}=$ 125]{\includegraphics[scale=0.36]{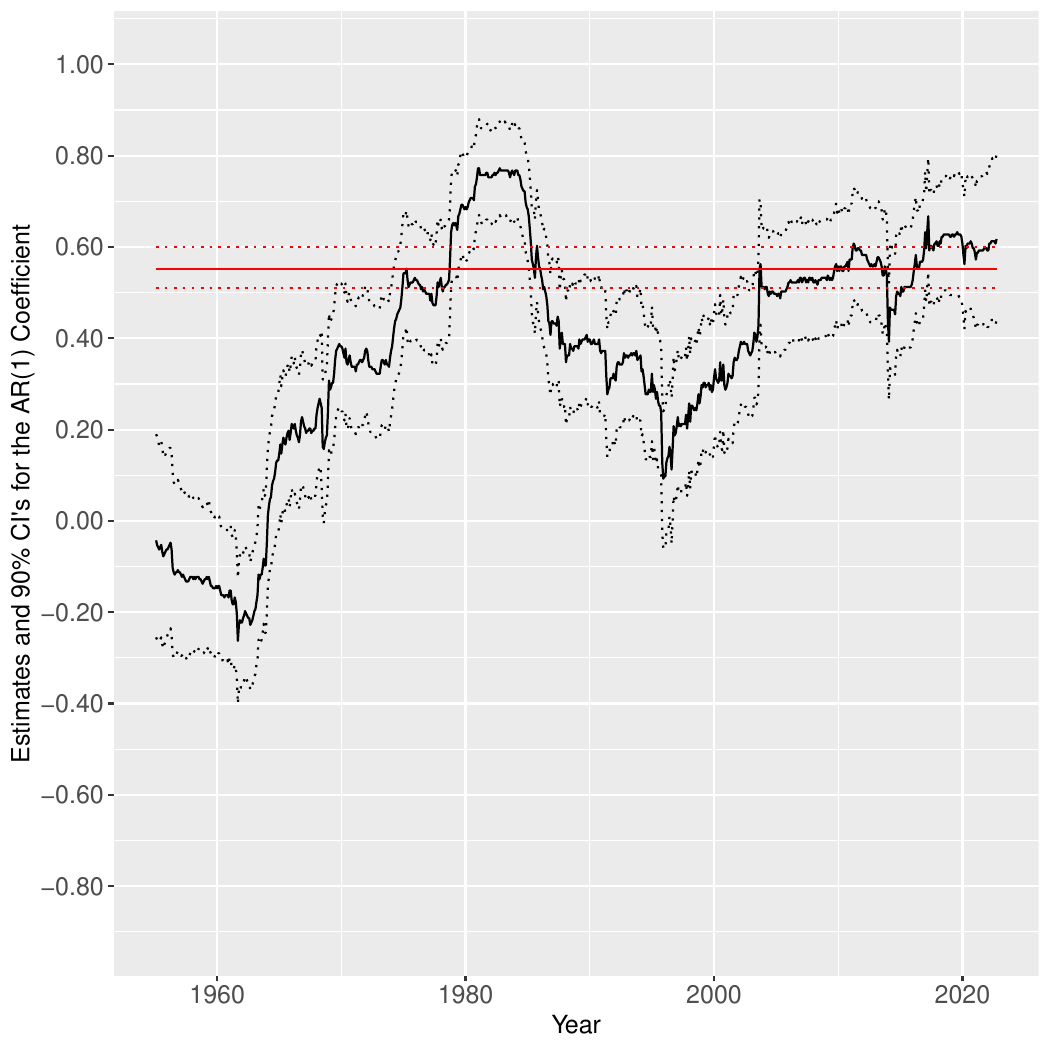}

}\quad{}\subfloat[US Inflation, $1.5n\widehat{h}_{us}=$ 188]{\includegraphics[scale=0.36]{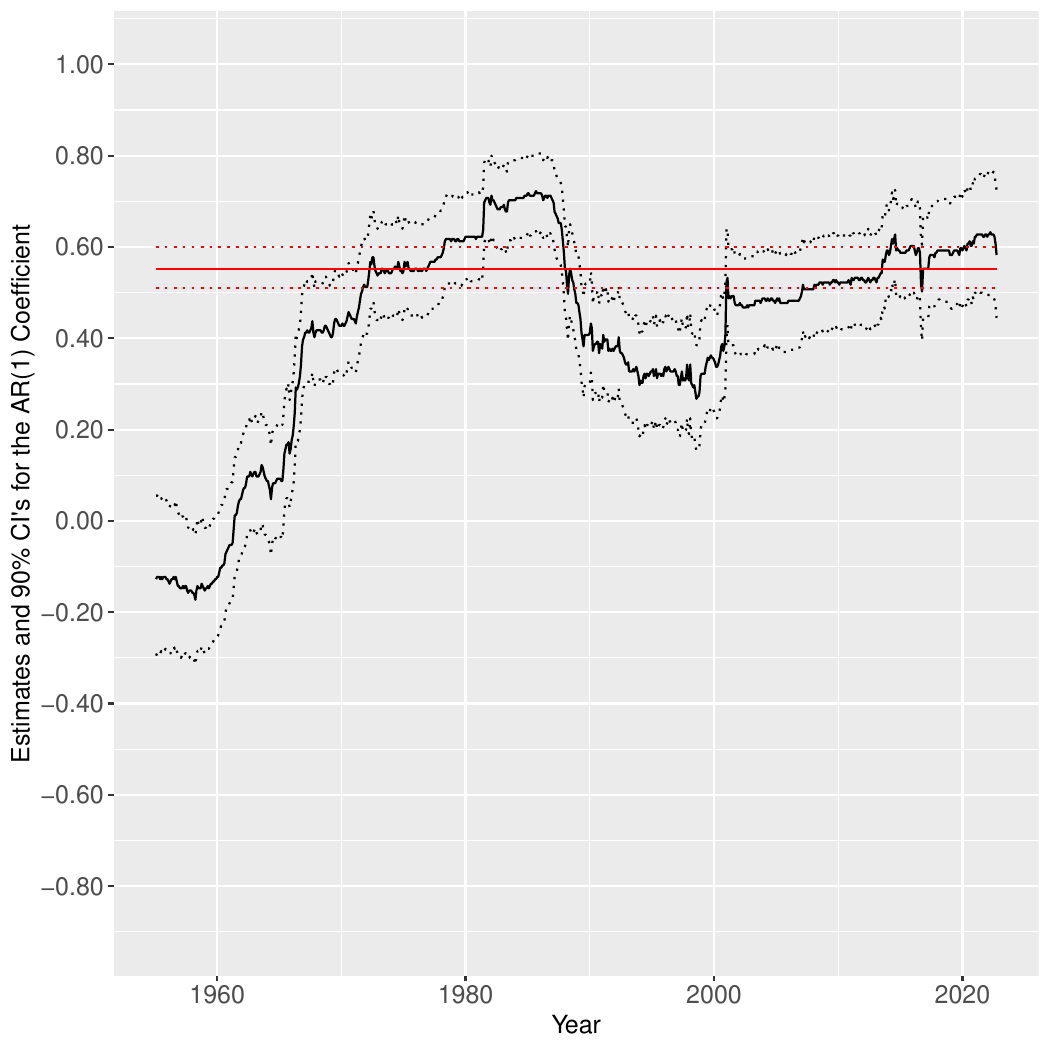}

}
\par\end{centering}
\begin{centering}
\subfloat[Canada Inflation, $n\widehat{h}_{us}=$ 125]{\noindent\includegraphics[scale=0.36]{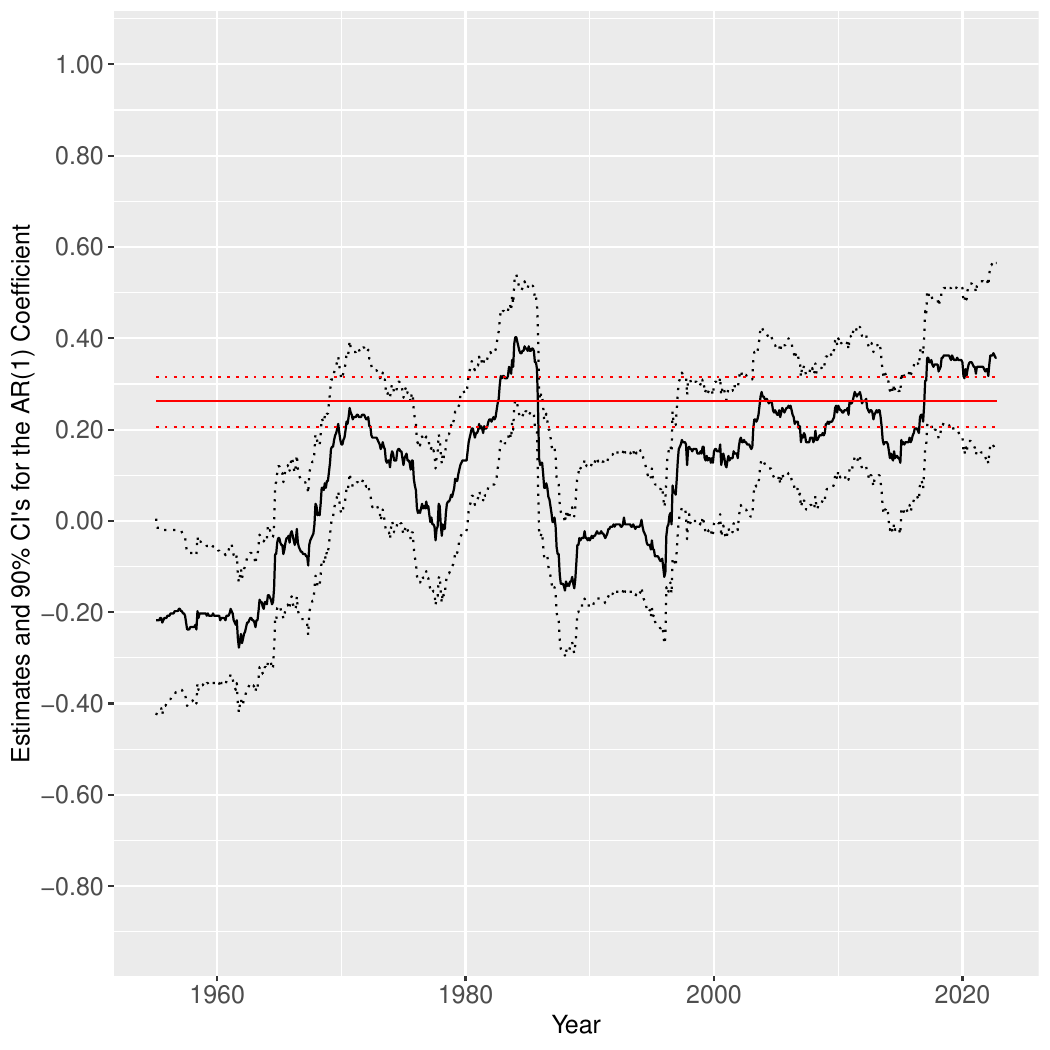}

}\quad{}\subfloat[Canada Inflation, $1.5n\widehat{h}_{us}=$ 188]{\includegraphics[scale=0.36]{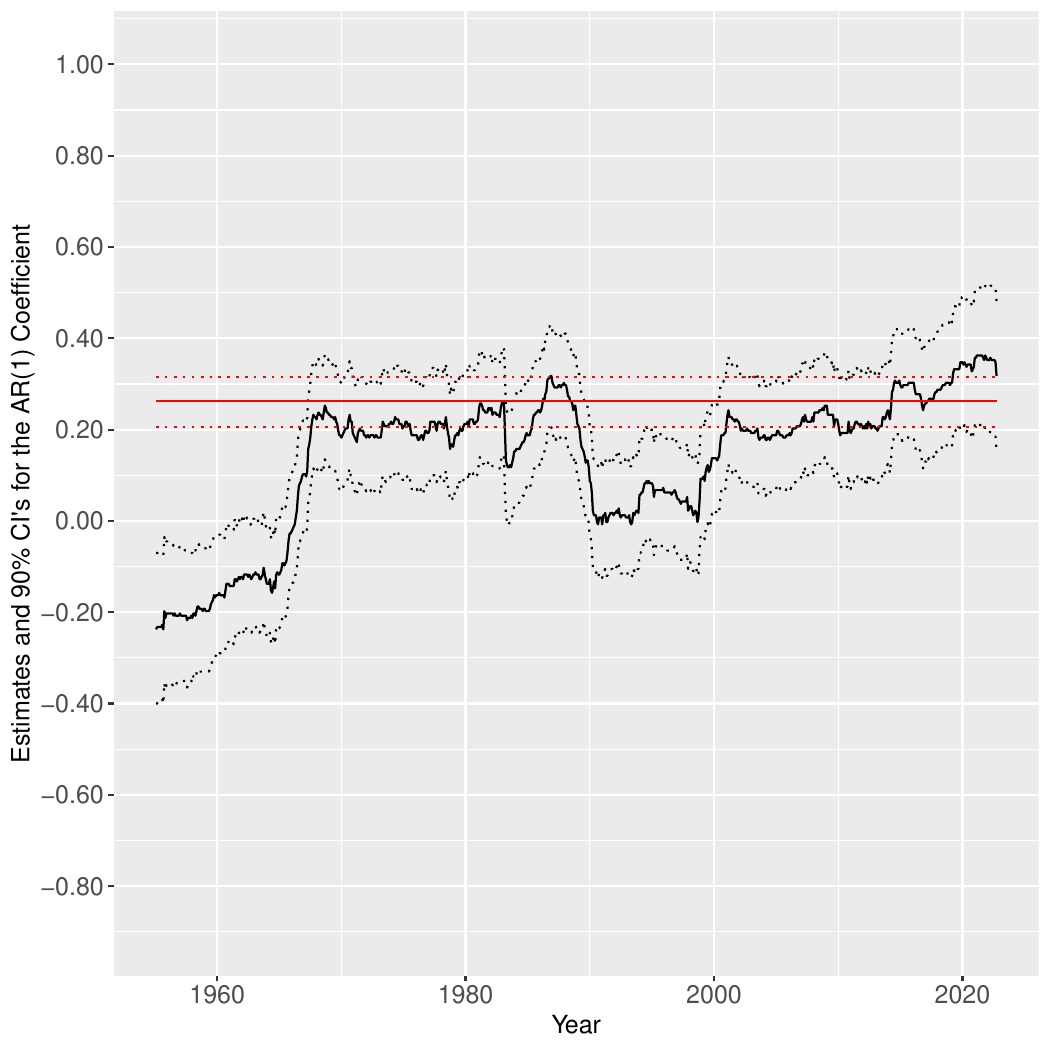}

}
\par\end{centering}
\begin{centering}
\subfloat[Germany Inflation, $n\widehat{h}_{us}=$ 125]{\includegraphics[scale=0.36]{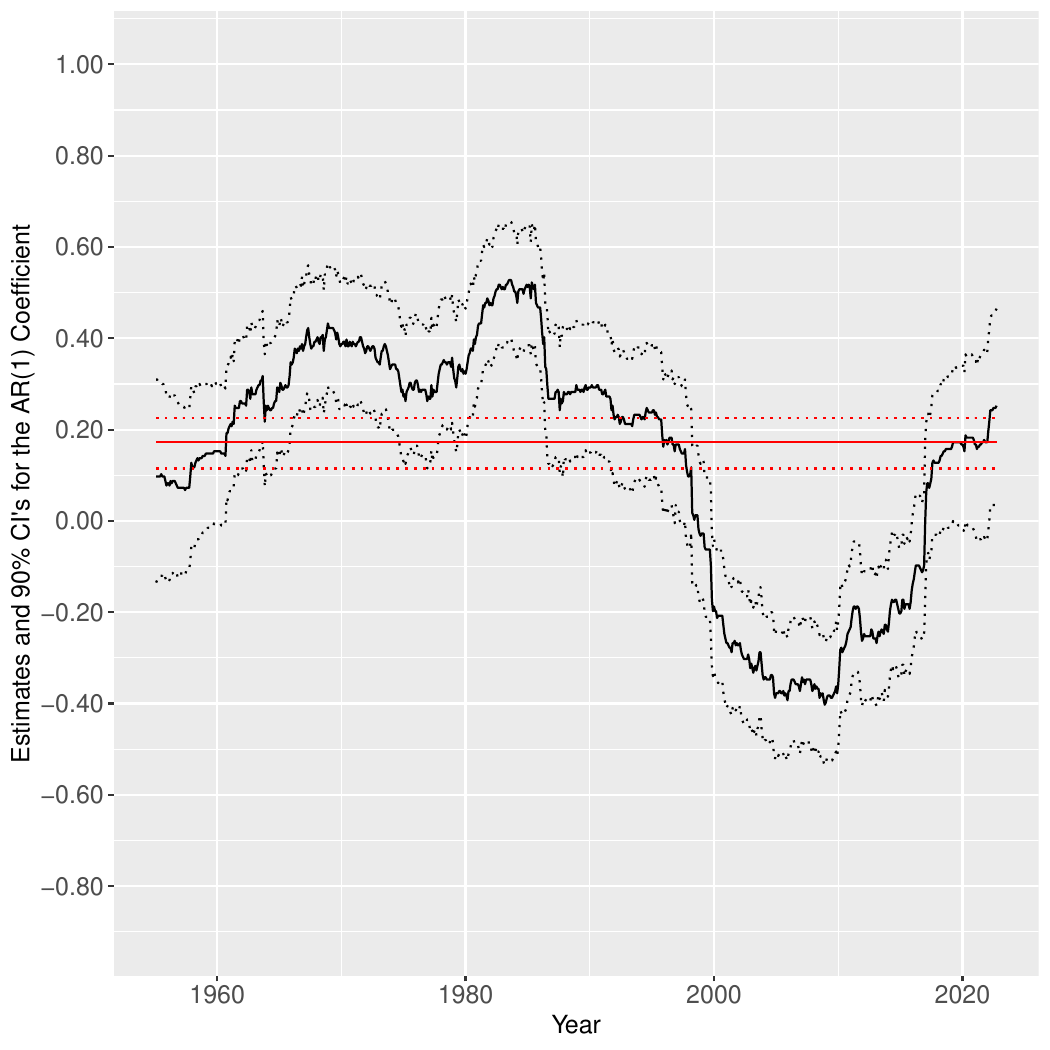}

}\quad{}\subfloat[Germany Inflation, $1.5n\widehat{h}_{us}=$ 188]{\includegraphics[scale=0.36]{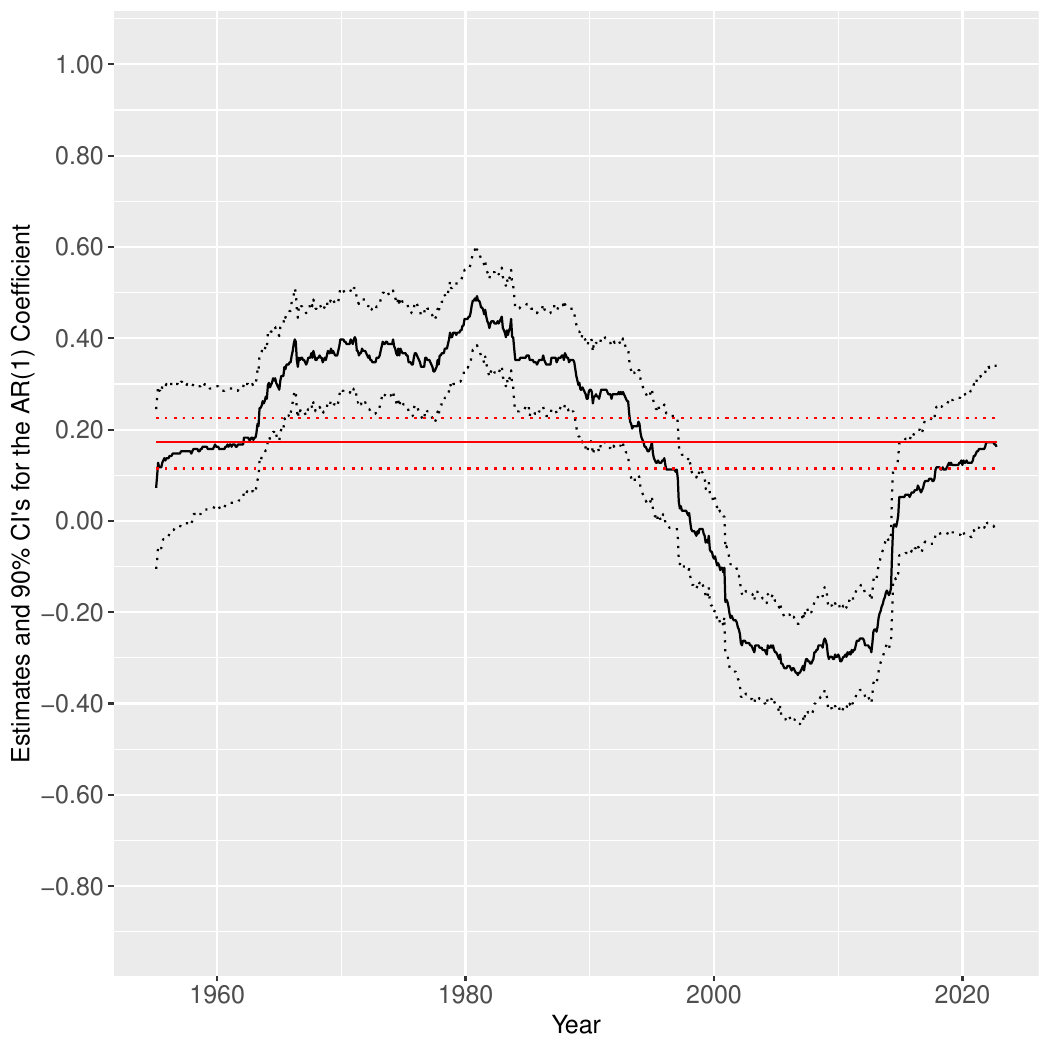}

}
\par\end{centering}
\caption{Estimates and 90\% CI's for the AR(1) Coefficient in TVP-AR(1) Models:
US, Canada, and Germany Inflation. The solid black line is the MUE's
of $\rho\left(t\right)$ in a TVP-AR(1) model, with the 90\% CI's
denoted by the dotted black line. The solid red line is the MUE's
of the AR coefficient in a constant parameter AR(1) model, with the
90\% CI's denoted by the dotted red line.\protect\label{fig:MT-EMP_AR1_1}}
\end{figure}
{\noindent}method to derive the constant parameter estimates is equivalent
to Mikusheva's \citeyearpar{mikusheva2007uniform} modification of
Stock's \citeyearpar{stock1991confidence} method. \citet{mikusheva2007uniform}
proves that these confidence sets are uniformly valid asymptotically
for non-time-varying $\rho\in(0,1]$.

Inflation persistence in the US has been extensively studied with
mixed results regarding its extent and causes. Studies like those
by \citet{cogley2001evolving} have noted a decrease in inflation
persistence following the early 1980s, attributing this change to
shifts in monetary policy, especially the Federal Reserve's (Fed)
increased focus on inflation targeting. This view is supported by
research that points to the Volcker disinflation period as a pivotal
time when the Fed's credibility was enhanced, leading to a stabilization
of inflation expectations, see \citet{sims1992interpreting}. We also
find a sharp decline from .78 to .12 in inflation persistence in the
US measured by the MUE of the TVP-AR coefficient between 1983 and
1995 in Figure \ref{fig:MT-EMP_AR1_1}(a), which is consistent with
the above empirical results.

Entering the era of 2000s, inflation persistence remains a central
topic in macroeconomic research. One line of research (e.g., \citet*{eggertsson2003zero,chung2012have})
concerns zero lower bound (ZLB) effects when the Fed set the interest
rates close to zero. The theory predicts that at ZLB, monetary policy
could be less effective in controlling inflation, thereby potentially
increasing inflation persistence if not coupled with assertive non-traditional
interventions. To overcome the challenges of the ZLB effects on the
effectiveness of its monetary policies, the Fed initiated a new practice
called ``forward guidance,'' where the future course of monetary
policy is communicated to the public by the central bank. Numerous
studies (e.g., \citet*{CAMPBELL2012BIP,negro2023JPE}) have suggested
that forward guidance is effective in enabling the Fed to better control
inflation, which could lead to lower inflation persistence. Along
this line, \citet{bernanke2020new} argues that forward guidance combined
with quantitative easing (QE) granted the Fed significantly more space
to provide accommodation when its standard policy rate was near zero.
More recently, \citet*{cole2023living} argue that despite the significant
efforts to make their policy credible, the credibility of most central
banks including the Fed has been generally declining, making monetary
policies aimed at controlling inflation less effective. A concrete
manifestation of their claim would be an increase in the inflation
persistence over a longer horizon.

We find an increase in inflation persistence in the US from .3 to
.6 measured by the MUE of TVP-AR coefficient between 2000 and 2008
in Figure \ref{fig:MT-EMP_AR1_1}(a), which seems to be mostly driven
by the ZLB effects despite the introduction of forward guidance during
this period. Later on, when more specific monetary policies (e.g.,
QE) were introduced to stimulate the economy following the 2008 financial
crisis, inflation persistence drops from .6 to .4 between 2010 and
2013, in line with the result of \citet{bernanke2020new}. With that
said, we observe in Figure \ref{fig:MT-EMP_AR1_1}(a) an upward trend
in inflation persistence in the US over a longer horizon between 2000
and 2022, which could be caused by a decline in the credibility of
the Fed as claimed by \citet*{cole2023living}. Overall, we conclude
that while the Fed's measures have been effective in controlling inflation
expectations and persistence, more nuanced policies are needed to
handle complications arising from the ZLB effects, heterogeneous beliefs
(\citet*{andrade2019forward}), structural changes (e.g., higher volatility
and shifts in consumer behavior and supply chains) caused by the Covid-19
and other factors.

We also apply our method to study inflation persistence in Canada
and Germany and obtain the following results. First, we find significant
time variation in inflation persistence as measured by the MUE's of
$\rho\left(t\right)$ for both countries since 1980s. Given the magnitude
of the time-variation, our findings show that inflation persistence
is more likely to be driven by changes in the regime of monetary policy
and credibility of central banks, rather than by nominal or real frictions
in the economy. Second, there seems to be a universal upward trend
since 2000 in inflation persistence for all countries under consideration,
echoing the findings of \citet*{cole2023living}, who argue that there
is a decline in credibility in central bank policies. Third, the introduction
of the Euro in 1999 marked a significant shift in inflation persistence
in Germany, with the European Central Bank (ECB) taking over monetary
policy. Inflation rates remained relatively stable but were influenced
by broader Eurozone policies and economic conditions. The early 2000s
saw moderate inflation, which aligned closely with the ECB\textquoteright s
target. Our method captures this regime change by showing drastically
different patterns of inflation persistence before and after 1999,
supporting the conclusion above that inflation persistence is more
likely to be driven by changes in the regime of monetary policy and
the credibility of central banks. 

\subsection{Real Exchange Rate\protect\label{subsec:MT-EMP_Real-Exchange-Rate}}

The second application concerns the real exchange rate. The results
are given in Figure \ref{fig:MT-EMP_AR1_2}. We use US dollars (USD)
as the benchmark currency, and calculate the bilateral real exchange
rate $re_{it}$ of country $i$ at time $t$ as $rex_{it}=nex_{it}\times\frac{CPI_{it}}{CPI_{0t}}$,
where $nex_{it}$ is the nominal exchange rate (USD per domestic currency)
at time $t$, $CPI_{it}$ is the price level of country $i$ at time
$t$, and $CPI_{0t}$ is the price level of the US at time $t$. Therefore,
an increase in $rex_{it}$ represents an appreciation of country $i$'s
currency against USD. We report results for the UK, Sweden, and Switzerland.
The data span for the monthly real exchange rate dataset is Jan 1957
to Aug 2022. Thus, $n=788$ for each country. Similar to the inflation
series, we compute the $n\widehat{h}_{us}$ values based on (\ref{eq:MT-hhat_us-def}).
Then, we compute the MUE's and 90\% CI's for the{\noindent}
\begin{figure}[H]
\captionsetup[subfigure]{position=top,font=scriptsize,singlelinecheck=off,justification=raggedright}\vspace{-1cm}

\begin{centering}
\subfloat[UK Real Exchange Rate, $n\widehat{h}_{us}=$ 823]{\noindent\includegraphics[scale=0.36]{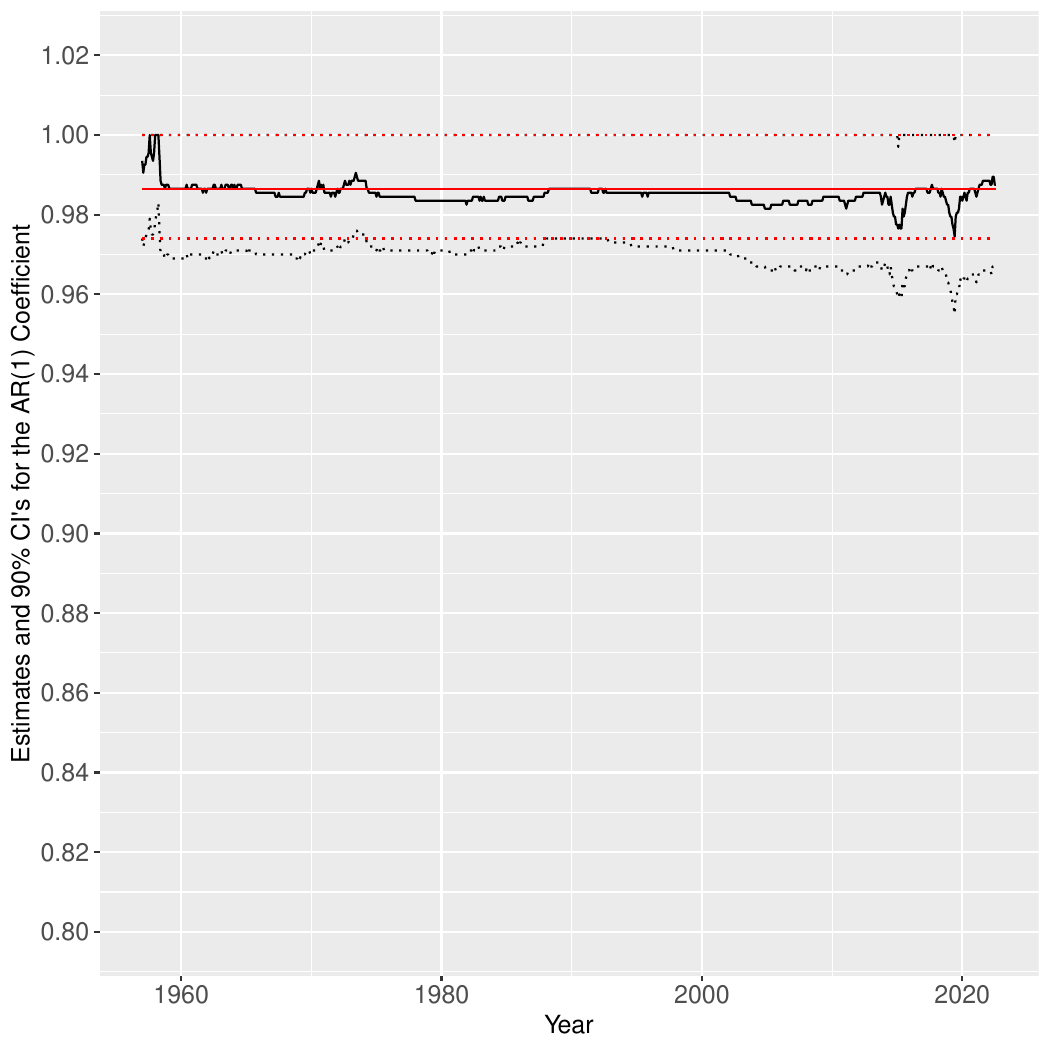}

}\quad{}\subfloat[UK Real Exchange Rate, $1.5n\widehat{h}_{us}=$ 1,234]{\includegraphics[scale=0.36]{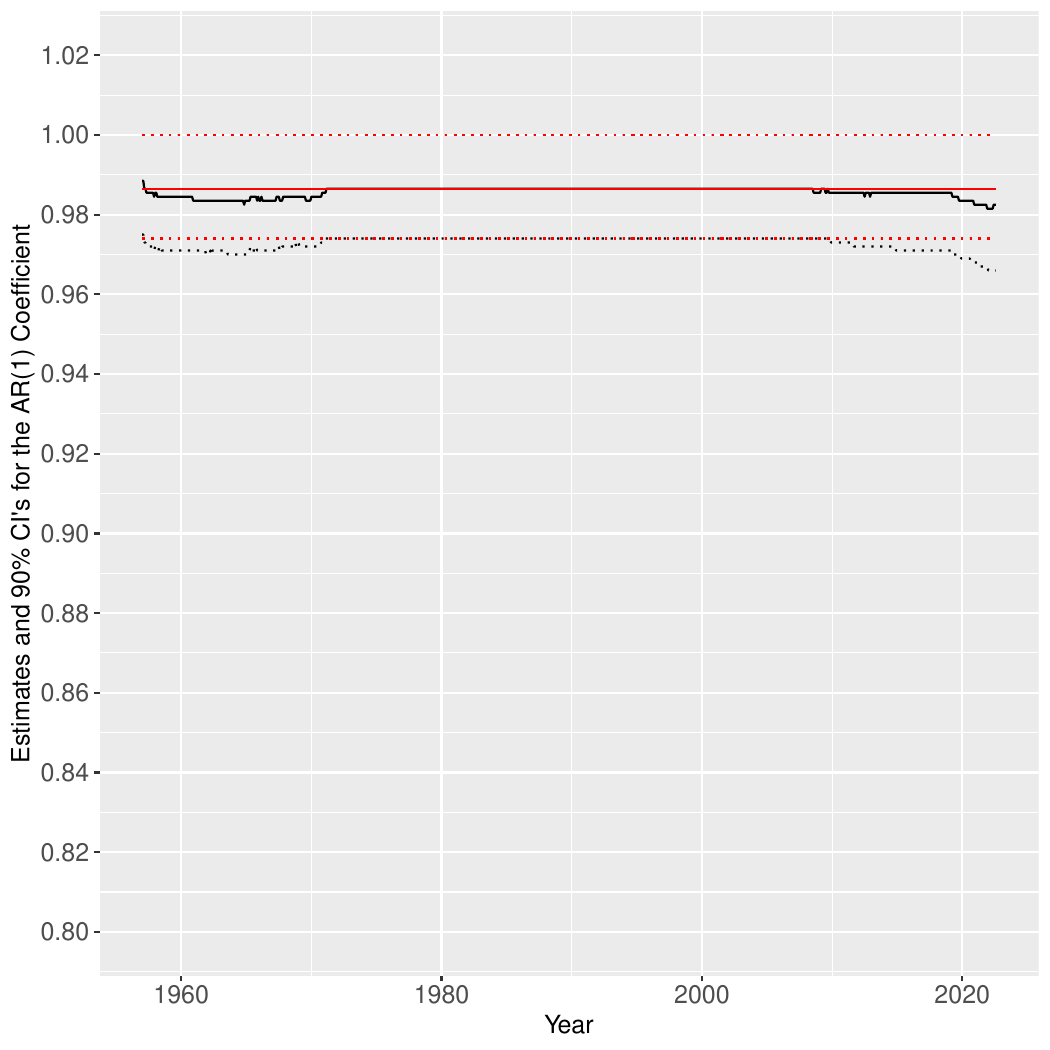}

}
\par\end{centering}
\begin{centering}
\subfloat[Sweden Real Exchange Rate, $n\widehat{h}_{us}=$ 823]{\includegraphics[scale=0.36]{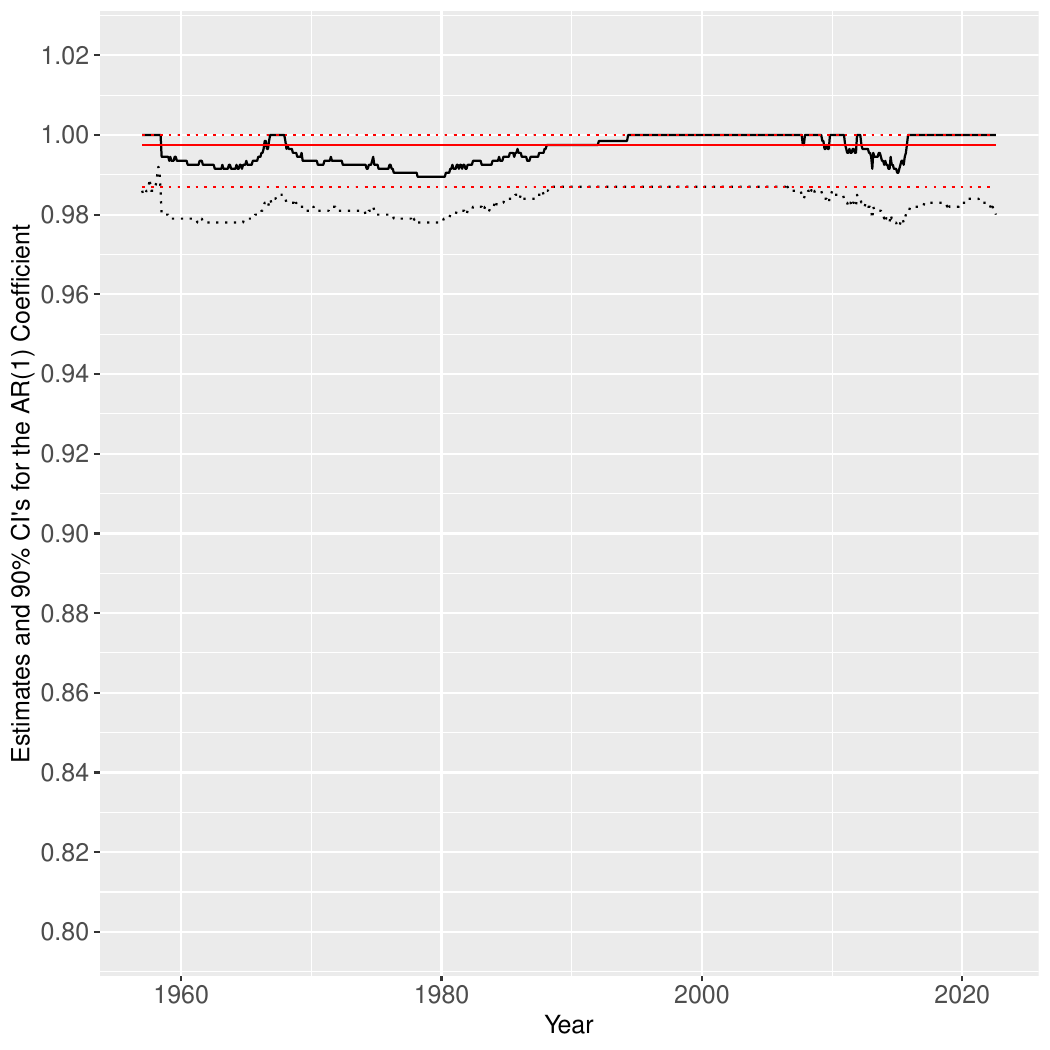}

}\quad{}\subfloat[Sweden Real Exchange Rate, $1.5n\widehat{h}_{us}=$ 1,234]{\includegraphics[scale=0.36]{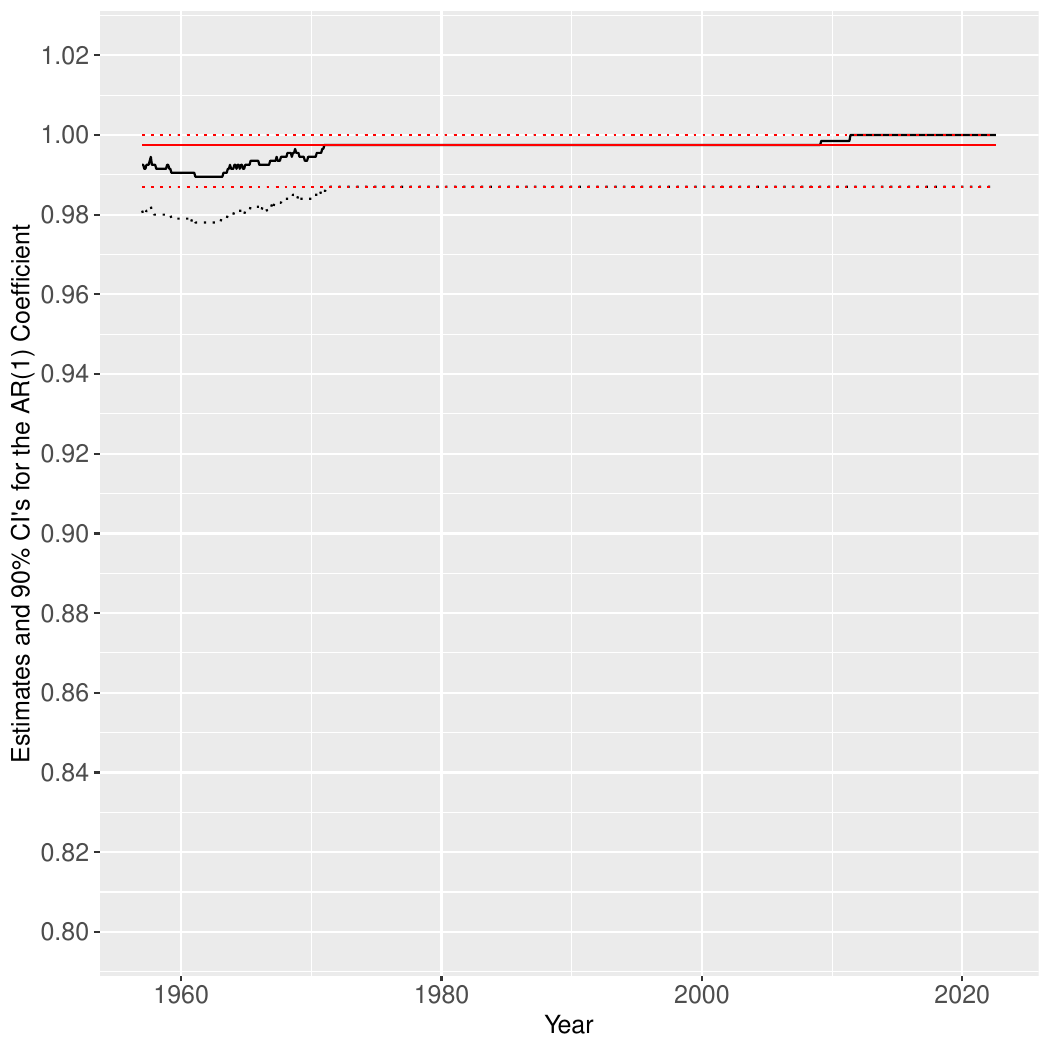}

}
\par\end{centering}
\begin{centering}
\subfloat[Switzerland Real Exchange Rate, $n\widehat{h}_{us}=$ 393]{\includegraphics[scale=0.36]{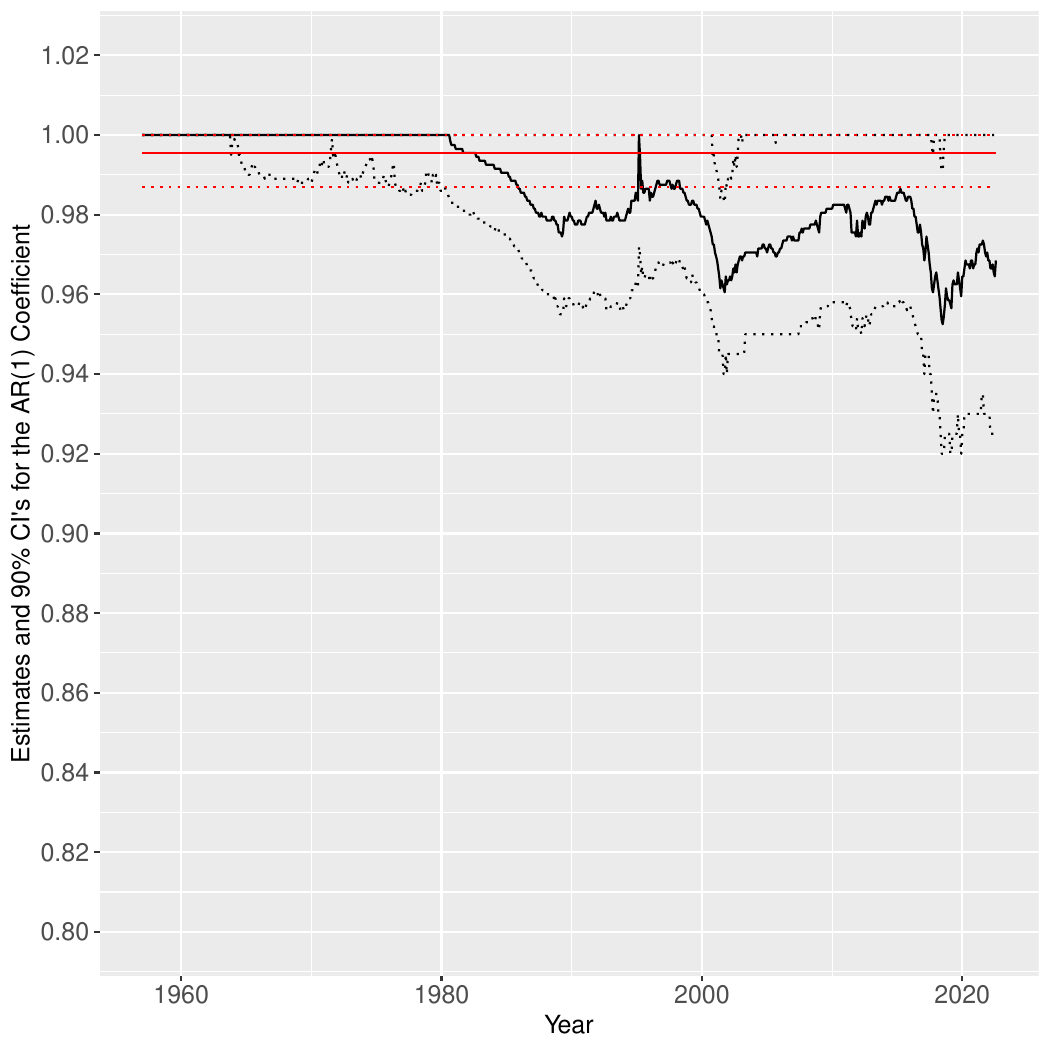}

}\quad{}\subfloat[Switzerland Real Exchange Rate, $1.5n\widehat{h}_{us}=$ 590]{\includegraphics[scale=0.36]{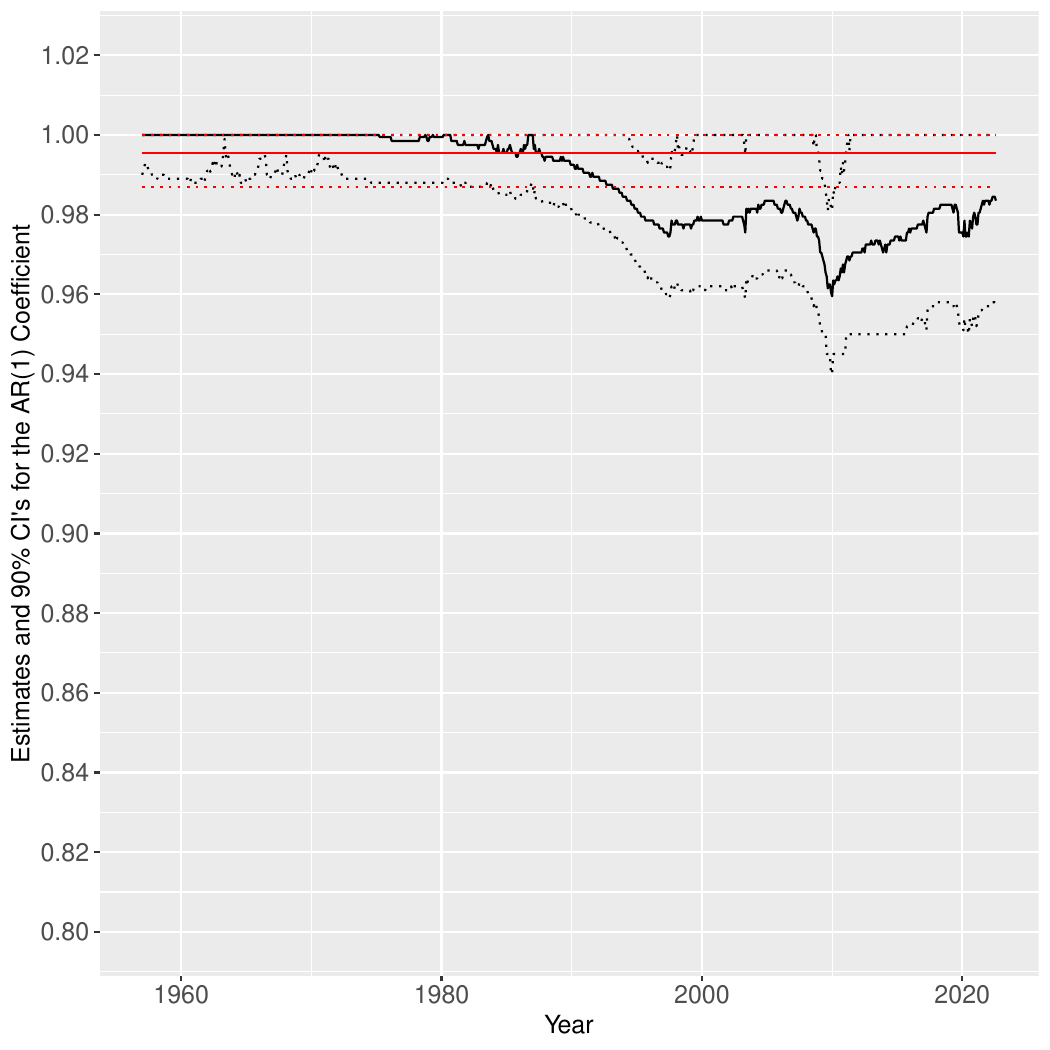}

}
\par\end{centering}
\caption{Estimates and 90\% CI's for the AR(1) Coefficient in TVP-AR(1) Models:
UK, Sweden, and Switzerland Real Exchange Rate\protect\label{fig:MT-EMP_AR1_2}}
\end{figure}
{\noindent}AR coefficient, $\rho\left(t\right)$, at each $t$ and
for each country. For robustness purposes, we multiply the undersmoothed
$n\widehat{h}_{us}$ by 1.5 and report the results in the right column
of each figure. For the real exchange rate series considered in this
section, Ljung-Box tests based on six lags of the residuals from the
AR(1) model fail to reject the null hypothesis of no autocorrelation
at the 5\% significance level, see Table \ref{tab:ACF_Test_AR1}.

It has been known in the literature that real exchange rates in developed
countries tend to be highly persistent, with deviations from the PPP
level taking a long time to disappear (\citet{rogoff1996purchasing,engel2014exchange}).
There are several explanations in the literature for the real exchange
rate persistence, including nominal price rigidities, interest rate
inertia in monetary policy (\citet{benigno2004real}), heterogeneous
dynamics in subcomponents (\citet*{imbs2005ppp}), Balassa--Samuelson
effects (\citet{balassa1964purchasing,samuelson1964theoretical}),
to name a few. Figure \ref{fig:MT-EMP_AR1_2} presents the results
on real exchange rate persistence measured by the MUE's of $\rho\left(t\right)$
for the UK, Sweden, and Switzerland. Across all three countries the
MUE's are close to 1, showing that our method yields near constant
graphs in scenarios where that seems to be suitable. The 90\% CI's
for $\rho\left(t\right)$ are also very short. \citet*{chari2002can}
report estimates from a constant parameter autoregressive process
to lie between .76 and .87 for US bilateral real exchange rates against
nine developed European countries using data between 1972 and 1994.
They developed a general equilibrium model where the firms can only
set price once per year, which implies a very high level of price
stickiness in the economy. Yet, their model cannot generate the level
of persistence observed in the data. Allowing for a time-varying parameter
autoregressive process and using data from a longer time horizon,
our MUE's of $\rho\left(t\right)$ for the three currencies are even
higher at close to one, supporting their claim that nominal price
rigidities are not enough to explain the high real exchange rate persistence.
Practitioners would need to seek other possible explanations such
as the Balassa--Samuelson effects or heterogeneous dynamics in subcomponents. 

While high real exchange rate persistence seems prevalent, there is
heterogeneity in the pattern of persistence across the countries we
consider. In particular, the MUE's of $\rho\left(t\right)$ for Switzerland
demonstrate a downward trend since early 1980s in Figure \ref{fig:MT-EMP_AR1_2}(e).
The Swiss National Bank (SNB) has adopted several significant monetary
policy changes since the 1990s, including the shift to a three-fold
target (price stability, 3-year inflation forecast, and a range for
the 3M Libor, see \citet*{jordan2010ten}) in the late 1990s and the
introduction of negative interest rates in 2015 when it was forced
to abandon a policy of defending the Swiss franc with a peg to the
euro. These policies aim to stabilize price levels and influence interest
rates, which can affect exchange rate dynamics and potentially reduce
persistence by promoting quicker adjustments to shocks. In Figure
\ref{fig:MT-EMP_AR1_2}(e) we observe a sharp decline in the MUE of
AR(1) coefficient between 1995 and 2002 and between 2015 and 2019,
consistent with the timing of SNB's monetary policy changes. Meanwhile,
the European debt crisis and subsequent economic turmoil in the Eurozone
led to significant safe-haven flows into the Swiss franc, prompting
the SNB to implement a cap on the franc\textquoteright s value against
the euro in 2011. This cap was removed in 2015. Such a cap on the
franc's value would increase its real exchange rate persistence since
the price of the franc could have fluctuated above the cap for the
duration of the policy. We also find the real exchange rate persistence
to be slightly increasing between 2011 and 2015 in Figure \ref{fig:MT-EMP_AR1_2}(e).
As shown by these results, our method can capture the major events
and policy changes that affect real exchange rate persistence reasonably
accurately. 

\section{Asymptotics\protect\label{sec:MT-Asymptotics}}

This section establishes the correct uniform asymptotic size and asymptotic
similarity of the confidence interval $CI_{n,\tau}$ for $\rho(\tau)$,
the median-unbiased property of $\widetilde{\rho}_{n\tau}$, and some
asymptotic properties of the data-dependent bandwidth parameter $\widehat{h}$.
To prove these results, this section provides asymptotic results for
the LS estimator $\widehat{\rho}_{n\tau}$ and t-statistic $T_{n}\left(\rho_{0,n}\right)$
under drifting sequences of parameter values. All proofs of the results
stated below are given in Section \ref{sec:Theory} of the Supplemental
Material.

\subsection{\protect\label{sec:MT-Parameter=000020Space}Parameter Space}

Let $I_{a,r}:=\left[a-r,a+r\right]$ for $a\in\R$ and $r>0$. Let
$\lfloor x\rfloor$ denote the integer part of $x$.

We impose the following structure on the $\rho$ function: for some
$\varepsilon_{2},\varepsilon_{3}>0$, 
\begin{equation}
\rho\left(s\right)=1-\kappa\left(s\right)/b\label{eq:MT-structure-of-rho}
\end{equation}
for $s\in I_{\tau,\varepsilon_{2}}$ and some $b\in\left[\varepsilon_{3},\infty\right]$,
where $\kappa\left(\cd\right)$ is a nonnegative twice continuously
differentiable function on $I_{\tau,\varepsilon_{2}}$.

The parameter space for $\left(\rho,\mu,\sigma^{2},\kappa,b,F\right)$
is given by

$\Lambda_{n}=\{\lambda=\left(\rho,\mu,\sigma^{2},\kappa,b,F\right)$:

(i) $\rho,\ \mu,$ and $\sigma^{2}$ are Lipschitz functions from
$\left[0,1\right]$ to $\left[-1+\varepsilon_{1},1\right],\ \left[C_{2,L},C_{2,U}\right]$,
and $\left[C_{3,L},C_{3,U}\right]$, respectively, with Lipschitz
constants bounded by $L_{1}$, $L_{2}$, and $L_{3}$, respectively;

(ii) $\rho\left(s\right)=1-\kappa\left(s\right)/b$ for $s\in I_{\tau,\varepsilon_{2}}$
and $b\in[\varepsilon_{3},\infty]$, where $\kappa\left(\cd\right)$
is a twice continuously differentiable function from $I_{\tau,\varepsilon_{2}}$
to $\left[0,C_{4}\right]$ with Lipschitz constant bounded by $L_{4}$
and $\kappa(\tau)\geq\varepsilon_{4}$;

(iii) $\mu\left(s\right)=C_{\mu1}\exp\left\{ -\eta\left(s\right)/b\right\} +C_{\mu2}$
for $s\in I_{\tau,\varepsilon_{2}}$ and $b\in[\varepsilon_{3},\infty]$,
where $\eta\left(\cd\right)$ is a nonnegative Lipschitz function
with Lipschitz constant bounded by $L_{5}$;

(iv) $\left\{ U_{t}:t=0,1,...,n\right\} $ is a stationary martingale
difference sequence under $F$ with $\E_{F}\left(\rest{U_{t}}\mathscr{G}_{t-1}\right)=0\ \text{a.s.}$,
$\E_{F}\left(\rest{U_{t}^{2}}\mathscr{G}_{t-1}\right)=1\ \text{a.s.}$,
and $\E_{F}\left[\rest{U_{t}^{4}}\mathscr{G}_{t-1}\right]<M\ \text{a.s.},$
where $\mathscr{G}_{t}$ is some non-decreasing sequence of $\sigma$-fields
for which $\sigma\left(U_{0},...,U_{t}\right)\subseteq\mathscr{G}_{t}$
and $\sigma\left(Y_{0}^{*}\right)\in\mathscr{G}_{t}$ for $t=1,...,n$;

(v) $\E_{F}\left(Y_{0}^{*}\right)^{2}\leq C_{5}n$;

for some constants $0<\varepsilon_{1}<2$, $\varepsilon_{2},\varepsilon_{3},\varepsilon_{4}>0$,
$-\infty<C_{j,L}\leq C_{j,U}<\infty$ for $j=2,3$, $C_{3,L}>0$,
$0\leq C_{4}<\infty$, $0\leq C_{5}<\infty$, $L_{j}<\infty\ \any j\leq5$,
$C_{\mu1},C_{\mu2}\in\left(-\infty,\infty\right),$ and $M\in\left(0,\infty\right)\}$.

$\ $

Note that the dependence on $n$ of $\Lambda_{n}$ is only through
part (v), which concerns the initial condition $Y_{0}^{*}$.\footnote{The parameter space $\Lambda_{n}$ could be made independent of $n$
by specifying that $\E\left(Y_{0}^{*}\right)^{2}\leq C_{5}$. But,
allowing the bound on $\E\left(Y_{0}^{*}\right)^{2}$ to be $C_{5}n$,
allows the parameter space to expand with $n$ and is in accord with
assumptions in the literature for AR models with a possible unit root
or near unit root.} Also, $\mu\left(s\right)$ is defined in such a way that it is allowed
to vary smoothly on $\R$. It is used in the proof of Lemma \ref{lem:MT-Max=000020Intertemporal=000020Difference}(d),
which is further used to bound the absolute difference between $\mu\left(t\right)$
and $\mu\left(s\right)$ for $t,s\in\left[T_{1},T_{2}\right]$ in
obtaining the asymptotic distribution of the t-statistic $T_{n}\left(\rho_{0,n}\right)$.

$\ $

To establish the correct asymptotic size of the CI for $\rho(\tau)$,
we need to consider sequences of parameters $\left\{ \lambda_{n}=\left(\rho_{n},\mu_{n},\sigma_{n}^{2},\kappa_{n},b_{n},F_{n}\right)\right\} _{n\geq1}$
in $\Lambda_{n}.$ The parameter space $\Lambda_{n}$ is defined to
allow the sequence $\left\{ \rho_{n}(\tau)\right\} _{n\geq1}$ to
equal one for all $n\geq1,$ converge to one at any rate, or be bounded
away from one. More specifically, by the definition of $\rho_{n}(\cd)$,
\begin{equation}
\rho_{n}\left(\tau\right)=1-\kappa_{n}\left(\tau\right)/b_{n}\ \text{for some }b_{n}\in\left[0,\infty\right].\label{eq:MT-rho_n(tau)}
\end{equation}
If $b_{n}\gto\infty$, then $\rho_{n}\left(\tau\right)$ is local
to one. If $\limsup_{n\gto\infty}b_{n}<\infty$, then $\rho_{n}\left(\tau\right)$
is bounded away from one. Furthermore, $\Lambda_{n}$ is defined so
that, for points $s\in\left[0,1\right]$ other than $\tau,$ $\rho_{n}\left(s\right)$
can converge to one at any rate or be bounded away from one, and the
distance-from-unity behavior may differ between different $s$ outside
of $I_{\tau,\varepsilon_{2}}.$

\subsection{Correct Asymptotic Size of the Confidence Interval for \texorpdfstring{$\boldsymbol{\rho(\tau)}$}{rho(tau)}}

The bandwidth $h$ is assumed to satisfy the following assumptions.
\begin{assumption}[\textbf{Bandwidth $\boldsymbol{h}$}]
 \label{assu:MT-h} $h\gto0$ and $nh\gto\infty$ as $n\gto\infty$. 
\end{assumption}
\begin{assumption}[\textbf{Order of $\boldsymbol{h}$}]
 \label{assu:MT-stationary-nh^5=000020->=0000200} $nh^{5}\rightarrow0$.\footnote{Assumption \ref{assu:MT-stationary-nh^5=000020->=0000200} is used
in the proofs of Lemmas \ref{lem:MT-Order=000020of=000020Y_T0}(b),
\ref{lem:MT-Order=000020of=000020Y_T0}(c), \ref{lem:MT-asympdist_dist_stationary-denom},
and \ref{lem:MT-asympdist_dist_stationary-numerator=000020}(a) below.} 
\end{assumption}
Let $P_{\lambda}\left(\cdot\right)$ denote probability under $\lambda\in\Lambda_{n}.$
The correct asymptotic size and asymptotic similarity of the CI $CI_{n,\tau}$
are established in the following theorem. The theorem's proof relies
on the ``drifting pointwise'' asymptotic results given in Sections
\ref{sec:MT-Preliminaries}-\ref{sec:MT-Stationary=000020Case} below
and the generic asymptotic size results in \citet*{andrews2020generic}.
\begin{thm}[\textbf{Correct Asymptotic Size of $\boldsymbol{CI_{\boldsymbol{n,}\tau}}$}]
\label{thm:MT-Cor=000020Asy=000020Size}Under Assumptions \ref{assu:MT-h}
and \ref{assu:MT-stationary-nh^5=000020->=0000200}, 
\[
\liminf\limits_{n\rightarrow\infty}\inf_{\lambda\in\Lambda_{n}}P_{\lambda}\left(\rho\left(\tau\right)\in CI_{n,\tau}\right)=\limsup\limits_{n\rightarrow\infty}\sup_{\lambda\in\Lambda_{n}}P_{\lambda}\left(\rho\left(\tau\right)\in CI_{n,\tau}\right)=1-\alpha.
\]
\end{thm}
The median-unbiased interval estimator $\widetilde{\rho}_{n\tau}$
has the following median-unbiasedness property. This result is a corollary
to one-sided versions of Theorem \ref{thm:MT-Cor=000020Asy=000020Size}.\footnote{Corollary \ref{cor:MT-Med-Unbiased} holds because $CI_{n,\tau}^{up}\left(.5\right)$
and $CI_{n,\tau}^{low}\left(.5\right)$ both have coverage probabilities
of $1/2$ or greater by the proof of Theorem \ref{thm:MT-Cor=000020Asy=000020Size}
applied to these one-sided CIs and $\left(\widetilde{\rho}_{n\tau}\geq\rho\left(\tau\right)\right)$$\supset$
$CI_{n,\tau}^{up}\left(.5\right)$ and $\left(\widetilde{\rho}_{n\tau}\leq\rho\left(\tau\right)\right)\supset CI_{n,\tau}^{low}\left(.5\right).$} 
\begin{cor}[\textbf{Asymptotic Median-Unbiasedness of $\boldsymbol{\boldsymbol{\widetilde{\rho}}_{\boldsymbol{n\tau}}}$}]
\label{cor:MT-Med-Unbiased}Under Assumptions \ref{assu:MT-h} and
\ref{assu:MT-stationary-nh^5=000020->=0000200}, 
\begin{align*}
\liminf\limits_{n\rightarrow\infty}\inf_{\lambda\in\Lambda_{n}}P_{\lambda}\left(\widetilde{\rho}_{n\tau,up}\geq\rho\left(\tau\right)\right) & \geq1/2\ \textup{and}\\
\liminf\limits_{n\rightarrow\infty}\inf_{\lambda\in\Lambda_{n}}P_{\lambda}\left(\widetilde{\rho}_{n\tau,low}\leq\rho\left(\tau\right)\right) & \geq1/2.
\end{align*}
\end{cor}

\subsection{\protect\label{sec:MT-Preliminaries}Preliminaries}

Define 
\begin{equation}
c_{t,j}:=\prod_{k=0}^{j-1}\rho_{t-k}\ \text{for}\ 1\leq j\leq t,\ 1\leq t\leq n,\ \ \text{and}\ c_{t,0}:=1.\label{eq:MT-def_c_t,j}
\end{equation}

We consider the interval $I_{n\tau,nh/2}$ and decompose $Y_{t}^{*}$
into the sum of two parts 
\begin{equation}
Y_{t}^{*}=Y_{t}^{0}+c_{t,t-T_{0}}Y_{T_{0}}^{*},\ \text{for }t\in I_{n\tau,nh/2},\label{eq:MT-decompY*}
\end{equation}
where $Y_{t}^{0}$ is an AR(1) process with the same time-varying
parameters $\rho_{t}$ as $Y_{t}^{*}$, but with a zero initial condition
at $T_{0}:=T_{1}-1$, and $c_{t,t-T_{0}}$ is defined in \eqref{eq:MT-def_c_t,j}.
By recursive substitution in \eqref{eq:MT-tvp-model}, we have 
\begin{equation}
Y_{t}^{0}=\sum_{j=0}^{t-T_{1}}c_{t,j}\sigma_{t-j}U_{t-j}\ \text{and }Y_{t}^{*}=\sum_{j=0}^{t-T_{1}}c_{t,j}\sigma_{t-j}U_{t-j}+c_{t,t-T_{0}}Y_{T_{0}}^{*}\ \text{for }t=T_{1},...,T_{2}.\label{eq:MT-Y_t^0=000020and=000020Y_t^star}
\end{equation}
Then, from \eqref{eq:MT-tvp-model} and \eqref{eq:MT-decompY*}, we
have 
\begin{align}
Y_{t} & =\mu_{t}+\sum_{j=0}^{t-T_{1}}c_{t,j}\sigma_{t-j}U_{t-j}+c_{t,t-T_{0}}Y_{T_{0}}^{*}\ \text{for }t=T_{1},...,T_{2}.\label{eq:MT-YtandYt0}
\end{align}

We consider sequences $\left\{ \lambda_{n}=\left(\rho_{n},\mu_{n},\sigma_{n}^{2},\kappa_{n},b_{n},F_{n}\right)\in\Lambda_{n}\right\} _{n\geq1}$$.$
The estimand of interest is $\rho_{n}\left(\tau\right)$ for $n\geq1.$

Let 
\begin{align}
\rho_{n,t}:=\rho_{n}\left(t/n\right),\  & \mu_{n,t}:=\mu_{n}\left(t/n\right),\ \sigma_{n,t}^{2}:=\sigma_{n}^{2}\left(t/n\right),\nonumber \\
\rho_{n\tau}:=\rho_{n}\left(\tau\right),\  & \mu_{n\tau}:=\mu_{n}\left(\tau\right),\ \text{and }\sigma_{n\tau}^{2}:=\sigma_{n}^{2}\left(\tau\right).\label{eq:MT-rho_n,t=000020defn=0000202nd=000020time}
\end{align}
By part (ii) of the definition of $\Lambda_{n},$ for $t\in I_{n\tau,nh/2}$,
\begin{equation}
\rho_{n,t}:=\rho_{n}\left(t/n\right)=1-\kappa_{n}\left(t/n\right)/b_{n}\label{eq:MT-rho_n}
\end{equation}
for $n$ sufficiently large that $h/2+1/n\leq\varepsilon_{2}.$ For
notational simplicity, we drop the $n$ subscripts on $\rho_{n,t},$
$\mu_{n,t},$ and $\sigma_{n,t}^{2}$ below.

We consider sequences of null hypotheses $H_{0}:\rho_{n\tau}=\rho_{0,n}$
for $n\geq1.$ As noted above, the CI $CI_{n,\tau}$ is defined by
the inversion of tests of such null hypotheses.

$\ $

For some results, we assume the functions $\kappa_{n}\left(\cd\right)$,
$\mu_{n}\left(\cd\right)$, and $\sigma_{n}^{2}\left(\cd\right)$
converge. 
\begin{assumption}[\textbf{Limits of $\boldsymbol{\kappa_{n}\left(\cd\right)}$, $\boldsymbol{\mu_{n}\left(\cd\right)}$,
and $\boldsymbol{\sigma_{n}^{2}\left(\cd\right)}$}]
\textbf{} \label{assu:MT-kappa_n_sigma_n=000020mu_n=000020convergence}
$\kappa_{n}\left(\cd\right)$, $\mu_{n}\left(\cd\right),$ and $\sigma_{n}^{2}\left(\cd\right)$
restricted to $I_{\tau,\varepsilon_{2}}$ converge uniformly to some
functions $\kappa_{0}\left(\cd\right)$, $\mu_{0}\left(\cd\right),$
and $\sigma_{0}^{2}\left(\cd\right)$ on $I_{\tau,\varepsilon_{2}}$,
respectively. 
\end{assumption}
$\ $

For some results, we assume that $\left\{ b_{n}/\left(nh^{1/2}\right)\right\} _{n\geq1}$
converges. 
\begin{assumption}[\textbf{Convergence of $\boldsymbol{b_{n}/\left(nh^{1/2}\right)}$}]
\textbf{} \label{assu:MT-bn=000020div=000020nh=000020converges}
$b_{n}/\left(nh^{1/2}\right)\to w_{0}$ for some $w_{0}\in\left[0,\infty\right]$. 
\end{assumption}
$\ $

Assumptions \ref{assu:MT-kappa_n_sigma_n=000020mu_n=000020convergence}
and \ref{assu:MT-bn=000020div=000020nh=000020converges} are innocuous
because, to establish the uniform inference results in Theorem \ref{thm:MT-Cor=000020Asy=000020Size},
we show that it suffices to consider sequences $\left\{ \l_{n}\in\L_{n}\right\} _{n\geq1}$
that satisfy these assumptions. Note that Theorem \ref{thm:MT-Cor=000020Asy=000020Size}
and Corollary \ref{cor:MT-Med-Unbiased} do not impose Assumption
\ref{assu:MT-kappa_n_sigma_n=000020mu_n=000020convergence} or \ref{assu:MT-bn=000020div=000020nh=000020converges}.

The following lemma bounds the maximum intertemporal difference between
the TVP $\rho_{t}$ and $\rho_{n\tau}$ for $t\in\left[T_{1},T_{2}\right]$. 
\begin{lem}[{{\textbf{Maximum Intertemporal Differences on $\boldsymbol{\left[T_{1},T_{2}\right]}$}}}]
\label{lem:MT-Max=000020Intertemporal=000020Difference} Under Assumptions
\ref{assu:MT-h} and \ref{assu:MT-kappa_n_sigma_n=000020mu_n=000020convergence},
for a sequence $\left\{ \lambda_{n}=\left(\rho_{n},\mu_{n},\sigma_{n}^{2},\kappa_{n},b_{n},F_{n}\right)\in\Lambda_{n}\right\} _{n\geq1}$, 
\begin{enumerate}
\item $\max_{t\in\left[T_{1},T_{2}\right]}\abs{\rho_{t}-\rho_{n\tau}}=O\left(h/b_{n}\right)$, 
\item $\max_{t\in\left[T_{1},T_{2}\right]}\abs{\sigma_{t}^{2}-\sigma_{n\tau}^{2}}=O\left(h\right)$, 
\item $\max_{t\in\left[T_{1},T_{2}\right]}\max_{0\leq j\leq t-T_{1}}\abs{c_{t,j}-\rho_{n\tau}^{j}}=O\left(nh^{2}/b_{n}\right)$,
$\text{ and}$ 
\item $\max_{t\in\left[T_{1},T_{2}\right]}\abs{\mu_{t}-\mu_{n\tau}}=O\left(h/b_{n}\right)$. 
\end{enumerate}
\end{lem}
The next lemma provides several bounds on the endogenous initial condition
$Y_{T_{0}}^{*}$. 
\begin{lem}[\textbf{Order of $\boldsymbol{Y_{T_{0}}^{*}}$}]
\label{lem:MT-Order=000020of=000020Y_T0} For a sequence $\left\{ \lambda_{n}=\left(\rho_{n},\mu_{n},\sigma_{n}^{2},\kappa_{n},b_{n},F_{n}\right)\in\Lambda_{n}\right\} _{n\geq1}$,
we have 
\begin{enumerate}
\item $Y_{T_{0}}^{*}=O_{p}\left(n^{1/2}\right)$ under Assumption \ref{assu:MT-h}, 
\item $Y_{T_{0}}^{*}=o_{p}\left(b_{n}/\left(nh\right){}^{1/2}\right)$ under
Assumptions \ref{assu:MT-h}, \ref{assu:MT-stationary-nh^5=000020->=0000200},
\ref{assu:MT-kappa_n_sigma_n=000020mu_n=000020convergence}, and \ref{assu:MT-bn=000020div=000020nh=000020converges}
and $nh/b_{n}\to r_{0}=0$, and 
\item $Y_{T_{0}}^{*}=O_{p}\left(b_{n}^{1/2}\right)$ under Assumptions \ref{assu:MT-h}
and \ref{assu:MT-stationary-nh^5=000020->=0000200} and $nh/b_{n}\to r_{0}=\infty$. 
\end{enumerate}
\end{lem}
\begin{rem}
Lemma \ref{lem:MT-Order=000020of=000020Y_T0}(a) shows that part (v)
of $\L_{n}$ implies that the endogenous initial condition $Y_{T_{0}}^{*}$
is $O_{p}\left(n^{1/2}\right)$. Part (b) of the lemma is used in
the ``very local-to-unity'' case in which $nh/b_{n}\to r_{0}=0$.
It provides a better bound than part (a) when $b_{n}/\left(nh^{1/2}\right)\to w_{0}<\infty$
because, in that case, the bound $o_{p}\left(b_{n}/\left(nh\right){}^{1/2}\right)$
is $o_{p}\left(n^{1/2}\right)$. Part (c) of the lemma provides a
better bound than part (a) in the ``stationary'' case in which $nh/b_{n}\to r_{0}=\infty$. 
\end{rem}

\subsection{\protect\label{sec:MT-Local-to-Unit-root} Asymptotics in the Local-to-Unity
Case}

The local-to-unity case is characterized by $nh/b_{n}\rightarrow r_{0}\in\left[0,\infty\right)$.
In this section, we determine the asymptotic distributions of the
local LS estimator $\widehat{\rho}_{n\tau}$ and corresponding t-statistic
$T_{n}\left(\rho_{0,n}\right)$ in the local-to-unity case.

Define 
\begin{equation}
t\left(s\right):=t_{n,\tau}\left(s\right)=T_{1}+\lfloor nhs\rfloor=\lfloor n\tau\rfloor-\lfloor nh/2\rfloor+\lfloor nhs\rfloor\label{eq:MT-def=000020t(s)}
\end{equation}
as a function of $s\in\left[0,1\right]$ for any fixed $\tau\in\left[0,1\right]$.
First, we obtain the asymptotic distribution of the zero-initial condition
process $\left(nh\right)^{-1/2}Y_{n,t\left(s\right)}^{0}$ for $s\in\left[0,1\right]$.
Then, we obtain the asymptotic distribution of the normalized initial
condition $Y_{T_{0}}^{*}$ in the $r_{0}\in\left(0,\infty\right)$
case. Lastly, we obtain the asymptotic distributions of $\widehat{\rho}_{n\tau}$
and $T_{n}\left(\rho_{0,n}\right)$. The last step includes dealing
with the initial condition $Y_{T_{0}}^{*}$ in the $r_{0}=0$ case. 
\begin{lem}[\textbf{Asymptotic Distribution of $\boldsymbol{Y_{n,t\left(s\right)}^{0}}$}]
\label{lem:MT-y0limit} Under Assumptions \ref{assu:MT-h} and \ref{assu:MT-kappa_n_sigma_n=000020mu_n=000020convergence},
for a sequence $\left\{ \lambda_{n}=\left(\rho_{n},\mu_{n},\sigma_{n}^{2},\kappa_{n},b_{n},F_{n}\right)\in\Lambda_{n}\right\} _{n\geq1}$
and $nh/b_{n}\rightarrow r_{0}\in\left[0,\infty\right)$ $\left(\text{local-to-unity case}\right)$,
we have 
\[
\left(nh\right)^{-1/2}Y_{n,t\left(s\right)}^{0}/\sigma_{0}\left(\tau\right)\Rightarrow I_{\psi}\left(s\right)\ for\ \psi=r_{0}\kappa_{0}(\tau),
\]
where $t\left(s\right)$ is defined in \eqref{eq:MT-def=000020t(s)},
$I_{\psi}\left(s\right)$ is defined in \eqref{eq:MT-I_D,tau(r)=000020def},
and ``$\Rightarrow$'' denotes weak convergence with respect to
the Skorohod metric.
\end{lem}
The following lemma is used to determine the effect of the initial
condition on the LS estimator and t-statistic when $r_{0}\in(0,\infty).$ 
\begin{lem}[\textbf{Asymptotic Distribution of the Initial Condition $\boldsymbol{Y_{T_{0}}^{*}}$}]
\label{lem:MT-T_0=000020Initial=000020Condition=000020Asy=000020Distn=000020Lem}\textbf{
}Under Assumptions \ref{assu:MT-h} and \ref{assu:MT-kappa_n_sigma_n=000020mu_n=000020convergence}\emph{,
}for a sequence $\{\lambda_{n}=(\rho_{n},\mu_{n},\sigma_{n}^{2},\kappa_{n},b_{n},F_{n})\in\Lambda_{n}\}_{n\geq1}$
and $nh/b_{n}\rightarrow r_{0}\in(0,\infty)$ \emph{(}local-to-unity
case\emph{), }we have 
\[
(2\psi/nh)^{1/2}Y_{T_{0}}^{\ast}/\sigma_{0}\left(\tau\right)\rightarrow_{d}Z_{1}\sim N(0,1),
\]
where $\psi:=r_{0}\kappa_{0}(\tau),$ $Z_{1}$ is independent of $B(\cdot),$
$I_{\psi}(\cdot)$ defined in \eqref{eq:MT-I_D,tau(r)=000020def}\emph{,}
and the convergence holds jointly with the convergence in Lemma \ref{lem:MT-y0limit}\emph{.} 
\end{lem}
The following lemma is useful in determining the asymptotic properties
of $\widehat{\rho}_{n\tau}$ in the local-to-unity case. 
\begin{lem}[\textbf{Convergence of Components in the Local-to-Unity Case}]
\label{lem:MT-asympdist_components_local=000020to=000020unity=000020case}
Under the null hypothesis $H_{0}:\rho_{n\tau}=\rho_{0,n}$, Assumptions
\ref{assu:MT-h} and \ref{assu:MT-kappa_n_sigma_n=000020mu_n=000020convergence},
and $nh/b_{n}\rightarrow r_{0}\in\left(0,\infty\right)$ $\left(\text{local-to-unity case}\right)$,
for a sequence $\left\{ \lambda_{n}=\left(\rho_{n},\mu_{n},\sigma_{n}^{2},\kappa_{n},b_{n},F_{n}\right)\in\Lambda_{n}\right\} _{n\geq1}$
and $\psi=r_{0}\kappa_{0}\left(\tau\right)$, the following results
hold jointly 
\begin{enumerate}
\item $\left(nh\right)^{-1/2}Y_{t\left(s\right)}/\sigma_{0}\left(\tau\right)\Rightarrow I_{\psi}^{*}\left(s\right),$
where $t\left(s\right)=T_{1}+\lfloor nhs\rfloor$ for $s\in\left[0,1\right],$ 
\item $\left(nh\right)^{-3/2}\sum_{t=T_{1}}^{T_{2}}Y_{t-1}/\sigma_{0}\left(\tau\right)\dto\int_{0}^{1}I_{\psi}^{*}\left(s\right)ds,$ 
\item $\left(nh\right)^{-2}\sum_{t=T_{1}}^{T_{2}}Y_{t-1}^{2}/\sigma_{0}^{2}\left(\tau\right)\dto\int_{0}^{1}I_{\psi}^{*2}\left(s\right)ds,$ 
\item $\left(nh\right)^{-1/2}\sum_{t=T_{1}}^{T_{2}}U_{t}\sigma_{t}/\sigma_{0}\left(\tau\right)\dto\int_{0}^{1}dB\left(s\right),$ 
\item $\left(nh\right)^{-1}\sum_{t=T_{1}}^{T_{2}}Y_{t-1}U_{t}\sigma_{t}/\sigma_{0}^{2}\left(\tau\right)\dto\int_{0}^{1}I_{\psi}^{*}\left(s\right)dB\left(s\right)$, 
\item $\left(nh\right)^{-2}\sum_{t=T_{1}}^{T_{2}}\left(Y_{t-1}^{0}/\sigma_{0}\left(\tau\right)\right)^{2}\dto\int_{0}^{1}I_{\psi}^{2}\left(s\right)ds$,
and 
\item when $r_{0}=0$, parts \textup{(a)--(c)} and \textup{(e)} hold with
$Y_{t\left(s\right)}$ \textup{(}$=\mu_{t\left(s\right)}+Y_{t\left(s\right)}^{0}+c_{t\left(s\right),t\left(s\right)-T_{0}}Y_{T_{0}}^{*}$\textup{)}
and $Y_{t-1}$ replaced by $\mu_{t\left(s\right)}+Y_{t\left(s\right)}^{0}$
and $\mu_{t-1}+Y_{t-1}^{0}$, respectively. 
\end{enumerate}
\end{lem}
After proper re-scaling, we have 
\begin{equation}
nh\left(\widehat{\rho}_{n\tau}-\rho_{0,n}\right)=\dfrac{\left(nh\right)^{-1}\sum_{t=T_{1}}^{T_{2}}\left(Y_{t-1}-\ol Y_{nh,-1}\right)\left(Y_{t}-\rho_{0,n}Y_{t-1}\right)}{\left(nh\right)^{-2}\sum_{t=T_{1}}^{T_{2}}\left(Y_{t-1}-\ol Y_{nh,-1}\right)^{2}}.\label{eq:MT-normalized=000020rho}
\end{equation}

The next theorem gives the asymptotic distribution of $nh\left(\widehat{\rho}_{n\tau}-\rho_{0,n}\right)$
and the t-statistic $T_{n}\left(\rho_{0,n}\right)$ in the local-to-unity
case. 
\begin{thm}[\textbf{Asymptotic Distribution of Normalized $\boldsymbol{\widehat{\rho}_{n\tau}}$
and $\boldsymbol{t}$-Statistic in the Local-to-Unity Case}]
\label{thm:MT-asym_rho_hat_local=000020to=000020unity=000020case}Under
the null hypothesis $H_{0}:\rho_{n\tau}=\rho_{0,n}$, Assumptions
\ref{assu:MT-h}, \ref{assu:MT-stationary-nh^5=000020->=0000200},
and \ref{assu:MT-kappa_n_sigma_n=000020mu_n=000020convergence}, $nh/b_{n}\rightarrow r_{0}\in\left[0,\infty\right)$
$\left(\text{local-to-unity case}\right)$, and Assumption \ref{assu:MT-bn=000020div=000020nh=000020converges}
if $r_{0}=0$, for a sequence $\left\{ \lambda_{n}=\left(\rho_{n},\mu_{n},\sigma_{n}^{2},\kappa_{n},b_{n},F_{n}\right)\in\Lambda_{n}\right\} _{n\geq1}$
and $\psi=r_{0}\kappa_{0}\left(\tau\right)$, we have 
\[
nh\left(\widehat{\rho}_{n\tau}-\rho_{0,n}\right)\dto\left(\int_{0}^{1}I_{D,\psi}^{*2}\left(s\right)ds\right)^{-1}\int_{0}^{1}I_{D,\psi}^{*}\left(s\right)dB\left(s\right)
\]
and 
\[
T_{n}\left(\rho_{0,n}\right)\dto\left(\int_{0}^{1}I_{D,\psi}^{*2}\left(s\right)ds\right)^{-1/2}\int_{0}^{1}I_{D,\psi}^{*}\left(s\right)dB\left(s\right).
\]
\end{thm}
\begin{rem}
\label{rem:MT-subseq=000020loc=000020unity=000020thm}For any subsequence
$\{p_{n}\}_{n\geq1}$ of $\{n\}_{n\geq1}$, Lemmas \ref{lem:MT-Max=000020Intertemporal=000020Difference}--\ref{lem:MT-asympdist_components_local=000020to=000020unity=000020case}
and Theorem \ref{thm:MT-asym_rho_hat_local=000020to=000020unity=000020case}
hold with $p_{n}$ in place of $n$ throughout and $h_{p_{n}}$ in
place of $h=h_{n}.$ 
\end{rem}

\subsection{\protect\label{sec:MT-Stationary=000020Case}Asymptotics in the Stationary
Case}

The ``stationary'' case is characterized by $b_{n}=o\left(nh\right)$,
or equivalently, $nh/b_{n}\rightarrow r_{0}=\infty.$ The stationary
case is defined such that the t-statistic has a standard normal distribution
under $H_{0}:\rho_{n\tau}=\rho_{0,n}$ in the stationary case. If
$b_{n}$ is bounded, then $\rho_{n\tau}\leq C_{\rho}$ for some constant
$C_{\rho}<1$ for all $n\geq1$. If $b_{n}$ diverges to infinity,
then $\rho_{n\tau}$ goes to one at a rate slower than $1/nh$. Thus,
the stationary case includes some scenarios where $\rho_{n\tau}$
goes to one. 

Define 
\begin{equation}
\ol{\rho}_{n}:=\max\left\{ \exp\left\{ -\kappa_{0}\left(\tau\right)/\left(2b_{n}\right)\right\} ,-1+\varepsilon_{1}\right\} ,\label{eq:MT-def=000020rho^bar=000020ol}
\end{equation}
where $-1+\varepsilon_{1}$ is a lower bound on $\rho_{t}$ by the
definition of $\Lambda_{n}.$ We can bound $|\rho_{t}|$ for $t\in\left[T_{1},T_{2}\right]$
using $\ol{\rho}_{n}$. First, suppose $\rho_{t}\leq0$, then $-1+\varepsilon_{1}\leq\rho_{t}\leq0$,
which implies that $|\rho_{t}|\leq$$\ol{\rho}_{n}$, as desired.
Given this, in the following calculations we suppose $\rho_{t}\geq0$
for all $t\in[0,1]$ without loss of generality. Then, for $n$ sufficiently
large, we have
\begin{align}
\max_{t\in\left[T_{1},T_{2}\right]}\abs{\rho_{t}} & \leq\max_{t\in\left[T_{1},T_{2}\right]}\abs{\rho_{t}-\rho_{n\tau}}+\abs{\rho_{n\tau}-\exp\left\{ -\kappa_{0}\left(\tau\right)/b_{n}\right\} }+\exp\left\{ -\kappa_{0}\left(\tau\right)/b_{n}\right\} ,\nonumber \\
 & \leq O\left(h/b_{n}\right)+o\left(1\right)/b_{n}+\exp\left\{ -\kappa_{0}\left(\tau\right)/b_{n}\right\} \leq\ol{\rho}_{n},\label{eq:MT-bound=000020max=000020rho_t=000020by=000020rho=000020bar}
\end{align}
where $b_{n}\geq\varepsilon_{3}>0$, the second inequality uses Lemma
\ref{lem:MT-Max=000020Intertemporal=000020Difference}(a), \eqref{eq:MT-rho_n},
Assumption \ref{assu:MT-kappa_n_sigma_n=000020mu_n=000020convergence},
and a mean value expansion of the $\exp\left\{ \cd\right\} $ function,
and the last inequality holds using $\kappa_{0}(\tau)>0$ and the
fact that when $b_{n}\rightarrow\infty$, $\overline{\rho}_{n}-\exp\{-\kappa_{0}(\tau)/b_{n}\}\geq K/b_{n}$
for some constant $K>0.$

Equation \eqref{eq:MT-bound=000020max=000020rho_t=000020by=000020rho=000020bar}
implies
\begin{equation}
\abs{c_{t,j}}=\abs{\prod_{k=0}^{j-1}\rho_{t-k}}\leq\overline{\rho}_{n}^{j}\ \text{for }j=1,...,t-T_{1}+1\ \text{and }t=T_{1},...,T_{2}.\label{eq:MT-c_tj=000020bounded-gto_1}
\end{equation}
Furthermore, using Lemma A.1 of \citet*{GIRAITISKapetaniosYates2014rckernel}
and the definition of the parameter space $\Lambda_{n}$, under $H_{0}:\rho_{n\tau}=\rho_{0,n}$,
we have 
\begin{equation}
\abs{c_{t,j}-\rho_{0,n}^{j}}\leq j\ol{\rho}_{n}^{j-1}\max_{k\in\left[T_{1},T_{2}\right]}\abs{\rho_{k}-\rho_{0,n}}=j\ol{\rho}_{n}^{j-1}O\left(h/b_{n}\right)\ \text{for }j=0,...,t-T_{0}\ \text{and }t=T_{1},...,T_{2},\label{eq:MT-ctj=000020-=000020rho^j=000020bound-gto_1}
\end{equation}
where the equality holds by Lemma \ref{lem:MT-Max=000020Intertemporal=000020Difference}(a). 

The following results are used in the analysis: 
\begin{align}
\sum_{t=0}^{\infty}\ol{\rho}_{n}^{t} & =\left(1-\ol{\rho}_{n}\right)^{-1}=O\left(b_{n}\right),\label{eq:MT-ps=000020rho^t}\\
\sum_{t=0}^{\infty}\ol{\rho}_{n}^{2t} & =\left(1-\ol{\rho}_{n}^{2}\right)^{-1}=O\left(b_{n}\right),\label{eq:MT-ps=000020rho^2t}\\
\sum_{t=0}^{\infty}t\ol{\rho}_{n}^{t} & =\ol{\rho}_{n}\left(1-\ol{\rho}_{n}\right)^{-2}=O\left(b_{n}^{2}\right),\label{eq:MT-ps=000020t*rho^t}\\
\sum_{t=0}^{\infty}t^{2}\ol{\rho}_{n}^{2t} & =\ol{\rho}_{n}^{2}\left(1+\ol{\rho}_{n}^{2}\right)\left(1-\ol{\rho}_{n}^{2}\right)^{-3}=O\left(b_{n}^{3}\right),\text{ and}\label{eq:MT-ps=000020t^2*rho^2t}\\
\sum_{t>s=0}^{nh}\ol{\rho}_{n}^{t-s} & =\sum_{t=1}^{nh}\sum_{s=0}^{t-1}\ol{\rho}_{n}^{t-s}=O\left(nhb_{n}\right).\label{eq:MT-ps=000020rho^(t-s)}
\end{align}
When $\rho_{0,n}=1-\kappa_{n}\left(\tau\right)/b_{n}$ with $b_{n}=o\left(nh\right)$
, one can replace $\ol{\rho}_{n}$ with $\rho_{0,n}$ in \eqref{eq:MT-ps=000020rho^t}--\eqref{eq:MT-ps=000020rho^(t-s)}
and the results still hold.

$\ $

We now establish the $N\left(0,1\right)$ asymptotic distribution
of $\widehat{\rho}_{n\tau}$ in the stationary case with the following
normalization: 
\begin{align}
 & \ \left(1-\rho_{0,n}^{2}\right)^{-1/2}\left(nh\right)^{1/2}\left(\widehat{\rho}_{n\tau}-\rho_{0,n}\right)\nonumber \\
= & \ \dfrac{\left(1-\rho_{0,n}^{2}\right)^{1/2}\left(nh\right)^{-1/2}\sum_{t=T_{1}}^{T_{2}}\left(Y_{t-1}-\ol Y_{nh,-1}\right)\left(Y_{t}-\rho_{0,n}Y_{t-1}\right)}{\left(1-\rho_{0,n}^{2}\right)\left(nh\right)^{-1}\sum_{t=T_{1}}^{T_{2}}\left(Y_{t-1}-\ol Y_{nh,-1}\right)^{2}}.\label{eq:MT-stationary=000020-=000020rho=000020hat-1}
\end{align}

Next, we prove a lemma on the asymptotic properties of the zero-initial
condition process $Y_{t}^{0}$ in the stationary case. 
\begin{lem}[\textbf{Asymptotics in the Stationary Case}]
\label{lem:MT-Asymptotic=000020Properties=000020of=000020the=000020Components=000020Stationary}
Under the null hypothesis $H_{0}:\rho_{n\tau}=\rho_{0,n}$ and Assumptions
\ref{assu:MT-h} and \ref{assu:MT-kappa_n_sigma_n=000020mu_n=000020convergence},
for a sequence $\left\{ \lambda_{n}=\left(\rho_{n},\mu_{n},\sigma_{n}^{2},\kappa_{n},b_{n},F_{n}\right)\in\Lambda_{n}\right\} _{n\geq1}$
where $nh/b_{n}\gto r_{0}=\infty$ and $\kappa_{0}\left(\tau\right)>0$
$\left(\text{stationary case}\right)$, the following results hold
jointly 
\begin{enumerate}
\item $\left(1-\rho_{0,n}^{2}\right)^{1/2}\left(nh\right)^{-1}\sum_{t=T_{1}}^{T_{2}}Y_{t-1}^{0}/\sigma_{0}\left(\tau\right)\pto0,$ 
\item $\left(1-\rho_{0,n}^{2}\right)\left(nh\right)^{-1}\sum_{t=T_{1}}^{T_{2}}\left(Y_{t-1}^{0}/\sigma_{0}\left(\tau\right)\right)^{2}\pto1,$
and 
\item $\left(1-\rho_{0,n}^{2}\right)^{1/2}\left(nh\right)^{-1/2}\sum_{t=T_{1}}^{T_{2}}Y_{t-1}^{0}\sigma_{t}U_{t}/\sigma_{0}^{2}\left(\tau\right)\dto N\left(0,1\right),$ 
\end{enumerate}
where $Y_{t}^{0}$ is defined in \eqref{eq:MT-Y_t^0=000020and=000020Y_t^star}. 
\end{lem}
Lemma \ref{lem:MT-Asymptotic=000020Properties=000020of=000020the=000020Components=000020Stationary}
concerns the zero-initial condition process $Y_{t}^{0}$ with time-varying
autoregressive parameter $\rho_{t}$, which is the foundation of our
asymptotic analysis of $\widehat{\rho}_{n\tau}$. Essentially Lemma
\ref{lem:MT-Asymptotic=000020Properties=000020of=000020the=000020Components=000020Stationary}
says that when $nh/b_{n}\gto r_{0}=\infty$, the asymptotic distributions
of normalized sums of $Y_{t}^{0}$ behave in the same way as in an
AR process with constant $\rho<1$. The proof involves applications
of appropriate weak laws of large numbers and central limit theorems
for martingale difference triangular arrays and approximations of
TVPs. One can extend the results to allow the $U_{t}$ process to
be conditionally heteroskedastic, but this requires a different definition
of $\widehat{\sigma}_{n\tau}^{2}$ and complicates the analysis, see
\citet{andrews2014conditional}. For simplicity, we do not consider
this extension here.

$\ $

Define 
\begin{equation}
\ol{\mu}_{nh}=\left(nh\right)^{-1}\sum_{t=T_{1}}^{T_{2}}\mu_{t}\ \text{and }\ol{\mu}_{nh,-1}=\left(nh\right)^{-1}\sum_{t=T_{1}}^{T_{2}}\mu_{t-1}.\label{eq:MT-def=000020-=000020olmu_nh}
\end{equation}

\begin{lem}[\textbf{Asymptotics of the Denominator in the Stationary Case}]
\label{lem:MT-asympdist_dist_stationary-denom} Under the null hypothesis
$H_{0}:\rho_{n\tau}=\rho_{0,n}$ and Assumptions \ref{assu:MT-h},
\ref{assu:MT-stationary-nh^5=000020->=0000200}, and \ref{assu:MT-kappa_n_sigma_n=000020mu_n=000020convergence},
for a sequence $\left\{ \lambda_{n}=\left(\rho_{n},\mu_{n},\sigma_{n}^{2},\kappa_{n},b_{n},F_{n}\right)\in\Lambda_{n}\right\} _{n\geq1}$
where $nh/b_{n}\gto r_{0}=\infty$ and $\kappa_{0}\left(\tau\right)>0$
$\left(\text{stationary case}\right)$, the following results hold 
\begin{enumerate}
\item $\left(1-\rho_{0,n}^{2}\right)\left(nh\right)^{-1}\sum_{t=T_{1}}^{T_{2}}\left(Y_{t-1}-\ol{\mu}_{nh,-1}\right)^{2}/\sigma_{0}^{2}\left(\tau\right)\pto1$,
and 
\item $\left(1-\rho_{0,n}^{2}\right)\left(\ol Y_{nh,-1}-\ol{\mu}_{nh,-1}\right)^{2}/\sigma_{0}^{2}\left(\tau\right)\pto0$. 
\end{enumerate}
\end{lem}
Lemma \ref{lem:MT-asympdist_dist_stationary-denom} concerns the asymptotic
distribution of the denominator of the normalized $\widehat{\rho}_{n\tau}$
in \eqref{eq:MT-stationary=000020-=000020rho=000020hat-1}. To bound
the difference between $Y_{t}^{0}$ and $Y_{t}$ and control for TVPs
asymptotically, we use various inequalities and approximations in
the proof of Lemma \ref{lem:MT-asympdist_dist_stationary-denom}. 
\begin{lem}[\textbf{Asymptotics of the Numerator in the Stationary Case}]
\label{lem:MT-asympdist_dist_stationary-numerator=000020} Under
the null hypothesis $H_{0}:\rho_{n\tau}=\rho_{0,n}$ and Assumptions
\ref{assu:MT-h}, \ref{assu:MT-stationary-nh^5=000020->=0000200},
and \ref{assu:MT-kappa_n_sigma_n=000020mu_n=000020convergence}, for
a sequence $\left\{ \lambda_{n}=\left(\rho_{n},\mu_{n},\sigma_{n}^{2},\kappa_{n},b_{n},F_{n}\right)\in\Lambda_{n}\right\} _{n\geq1}$
where $nh/b_{n}\gto r_{0}=\infty$ and $\kappa_{0}\left(\tau\right)>0$
$\left(\text{stationary case}\right)$, the following results hold 
\begin{enumerate}
\item $\left(1-\rho_{0,n}^{2}\right)^{1/2}\left(nh\right)^{-1/2}\sum_{t=T_{1}}^{T_{2}}\left(Y_{t-1}-\ol{\mu}_{nh,-1}\right)\left[\begin{array}{c}
Y_{t}-\ol{\mu}_{nh}-\rho_{0,n}\left(Y_{t-1}-\ol{\mu}_{nh,-1}\right)\end{array}\right]/\sigma_{0}^{2}\left(\tau\right)$

$\dto N\left(0,1\right)$, and 
\item $\left(1-\rho_{0,n}^{2}\right)^{1/2}\left(nh\right)^{-1/2}\sum_{t=T_{1}}^{T_{2}}\left(\ol Y_{nh,-1}-\ol{\mu}_{nh,-1}\right)\left[\begin{array}{c}
Y_{t}-\ol{\mu}_{nh}-\rho_{0,n}\left(Y_{t-1}-\ol{\mu}_{nh,-1}\right)\end{array}\right]/\sigma_{0}^{2}\left(\tau\right)$

$\pto0$. 
\end{enumerate}
\end{lem}
Lemma \ref{lem:MT-asympdist_dist_stationary-numerator=000020} concerns
the asymptotic distribution of the numerator of the normalized $\widehat{\rho}_{n\tau}$
in \eqref{eq:MT-stationary=000020-=000020rho=000020hat-1}. The proof
of this lemma uses the results in Lemmas \ref{lem:MT-Order=000020of=000020Y_T0}(c)
and \ref{lem:MT-Asymptotic=000020Properties=000020of=000020the=000020Components=000020Stationary}.

$\ $

The next theorem provides the limit distributions of $\left(1-\rho_{0,n}^{2}\right)^{-1/2}\left(nh\right)^{1/2}\left(\widehat{\rho}_{n\tau}-\rho_{0,n}\right)$
and $T_{n}\left(\rho_{0,n}\right)$ in the stationary $\rho_{n\tau}$
case. 
\begin{thm}[\textbf{Asymptotic Distribution of Normalized $\boldsymbol{\widehat{\rho}_{n\tau}}$
and $\boldsymbol{t}$-Statistic in the Stationary Case}]
\label{thm:MT-asympdist_dist_rho^hat-stationary} Under the null
hypothesis $H_{0}:\rho_{n\tau}=\rho_{0,n}$ and Assumptions \ref{assu:MT-h},
\ref{assu:MT-stationary-nh^5=000020->=0000200}, and \ref{assu:MT-kappa_n_sigma_n=000020mu_n=000020convergence},
for a sequence $\left\{ \lambda_{n}=\left(\rho_{n},\mu_{n},\sigma_{n}^{2},\kappa_{n},b_{n},F_{n}\right)\in\Lambda_{n}\right\} _{n\geq1}$
where $nh/b_{n}\gto r_{0}=\infty$ and $\kappa_{0}\left(\tau\right)>0$
$\left(\text{stationary case}\right)$, we have 
\[
\left(1-\rho_{0,n}^{2}\right)^{-1/2}\left(nh\right)^{1/2}\left(\widehat{\rho}_{n\tau}-\rho_{0,n}\right)\dto N\left(0,1\right)
\]
and 
\[
T_{n}\left(\rho_{0,n}\right)\dto N\left(0,1\right).
\]
\end{thm}
\begin{rem}
\label{rem:MT-subseq=000020staty=000020case}For any subsequence $\{p_{n}\}_{n\geq1}$
of $\{n\}_{n\geq1}$, Lemmas \ref{lem:MT-Asymptotic=000020Properties=000020of=000020the=000020Components=000020Stationary}--\ref{lem:MT-asympdist_dist_stationary-numerator=000020}
and Theorem \ref{thm:MT-asympdist_dist_rho^hat-stationary} hold with
$p_{n}$ in place of $n$ throughout and $h_{p_{n}}$ in place of
$h=h_{n}.$ 
\end{rem}

\subsection{Asymptotic Results for \texorpdfstring{$\boldsymbol{\widehat{h}}$}{h}
\protect\label{subsec:MT-Asymptotic-Results-for-hhat}}

Here we give conditions under which $\widehat{h}$ defined in \eqref{eq:MT-h^hat=000020defn}
is asymptotically equivalent to the value $\widehat{h}_{opt}$ that
minimizes the ``empirical loss,'' which is unobserved. See \citet{li1987asymptotic}
and \citet{andrews1991asymptotic} for analogous results in i.i.d.
models. The empirical loss, $L_{n}(h),$ is 
\begin{equation}
L_{n}(h):=n^{-1}\sum_{t=1}^{n}(\widehat{\mu}_{t-1}(h)-\mu_{t}+Y_{t-1}(\widehat{\rho}_{t-1}(h)-\rho_{t}))^{2}.\label{eq:MT-Empirical=000020loss=000020defn}
\end{equation}
We give a simple high-level condition under which 
\begin{equation}
\frac{L_{n}(\widehat{h})}{L_{n}(\widehat{h}_{opt})}=\frac{L_{n}(\widehat{h})}{\inf_{h\in\mathcal{H}_{n}}L_{n}(h)}\rightarrow_{p}1.\label{eq:MT-Optimality=000020criterion}
\end{equation}

The $FE_{n}(h)$ criterion can be written as follows: 
\begin{align}
FE_{n}(h)=\  & L_{n}(h)-2C_{n}(h)+E_{n},\text{ where }E_{n}:=n^{-1}\sum_{t=1}^{n}\sigma_{t}^{2}U_{t}^{2}\text{ and}\nonumber \\
C_{n}(h):=\  & n^{-1}\sum_{t=1}^{n}\sigma_{t}U_{t}(\widehat{\mu}_{t-1}(h)-\mu_{t}+Y_{t-1}(\widehat{\rho}_{t-1}(h)-\rho_{t})).\label{eq:MT-Decomp=000020os=000020FE=000020criterion}
\end{align}
The empirical loss $L_{n}(h)$ and the cross-product term $C_{n}(h)$
depend on $h,$ but the average squared error $E_{n}$ does not. Since
$E_{n}$ does not depend on $h,$ $\widehat{h}$ minimizes $L_{n}(h)-2C_{n}(h)$
over $\mathcal{H}_{n}.$

Under the following condition, $\widehat{h}$ is asymptotically equivalent
to the infeasible value $\widehat{h}_{opt}$ that minimizes $L_{n}(h).$ 
\begin{assumption}
\label{assu:MT-Asymp_h1}$\sup_{h\in\mathcal{H}_{n}}\frac{|C_{n}(h)|}{L_{n}(h)}\rightarrow_{p}0.$ 
\end{assumption}
\begin{lem}
\label{lem:MT-First=000020Data_Dependent=000020Bandwidth=000020Lem}Under
\emph{Assumption }\ref{assu:MT-Asymp_h1}, $\frac{L_{n}(\widehat{h})}{L_{n}(\widehat{h}_{opt})}\rightarrow_{p}1.$ 
\end{lem}
Now, we give a set of sufficient conditions for Assumption \ref{assu:MT-Asymp_h1}.
Let $h_{\min}=\min\{h:h\in\mathcal{H}_{n}\}$ and $h_{\max}=\max\{h:h\in\mathcal{H}_{n}\}.$
Typically, $h_{\min}$ and $h_{\max}$ depend on $n$ and decrease
to $0$ as $n\rightarrow\infty.$ Let $\xi_{n}$ denote the cardinality
of $\mathcal{H}_{n}.$ We decompose $L_{n}(h)$ and $C_{n}(h)$ into
the main components $L_{2n}(h)$ and $C_{2n}(h),$ respectively, which
depend on $t=nh_{\max}+1,...,n$ and for which $(\widehat{\mu}_{t-1}(h),\widehat{\rho}_{t-1}(h))$
depend only on random variables in $\mathcal{G}_{t-1},$ and the ``small
$t$'' boundary components $L_{1n}(h)$ and $C_{1n}(h),$ respectively,
which depend on $t=1,...,nh_{\max}$ for which $(\widehat{\mu}_{t-1}(h),\widehat{\rho}_{t-1}(h))$
depend on some random variables that are not in $\mathcal{G}_{t-1}.$
Define 
\begin{align}
L_{1n}(h):=\  & n^{-1}\sum_{t=1}^{nh_{\max}}(\widehat{\mu}_{t-1}(h)-\mu_{t}+Y_{t-1}(\widehat{\rho}_{t-1}(h)-\rho_{t}))^{2},\text{ }\nonumber \\
L_{2n}(h):=\  & n^{-1}\sum_{t=nh_{\max}+1}^{n}(\widehat{\mu}_{t-1}(h)-\mu_{t}+Y_{t-1}(\widehat{\rho}_{t-1}(h)-\rho_{t}))^{2},\nonumber \\
C_{1n}(h):=\  & n^{-1}\sum_{t=1}^{nh_{\max}}\sigma_{t}U_{t}(\widehat{\mu}_{t-1}(h)-\mu_{t}+Y_{t-1}(\widehat{\rho}_{t-1}(h)-\rho_{t})),\text{ and}\nonumber \\
C_{2n}(h):=\  & n^{-1}\sum_{t=nh_{\max}+1}^{n}\sigma_{t}U_{t}(\widehat{\mu}_{t-1}(h)-\mu_{t}+Y_{t-1}(\widehat{\rho}_{t-1}(h)-\rho_{t})).\label{eq:MT-Partition=000020of=000020L=000020and=000020C}
\end{align}

The risk as a function of $h$ is $R_{n}(h)=\E L_{n}(h).$ Let $R_{2n}(h)=\E L_{2n}(h).$

The following assumption is sufficient for Assumption \ref{assu:MT-Asymp_h1}. 
\begin{assumption}[\textbf{Sufficient Conditions for Assumption \ref{assu:MT-Asymp_h1}}]
 \label{assu:MT-Assu_h2} 
\end{assumption}
\begin{enumerate}
\item The true sequence of distributions is from the parameter spaces $\{\Lambda_{n}:n\geq1\}.$ 
\item $\sup_{h\in\mathcal{H}_{n}}\frac{|C_{1n}(h)|}{L_{n}(h)}\rightarrow_{p}0.$ 
\item $\sup_{h\in\mathcal{H}_{n}}|\frac{L_{2n}(h)}{R_{2n}(h)}-1|\rightarrow_{p}0.$ 
\item $h_{\max}\leq1-\varepsilon$ for $n$ large for some $\varepsilon>0.$ 
\item $\frac{n\inf_{h\in\mathcal{H}_{n}}R_{2n}(h)}{\xi_{n}}\rightarrow\infty.$ 
\end{enumerate}
Assumption \ref{assu:MT-Assu_h2}(b) implies that the ``small $t$''
boundary component of $C_{n}(h),$ which only depends on $t\leq nh_{\max},$
is asymptotically dominated by $L_{n}(h),$ which is based on all
$t\leq n.$ Assumption \ref{assu:MT-Assu_h2}(c) requires that the
variability of $L_{2n}(h)$ is small relative to its mean. In particular,
Assumption \ref{assu:MT-Assu_h2}(c) holds if $StdDev(L_{2n}(h))/\E L_{2n}(h)=o(1).$
Assumption \ref{assu:MT-Assu_h2}(d) implies that the elements of
$\mathcal{H}_{n}$ are bounded away from one, which implies that $nh=n$
is not a feasible choice.

To interpret Assumption \ref{assu:MT-Assu_h2}(e), we give an intuitive
discussion of the magnitudes of $\inf_{h\in\mathcal{H}_{n}}L_{2n}(h)$
and $\inf_{h\in\mathcal{H}_{n}}R_{2n}(h).$ First, consider the case
where $\{Y_{t}\}_{t\leq n}$ displays unit root or local-to-unity
behavior across the whole time series. Then, $\widehat{\mu}_{t-1}(h)-\mu_{t}\thickapprox(nh)^{-1/2}$
and $Y_{t-1}(\widehat{\rho}_{t-1}(h)-\rho_{t})\thickapprox n^{1/2}(nh)^{-1}=(nh^{2})^{-1/2},$
$L_{2n}(h)\thickapprox(nh^{2})^{-1},$ and under suitable conditions,
$R_{2n}(h)\thickapprox(nh^{2})^{-1}.$ Second, in the case where $\{Y_{t}\}_{t\leq n}$
displays stationary behavior across the whole time series, $\widehat{\mu}_{t-1}(h)-\mu_{t}\thickapprox(nh)^{-1/2},$
$Y_{t-1}(\widehat{\rho}_{t-1}(h)-\rho_{t})\thickapprox(nh)^{-1/2},$
$L_{2n}(h)\thickapprox(nh)^{-1},$ and under suitable conditions $R_{2n}(h)\thickapprox(nh)^{-1}.$
Third, in the case where $\{Y_{t}\}_{t\leq n}$ displays behavior
that varies between unit root and stationarity across the time series,
the order of magnitude of $R_{2n}(h)$ is between $(nh^{2})^{-1}$
and $(nh)^{-1},$ which is bounded below by $(nh_{\max})^{-1}.$ Hence,
Assumption \ref{assu:MT-Assu_h2}(e) requires $n\cdot(nh_{\max})^{-1}/\xi_{n}\rightarrow\infty$
or $\xi_{n}=o(h_{\max}^{-1}).$ That is, the number $\xi_{n}$ of
values $h$ in $\mathcal{H}_{n}$ needs to be of smaller order than
the reciprocal of the maximum value $h_{\max}$ in $\mathcal{H}_{n}.$ 
\begin{lem}
\label{lem:MT-Second=000020Data_Dependent=000020Bandwidth=000020Lem}\emph{Assumption
}\ref{assu:MT-Assu_h2}\emph{ }implies \emph{Assumption }\ref{assu:MT-Asymp_h1}\emph{.} 
\end{lem}
\begin{rem}
\noindent The proof of Lemma \ref{lem:MT-Second=000020Data_Dependent=000020Bandwidth=000020Lem}
shows that $\E C_{2n}(h)=0$ and $Var(C_{2n}(h))\leq C_{3U}\times(n-nh)^{-1}\E L_{2n}(h)$
$\forall h\in\mathcal{H}_{n},n\geq1,$ where $C_{3U}<\infty$ is the
bound on the variance function $\sigma^{2}$ in the definition of
$\Lambda_{n}.$
\end{rem}
\begin{rem}
\noindent Let $R_{1n}(h)=\E L_{1n}(h).$ Under Assumption \ref{assu:MT-Assu_h2}
and $\sup_{h\in\mathcal{H}_{n}}\frac{R_{1n}(h)+\E|C_{1n}(h)|}{R_{n}(h)}\rightarrow0,$
we also have: $\frac{R_{n}(\widehat{h})}{R_{n}(\widehat{h}_{opt})}\rightarrow_{p}1.$ 
\end{rem}
\newpage{}


\appendix

\part*{\centering Supplemental Material}

\global\long\def\thesection{\Alph{section}}%

Section \ref{sec:SM-Critical-Values-J_psi} provides the critical
values $c_{\psi}\left(\a\right)$, the $\a^{\text{th}}$ Quantile
of $J_{\psi}$, for constructing CI's and MUE's for $\rho\left(\tau\right)$.
Section \ref{sec:Theory} provides asymptotic theory. Section \ref{sec:Additional-Simulation-Results}
is concerned with the simulation results. Section \ref{sec:Extension-to-TVP-AR(p)}
extends the methods of the paper for TVP-AR(1) models to TVP-AR(p)
models for $p>1.$ Section \ref{sec:addl-Empirical-Results} provides
additional empirical results.

\numberwithin{equation}{subsection}

\section{Critical Values \texorpdfstring{$\boldsymbol{c_{\psi}\left(\alpha\right)}$}{c}\protect\label{sec:SM-Critical-Values-J_psi}}

Table \ref{tab:SM-Critical-Values_psi} provides the critical values
$c_{\psi}\left(\alpha\right)$ for $\a=.025$, $.05$, $.5$, .95,
and .975 and for $\psi$ between 0 and 500. Given these critical values,
one can compute equal-tailed two-sided CI's and MUE's for $\rho\left(\tau\right)$
based on (\ref{eq:MT-CI=000020defn}) and (\ref{eq:MT-MUE_Def}),
respectively.

$\ $

\global\long\def\thetable{SM.\arabic{table}}%
 \setcounter{table}{0}

\begin{table}[H]
\begin{centering}
\caption{\protect\label{tab:SM-Critical-Values_psi}Values of Relevant Quantiles
of $J_{\psi}$ for Use with 90\% and 95\% Equal-Tailed Two-Sided CI's
and MUE's}
\par\end{centering}
\centering%
\begin{tabular}{>{\centering}m{1.1cm}>{\centering}m{0.7cm}>{\centering}m{0.7cm}>{\centering}m{0.7cm}>{\centering}m{0.7cm}>{\centering}m{0.7cm}>{\centering}m{0.7cm}>{\centering}m{0.7cm}>{\centering}m{0.7cm}>{\centering}m{0.7cm}>{\centering}m{0.7cm}>{\centering}m{0.7cm}>{\centering}m{0.7cm}>{\centering}m{0.7cm}}
\toprule 
\multicolumn{14}{c}{{\scriptsize Values of $c_{\psi}\left(\alpha\right)$, the $\a^{\text{th}}$
Quantile of $J_{\psi}$, for Use with 90\% and 95\% Equal-Tailed Two-Sided
CI's and MUE's}}\tabularnewline
\midrule 
{\scriptsize$\psi$} & {\scriptsize 0} & {\scriptsize 0.2} & {\scriptsize 0.4} & {\scriptsize 0.6} & {\scriptsize 0.8} & {\scriptsize 1} & {\scriptsize 1.4} & {\scriptsize 1.8} & {\scriptsize 2.2} & {\scriptsize 2.6} & {\scriptsize 3} & {\scriptsize 3.4} & {\scriptsize 3.8}\tabularnewline
\midrule
{\scriptsize$c_{\psi}\left(.025\right)$} & {\scriptsize -3.12} & {\scriptsize -3.09} & {\scriptsize -3.05} & {\scriptsize -3.03} & {\scriptsize -2.99} & {\scriptsize -2.98} & {\scriptsize -2.93} & {\scriptsize -2.89} & {\scriptsize -2.85} & {\scriptsize -2.82} & {\scriptsize -2.79} & {\scriptsize -2.77} & {\scriptsize -2.74}\tabularnewline
{\scriptsize$c_{\psi}\left(.05\right)$} & {\scriptsize -2.86} & {\scriptsize -2.83} & {\scriptsize -2.79} & {\scriptsize -2.76} & {\scriptsize -2.72} & {\scriptsize -2.70} & {\scriptsize -2.65} & {\scriptsize -2.61} & {\scriptsize -2.57} & {\scriptsize -2.53} & {\scriptsize -2.51} & {\scriptsize -2.48} & {\scriptsize -2.46}\tabularnewline
{\scriptsize$c_{\psi}\left(.5\right)$} & {\scriptsize -1.57} & {\scriptsize -1.51} & {\scriptsize -1.47} & {\scriptsize -1.42} & {\scriptsize -1.37} & {\scriptsize -1.34} & {\scriptsize -1.26} & {\scriptsize -1.20} & {\scriptsize -1.14} & {\scriptsize -1.08} & {\scriptsize -1.03} & {\scriptsize -1.00} & {\scriptsize -0.96}\tabularnewline
{\scriptsize$c_{\psi}\left(.95\right)$} & {\scriptsize -0.09} & {\scriptsize -0.02} & {\scriptsize 0.03} & {\scriptsize 0.08} & {\scriptsize 0.13} & {\scriptsize 0.17} & {\scriptsize 0.24} & {\scriptsize 0.31} & {\scriptsize 0.37} & {\scriptsize 0.42} & {\scriptsize 0.48} & {\scriptsize 0.53} & {\scriptsize 0.56}\tabularnewline
{\scriptsize$c_{\psi}\left(.975\right)$} & {\scriptsize 0.23} & {\scriptsize 0.30} & {\scriptsize 0.36} & {\scriptsize 0.40} & {\scriptsize 0.45} & {\scriptsize 0.49} & {\scriptsize 0.55} & {\scriptsize 0.63} & {\scriptsize 0.69} & {\scriptsize 0.74} & {\scriptsize 0.79} & {\scriptsize 0.84} & {\scriptsize 0.87}\tabularnewline
\midrule
{\scriptsize$\psi$} & {\scriptsize 4.2} & {\scriptsize 4.6} & {\scriptsize 5} & {\scriptsize 6} & {\scriptsize 7} & {\scriptsize 8} & {\scriptsize 9} & {\scriptsize 10} & {\scriptsize 11} & {\scriptsize 12} & {\scriptsize 13} & {\scriptsize 14} & {\scriptsize 15}\tabularnewline
\midrule
{\scriptsize$c_{\psi}\left(.025\right)$} & {\scriptsize -2.72} & {\scriptsize -2.70} & {\scriptsize -2.68} & {\scriptsize -2.65} & {\scriptsize -2.60} & {\scriptsize -2.58} & {\scriptsize -2.56} & {\scriptsize -2.54} & {\scriptsize -2.51} & {\scriptsize -2.50} & {\scriptsize -2.48} & {\scriptsize -2.46} & {\scriptsize -2.45}\tabularnewline
{\scriptsize$c_{\psi}\left(.05\right)$} & {\scriptsize -2.44} & {\scriptsize -2.41} & {\scriptsize -2.39} & {\scriptsize -2.35} & {\scriptsize -2.31} & {\scriptsize -2.28} & {\scriptsize -2.26} & {\scriptsize -2.23} & {\scriptsize -2.21} & {\scriptsize -2.19} & {\scriptsize -2.18} & {\scriptsize -2.15} & {\scriptsize -2.14}\tabularnewline
{\scriptsize$c_{\psi}\left(.5\right)$} & {\scriptsize -0.92} & {\scriptsize -0.90} & {\scriptsize -0.86} & {\scriptsize -0.81} & {\scriptsize -0.75} & {\scriptsize -0.71} & {\scriptsize -0.68} & {\scriptsize -0.65} & {\scriptsize -0.62} & {\scriptsize -0.59} & {\scriptsize -0.58} & {\scriptsize -0.55} & {\scriptsize -0.54}\tabularnewline
{\scriptsize$c_{\psi}\left(.95\right)$} & {\scriptsize 0.60} & {\scriptsize 0.64} & {\scriptsize 0.68} & {\scriptsize 0.75} & {\scriptsize 0.82} & {\scriptsize 0.86} & {\scriptsize 0.91} & {\scriptsize 0.95} & {\scriptsize 0.98} & {\scriptsize 1.02} & {\scriptsize 1.04} & {\scriptsize 1.05} & {\scriptsize 1.08}\tabularnewline
{\scriptsize$c_{\psi}\left(.975\right)$} & {\scriptsize 0.91} & {\scriptsize 0.94} & {\scriptsize 0.99} & {\scriptsize 1.05} & {\scriptsize 1.12} & {\scriptsize 1.17} & {\scriptsize 1.22} & {\scriptsize 1.25} & {\scriptsize 1.29} & {\scriptsize 1.32} & {\scriptsize 1.34} & {\scriptsize 1.37} & {\scriptsize 1.39}\tabularnewline
\midrule 
{\scriptsize$\psi$} & {\scriptsize 20} & {\scriptsize 25} & {\scriptsize 30} & {\scriptsize 40} & {\scriptsize 50} & {\scriptsize 60} & {\scriptsize 70} & {\scriptsize 80} & {\scriptsize 90} & {\scriptsize 100} & {\scriptsize 200} & {\scriptsize 300} & {\scriptsize 500}\tabularnewline
\midrule
{\scriptsize$c_{\psi}\left(.025\right)$} & {\scriptsize -2.39} & {\scriptsize -2.35} & {\scriptsize -2.32} & {\scriptsize -2.28} & {\scriptsize -2.25} & {\scriptsize -2.23} & {\scriptsize -2.20} & {\scriptsize -2.19} & {\scriptsize -2.17} & {\scriptsize -2.17} & {\scriptsize -2.11} & {\scriptsize -2.08} & {\scriptsize -2.05}\tabularnewline
{\scriptsize$c_{\psi}\left(.05\right)$} & {\scriptsize -2.08} & {\scriptsize -2.05} & {\scriptsize -2.02} & {\scriptsize -1.96} & {\scriptsize -1.94} & {\scriptsize -1.91} & {\scriptsize -1.89} & {\scriptsize -1.88} & {\scriptsize -1.86} & {\scriptsize -1.85} & {\scriptsize -1.80} & {\scriptsize -1.76} & {\scriptsize -1.74}\tabularnewline
{\scriptsize$c_{\psi}\left(.5\right)$} & {\scriptsize -0.47} & {\scriptsize -0.42} & {\scriptsize -0.39} & {\scriptsize -0.33} & {\scriptsize -0.30} & {\scriptsize -0.27} & {\scriptsize -0.25} & {\scriptsize -0.23} & {\scriptsize -0.23} & {\scriptsize -0.21} & {\scriptsize -0.15} & {\scriptsize -0.12} & {\scriptsize -0.09}\tabularnewline
{\scriptsize$c_{\psi}\left(.95\right)$} & {\scriptsize 1.15} & {\scriptsize 1.20} & {\scriptsize 1.24} & {\scriptsize 1.30} & {\scriptsize 1.33} & {\scriptsize 1.37} & {\scriptsize 1.39} & {\scriptsize 1.40} & {\scriptsize 1.41} & {\scriptsize 1.43} & {\scriptsize 1.49} & {\scriptsize 1.52} & {\scriptsize 1.55}\tabularnewline
{\scriptsize$c_{\psi}\left(.975\right)$} & {\scriptsize 1.46} & {\scriptsize 1.51} & {\scriptsize 1.56} & {\scriptsize 1.61} & {\scriptsize 1.65} & {\scriptsize 1.67} & {\scriptsize 1.71} & {\scriptsize 1.72} & {\scriptsize 1.72} & {\scriptsize 1.74} & {\scriptsize 1.80} & {\scriptsize 1.83} & {\scriptsize 1.87}\tabularnewline
\bottomrule
\end{tabular}
\end{table}

\section{Theory\protect\label{sec:Theory}}

\subsection{Weaker Assumptions on \texorpdfstring{$\boldsymbol{h}$}{h} for
the Case where \texorpdfstring{$\boldsymbol{\rho_{n}}$}{rho},
\texorpdfstring{$\boldsymbol{\mu_{n}}$}{mu}, and \texorpdfstring{$\boldsymbol{\sigma_{n}^{2}}$}{sigma}
Are Asymptotically Locally Constant}

The asymptotic results in Section \ref{sec:MT-Asymptotics} rely on
Assumption \ref{assu:MT-stationary-nh^5=000020->=0000200}, which
requires that the bandwidth $h$ is small enough that $nh^{5}\rightarrow0$
as $n\rightarrow\infty.$ This assumption is suitable when the functions
$\rho_{n},$ $\mu_{n},$ and $\sigma_{n}^{2}$ are asymptotically
non-constant in a neighborhood of the time point of interest $\tau.$
However, this condition can be relaxed if the functions $\rho_{n},$
$\mu_{n},$ and $\sigma_{n}^{2}$ are constant or asymptotically constant
in a neighborhood of $\tau.$ In this section, we state an alternative
to Assumption \ref{assu:MT-stationary-nh^5=000020->=0000200} that
imposes weaker conditions on $h$ that depend on how close the functions
$\rho_{n},$ $\mu_{n},$ and $\sigma_{n}^{2}$ are to being asymptotically
constant in a neighborhood of $\tau,$ but still allows us to establish
the results of Theorems \ref{thm:MT-asym_rho_hat_local=000020to=000020unity=000020case}
and \ref{thm:MT-asympdist_dist_rho^hat-stationary}. In the case of
locally constant functions, the condition on $h$ is just $h=o(\ln^{-2}(n)).$

The functions $\rho_{n}(\cdot)$ and $\mu_{n}(\cdot)$ depend on $\kappa_{n}(\cdot)$
and $\eta_{n}(\cdot)$ by the definition of the parameter space $\Lambda_{n}$
following (\ref{eq:MT-structure-of-rho}).
\begin{defn*}
\noindent Let $\ell_{n}$ be the supremum of the Lipschitz constants
of $\kappa_{n}(\cdot),$ $\eta_{n}(\cdot),$ and $\sigma_{n}^{2}(\cdot)$
and the absolute value of the second derivative of $\kappa_{n}(\cdot)$
over the interval $I_{\tau,\varepsilon_{2}}$ that contains $\tau.$
\end{defn*}
\noindent In this definition, $I_{\tau,\varepsilon_{2}}$ is as defined
in the parameter space $\Lambda_{n}$ in Section \ref{sec:MT-Parameter=000020Space}.

As defined, $\ell_{n}$ is a measure of the non-constancy of the functions
$\rho_{n}(\cdot),$ $\mu_{n}(\cdot),$ and $\sigma_{n}^{2}(\cdot).$
When $\ell_{n}\rightarrow0$ as $n\rightarrow\infty,$ these functions
become closer to being constant functions in a neighborhood of $\tau$
as the sample size increases.

\renewcommand{\theassumption}{2*}
\begin{assumption}
\noindent\label{assu:assumption2*}$\textup{(i)}$ $n\ell_{n}^{2}h^{5}\rightarrow0$
and $\textup{(ii)}$ $h=o(\ln^{-2}(n)).$
\end{assumption}
\noindent Assumption \ref{assu:assumption2*}(i) allows $h$ to converge
to zero slower than the condition $nh^{5}\rightarrow0$ when $\ell_{n}\rightarrow0.$
If $n\ell_{n}^{2}\rightarrow0,$ then $h$ can converge to zero as
slowly as $o(\ln^{-2}(n)).$

The following result shows that Assumption \ref{assu:MT-stationary-nh^5=000020->=0000200}
can be replaced by the weaker condition Assumption \ref{assu:assumption2*}
and Theorems \ref{thm:MT-asym_rho_hat_local=000020to=000020unity=000020case}
and \ref{thm:MT-asympdist_dist_rho^hat-stationary} still hold.
\begin{thm}
\noindent\label{thm:ThmA.1_As2star}Theorems \ref{thm:MT-asym_rho_hat_local=000020to=000020unity=000020case}
and \ref{thm:MT-asympdist_dist_rho^hat-stationary} hold with Assumption
\ref{assu:MT-stationary-nh^5=000020->=0000200} replaced by Assumption
\ref{assu:assumption2*}.
\end{thm}
The proof of Theorem \ref{thm:ThmA.1_As2star} is given in Section
\ref{subsec:Pf=000020NEW=000020Thm=000020A.1} below.

\subsection{\protect\label{subsec:Proof-of-Theorem=000020Cor=000020Asy=000020Size}Proof
of Theorem \ref{thm:MT-Cor=000020Asy=000020Size}}

The proof of Theorem \ref{thm:MT-Cor=000020Asy=000020Size} uses the
following lemma (which is also used elsewhere below). 

By definition, $\rho_{n}\left(s\right)=1-\kappa_{n}\left(s\right)/b_{n},$
see \eqref{eq:MT-rho_n}. When $b_{n}\rightarrow\infty,$ for $s\in I_{\tau,\varepsilon_{2}},$
define $\kappa_{n}^{\ast}\left(s\right)$ by 
\begin{equation}
\rho_{n}\left(s\right)=1-\kappa_{n}\left(s\right)/b_{n}=\exp\{-\kappa_{n}^{\ast}\left(s\right)/b_{n}\}.\label{Defn=000020of=000020cappa*_n}
\end{equation}
The function $\kappa_{n}^{\ast}\left(s\right)$ on $I_{\tau,\varepsilon_{2}}$
has the following property when $b_{n}\rightarrow\infty$. 
\begin{lem}
\label{lem:A.1=000020Kappa_n*=000020go=000020to=000020Kappa_n}If
$b_{n}\rightarrow\infty$ and $\kappa_{n}\left(s\right)$ satisfies
part \textup{(ii)} of the parameter space $\Lambda_{n}$ defined following
\eqref{eq:MT-structure-of-rho}, then 
\[
\sup_{s\in I_{\tau,\varepsilon_{2}}}\left\vert \frac{\kappa_{n}^{\ast}\left(s\right)}{\kappa_{n}\left(s\right)}-1\right\vert \rightarrow0.
\]
\end{lem}
The proof of Lemma \ref{lem:A.1=000020Kappa_n*=000020go=000020to=000020Kappa_n}
follows that of Theorem \ref{thm:MT-Cor=000020Asy=000020Size}.
\begin{proof}[\textbf{Proof of Theorem \ref{thm:MT-Cor=000020Asy=000020Size}}]
Let 
\begin{equation}
CP_{n}\left(\lambda\right):=P_{\lambda}\left(\rho\left(\tau\right)\in CI_{n,\tau}\right),\label{eq:CP_lambda=000020defn}
\end{equation}
where $P_{\lambda}\left(\cdot\right)$ denotes probability under $\lambda\in\Lambda_{n}.$
The results of Theorem \ref{thm:MT-Cor=000020Asy=000020Size} hold
by Theorem 2.1 of \citet*{andrews2020generic} (ACG) provided Assumptions
A1 and S of ACG hold with $CP$ in Assumption S equal to $1-\alpha$.
In applying Theorem 2.1 of ACG, we let the parameter space $\Lambda$
in that paper depend on $n$, as it does in the present paper, which
does not cause any complications for the results of that paper. Sufficient
conditions for these assumptions are Assumptions B and S of ACG by
Theorem 2.2 of ACG. Assumptions B and S of ACG combine to require:
For any subsequence $\left\{ p_{n}\right\} _{n\geq1}$ of $\left\{ n\right\} _{n\geq1}$
and any sequence $\left\{ \lambda_{p_{n}}\in\Lambda_{p_{n}}\right\} _{n\geq1}$
\begin{equation}
\text{for which}\ h_{p_{n}}^{*}\left(\lambda_{p_{n}}\right)\rightarrow h^{*}\in H^{*},\ \text{we have}\ CP_{p_{n}}\left(\lambda_{p_{n}}\right)\rightarrow1-\alpha,\label{eq:ACG=000020requirement}
\end{equation}
where $\left\{ h_{n}^{*}\left(\lambda\right)\right\} _{n\geq1}$ is
a suitably chosen sequence of functions.\footnote{The asterisks on $h_{n}^{*}(\lambda_{p_{n}}),h^{*},$ and $H^{*}$
do not appear in ACG. They are added here to avoid confusion with
the smoothing parameter $h$ used in this paper.} In the present case, we take 
\begin{align}
h_{n}^{*}\left(\lambda\right) & :=\left(h_{n,1}^{*}\left(\lambda\right),h_{n,2}^{*}\left(\lambda\right),h_{n,3}^{*}\left(\l\right),h_{n,4}^{*}\left(\lambda\right)\right)',\ \text{where}\nonumber \\
h_{n,1}^{*}\left(\lambda\right) & :=nh\left(1-\rho\left(\tau\right)\right)=nh\left(1-\left(1-\frac{\kappa\left(\tau\right)}{b}\right)\right)=\frac{nh}{b}\kappa\left(\tau\right),\nonumber \\
h_{n,2}^{*}\left(\lambda\right) & :=\frac{nh}{b},\nonumber \\
h_{n,3}^{\ast}\left(\lambda\right) & :=b/\left(nh^{1/2}\right),\ \text{and}\nonumber \\
h_{n,4}^{*}\left(\lambda\right) & :=\left(\kappa,\mu,\sigma^{2}\right),\label{eq:h^asterisk=000020defns}
\end{align}
where $\left(\kappa,\mu,\sigma^{2}\right)$ are viewed as functions
on $I_{\tau,\varepsilon_{2}}$, rather than on $\left[0,1\right]$.
As required by ACG, the functions $\left(\kappa,\mu,\sigma^{2}\right)$
in $h_{n,4}^{*}\left(\lambda\right)$ lie in a compact metric space
$\mathscr{T}$ (under the sup norm) by the definition of $\Lambda_{n}.$
We can write the smoothing parameter $h$ as $h=h_{n}.$ This implies
that for the subsequence $\left\{ p_{n}\right\} _{n\geq1}$, $nh$
becomes $p_{n}h_{p_{n}}.$

The condition $h_{p_{n}}^{*}\left(\lambda_{p_{n}}\right)\rightarrow h^{*}\in H^{*}$
in (\ref{eq:ACG=000020requirement}) implies (i) $\left(\kappa_{p_{n}},\mu_{p_{n}},\sigma_{p_{n}}^{2}\right)\rightarrow\left(\kappa_{0},\mu_{0},\sigma_{0}^{2}\right)\in\mathscr{T}$
under the sup norm, which implies Assumption \ref{assu:MT-kappa_n_sigma_n=000020mu_n=000020convergence},
(ii) $p_{n}h_{p_{n}}/b_{p_{n}}\rightarrow r_{0}\in\left[0,\infty\right],$
which is imposed in the subsequence versions of Theorems \ref{thm:MT-asym_rho_hat_local=000020to=000020unity=000020case}
and \ref{thm:MT-asympdist_dist_rho^hat-stationary} with $r_{0}\in\left[0,\infty\right)$
and $r_{0}=\infty$, respectively, (iii) $b_{p_{n}}/\left(p_{n}h_{p_{n}}^{1/2}\right)\to w_{0}\in\left[0,\infty\right]$,
which implies Assumption \ref{assu:MT-bn=000020div=000020nh=000020converges}
and is imposed in the subsequence versions of Lemma \ref{lem:MT-Order=000020of=000020Y_T0}(b)
and Theorem \ref{thm:MT-asym_rho_hat_local=000020to=000020unity=000020case}
when $r_{0}=0$, and (iv) $p_{n}h_{p_{n}}\left(1-\rho_{p_{n}}\left(\tau\right)\right)=\frac{p_{n}h_{p_{n}}}{b_{p_{n}}}\kappa_{p_{n}}\left(\tau\right)\rightarrow r_{0}\kappa_{0}\left(\tau\right)\in\left[0,\infty\right]$,
which holds by (i) and (ii). In the present case, $h^{*}=\left(r_{0}\kappa_{0}\left(\tau\right),r_{0},w_{0},\left(\kappa_{0},\mu_{0},\sigma_{0}^{2}\right)\right)'$
and $H^{*}=\left[0,\infty\right]\times\left[0,\infty\right]\times\left[0,\infty\right]\times\mathscr{T}.$

By Theorem \ref{thm:MT-asym_rho_hat_local=000020to=000020unity=000020case},
Assumptions \ref{assu:MT-h}, \ref{assu:MT-stationary-nh^5=000020->=0000200},
and \ref{assu:MT-kappa_n_sigma_n=000020mu_n=000020convergence}, and
$p_{n}h_{p_{n}}/b_{p_{n}}\rightarrow r_{0}\in\left[0,\infty\right),$
which are implied by $h_{p_{n}}^{*}\left(\lambda_{p_{n}}\right)\rightarrow h^{*}\in H^{*}$
with $r_{0}\in\left[0,\infty\right),$ we have: $T_{p_{n}}\left(\rho_{p_{n}}\left(\tau\right)\right)\rightarrow_{d}J_{\psi}$
with $\psi=r_{0}\kappa_{0}\left(\tau\right)$ for $J_{\psi}$ defined
in (\ref{eq:MT-Distn=000020J_psi}). By Theorem \ref{thm:MT-asympdist_dist_rho^hat-stationary},
Assumptions \ref{assu:MT-h}, \ref{assu:MT-stationary-nh^5=000020->=0000200},
and \ref{assu:MT-kappa_n_sigma_n=000020mu_n=000020convergence}, and
$p_{n}h_{p_{n}}/b_{p_{n}}\rightarrow r_{0}=\infty,$ which are implied
by $h_{p_{n}}^{*}\left(\lambda_{p_{n}}\right)\rightarrow h^{*}\in H^{*}$
with $r_{0}=\infty,$ we have $T_{p_{n}}\left(\rho_{p_{n}}\left(\tau\right)\right)\rightarrow_{d}J_{\psi}$
for $\psi=\infty$ and $\ J_{\infty}\sim N(0,1).$

For $\psi\in\left[0,\infty\right]$, the quantiles $c_{\psi}\left(\alpha/2\right)$
and $c_{\psi}\left(1-\alpha/2\right)$ of the distribution of $J_{\psi}$,
which appear in the definition of $CI_{n,\tau},$ are continuous at
all $\psi\in\left[0,\infty\right].$ The proof of this is given in
the proof of Lemma A.7 of ACG with $Z\sim N\left(0,1\right)$ replaced
by $Z=0$ (which simplifies the proof because some terms that need
to be shown to be $o_{p}\left(1\right)$ are immediately $0$ and
the case $\psi=0$ is trivial when $Z=0$). In consequence, under
$\left\{ \lambda_{p_{n}}\in\Lambda_{n}\right\} _{n\geq1}$,
\begin{align}
\psi_{p_{n}h_{p_{n}},\rho_{p_{n}}\left(\tau\right)} & \rightarrow r_{0}\kappa_{0}\left(\tau\right)\ \text{implies\ that}\nonumber \\
c_{\psi_{p_{n}h_{p_{n}},\rho_{p_{n}}\left(\tau\right)}}\left(\alpha/2\right) & \rightarrow c_{r_{0}\kappa_{0}\left(\tau\right)}\left(\alpha/2\right)\ \text{and}\ c_{\psi_{p_{n}h_{p_{n}}\rho_{p_{n}}\left(\tau\right)}}\left(1-\alpha/2\right)\rightarrow c_{r_{0}\kappa_{0}\left(\tau\right)}\left(1-\alpha/2\right).\label{eq:quant=000020conv}
\end{align}

Now, we show that the convergence in the first line of (\ref{eq:quant=000020conv})
holds. For notational simplicity, we replace $p_{n}$ by $n$ in the
proof of this convergence. We consider three cases. Case 1: $r_{0}\in\left[0,\infty\right)$.
Case 2: (i) $r_{0}=\infty$ and (ii) $\rho_{n}\left(\tau\right)>0$
for $n$ sufficiently large. Case 3: (i) $r_{0}=\infty$ and (ii)
$\rho_{n}\left(\tau\right)\leq0$ infinitely often as $n\rightarrow\infty.$
In case 1, $nh/b_{n}\rightarrow r_{0}<\infty,$ and so, $b_{n}\rightarrow\infty$,
Lemma \ref{lem:A.1=000020Kappa_n*=000020go=000020to=000020Kappa_n}
applies, and $\rho_{n}\left(\tau\right)>0$ for $n$ sufficiently
large. Thus, we have
\begin{align}
\psi_{p_{n}h_{p_{n}},\rho_{p_{n}}\left(\tau\right)} & =-nh\ln\left(\rho_{n}\left(\tau\right)\right)=-nh\ln\left(\exp\{-\kappa_{n}^{\ast}\left(\tau\right)/b_{n}\}\right)\nonumber \\
 & =nh\kappa_{n}\left(\tau\right)\left(1+o\left(1\right)\right)/b_{n}\rightarrow r_{0}\kappa_{0}\left(\tau\right),\label{eq:Thm7.1-case1}
\end{align}
where the second equality holds by (\ref{Defn=000020of=000020cappa*_n}),
the third equality holds by Lemma \ref{lem:A.1=000020Kappa_n*=000020go=000020to=000020Kappa_n},
and the convergence holds by $h_{n}^{*}\left(\lambda_{n}\right)\rightarrow h^{*}\in H^{*}$.
Thus, case 1 is proved.

For case 2, by implication (iv) listed above, we have $d_{n}:=nh\left(1-\rho_{n}\left(\tau\right)\right)\rightarrow r_{0}\kappa_{0}\left(\tau\right)=\infty.$
We have $\rho_{n}\left(\tau\right)=1-d_{n}/nh$ and $\psi_{p_{n}h_{p_{n}},\rho_{p_{n}}\left(\tau\right)}=-nh\ln\left(1-d_{n}/nh\right)$
for $n$ large by the definition of $\psi_{p_{n}h_{p_{n}},\rho_{p_{n}}\left(\tau\right)}$
and the assumption that $\rho_{n}\left(\tau\right)>0$ for $n$ sufficiently
large by condition (ii) of case 2. For all $d\in(0,\infty)$, we have
\begin{equation}
\liminf_{n\rightarrow\infty}\psi_{p_{n}h_{p_{n}},\rho_{p_{n}}\left(\tau\right)}=\liminf_{n\to\infty}\left[-nh\ln\left(1-d_{n}/nh\right)\right]\geq\liminf_{n\to\infty}\left[-nh\ln\left(1-d/nh\right)\right]\label{eq:Thm7.1-case2}
\end{equation}
because the right-hand side quantity is increasing in $d$ and $d_{n}\rightarrow\infty.$
By a mean value expansion around $1,$ $\ln\left(1-d/nh\right)=(1/d_{n*})(-d/nh)$,
where $d_{n*}\in[1-d/nh,1]$ and $d_{n*}\rightarrow1.$ Hence,
\begin{equation}
\liminf_{n\to\infty}\left[-nh\ln\left(1-d/nh\right)\right]\rightarrow d.\label{eq:Thm7.1-case2-2}
\end{equation}

Since this holds for all $d\in(0,\infty)$, using (\ref{eq:Thm7.1-case2}),
we get $\liminf_{n\rightarrow\infty}\psi_{p_{n}h_{p_{n}},\rho_{p_{n}}\left(\tau\right)}=\infty=r_{0}\kappa_{0}\left(\tau\right),$
as desired and case 2 is proved.

For case 3, for the subsequence of indices for which $\rho_{n}\left(\tau\right)\leq0$,
we have $\psi_{p_{n}h_{p_{n}},\rho_{p_{n}}\left(\tau\right)}=\infty$
for all indices, and so, its limit equals $\infty=r_{0}\kappa_{0}\left(\tau\right),$
as desired. For the subsequence of indices for which $\rho_{n}\left(\tau\right)>0$,
the argument used to prove case 2 applies and the limit is $\infty=r_{0}\kappa_{0}\left(\tau\right).$
Thus, case 3 is proved.

Now, given the convergence in the second line of (\ref{eq:quant=000020conv}),
we have 
\begin{align}
 & \hspace{6mm}CP_{p_{n}}\left(\lambda_{p_{n}}\right)\nonumber \\
 & =P_{\lambda_{p_{n}}}\left(\rho_{p_{n}}\left(\tau\right)\in CI_{p_{n},\tau}\right)\nonumber \\
 & =P_{\lambda_{p_{n}}}\left(c_{\psi_{p_{n}h_{p_{n}},\rho_{p_{n}}\left(\tau\right)}}\left(\alpha/2\right)\leq T_{p_{n}}\left(\rho_{p_{n}}\left(\tau\right)\right)\leq c_{\psi_{p_{n}h_{p_{n}},\rho_{p_{n}}\left(\tau\right)}}\left(1-\alpha/2\right)\right)\nonumber \\
 & \rightarrow P\left(c_{r_{0}\kappa_{0}\left(\tau\right)}\left(\alpha/2\right)\leq J_{r_{0}\kappa_{0}\left(\tau\right)}\leq c_{r_{0}\kappa_{0}\left(\tau\right)}\left(1-\alpha/2\right)\right)\nonumber \\
 & =1-\text{\ensuremath{\a}},\label{eq:CP=000020conv}
\end{align}
where the second equality holds by (\ref{eq:MT-CI=000020defn}), the
convergence holds by $T_{p_{n}}\left(\rho_{p_{n}}\left(\tau\right)\right)\rightarrow_{d}J_{\psi}$
with $\psi=r_{0}\kappa_{0}\left(\tau\right)$ and (\ref{eq:quant=000020conv}),
and the last equality holds by the definition of the quantile $c_{\psi}\left(\alpha\right)$
following (\ref{eq:MT-Distn=000020J_psi}). This verifies (\ref{eq:ACG=000020requirement})
and completes the verification of Assumptions B and S of ACG with
$CP=1-\alpha$ in Assumption S, which completes the proof. 
\end{proof}
\begin{proof}[\textbf{Proof of Lemma \ref{lem:A.1=000020Kappa_n*=000020go=000020to=000020Kappa_n}}]
Because $b_{n}\rightarrow\infty$ and $\kappa_{n}\left(\cd\right)$
is nonnegative and bounded on $I_{\tau,\varepsilon_{2}}$ uniformly
over $n$ by part (ii) of $\Lambda_{n},$ 
\begin{equation}
\sup_{s\in I_{\tau,\varepsilon_{2}}}|\rho_{n}\left(s\right)-1|=\sup_{s\in I_{\tau,\varepsilon_{2}}}\frac{\kappa_{n}\left(s\right)}{b_{n}}\rightarrow0.\label{1st=000020eqn}
\end{equation}

\noindent We prove by contradiction that 
\begin{equation}
\sup_{s\in I_{\tau,\varepsilon_{2}}}\frac{\kappa_{n}^{\ast}\left(s\right)}{b_{n}}\rightarrow0.\label{2nd=000020eqn}
\end{equation}
Let $\{s_{n}\}_{n\geq1}$ be a sequence in $I_{\tau,\varepsilon_{2}}$
such that $\sup_{s\in I_{\tau,\varepsilon_{2}}}\kappa_{n}^{\ast}\left(s\right)/b_{n}-\kappa_{n}^{\ast}\left(s_{n}\right)/b_{n}\rightarrow0.$
Suppose the claim does not hold. Then, there exists a subsequence
$\{n_{k}\}_{k\geq1}$ of $\{n\}_{n\geq1}$ and a constant $\delta>0$
such that $\sup_{s\in I_{\tau,\varepsilon_{2}}}\kappa_{n_{k}}^{\ast}\left(s_{n_{k}}\right)/b_{n_{k}}\geq\delta$
$\forall k\geq1,$ and hence, $\kappa_{n_{k}}^{\ast}\left(s_{n_{k}}\right)/b_{n_{k}}\geq\delta/2$
for all $k$ large. In consequence, 
\begin{equation}
\rho_{n_{k}}\left(s_{n_{k}}\right):=\exp\{-\kappa_{n_{k}}^{\ast}\left(s_{n_{k}}\right)/b_{n_{k}}\}\leq\exp\{-\delta/2\}<1\label{3rd=000020eqn}
\end{equation}
for all $k$ large, where the equality holds by the definition of
$\kappa_{n}^{\ast}\left(s\right)$ in (\ref{Defn=000020of=000020cappa*_n}),
which contradicts (\ref{1st=000020eqn}) and establishes (\ref{2nd=000020eqn}).

For each $s\in I_{\tau,\varepsilon_{2}},$ a mean value expansion
gives 
\begin{equation}
\exp\{-\kappa_{n}^{\ast}\left(s\right)/b_{n}\}=1-\exp\{-\kappa_{n}^{\ast\ast}\left(s\right)/b_{n}\}\frac{\kappa_{n}^{\ast}\left(s\right)}{b_{n}},\label{4th=000020eqn}
\end{equation}
where $\kappa_{n}^{\ast\ast}\left(s\right)$ lies between $\kappa_{n}^{\ast}\left(s\right)$
and $0.$ The latter and (\ref{2nd=000020eqn}) give: $\sup_{s\in I_{\tau,\varepsilon_{2}}}\kappa_{n}^{\ast\ast}\left(s\right)/b_{n}=o(1).$

We have 
\begin{eqnarray}
0\overset{}{=}\rho_{n}\left(s\right)-\rho_{n}\left(s\right) & = & \exp\{-\kappa_{n}^{\ast}\left(s\right)/b_{n}\}-(1-\kappa_{n}\left(s\right)/b_{n})\nonumber \\
 & = & -\exp\{-\kappa_{n}^{\ast\ast}\left(s\right)/b_{n}\}\frac{\kappa_{n}^{\ast}\left(s\right)}{b_{n}}+\frac{\kappa_{n}\left(s\right)}{b_{n}},\label{5th=000020eqn}
\end{eqnarray}
where the third equality uses (\ref{4th=000020eqn}). Equation (\ref{5th=000020eqn})
gives 
\begin{eqnarray}
 &  & \frac{\kappa_{n}^{\ast}\left(s\right)}{\kappa_{n}\left(s\right)}\overset{}{=}\exp\{\kappa_{n}^{\ast\ast}\left(s\right)/b_{n}\},\text{ and so,}\nonumber \\
 &  & \sup_{s\in I_{\tau,\varepsilon_{2}}}\left\vert \frac{\kappa_{n}^{\ast}\left(s\right)}{\kappa_{n}\left(s\right)}-1\right\vert \overset{}{=}\sup_{s\in I_{\tau,\varepsilon_{2}}}\left\vert \exp\{\kappa_{n}^{\ast\ast}\left(s\right)/b_{n}\}-1\right\vert \overset{}{\rightarrow}0,
\end{eqnarray}
where the convergence holds by a mean value expansion using $\sup_{s\in I_{\tau,\varepsilon_{2}}}\kappa_{n}^{\ast\ast}\left(s\right)/b_{n}=o(1).$
\end{proof}

\subsection{\protect\label{subsec:Proof-of-Lemma-1.1=000020Intertemporal=000020Difference}
Proof of Lemma \ref{lem:MT-Max=000020Intertemporal=000020Difference}}
\begin{proof}[\textbf{Proof of Lemma \ref{lem:MT-Max=000020Intertemporal=000020Difference}}]
First, we prove part (a). By definition, $\left[T_{1},T_{2}\right]=I_{n\tau,nh/2}$.
Because $\kappa_{n}\left(\cd\right)$ is Lipschitz on $I_{\tau,\varepsilon_{2}}$,
we have 
\begin{equation}
\max_{t\in I_{n\tau,nh/2}}\abs{\kappa_{n}\left(t/n\right)-\kappa_{n}\left(\tau\right)}\leq L_{4}\max_{t\in I_{n\tau,nh/2}}\abs{t/n-\tau}=O\left(h\right),\label{eq:kappanapprox_a.5.8}
\end{equation}
where the inequality holds for $n$ sufficiently large such that $h/2+1/n\leq\varepsilon_{2}$
(because then $t\in I_{n\tau,nh/2}$ implies that $t/n\in I_{\tau,\varepsilon_{2}}$).

Next, we have 
\begin{align}
\max_{t\in I_{n\tau,nh/2}}\abs{\rho_{t}-\rho_{n\tau}} & =\max_{t\in I_{n\tau,nh/2}}\abs{\kappa_{n}\left(t/n\right)-\kappa_{n}\left(\tau\right)}/b_{n}\nonumber \\
 & =O\left(h/b_{n}\right),\label{eq:rho_t-rho_ntau_A.5.9}
\end{align}
where the first equality holds because $\rho_{t}=\rho_{n}(t/n)=1-\kappa_{n}(t/n)/b_{n}$
and $\rho_{n\tau}=\rho_{n}(\tau)=1-\kappa_{n}(\tau)/b_{n}$ by \eqref{eq:MT-structure-of-rho}
and the second equality uses \eqref{eq:kappanapprox_a.5.8}. This
establishes part (a).

Part (b) holds by (\ref{eq:kappanapprox_a.5.8}) with $\sigma_{n}^{2}\left(\cd\right)$
in place of $\kappa_{n}\left(\cd\right)$.

Using \eqref{eq:MT-def_c_t,j}, we have 
\begin{align}
\max_{t\in I_{n\tau,nh/2}}\max_{0\leq j\leq t-T_{1}}\abs{c_{t,j}-\rho_{n\tau}^{j}} & :=\max_{t\in I_{n\tau,nh/2}}\max_{0\leq j\leq t-T_{1}}\abs{\prod_{k=1}^{j}\rho_{t-k}-\rho_{n\tau}^{j}}\nonumber \\
 & \leq\max_{t\in I_{n\tau,nh/2}}\max_{0\leq j\leq t-T_{1}}\left(j+1\right)\max_{1\leq k\leq j}\abs{\rho_{t-k}-\rho_{n\tau}}\nonumber \\
 & =O\left(nh\cd h/b_{n}\right),\label{eq:c_t,j=000020-=000020pho_ntauj-A.5.10}
\end{align}
where the first equality holds by \eqref{eq:MT-def_c_t,j}, the inequality
uses standard manipulations and $\max\left(\abs{\rho_{t-k}},\abs{\rho_{n\tau}}\right)\leq1$,
and the second equality uses part (a) and $j+1\leq T_{2}-T_{1}=O\left(nh\right)$.
Hence, part (c) holds.

Part (d) holds by \eqref{eq:rho_t-rho_ntau_A.5.9} with $\eta_{n}\left(\cd\right)$
in place of $\kappa_{n}\left(\cd\right)$. 
\end{proof}

\subsection{\protect\label{subsec:Proof_of_Lemma_1.2=000020order=000020of=000020Y_T0}Proof
of Lemma \ref{lem:MT-Order=000020of=000020Y_T0}}

The proofs of Lemma \ref{lem:MT-Order=000020of=000020Y_T0}(b) and
Lemma \ref{lem:MT-T_0=000020Initial=000020Condition=000020Asy=000020Distn=000020Lem}
use the following lemma, which is an extension of Lemma \ref{lem:MT-Max=000020Intertemporal=000020Difference}(a)--(c). 
\begin{lem}
\label{Bd=000020on=000020rho=000020deviations=000020Lem}Under Assumptions
\emph{\ref{assu:MT-h} }and \emph{\ref{assu:MT-kappa_n_sigma_n=000020mu_n=000020convergence},
}for a sequence $\{\lambda_{n}=(\rho_{n},\mu_{n},\sigma_{n}^{2},\kappa_{n},b_{n},F_{n})\in\Lambda_{n}\}_{n\geq1}$
and a sequence of integer constants $\{m_{n}\}_{n\geq1}$\emph{\ }for
which $m_{n}\rightarrow\infty$ and $m_{n}/n\rightarrow0,$ 
\begin{enumerate}
\item $\max_{t\in I_{n\tau,nh/2+2m_{n}}}|\rho_{t}-\rho_{n\tau}|=O((h+\frac{m_{n}}{n})/b_{n}),$ 
\item $\max_{t\in I_{n\tau,nh/2+2m_{n}}}|\sigma_{t}^{2}-\sigma_{n\tau}^{2}|=O(h+\frac{m_{n}}{n}),$
and 
\item $\max_{t\in\{T_{0}-m_{n},T_{0}\}}\max_{0\leq j\leq m_{n}}|c_{t,j}-\rho_{n\tau}^{j}|=O(m_{n}(h+\frac{m_{n}}{n})/b_{n}).$ 
\end{enumerate}
\end{lem}
\begin{proof}[\textbf{Proof of Lemma \ref{Bd=000020on=000020rho=000020deviations=000020Lem}}]
 The proof of Lemma \ref{Bd=000020on=000020rho=000020deviations=000020Lem}(a)
is the same as that of Lemma \ref{lem:MT-Max=000020Intertemporal=000020Difference}(a)
with the following changes. In \eqref{eq:kappanapprox_a.5.8}, $I_{n\tau,nh/2}$
is replaced by $I_{n\tau,nh/2+2m_{n}}$ and $L_{4}\max_{t\in I_{n\tau,nh/2}}|t/n-\tau|=O(h)$
is replaced by $L_{4}\max_{t\in I_{n\tau,nh/2+2m_{n}}}|t/n-\tau|=O(h+\frac{m_{n}}{n})$
and the inequality in \eqref{eq:kappanapprox_a.5.8} holds for $n$
sufficiently large that $h/2+2m_{n}/n+1/n\leq\varepsilon_{2},$ which
uses the assumption that $m_{n}/n\rightarrow0$ (because then $t\in I_{n\tau,nh/2+2m_{n}}$
implies that $t/n\in I_{\tau,\varepsilon_{2}}).$ In \eqref{eq:rho_t-rho_ntau_A.5.9},
$I_{n\tau,nh/2}$ is replaced by $I_{n\tau,nh/2+2m_{n}}$ and $\max_{t\in I_{n\tau,nh/2}}\exp\{x_{nt}\}|\kappa_{n}^{\ast}(t/n)-\kappa_{n}^{\ast}(\tau)|/b_{n}=O(h/b_{n})$
is replaced by $\max_{t\in I_{n\tau,nh/2+2m_{n}}}\exp\{x_{nt}\}|\kappa_{n}^{\ast}(t/n)-\kappa_{n}^{\ast}(\tau)|/b_{n}=O((h+\frac{m_{n}}{n})/b_{n})$
using the revised version of \eqref{eq:kappanapprox_a.5.8}. This
establishes Lemma \ref{Bd=000020on=000020rho=000020deviations=000020Lem}(a).

The proof of Lemma \ref{Bd=000020on=000020rho=000020deviations=000020Lem}(b)
is the same as that of part (a) with $|\kappa_{n}(t/n)-\kappa_{n}(\tau)|/b_{n}$
replaced by $|\sigma_{t}^{2}-\sigma_{n\tau}^{2}|.$

The proof of Lemma \ref{Bd=000020on=000020rho=000020deviations=000020Lem}(c)
is the same as that of Lemma \ref{lem:MT-Max=000020Intertemporal=000020Difference}(c)
with $\max_{t\in I_{n\tau,nh/2}}\allowbreak\max_{0\leq j\leq t-T_{1}}$
replaced by $\max_{t\in\{T_{0},T_{0}-m_{n}\}}\max_{0\leq j\leq2m_{n}},$
which implies that the largest value of $j$ considered is bounded
by $2m_{n},$ rather than $nh.$ This implies that the rhs of \eqref{eq:c_t,j=000020-=000020pho_ntauj-A.5.10}
is changed from $O(nh\cdot h/b_{n})$ to $O(m_{n}\cdot(h+\frac{m_{n}}{n})/b_{n}),$
which establishes Lemma \ref{Bd=000020on=000020rho=000020deviations=000020Lem}(c). 
\end{proof}
\begin{proof}[\textbf{Proof of Lemma \ref{lem:MT-Order=000020of=000020Y_T0}}]
First, we prove part (a). By recursive substitution of \eqref{eq:MT-tvp-model},
we have 
\begin{align}
Y_{T_{0}}^{*} & =\sum_{j=0}^{T_{0}-1}c_{T_{0},j}\sigma_{T_{0}-j}U_{T_{0}-j}+c_{T_{0},T_{0}}Y_{0}^{*}.\label{eq:A.6.1-recursive}
\end{align}

By Markov's inequality, we only need to show $\E Y_{T_{0}}^{*2}/n=O\left(1\right)$,
which is true because 
\begin{align}
\E Y_{T_{0}}^{*2}/n= & \ n^{-1}\E\left(\sum_{j=0}^{T_{0}-1}c_{T_{0},j}\sigma_{T_{0}-j}U_{T_{0}-j}+c_{T_{0},T_{0}}Y_{0}^{*}\right)^{2}\nonumber \\
\leq & \ 2n^{-1}\E\left(\sum_{j=0}^{T_{0}-1}c_{T_{0},j}\sigma_{T_{0}-j}U_{T_{0}-j}\right)^{2}+2n^{-1}\E\left(c_{T_{0},T_{0}}Y_{0}^{*}\right)^{2}\nonumber \\
= & \ 2\sum_{j=0}^{T_{0}-1}c_{T_{0},j}^{2}\sigma_{T_{0}-j}^{2}/n+2c_{T_{0},T_{0}}^{2}\E Y_{0}^{*}{}^{2}/n\nonumber \\
\leq & \ 2C_{3,U}\left(T_{0}/n\right)+O(1)=O\left(1\right),\label{eq:=000020bound=000020order=000020of=000020E(Y_T0*)^2=000020O(n)=000020A.6.2}
\end{align}
where the first inequality holds by $\left(a+b\right)^{2}\leq2a^{2}+2b^{2}$,
the second equality uses the fact that $\left\{ U_{t}\right\} _{t=1}^{n}$
is a martingale difference sequence and $\E U_{t}^{2}=1$ by the definition
of the parameter space $\Lambda_{n}$, and the last inequality holds
by $\max_{t\in\left[1,n\right]}\sigma_{t}^{2}\leq C_{3,U}$, $\max_{j\in\left[0,T_{0}\right]}\abs{c_{T_{0},j}}\leq1$,
$T_{0}=\lfloor n\tau\rfloor-\lfloor nh/2\rfloor-1=O\left(n\right)$,
and part (v) of $\Lambda_{n}$.

In consequence, by Markov's inequality, we have 
\begin{equation}
Y_{T_{0}}^{*}=O_{p}\left(n^{1/2}\right),\label{eq:=000020bound=000020order=000020of=000020Y_T0*=000020Op(n^1/2)}
\end{equation}
which proves part (a).

Next, we prove part (b). By part (a), $Y_{T_{0}}^{\ast}=O_{p}(n^{1/2}).$
Hence, if $w_{0}=\infty$ (in Assumption \ref{assu:MT-bn=000020div=000020nh=000020converges}),
then 
\begin{equation}
\frac{(nh)^{1/2}}{b_{n}}Y_{T_{0}}^{\ast}=\frac{nh^{1/2}}{b_{n}}O_{p}(1)=o_{p}(1)\label{r_0=000020=00003D00003D=0000200=000020eqn=0000201}
\end{equation}
and the result of part (b) of the lemma is proved.

Hence, to prove part (b), it remains to consider the case where $w_{0}<\infty.$
For notational simplicity, we suppose $\sigma_{0}(\tau)=1.$ By recursive
substitution, as in \eqref{eq:A.6.1-recursive}, for a sequence of
integer constants $\{m_{n}\}_{n\geq1}$\emph{\ }for which $m_{n}\rightarrow\infty$
and $m_{n}/n\rightarrow0,$ we have 
\begin{equation}
Y_{T_{0}}^{\ast}=\sum_{j=0}^{m_{n}-1}c_{T_{0},j}\sigma_{T_{0}-j}U_{T_{0}-j}+c_{T_{0},m_{n}}Y_{T_{0}-m_{n}}^{\ast}.\label{r_0=000020=00003D00003D=0000200=000020eqn=0000202}
\end{equation}
Similarly to \eqref{eq:MT-bound=000020max=000020rho_t=000020by=000020rho=000020bar},
we bound $|\rho_{t}|$ for $t\in[T_{0}-m_{n},T_{0}]$ by $\overline{\rho}_{n}:=\max\{\exp\{-\kappa_{0}(\tau)/(2b_{n})\},-1+\varepsilon_{1}\}$.
As in \eqref{eq:MT-bound=000020max=000020rho_t=000020by=000020rho=000020bar},
it suffices to consider the case where $\rho_{t}\geq0$ for all $t\in[0,1]$.
We have
\begin{eqnarray}
\max_{t\in[T_{0}-m_{n},T_{0}]}|\rho_{t}|\hspace{-0.08in} & \leq & \hspace{-0.08in}\max_{t\in[T_{0}-m_{n},T_{0}]}|\rho_{t}-\rho_{n\tau}|+|\rho_{n\tau}-\exp\{-\kappa_{0}(\tau)/b_{n}\}|+\exp\{-\kappa_{0}(\tau)/b_{n}\}\nonumber \\
 & \leq & \hspace{-0.08in}O\left((h+\frac{m_{n}}{n})/b_{n}\right)+o(1)/b_{n}+\exp\{-\kappa_{0}(\tau)/b_{n}\}\leq\overline{\rho}_{n},\label{r_0=000020=00003D00003D=0000200=000020eqn=0000203}
\end{eqnarray}
where $b_{n}\geq\varepsilon_{3}>0,$ the second inequality uses Lemma
\ref{Bd=000020on=000020rho=000020deviations=000020Lem}(a), \eqref{eq:MT-rho_n},
Assumption \ref{assu:MT-kappa_n_sigma_n=000020mu_n=000020convergence},
and a mean value expansion, and the last inequality holds using $\kappa_{0}(\tau)>0$
and the fact that, when $b_{n}\rightarrow\infty,$ $\overline{\rho}_{n}-\exp\{-\kappa_{0}(\tau)/b_{n}\}\geq K/b_{n}$
for some constant $K>0.$ Hence, for $j=0,...,m_{n},$ 
\begin{equation}
c_{T_{0},j}\leq\overline{\rho}_{n}^{j}.\label{r_0=000020=00003D00003D=0000200=000020eqn=0000204}
\end{equation}

We have 
\begin{eqnarray}
EY_{T_{0}}^{\ast2}\hspace{-0.08in} & = & \hspace{-0.08in}\sum_{j=0}^{m_{n}-1}c_{T_{0},j}^{2}\sigma_{T_{0}-j}^{2}+c_{T_{0},m_{n}}^{2}EY_{T_{0}-m_{n}}^{\ast2}\leq O(1)\sum_{j=0}^{m_{n}-1}\overline{\rho}_{n}^{2j}+\overline{\rho}_{n}^{2m_{n}}O(n)\nonumber \\
 & \leq & \hspace{-0.08in}O(1)\frac{1}{1-\overline{\rho}_{n}^{2}}+\ensuremath{\omega}{}^{m_{n}/b_{n}}O(n)=O(b_{n})+\ensuremath{\omega}^{m_{n}/b_{n}}O(n),\text{ where}\nonumber \\
\ensuremath{\omega}\hspace{-0.08in} & : & \hspace{-0.13in}=\text{ }\exp\{-\kappa_{0}(\tau)\},\label{r_0=000020=00003D00003D=0000200=000020eqn=0000205}
\end{eqnarray}
the first equality uses (\ref{r_0=000020=00003D00003D=0000200=000020eqn=0000202})
and the martingale difference property of $\{U_{t}\}_{t\geq1},$ the
first inequality uses (\ref{r_0=000020=00003D00003D=0000200=000020eqn=0000204}),
$\max_{t\leq n}\sigma_{t}^{2}\leq C_{3,U}<\infty$ (by part (i) of
$\Lambda_{n}),$ and $EY_{T_{0}-m_{n}}^{\ast2}=O(n)$ (which holds
by \eqref{eq:=000020bound=000020order=000020of=000020E(Y_T0*)^2=000020O(n)=000020A.6.2}
with $T_{0}-m_{n}$ in place of $T_{0}),$ the second inequality holds
by a bound on the geometric sum and uses the definition of $\ensuremath{\omega},$
and the last equality uses \eqref{eq:MT-ps=000020rho^2t}.

We are considering the case where $\frac{b_{n}}{nh^{1/2}}\rightarrow w_{0}<\infty.$
We take $m_{n}=b_{n}c\ln(n)$ for some finite positive constant $c\in(0,-1/\ln(\omega)).$
We have $m_{n}\rightarrow\infty,$ because $b_{n}$ is bounded away
from zero by condition (ii) of $\Lambda_{n}$ and $c$ is positive.
We have 
\begin{equation}
\frac{m_{n}}{n}=\frac{b_{n}}{nh^{1/2}}ch^{1/2}\ln(n)\rightarrow0\text{ provided }h^{1/2}\ln(n)=o(1),\label{Pf=000020Lem=0000207.2(b)=000020eqn=00002010}
\end{equation}
because $\frac{b_{n}}{nh^{1/2}}\rightarrow w_{0}<\infty$ in the case
that we are considering. By Assumption \ref{assu:MT-stationary-nh^5=000020->=0000200},
the condition $h^{1/2}\ln(n)=o(1)$ holds. Thus, $m_{n}$ satisfies
the required conditions that $m_{n}\rightarrow\infty$ and $\frac{m_{n}}{n}\rightarrow0.$

Next, we have 
\begin{equation}
\omega^{m_{n}/b_{n}}n\rightarrow0\text{ iff }\frac{m_{n}}{b_{n}}\ln(\omega)+\ln(n)=(c\ln(\omega)+1)\ln(n)\rightarrow-\infty\label{Pf=000020Lem=0000207.2(b)=000020eqn=00002011}
\end{equation}
and the latter holds because $c\ln(\omega)+1<0$ by the definition
of $c.$

By (\ref{r_0=000020=00003D00003D=0000200=000020eqn=0000205}) and
Markov's inequality, $Y_{T_{0}}^{\ast2}=O_{p}(b_{n})+O_{p}(\omega^{m_{n}/b_{n}}n).$
Using this, we have
\begin{equation}
\frac{nh}{b_{n}^{2}}Y_{T_{0}}^{\ast2}=O_{p}(\frac{nh}{b_{n}})+\frac{nh}{b_{n}^{2}}O_{p}(\omega^{m_{n}/b_{n}}n)=o_{p}(1),\label{Pf=000020Lem=0000207.2(b)=000020eqn=00002012}
\end{equation}
where the second equality uses $\frac{nh}{b_{n}}\rightarrow r_{0}=0$
in the present case, $\frac{nh}{b_{n}^{2}}=\frac{nh}{b_{n}}\frac{1}{b_{n}}\rightarrow0,$
and $\omega^{m_{n}/b_{n}}n\rightarrow0$ by (\ref{Pf=000020Lem=0000207.2(b)=000020eqn=00002011}).

Equation (\ref{Pf=000020Lem=0000207.2(b)=000020eqn=00002012}) establishes
the desired results for the case $w_{0}<\infty,$ which completes
the proof of part (b).

Next, we prove part (c), which considers the case $nh/b_{n}\rightarrow r_{0}=\infty.$
Let $\{m_{n}\}_{n\geq1}$ be an arbitrary sequence of positive integers
for which $m_{n}\rightarrow\infty$ and $m_{n}/n\rightarrow0.$ Then,
by (\ref{r_0=000020=00003D00003D=0000200=000020eqn=0000205}), $EY_{T_{0}}^{\ast2}=O(b_{n})+\omega^{m_{n}/b_{n}}O(n).$
Now, take $m_{n}$ as defined above, i.e., $m_{n}=b_{n}c\ln(n)$ for
$c\in(0,-1/\ln(\omega)).$ As above, $m_{n}\rightarrow\infty.$ In
addition, 
\begin{equation}
\frac{m_{n}}{n}=\frac{b_{n}}{nh}ch\ln(n)\rightarrow0\text{ provided }h\ln(n)=O(1),\label{Pf=000020Lem=0000207.2(c)=000020eqn=00002015}
\end{equation}
where the convergence holds because $\frac{nh}{b_{n}}\rightarrow r_{0}=\infty,$
which holds by assumption, iff $\frac{b_{n}}{nh}\rightarrow0.$ The
condition $h\ln(n)=O(1)$ holds by Assumption \ref{assu:MT-stationary-nh^5=000020->=0000200}.
Now, $\omega^{m_{n}/b_{n}}O(n)=o(1)$ by the argument in (\ref{Pf=000020Lem=0000207.2(b)=000020eqn=00002011})
above. Combining this and (\ref{r_0=000020=00003D00003D=0000200=000020eqn=0000205})
gives
\begin{equation}
EY_{T_{0}}^{\ast2}=O(b_{n}).
\end{equation}
This and Markov's inequality proves part (c) of the lemma.
\end{proof}

\subsection{\protect\label{subsec:Proof_of_Lemma_1}Proof of Lemma \ref{lem:MT-y0limit}}

When $b_{n}\rightarrow\infty$, Lemma \ref{lem:A.1=000020Kappa_n*=000020go=000020to=000020Kappa_n}
and $\kappa_{n}\left(\cd\right)\rightarrow\kappa_{0}\left(\cd\right)$
imply that $\kappa_{n}^{\ast}\left(\cd\right)\rightarrow\kappa_{0}\left(\cd\right)$
uniformly over $I_{\tau,\varepsilon_{2}}$ and $\kappa_{n}^{\ast}\left(\cd\right)$
is Lipschitz with Lipschitz constant less than $2L_{4}$ because $\kappa_{n}\left(\cd\right)$
is Lipschitz with Lipschitz constant $L_{4}.$ In the following proofs,
for notational simplicity, we let $L_{4}$ be the Lipschitz constant
$2L_{4}$ for $\kappa_{n}^{\ast}\left(\cd\right)$ since the factor
2 does not affect the asymptotic results. Now, we prove the local-to-unity
asymptotic results with $\rho_{n}\left(\cdot\right)$ expressed in
terms of $\kappa_{n}^{\ast}\left(\cd\right)$, rather than $\kappa_{n}\left(\cd\right).$

$\ $
\begin{proof}[\textbf{Proof of Lemma \ref{lem:MT-y0limit}}]
We suppress $n$ from $Y_{n,t\left(s\right)}^{0}$ and $U_{n,t}$
in the proof. Denote $\E\left(\rest X\mathscr{G}_{i}\right)$ by $\E_{i}X$.
By the definition of the parameter space $\Lambda_{n}$, we have $\left\{ U_{t}\right\} _{t=1}^{n}$
is a martingale difference sequence.

We adopt the following notational convention. When the lower index
of a sum exceeds the upper index, the sum is defined to equal zero.
In particular, $\sum_{j=T_{1}+\lfloor nhs\rfloor+1}^{T_{1}+\lfloor nhs\rfloor}\kappa_{n}^{*}\left(j/n\right)=0$.
Then, by the definitions of $Y_{t}^{0}$ in (\ref{eq:MT-Y_t^0=000020and=000020Y_t^star}),
$\kappa_{n}^{*}$ in (\ref{Defn=000020of=000020cappa*_n}), and $t(s)$
in (\ref{eq:MT-def=000020t(s)}), we have 
\begin{align}
 & \left(nh\right)^{-1/2}Y_{t\left(s\right)}^{0}\nonumber \\
= & \left(nh\right)^{-1/2}\sum_{k=T_{1}}^{T_{1}+\lfloor nhs\rfloor}\exp\left\{ -\dfrac{1}{b_{n}}\sum_{j=k+1}^{T_{1}+\lfloor nhs\rfloor}\kappa_{n}^{*}\left(\dfrac{j}{n}\right)\right\} \sigma_{n}\left(\dfrac{k}{n}\right)U_{k}\nonumber \\
= & \left(nh\right)^{-1/2}\sum_{k=T_{1}}^{T_{1}+\lfloor nhs\rfloor}\exp\left\{ -\dfrac{1}{b_{n}}\left(\sum_{j=T_{1}}^{T_{1}+\lfloor nhs\rfloor}\kappa_{n}^{*}\left(\dfrac{j}{n}\right)-\sum_{j=T_{1}}^{k}\kappa_{n}^{*}\left(\dfrac{j}{n}\right)\right)\right\} \sigma_{n}\left(\dfrac{k}{n}\right)U_{k}\nonumber \\
= & \exp\left\{ -\dfrac{1}{b_{n}}\sum_{j=T_{1}}^{T_{1}+\lfloor nhs\rfloor}\kappa_{n}^{*}\left(\dfrac{j}{n}\right)\right\} \nonumber \\
 & \times\left[\left(nh\right)^{-1/2}\sum_{k=T_{1}}^{T_{1}+\lfloor nhs\rfloor}\exp\left\{ \dfrac{1}{b_{n}}\sum_{j=T_{1}}^{k}\kappa_{n}^{*}\left(\dfrac{j}{n}\right)\right\} \sigma_{n}\left(\dfrac{k}{n}\right)U_{k}\right]\nonumber \\
=: & A_{1s}\times A_{2s}.\label{eq:A_1s=000020=00003D000026=000020A_2s}
\end{align}

Because $\kappa_{n}^{*}\left(\cd\right)$ is Lipschitz on $I_{\tau,\varepsilon_{2}}$,
we have 
\begin{equation}
\max_{t\in I_{n\tau,nh/2}}\abs{\kappa_{n}^{*}\left(t/n\right)-\kappa_{n}^{*}\left(\tau\right)}\leq L_{4}\max_{t\in I_{n\tau,nh/2}}\abs{t/n-\tau}=O\left(h\right),\label{eq:kappa_n}
\end{equation}
where the inequality holds for $n$ sufficiently large such that $h/2+1/n\leq\varepsilon_{2}$.
Equation \eqref{eq:kappa_n} implies that 
\begin{equation}
\max_{s\in\left[0,1\right]}\abs{\left[\dfrac{1}{\lfloor nhs\rfloor+1}\sum_{j=T_{1}}^{T_{1}+\lfloor nhs\rfloor}\kappa_{n}^{*}\left(\dfrac{j}{n}\right)\right]-\kappa_{n}^{*}\left(\tau\right)}\leq\max_{j\in\left[T_{1},T_{1}+\lfloor nh\rfloor\right]}\abs{\kappa_{n}^{*}\left(\dfrac{j}{n}\right)-\kappa_{n}^{*}\left(\tau\right)}=O\left(h\right)\label{eq:kappa_conv}
\end{equation}
because the range of the summation is $\left[T_{1},T_{1}+\lfloor nhs\rfloor\right]\subset\left[T_{1},T_{1}+\lfloor nh\rfloor\right]\subset I_{n\tau,nh/2}$
for $s\in\left[0,1\right]$ by construction. Therefore, 
\begin{align}
A_{1s}= & \exp\left\{ -\dfrac{\lfloor nhs\rfloor+1}{b_{n}}\left[\dfrac{1}{\lfloor nhs\rfloor+1}\sum_{j=T_{1}}^{T_{1}+\lfloor nhs\rfloor}\kappa_{n}^{*}\left(\dfrac{j}{n}\right)\right]\right\} \nonumber \\
= & \exp\left\{ -\dfrac{\lfloor nhs+1\rfloor}{nh}\frac{nh}{b_{n}}\left(\kappa_{n}^{*}\left(\tau\right)+O\left(h\right)\right)\right\} \nonumber \\
\Rightarrow & \exp\left\{ -sr_{0}\kappa_{0}\left(\tau\right)\right\} ,\label{eq:A11}
\end{align}
where the convergence holds by Lemma \ref{lem:A.1=000020Kappa_n*=000020go=000020to=000020Kappa_n},
Assumptions \ref{assu:MT-h} and \ref{assu:MT-kappa_n_sigma_n=000020mu_n=000020convergence},
$nh/b_{n}\rightarrow r_{0}$ as $n\gto\infty$, and the continuous
mapping theorem (CMT).

To derive the limit distribution of 
\begin{equation}
A_{2s}=\left(nh\right)^{-1/2}\sum_{k=T_{1}}^{T_{1}+\lfloor nhs\rfloor}\exp\left\{ \dfrac{1}{b_{n}}\sum_{j=T_{1}}^{k}\kappa_{n}^{*}\left(\dfrac{j}{n}\right)\right\} \sigma_{n}\left(\dfrac{k}{n}\right)U_{k},
\end{equation}
we use Theorem 2.1 of \citet{hansen1992convergence}. First, we present
a few definitions. For any random arrays $\left\{ D_{n,k},W_{n,k}:T_{1}\leq k\leq T_{2};n\geq1\right\} $,
we transform the arrays into random elements on $\left[0,1\right]$
by defining 
\begin{equation}
D_{n}\left(u\right):=D_{n,T_{1}+\lfloor nhu\rfloor}\ \text{and}\ W_{n}\left(u\right):=W_{n,T_{1}+\lfloor nhu\rfloor}.\label{eq:random_transform}
\end{equation}
for $u\in\left[0,1\right]$. Define the differences $\delta_{n,k}:=W_{n,k}-W_{n,k-1}$.
Then, we define the stochastic integral 
\begin{equation}
\int_{0}^{s}D_{n}\left(u\right)dW_{n}\left(u\right):=\sum_{k=T_{1}}^{T_{1}+\lfloor nhs\rfloor}D_{n,k}\delta_{n,k+1}\label{eq:stochastic_integral}
\end{equation}
for $s\in\left[0,1\right]$.

We let 
\begin{align}
D_{n,k}^{a}:= & \exp\left\{ \dfrac{k-T_{1}+1}{b_{n}}\left(\dfrac{1}{k-T_{1}+1}\sum_{j=T_{1}}^{k}\kappa_{n}^{*}\left(\dfrac{j}{n}\right)-\kappa_{n}^{*}\left(\tau\right)\right)\right\} \nonumber \\
 & \times\exp\left\{ \dfrac{k-T_{1}+1}{b_{n}}\kappa_{n}^{*}\left(\tau\right)\right\} \sigma_{n}\left(\tau\right),\\
D_{n,k}^{b}:= & \exp\left\{ \dfrac{k-T_{1}+1}{b_{n}}\left(\dfrac{1}{k-T_{1}+1}\sum_{j=T_{1}}^{k}\kappa_{n}^{*}\left(\dfrac{j}{n}\right)-\kappa_{n}^{*}\left(\tau\right)\right)\right\} \nonumber \\
 & \times\exp\left\{ \dfrac{k-T_{1}+1}{b_{n}}\kappa_{n}^{*}\left(\tau\right)\right\} \left(\sigma_{n}\left(\dfrac{k}{n}\right)-\sigma_{n}\left(\tau\right)\right),\ \text{and}\\
W_{n,k}:= & \left(nh\right)^{-1/2}\sum_{j=T_{1}}^{k-1}U_{j}
\end{align}
for $T_{1}\leq k\leq T_{2}$. Next, transform these quantities into
random elements on $\left[0,1\right]$ by defining 
\begin{align}
D_{n}^{a}\left(u\right):= & \exp\left\{ \left(u+\dfrac{1}{nh}\right)\left(\frac{nh}{b_{n}}\right)\left(\dfrac{1}{1+nhu}\sum_{j=T_{1}}^{T_{1}+\lfloor nhu\rfloor}\left[\kappa_{n}^{*}\left(\dfrac{j}{n}\right)-\kappa_{n}^{*}\left(\tau\right)\right]\right)\right\} \nonumber \\
 & \times\exp\left\{ \left(u+\dfrac{1}{nh}\right)\left(\frac{nh}{b_{n}}\right)\kappa_{n}^{*}\left(\tau\right)\right\} \sigma_{n}\left(\tau\right),\label{eq:D_a}\\
D_{n}^{b}\left(u\right):= & \exp\left\{ \left(u+\dfrac{1}{nh}\right)\left(\frac{nh}{b_{n}}\right)\left(\dfrac{1}{1+nhu}\sum_{j=T_{1}}^{T_{1}+\lfloor nhu\rfloor}\left[\kappa_{n}^{*}\left(\dfrac{j}{n}\right)-\kappa_{n}^{*}\left(\tau\right)\right]\right)\right\} \nonumber \\
 & \times\exp\left\{ \left(u+\dfrac{1}{nh}\right)\left(\frac{nh}{b_{n}}\right)\kappa_{n}^{*}\left(\tau\right)\right\} \left(\sigma_{n}\left(\dfrac{T_{1}+\lfloor nhu\rfloor}{n}\right)-\sigma_{n}\left(\tau\right)\right),\ \text{and}\label{eq:D_b}\\
W_{n}\left(u\right):= & \left(nh\right)^{-1/2}\sum_{j=T_{1}}^{T_{1}+\lfloor nhu\rfloor-1}U_{j}\label{eq:W_n}
\end{align}
for $u\in\left[0,1\right]$.

Then, we have 
\begin{align}
A_{2s}= & \left(nh\right)^{-1/2}\sum_{k=T_{1}}^{T_{1}+\lfloor nhs\rfloor}\exp\left\{ \dfrac{1}{b_{n}}\sum_{j=T_{1}}^{k}\kappa_{n}^{*}\left(\dfrac{j}{n}\right)\right\} \sigma_{n}\left(\dfrac{k}{n}\right)U_{k}\nonumber \\
= & \left(nh\right)^{-1/2}\sum_{k=T_{1}}^{T_{1}+\lfloor nhs\rfloor}\exp\left\{ \dfrac{k-T_{1}+1}{b_{n}}\left(\dfrac{1}{k-T_{1}+1}\sum_{j=T_{1}}^{k}\kappa_{n}^{*}\left(\dfrac{j}{n}\right)-\kappa_{n}^{*}\left(\tau\right)\right)\right\} \nonumber \\
 & \times\exp\left\{ \dfrac{k-T_{1}+1}{b_{n}}\kappa_{n}^{*}\left(\tau\right)\right\} \sigma_{n}\left(\tau\right)U_{k}\nonumber \\
 & +\left(nh\right)^{-1/2}\sum_{k=T_{1}}^{T_{1}+\lfloor nhs\rfloor}\exp\left\{ \dfrac{k-T_{1}+1}{b_{n}}\left(\dfrac{1}{k-T_{1}+1}\sum_{j=T_{1}}^{k}\kappa_{n}^{*}\left(\dfrac{j}{n}\right)-\kappa_{n}^{*}\left(\tau\right)\right)\right\} \nonumber \\
 & \times\exp\left\{ \dfrac{k-T_{1}+1}{b_{n}}\kappa_{n}^{*}\left(\tau\right)\right\} \left(\sigma_{n}\left(\dfrac{k}{n}\right)-\sigma_{n}\left(\tau\right)\right)U_{k}\nonumber \\
= & \sum_{k=T_{1}}^{T_{1}+\lfloor nhs\rfloor}D_{n,k}^{a}\left(W_{n,k+1}-W_{n,k}\right)+\sum_{k=T_{1}}^{T_{1}+\lfloor nhs\rfloor}D_{n,k}^{b}\left(W_{n,k+1}-W_{n,k}\right)\nonumber \\
= & \int_{0}^{s}D_{n}^{a}\left(u\right)dW_{n}\left(u\right)+\int_{0}^{s}D_{n}^{b}\left(u\right)dW_{n}\left(u\right)\nonumber \\
=: & A_{2as}+A_{2bs}.\label{eq:A12ab}
\end{align}

We obtain 
\begin{equation}
\max_{u\in\left[0,1\right]}\abs{\dfrac{1}{1+nhu}\sum_{j=T_{1}}^{T_{1}+\lfloor nhu\rfloor}\left[\kappa_{n}\left(\dfrac{j}{n}\right)-\kappa_{n}\left(\tau\right)\right]}\leq O\left(h\right)\gto0,\label{eq:kappa_n=000020convergence=000020result.}
\end{equation}
by \eqref{eq:kappa_n} with $\kappa_{n}\left(\cd\right)$ in place
of $\kappa_{n}^{*}\left(\cd\right)$, $\left[T_{1},T_{1}+\lfloor nhu\rfloor\right]\subset\left[T_{1},T_{1}+\lfloor nh\rfloor\right]\subset I_{n\tau,nh/2}$,
and Assumption \ref{assu:MT-h}. In addition, Lemma \ref{lem:A.1=000020Kappa_n*=000020go=000020to=000020Kappa_n}
and Assumption \ref{assu:MT-kappa_n_sigma_n=000020mu_n=000020convergence}
imply that 
\begin{equation}
\sup_{s\in I_{\tau,\varepsilon_{2}}}|\kappa_{n}^{\ast}\left(s\right)-\kappa_{0}\left(s\right)|\rightarrow0.\label{eq:kappa_n*=000020converge=000020to=000020kappa0}
\end{equation}
Combining these two results, we have 
\begin{equation}
\max_{u\in\left[0,1\right]}\abs{\dfrac{1}{1+nhu}\sum_{j=T_{1}}^{T_{1}+\lfloor nhu\rfloor}\left[\kappa_{n}^{*}\left(\dfrac{j}{n}\right)-\kappa_{n}^{*}\left(\tau\right)\right]}\leq O\left(h\right)\gto0.\label{eq:A12bkappa}
\end{equation}

Then, we have 
\begin{equation}
D_{n}^{a}\left(u\right)\Rightarrow D^{a}\left(u\right):=\exp\left\{ ur_{0}\kappa_{0}\left(\tau\right)\right\} \sigma_{0}\left(\tau\right)\label{eq:Ualambda}
\end{equation}
by \eqref{eq:kappa_n*=000020converge=000020to=000020kappa0}, \eqref{eq:A12bkappa},
and the CMT. We also have 
\begin{equation}
W_{n}\left(\cdot\right)\Rightarrow B\left(\cdot\right),\label{eq:A.7.19_Wn=000020->=000020Bn}
\end{equation}
where $B\left(\cdot\right)$ is a standard Brownian motion on $\left[0,1\right]$,
by Theorem 2.3 of \citet{mcleish1974dependent}.

Therefore, 
\begin{align}
A_{2as}=\int_{0}^{s}D_{n}^{a}\left(u\right)dW_{n}\left(u\right)\Rightarrow & \int_{0}^{s}\exp\left\{ ur_{0}\kappa_{0}\left(\tau\right)\right\} \sigma_{0}\left(\tau\right)dB\left(u\right)\label{eq:A12a}
\end{align}
by \eqref{eq:Ualambda}, \eqref{eq:A.7.19_Wn=000020->=000020Bn},
the definition of the parameter space $\Lambda_{n}$, and Theorem
2.1 of \citet{hansen1992convergence}, with $D_{n}^{a}\left(\cd\right)$
in \eqref{eq:D_a} and $D^{a}\left(\cd\right)$ in \eqref{eq:Ualambda}
in the roles of $U_{n}\left(\cd\right)$ and $U\left(\cd\right)$,
respectively, and $W_{n}\left(\cd\right)$ in \eqref{eq:W_n} and
$B\left(\cd\right)$ in the roles of $Y_{n}\left(\cd\right)$ and
$Y\left(\cd\right)$, respectively, in Theorem 2.1 of \citet{hansen1992convergence}.

For $A_{2bs}$, because $\sigma_{n}^{2}\left(\cd\right)$ is a bounded
Lipschitz function on $\left[0,1\right]$ and $I_{\tau,h/2}\subset\left[0,1\right]$,
we have 
\begin{equation}
\max_{t\in I_{n\tau,nh/2}}\abs{\sigma_{n}\left(t/n\right)-\sigma_{n}\left(\tau\right)}\leq C_{\sigma}\max_{t\in I_{n\tau,nh/2}}\abs{t/n-\tau}=O\left(h\right),\label{eq:sigma_conv}
\end{equation}
where $C_{\sigma}:=L_{3}/\left(2C_{3,L}\right)$. Equation \eqref{eq:sigma_conv}
implies that 
\begin{equation}
\max_{u\in\left[0,1\right]}\abs{\sigma_{n}\left(\dfrac{T_{1}+\lfloor nhu\rfloor}{n}\right)-\sigma_{n}\left(\tau\right)}\leq O\left(h\right)\gto0\label{eq:A12bsigma}
\end{equation}
by $\left[T_{1},T_{1}+\lfloor nhu\rfloor\right]\subset\left[T_{1},T_{1}+\lfloor nh\rfloor\right]\subset I_{n\tau,nh/2}$
for $u\in\left[0,1\right]$, and Assumption \ref{assu:MT-h}. Then,
\begin{equation}
D_{n,\tau}^{b}\left(u\right)\Rightarrow0\label{eq:Ublambda}
\end{equation}
by \eqref{eq:A12bkappa}, \eqref{eq:A12bsigma}, Assumption \ref{assu:MT-kappa_n_sigma_n=000020mu_n=000020convergence},
$nh/b_{n}\rightarrow r_{0}\in\left[0,\infty\right)$ as $n\to\infty$,
and the CMT. Therefore,

\begin{align}
A_{2bs}=\int_{0}^{s}D_{n}^{b}\left(u\right)dW_{n}\left(u\right)\Rightarrow & \ 0\label{eq:A12b}
\end{align}
by the convergence result \eqref{eq:Ublambda} and Theorem 2.1 of
\citet{hansen1992convergence}, with $D_{n}^{b}\left(\cd\right)$
in \eqref{eq:D_b} and the zero function on $\left[0,1\right]$ in
the roles of $U_{n}\left(\cd\right)$ and $U\left(\cd\right)$, respectively,
and $W_{n}\left(\cd\right)$ in \eqref{eq:W_n} and $B\left(\cd\right)$
in the roles of $Y_{n}\left(\cd\right)$ and $Y\left(\cd\right)$,
respectively, in Theorem 2.1 of \citet{hansen1992convergence}.

Combining \eqref{eq:A12a} and \eqref{eq:A12b}, we obtain

\begin{equation}
A_{2s}\Rightarrow\int_{0}^{s}\exp\left(ur_{0}\kappa_{0}\left(\tau\right)\right)\sigma_{0}\left(\tau\right)dB\left(u\right),\label{eq:A12}
\end{equation}
where $B\left(u\right)$ is standard Brownian motion.

Therefore, by \eqref{eq:A_1s=000020=00003D000026=000020A_2s}, \eqref{eq:A11},
and \eqref{eq:A12}, we have 
\begin{equation}
\left(nh\right)^{-1/2}Y_{n,t\left(s\right)}^{0}/\sigma_{0}\left(\tau\right)=A_{1s}+A_{2s}\Rightarrow\int_{0}^{s}\exp\left\{ -\left(s-u\right)r_{0}\kappa_{0}\left(\tau\right)\right\} dB\left(u\right)\label{eq:A1limit}
\end{equation}
using the CMT and the assumption that $nh/b_{n}\rightarrow r_{0}$
as $n\to\infty$.

The subsequence version of Lemma \ref{lem:MT-y0limit}, see Remark
\ref{rem:MT-subseq=000020loc=000020unity=000020thm}, which has $\{p_{n}\}_{n\geq1}$
in place of $\{n\}_{n\geq1},$ is proved by replacing $n$ by $p_{n}$
and $h=h_{n}$ by $h_{p_{n}}$ throughout the proof above. 
\end{proof}

\subsection{\protect\label{subsec:Proof-of-Lemma=000020initial=000020condition=000020T0*=000020dist}Proof
of Lemma \ref{lem:MT-T_0=000020Initial=000020Condition=000020Asy=000020Distn=000020Lem}}
\begin{proof}[\textbf{Proof of Lemma }\foreignlanguage{american}{\ref{lem:MT-T_0=000020Initial=000020Condition=000020Asy=000020Distn=000020Lem}}]
\textbf{ }By recursive substitution, for a sequence of integers $\{m_{n}\}_{n\geq1}$
such that $m_{n}\rightarrow\infty,$ we write 
\begin{eqnarray}
 &  & (2\psi/nh)^{1/2}Y_{T_{0}}^{\ast}/\sigma_{0}\left(\tau\right)\overset{}{=}D_{1n}+D_{2n},\text{ where}\nonumber \\
 &  & D_{1n}\overset{}{:=}(2\psi/nh)^{1/2}\sum_{j=0}^{m_{n}-1}c_{T_{0},j}\sigma_{T_{0}-j}U_{T_{0}-j}/\sigma_{0}\left(\tau\right)\text{ and }\nonumber \\
 &  & D_{2n}\overset{}{:=}(2\psi/nh)^{1/2}c_{T_{0},m_{n}}Y_{T_{0}-m_{n}}^{\ast}/\sigma_{0}\left(\tau\right).\label{A_1=000020and=000020A_2=000020defns}
\end{eqnarray}
We show that $D_{1n}\rightarrow_{d}Z_{1}\sim N(0,1).$ We choose $m_{n}$
such that $D_{2n}=o_{p}(1).$ This requires that $m_{n}$ is large
enough that $c_{T_{0},m_{n}}$ is sufficiently small, but small enough
that $\rho_{t}$ is close to $\rho_{n\tau}$ for all $t\in[T_{0}-2m_{n},T_{0}].$

Define 
\begin{equation}
m_{n}=b_{n}/h^{1/5}.\label{defn=000020of=000020m_n}
\end{equation}
For this choice of $m_{n},$ we have 
\begin{equation}
\begin{array}{cl}
\text{(i)} & \frac{m_{n}}{n}=\left(\frac{b_{n}}{nh}\right)h^{4/5}=o(1),\\
\text{(ii)} & \frac{m_{n}h}{b_{n}}=h^{4/5}=o(h^{1/2}),\text{ and}\\
\text{(iii)} & \frac{m_{n}^{2}}{nb_{n}}=\left(\frac{b_{n}}{nh}\right)h^{3/5}=o(h^{1/2}),
\end{array}\label{properties=000020of=000020m_n}
\end{equation}
where (i) and (iii) use $\frac{b_{n}}{nh}\rightarrow\frac{1}{r_{0}}\in(0,\infty)$
and (i)--(iii) use $h=o(1)$ by Assumption \ref{assu:MT-h}. Given
(\ref{properties=000020of=000020m_n}), Lemma \ref{Bd=000020on=000020rho=000020deviations=000020Lem}
applies and the error on the rhs of its part (c) is $O(m_{n}(h+\frac{m_{n}}{n})/b_{n})=o(h^{1/2}).$
Thus, 
\begin{equation}
c_{T_{0},m_{n}}=\rho_{n\tau}^{m_{n}}+o(h^{1/2})\text{ and }c_{T_{0}-m_{n},m_{n}}=\rho_{n\tau}^{m_{n}}+o(h^{1/2}).\label{bounds=000020on=000020c}
\end{equation}

In addition, we have 
\begin{eqnarray}
\rho_{n\tau}^{m_{n}}\hspace{-0.08in} & = & \hspace{-0.08in}(1-\kappa_{n}(\tau)/b_{n})^{m_{n}}=\exp\{-\kappa_{n}^{\ast}(\tau)m_{n}/b_{n}\}=\exp\{-\kappa_{n}^{\ast}(\tau)h^{-1/5}\}\nonumber \\
 & = & \hspace{-0.08in}\exp\{-\kappa_{0}(\tau)\}^{h^{-1/5}\kappa_{n}^{\ast}(\tau)/\kappa_{0}(\tau)}=(\ensuremath{\omega}^{\gamma_{n}}\gamma_{n}^{5/2})\gamma_{n}^{-5/2}=o(h^{1/2}),\text{ where}\nonumber \\
\ensuremath{\omega}\hspace{-0.08in} & : & \hspace{-0.13in}=\text{ }\exp\{-\kappa_{0}(\tau)\},\text{ }\gamma_{n}:=h^{-1/5}\kappa_{n}^{\ast}(\tau)/\kappa_{0}(\tau),\label{bd=000020on=000020rho^m_n}
\end{eqnarray}
the second equality uses the definition of $\kappa_{n}^{\ast}(\cdot)$
in (\ref{Defn=000020of=000020cappa*_n}), the third equality uses
(\ref{defn=000020of=000020m_n}), the fifth equality uses the definitions
of $\ensuremath{\omega}$ and $\gamma_{n},$ and the last equality
on the second line holds because $\gamma_{n}\rightarrow\infty$ (using
$\kappa_{n}^{\ast}(\tau)/\kappa_{0}(\tau)\rightarrow1$ by Lemma \ref{lem:A.1=000020Kappa_n*=000020go=000020to=000020Kappa_n}),
$\ensuremath{\omega}\in(0,1)$ (since $\kappa(\tau)\geq\varepsilon_{4}>0$
by part (ii) of $\Lambda_{n}),$ $\ensuremath{\omega}^{x}x^{5/2}=o(1)$
as $x=\gamma_{n}\rightarrow\infty$ (since $\ln(\ensuremath{\omega}^{x}x^{5/2})=x\ln\ensuremath{\omega}+(5/2)\ln x\rightarrow-\infty$
as $x\rightarrow\infty),$ and $\gamma_{n}^{-5/2}=h^{1/2}(\kappa_{n}^{\ast}(\tau)/\kappa_{0}(\tau))^{-5/2}=O(h^{1/2}).$

Equations (\ref{bounds=000020on=000020c}) and (\ref{bd=000020on=000020rho^m_n})
combine to give 
\begin{equation}
c_{T_{0},m_{n}}=o(h^{1/2})\text{ and }c_{T_{0}-m_{n},m_{n}}=o(h^{1/2}).\label{bounds=000020on=000020c=000020eqn=0000202}
\end{equation}

Next, by recursive substitution, we write the multiplicand $Y_{T_{0}-m_{n}}^{\ast}$
in $D_{2n}$ as the sum of two quantities, as in (\ref{A_1=000020and=000020A_2=000020defns}),
as follows: 
\begin{eqnarray}
 &  & Y_{T_{0}-m_{n}}^{\ast}\overset{}{:=}\sum_{j=0}^{T_{0}-m_{n}-1}c_{T_{0}-m_{n},j}\sigma_{T_{0}-m_{n}-j}U_{T_{0}-m_{n}-j}+c_{T_{0}-m_{n},T_{0}-m_{n}}Y_{0}^{\ast}\text{ and}\nonumber \\
 &  & D_{2n}\overset{}{=}D_{21n}+D_{22n},\text{ where}\nonumber \\
 &  & D_{21n}\overset{}{:=}(2\psi/nh)^{1/2}c_{T_{0},m_{n}}\sum_{j=0}^{T_{0}-m_{n}-1}c_{T_{0}-m_{n},j}\sigma_{T_{0}-m_{n}-j}U_{T_{0}-m_{n}-j}/\sigma_{0}\left(\tau\right)\text{ and}\nonumber \\
 &  & D_{22n}\overset{}{:=}(2\psi/nh)^{1/2}c_{T_{0},m_{n}}c_{T_{0}-m_{n},T_{0}-m_{n}}Y_{0}^{\ast}/\sigma_{0}\left(\tau\right).\label{defns=000020of=000020A_21n=000020and=000020A_22n}
\end{eqnarray}

By part (v) of $\Lambda_{n},$ $E_{F_{n}}(Y_{0}^{\ast})^{2}\leq C_{5}n.$
Hence, by Markov's inequality, $Y_{0}^{\ast}=O_{p}(n^{1/2})$ and
\begin{equation}
D_{22n}=h^{-1/2}c_{T_{0},m_{n}}c_{T_{0}-m_{n},T_{0}-m_{n}}O_{p}(1)/\sigma_{0}\left(\tau\right)=o_{p}(h^{1/2})=o_{p}(1),\label{bd=000020on=000020A_22n}
\end{equation}
where the second last equality holds by (\ref{bounds=000020on=000020c=000020eqn=0000202})
and $\sigma_{0}\left(\tau\right)>0$ (because $\sigma_{0}\left(\tau\right)$
is bounded below by $C_{3,L}$ by part (i) of $\Lambda_{n}$ and $C_{3,L}>0$
by assumption).

To show that $D_{2n}=o_{p}(1),$ it remains to show that $D_{21n}=o_{p}(1).$
By Markov's inequality, it suffices to show that $E_{F_{n}}D_{21n}^{2}\rightarrow0.$
Since $\{U_{t}:t=1,...,n\}$ is a stationary martingale difference
sequence by part (iv) of $\Lambda_{n},$ its elements are uncorrelated.
Thus, we have 
\begin{eqnarray}
ED_{21n}^{2}\hspace{-0.08in} & = & \hspace{-0.08in}(2\psi/nh)c_{T_{0},m_{n}}^{2}\sum_{j=0}^{T_{0}-m_{n}-1}c_{T_{0}-m_{n},j}^{2}\sigma_{T_{0}-m_{n}-j}^{2}EU_{T_{0}-m_{n}-j}^{2}/\sigma_{0}\left(\tau\right)\nonumber \\
 & = & \hspace{-0.08in}(2\psi/nh)o(h)\sum_{j=0}^{T_{0}-m_{n}-1}\sigma_{T_{0}-m_{n}-j}^{2}/\sigma_{0}\left(\tau\right)=o(1),\label{bd=000020on=000020E(A_21n)^2}
\end{eqnarray}
where the second equality holds by the first result in (\ref{bounds=000020on=000020c=000020eqn=0000202}),
$c_{T_{0}-m_{n},j}^{2}\leq1$ (since $|\rho_{t}|\leq1$ for all $t\leq n$
by part (i) of $\Lambda_{n}),$ and $EU_{t}^{2}=1$ for all $t\leq n$
(by part (iv) of $\Lambda_{n})$ and the third equality holds because
$\sigma_{T_{0}-m_{n}-j}^{2}$ is bounded by $C_{3,U}<\infty$ (by
part (i) of $\Lambda_{n}),$ $\sigma_{0}\left(\tau\right)>0$ (as
noted above), and $T_{0}-m_{n}\leq n.$ This completes the proof that
$D_{21n}=o_{p}(1)$ and $D_{2n}=o_{p}(1).$

Next, we consider $D_{1n}.$ By change of variables with $i=T_{0}-j,$
we have 
\begin{equation}
D_{1n}=(2\psi/nh)^{1/2}\sum_{i=T_{0}-m_{n}+1}^{T_{0}}c_{T_{0},T_{0}-i}\sigma_{i}U_{i}/\sigma_{0}\left(\tau\right),\label{A_1n=000020ch=000020of=000020varis}
\end{equation}
where $\{U_{t}:t=0,...,n\}$ is a stationary martingale difference
sequence under $F_{n}.$ We apply the CLT in \citet{hall1980martingale}
with $X_{ni}=(2\psi/nh)^{1/2}c_{T_{0},T_{0}-i}\sigma_{i}U_{i}/\sigma_{0}\left(\tau\right),$
with the number of summands being $m_{n}-1,$ rather than $n,$ and
with the $\sigma$-fields $\mathcal{F}_{ni}$ being the $\sigma$-fields
$\mathcal{G}_{i}$ in part (iv) of $\Lambda_{n}.$ We need to verify
a Lindeberg condition and a conditional variance condition. To verify
the former, for any $\varepsilon,\delta>0,$ we have 
\begin{eqnarray}
 &  & \hspace{-0.08in}P\left(\sum_{i=T_{0}-m_{n}+1}^{T_{0}}E(X_{n,i}^{2}1(|X_{n,i}|>\varepsilon)|\mathcal{F}_{n,i-1})>\delta\right)\nonumber \\
 & \leq & \hspace{-0.08in}\delta^{-1}\sum_{i=T_{0}-m_{n}+1}^{T_{0}}EX_{n,i}^{2}1(X_{n,i}^{2}>\varepsilon^{2})\nonumber \\
 & = & \hspace{-0.08in}\delta^{-1}(2\psi/nh)\sum_{i=T_{0}-m_{n}+1}^{T_{0}}c_{T_{0},T_{0}-i}^{2}(\sigma_{i}^{2}/\sigma_{0}^{2}\left(\tau\right))EU_{i}^{2}1((2\psi/nh)c_{T_{0},T_{0}-i}^{2}(\sigma_{i}^{2}/\sigma_{0}^{2}\left(\tau\right))U_{i}^{2}>\varepsilon^{2})\nonumber \\
 & \leq & \hspace{-0.08in}\delta^{-1}\left[(2\psi/nh)\sum_{i=T_{0}-m_{n}+1}^{T_{0}}c_{T_{0},T_{0}-i}^{2}\right](C_{3,U}/\sigma_{0}^{2}\left(\tau\right))EU_{1}^{2}1(2\psi(C_{3,U}/\sigma_{0}^{2}\left(\tau\right))U_{1}^{2}>nh\varepsilon^{2})\nonumber \\
 & = & \hspace{-0.08in}O(1)EU_{1}^{2}1(2\psi(C_{3,U}/\sigma_{0}^{2}\left(\tau\right))U_{1}^{2}>nh\varepsilon^{2})\nonumber \\
 & = & \hspace{-0.08in}o(1),\label{Lindeberg=000020condition}
\end{eqnarray}
where the first inequality holds by Markov's inequality, the first
equality holds by the definition of $X_{ni},$ and the second inequality
holds because $\sigma_{i}^{2}\leq C_{3,U}$ by part (i) of $\Lambda_{n},$
$c_{T_{0},T_{0}-i}^{2}\leq1$ (since $|\rho_{t}|\leq1),$ and $\{U_{i}\}$
are identically distributed by part (iv) of $\Lambda_{n}.$ The second
last equality in (\ref{Lindeberg=000020condition}) holds because
\begin{equation}
(2\psi/nh)\sum_{i=T_{0}-m_{n}+1}^{T_{0}}c_{T_{0},T_{0}-i}^{2}=1+o(1),\label{Lindeberg=000020pf=0000201}
\end{equation}
as shown below. The last equality in (\ref{Lindeberg=000020condition})
holds because, for $\xi:=\varepsilon^{2}/(2\psi C_{3,U}/\sigma_{0}^{2}\left(\tau\right))>0,$
\begin{equation}
EU_{1}^{2}1(U_{1}^{2}>nh\xi)=EU_{1}^{2}1([U_{1}^{2}/(nh\xi)]>1)\leq EU_{1}^{4}(nh\xi)^{-1}\rightarrow0\label{Lindeberg=000020pf=0000202}
\end{equation}
using $EU_{1}^{4}\leq M<\infty$ by part (iv) of $\Lambda_{n}$ and
$nh\rightarrow\infty$ by Assumption \ref{assu:MT-h}.

To show (\ref{Lindeberg=000020pf=0000201}), we have 
\begin{equation}
(nh)^{-1}\sum_{i=T_{0}-m_{n}+1}^{T_{0}}(c_{T_{0},T_{0}-i}^{2}-\rho_{n\tau}^{2(T_{0}-i)})=(nh)^{-1}m_{n}o(h^{1/2})=o(1),\label{Lindeberg=000020pf=0000203}
\end{equation}
where the first equality uses $|c_{T_{0},T_{0}-i}^{2}-\rho_{n\tau}^{2(T_{0}-i)}|=|c_{T_{0},T_{0}-i}-\rho_{n\tau}^{T_{0}-i}|\cdot|c_{T_{0},T_{0}-i}+\rho_{n\tau}^{T_{0}-i}|\leq2|c_{T_{0},T_{0}-i}-\rho_{n\tau}^{T_{0}-i}|$
(since $|\rho_{n\tau}^{T_{0}-i}|,|c_{T_{0},T_{0}-i}|\leq1)$ and Lemma
\ref{Bd=000020on=000020rho=000020deviations=000020Lem}(c) with an
error $O(m_{n}(h+\frac{m_{n}}{n})/b_{n})$ that is shown above to
be $o(h^{1/2})$ and the last equality holds because $m_{n}=b_{n}/h^{1/5}$
by (\ref{defn=000020of=000020m_n}), and so, $(nh)^{-1}m_{n}h^{1/2}=b_{n}/(nh^{7/10})=(b_{n}/nh)h^{3/10}=(1/r_{0}+o(1))h^{3/10}=o(1)$
since $r_{0}>0.$

Next, we show that 
\begin{equation}
(2\psi/nh)\sum_{i=T_{0}-m_{n}+1}^{T_{0}}\rho_{n\tau}^{2(T_{0}-i)}=1+o(1).\label{Lindeberg=000020pf=0000204}
\end{equation}
This holds because (i) $\sum_{i=T_{0}-m_{n}+1}^{T_{0}}\rho_{n\tau}^{2(T_{0}-i)}=\sum_{j=0}^{m_{n}-1}\rho_{n\tau}^{2j}=(1-\rho_{n\tau}^{2(m_{n}+1)})/(1-\rho_{n\tau}^{2})$
using a change of variables, (ii) $\rho_{n\tau}^{2(m_{n}+1)}=\exp\{-2\kappa_{n}^{\ast}(\tau)(m_{n}+1)/b_{n}\}\rightarrow0,$
for $\kappa_{n}^{\ast}(\cdot)$ defined in (\ref{Defn=000020of=000020cappa*_n}),
using $m_{n}/b_{n}=1/h^{1/5}\rightarrow\infty$ and $\kappa_{n}^{\ast}(\tau)\rightarrow\kappa_{0}(\tau)>0$
(by Lemma \ref{lem:A.1=000020Kappa_n*=000020go=000020to=000020Kappa_n},
Assumption \ref{assu:MT-kappa_n_sigma_n=000020mu_n=000020convergence},
and part (ii) of $\Lambda_{n},$ which guarantees that $\kappa_{0}(\tau)>0),$
and 
\begin{equation}
\text{(iii) }nh(1-\rho_{n\tau}^{2})=nh(1-\rho_{n\tau})(1+\rho_{n\tau})=(1+\rho_{n\tau})nh\kappa_{n}(\tau)/b_{n}\rightarrow2\kappa_{0}(\tau)r_{0}=2\psi.\label{Lindeberg=000020pf=0000205}
\end{equation}
Equations (\ref{Lindeberg=000020pf=0000203}) and (\ref{Lindeberg=000020pf=0000204})
combine to establish (\ref{Lindeberg=000020pf=0000201}). This completes
the verification of the Lindeberg condition in (\ref{Lindeberg=000020condition}).

Now, we prove the conditional variance condition. We have 
\begin{eqnarray}
\sum_{i=T_{0}-m_{n}+1}^{T_{0}}E(X_{n,i}^{2}|\mathcal{F}_{n,i-1})-1\hspace{-0.08in} & = & \hspace{-0.08in}(2\psi/nh)\sum_{i=T_{0}-m_{n}+1}^{T_{0}}c_{T_{0},T_{0}-i}^{2}\sigma_{i}^{2}/\sigma_{0}^{2}\left(\tau\right)-1\nonumber \\
 & = & \hspace{-0.08in}(2\psi/nh)\sum_{i=T_{0}-m_{n}+1}^{T_{0}}\rho_{n\tau}^{2(T_{0}-i)}\sigma_{i}^{2}/\sigma_{0}^{2}\left(\tau\right)-1+o(1)\nonumber \\
 & = & \hspace{-0.08in}(2\psi/nh)\sum_{i=T_{0}-m_{n}+1}^{T_{0}}\rho_{n\tau}^{2(T_{0}-i)}-1+o(1)\nonumber \\
 & = & \hspace{-0.08in}o(1),\label{Condl=000020var=000020pf=0000201}
\end{eqnarray}
where the first equality uses $E(U_{i}^{2}|\mathcal{G}_{i-1})=1$
a.s. by part (iv) of $\Lambda_{n}$ and the second equality holds
by the same argument as used to show (\ref{Lindeberg=000020pf=0000203})
since $\sigma_{i}^{2}/\sigma_{0}^{2}\left(\tau\right)$ is uniformly
bounded. The third equality in (\ref{Condl=000020var=000020pf=0000201})
holds because 
\begin{eqnarray}
 &  & \hspace{-0.08in}(2\psi/nh)\sum_{i=T_{0}-m_{n}+1}^{T_{0}}\rho_{n\tau}^{2(T_{0}-i)}|\sigma_{i}^{2}/\sigma_{0}^{2}\left(\tau\right)-1|\nonumber \\
 & \leq & \hspace{-0.08in}(2\psi/nh)\sum_{i=T_{0}-m_{n}+1}^{T_{0}}\rho_{n\tau}^{2(T_{0}-i)}\cdot\max_{j\in[T_{0}-m_{n},T_{0}]}|\sigma_{j}^{2}-\sigma_{0}^{2}\left(\tau\right)|/\sigma_{0}^{2}\left(\tau\right)\nonumber \\
 & = & \hspace{-0.08in}(1+o(1))\max_{j\in[T_{0}-m_{n},T_{0}]}|\sigma_{j}^{2}-\sigma_{0}^{2}\left(\tau\right)|/\sigma_{0}^{2}\left(\tau\right)\nonumber \\
 & = & \hspace{-0.08in}o(1),\label{Condl=000020variance=000020pf=0000202}
\end{eqnarray}
where the second last equality in (\ref{Condl=000020variance=000020pf=0000202})
holds by (\ref{Lindeberg=000020pf=0000204}) and the last equality
in (\ref{Condl=000020variance=000020pf=0000202}) holds by Lemma \ref{Bd=000020on=000020rho=000020deviations=000020Lem}(b),
$O(h+\frac{m_{n}}{n})=o(1)$ (since $h\rightarrow0$ and $m_{n}/n\rightarrow0),$
and $\sigma_{n\tau}^{2}=\sigma_{n}^{2}(\tau)\rightarrow\sigma_{0}^{2}\left(\tau\right)$
(by Assumption \ref{assu:MT-kappa_n_sigma_n=000020mu_n=000020convergence}).
Hence, the conditional variance condition in (\ref{Condl=000020var=000020pf=0000201})
holds. By the CLT of \citet{hall1980martingale}, we have 
\begin{equation}
D_{1n}\rightarrow_{d}Z_{1}\sim N(0,1),\label{S_1,n=000020asy=000020normal}
\end{equation}
as desired.

By (\ref{eq:A.7.19_Wn=000020->=000020Bn}), $W_{n}(\cdot)\Rightarrow B(\cdot)$
and by (\ref{S_1,n=000020asy=000020normal}), $D_{1n}\rightarrow_{d}Z_{1}.$
These results can be shown to hold jointly because $W_{n}(\cdot)$
and $D_{1n}$ depend on the same random variables $\{U_{t}\}_{t\leq n}.$
By (\ref{eq:W_n}), $W_{n}(u)$ is a linear function of $\{U_{t}:t\geq T_{1}\}$
for all $u\in[0,1].$ By (\ref{A_1n=000020ch=000020of=000020varis}),
$D_{1n}$ is a linear function of $\{U_{t}:t<T_{1}\}.$ Since $\{U_{t}\}_{t\leq n}$
is a martingale difference sequence, these properties imply that $Cov(W_{n}(u),D_{1n})=0$
for all $u\in[0,1].$ In consequence, $W_{n}(u)$ and $D_{1n}$ are
asymptotically independent, i.e., $B(\cdot)$ and $Z_{1}$ are independent. 
\end{proof}

\subsection{\protect\label{subsec:Proof-of-Lemma-3}Proof of Lemma \ref{lem:MT-asympdist_components_local=000020to=000020unity=000020case}}

In this section, for notational simplicity in the proof, we assume
that $\sigma_{0}^{2}\left(\tau\right)=1.$ If this is not assumed,
numerous quantities in the proof needed to be rescaled by $1/\sigma_{0}\left(\tau\right).$ 
\begin{proof}[\textbf{Proof of Lemma \ref{lem:MT-asympdist_components_local=000020to=000020unity=000020case}}]
First, we prove part (a). We have 
\begin{eqnarray}
(nh)^{-1/2}Y_{t(s)}\hspace{-0.08in} & = & \hspace{-0.08in}(nh)^{-1/2}(\mu_{t(s)}+Y_{t(s)}^{0}+c_{t(s),t(s)-T_{0}}Y_{T_{0}}^{\ast})\nonumber \\
 & = & \hspace{-0.08in}o(1)+(nh)^{-1/2}Y_{t(s)}^{0}+(2\psi)^{-1/2}c_{t(s),t(s)-T_{0}}(2\psi/nh)^{1/2}Y_{T_{0}}^{\ast}\nonumber \\
 & \Rightarrow & \hspace{-0.08in}I_{\psi}(s)+(2\psi)^{-1/2}\exp\{-\psi s\}Z_{1}=:I_{\psi}^{\ast}(s),
\end{eqnarray}
where the first equality holds by (\ref{eq:MT-tvp-model}) and (\ref{eq:MT-decompY*}),
the second equality holds by part (i) of $\Lambda_{n}$ and Assumption
\ref{assu:MT-h}, the convergence holds by Lemma \ref{lem:MT-y0limit},
Lemma \ref{lem:MT-T_0=000020Initial=000020Condition=000020Asy=000020Distn=000020Lem},
which uses the assumption that $r_{0}\in(0,\infty),$ and $c_{t(s),t(s)-T_{0}}\Rightarrow\exp\{-\psi s\},$
which we now establish.

By Lemma \ref{lem:MT-Max=000020Intertemporal=000020Difference}(c),
\begin{eqnarray}
c_{t(s),t(s)-T_{0}}\hspace{-0.08in} & = & \hspace{-0.08in}\rho_{n\tau}^{t(s)-T_{0}}+O(nh^{2}/b_{n})=\exp\{-\kappa_{n}^{\ast}(\tau)(t(s)-T_{0})/b_{n}\}+o(1)\\
 & = & \hspace{-0.08in}\exp\{-\kappa_{n}^{\ast}(\tau)(\lfloor nhs\rfloor+1)/b_{n}\}+o(1)\rightarrow\exp\{-\kappa_{0}(\tau)r_{0}s\}=\exp\{-\psi s\},\nonumber 
\end{eqnarray}
where the $O(nh^{2}/b_{n})$ term holds uniformly over $s\in[0,1],$
$\kappa_{n}^{\ast}(\tau)$ is defined in (\ref{Defn=000020of=000020cappa*_n}),
the second equality uses $nh/b_{n}\rightarrow r_{0}<\infty$ and $h\rightarrow0$
by Assumption \ref{assu:MT-h}, the third equality holds by the definition
of $t(s)$ in (\ref{eq:MT-def=000020t(s)}) and $T_{0}=T_{1}-1,$
the convergence uses Lemma \ref{lem:A.1=000020Kappa_n*=000020go=000020to=000020Kappa_n},
Assumption \ref{assu:MT-kappa_n_sigma_n=000020mu_n=000020convergence},
and $nh/b_{n}\rightarrow r_{0},$ and the final equality uses the
definition of $\psi=r_{0}\kappa_{0}(\tau).$ This completes the proof
of part (a).

The proofs of parts (b) and (c) use the technique developed in \citet{phillips1987towards}.
For example, for part (b), we have 
\begin{equation}
\left(nh\right)^{-3/2}\sum_{t=T_{1}}^{T_{2}}Y_{t-1}=\int_{0}^{1}\left[\left(nh\right)^{-1/2}Y_{t\left(s\right)}\right]ds\dto\int_{0}^{1}I_{\psi}^{*}\left(s\right)ds,\label{eq:lemma3.2}
\end{equation}
where the convergence holds by the CMT and part (a) and $\psi=r_{0}\kappa_{0}\left(\tau\right)$.

For part (c), we have 
\begin{equation}
\left(nh\right)^{-2}\sum_{t=T_{1}}^{T_{2}}Y_{t-1}^{2}=\int_{0}^{1}\left[\left(nh\right)^{-1/2}Y_{t\left(s\right)}\right]^{2}ds\dto\int_{0}^{1}I_{\psi}^{*2}\left(s\right)ds,\label{eq:lemma3.3}
\end{equation}
where the convergence holds by the CMT and part (a).

To prove part (d), define $W_{n}\left(s\right):=\left(nh\right)^{-1/2}\sum_{t=T_{1}}^{T_{1}+\lfloor nhs\rfloor-1}U_{t}$
for $s\in\left[0,1\right]$ and define the stochastic integral 
\begin{equation}
\int_{0}^{1}\left[\sigma_{t\left(s\right)}\right]dW_{n}\left(s\right):=\left(nh\right)^{-1/2}\sum_{t=T_{1}}^{T_{2}}U_{t}\sigma_{t}.\label{eq:partg}
\end{equation}
By Assumption \ref{assu:MT-kappa_n_sigma_n=000020mu_n=000020convergence},
Lemma \ref{lem:MT-Max=000020Intertemporal=000020Difference}(b), and
the triangle inequality, we have 
\begin{equation}
\max_{s\in\left[0,1\right]}\abs{\sigma_{t\left(s\right)}^{2}-1}\pto0,\label{eq:sigma_t^2=000020converge=000020to=000020sigma_0}
\end{equation}
which implies 
\begin{equation}
\max_{s\in\left[0,1\right]}\abs{\sigma_{t\left(s\right)}-1}\pto0\label{eq:sigma_t=000020converge=000020to=000020sigma_0}
\end{equation}
since $\s_{t}$ is nonnegative. By the functional central limit theorem
for martingale difference sequences, we have 
\begin{equation}
W_{n}\left(s\right)\Rightarrow B\left(s\right).\label{eq:partg2}
\end{equation}
Therefore, we use Theorem 2.1 of \citet{hansen1992convergence} and
obtain 
\begin{equation}
\left(nh\right)^{-1/2}\sum_{t=T_{1}}^{T_{2}}U_{t}\sigma_{t}=\int_{0}^{1}\left[\sigma_{t\left(s\right)}\right]dW_{n}\left(s\right)\dto\int_{0}^{1}dB\left(s\right).\label{eq:partg3}
\end{equation}

To prove part (e), we define the stochastic integral 
\begin{equation}
\int_{0}^{1}\left[\left(nh\right)^{-1/2}Y_{t\left(s\right)}\s_{t\left(s\right)}\right]dW_{n}\left(s\right):=\left(nh\right)^{-1}\sum_{t=T_{1}}^{T_{2}}Y_{t-1}U_{t}\sigma_{t}.\label{eq:=000020lem3=000020part(d)}
\end{equation}
By part (a), we have 
\begin{equation}
\left(nh\right)^{-1/2}Y_{t\left(s\right)}\Rightarrow I_{\psi}^{*}\left(s\right).\label{eq:Yt(s)=000020conv}
\end{equation}
We also know that $\max_{s\in\left[0,1\right]}\abs{\sigma_{t\left(s\right)}-1}\pto0$
by \eqref{eq:sigma_t=000020converge=000020to=000020sigma_0}. Thus,
by the CMT we have 
\begin{equation}
\left(nh\right)^{-1/2}Y_{t\left(s\right)}\s_{t\left(s\right)}=\left(nh\right)^{-1/2}Y_{t\left(s\right)}+\left(nh\right)^{-1/2}Y_{t\left(s\right)}\left(\s_{t\left(s\right)}-1\right)\Rightarrow I_{\psi}^{*}\left(s\right).\label{eq:Lem7.4=000020(e)=000020pf=000020integrand}
\end{equation}

Equations (\ref{eq:partg2}) and \eqref{eq:Lem7.4=000020(e)=000020pf=000020integrand}
hold jointly due to the common underlying martingale difference process
$\left\{ U_{t}\right\} _{t=T_{1},...,T_{2}}$. Therefore, we have
\begin{equation}
\left(nh\right)^{-1}\sum_{t=T_{1}}^{T_{2}}Y_{t-1}U_{t}\sigma_{t}=\int_{0}^{1}\left[\left(nh\right)^{-1/2}Y_{t\left(s\right)}\s_{t\left(s\right)}\right]dW_{n}\left(s\right)\dto\int_{0}^{1}I_{\psi}^{*}\left(s\right)dB\left(s\right),\label{eq:lemma3.4}
\end{equation}
where the convergence holds by Theorem 2.1 of \citet{hansen1992convergence},
with $\left(nh\right)^{-1/2}Y_{t\left(\cd\right)}\s_{t\left(\cd\right)}$
and $I_{\psi}^{*}\left(\cd\right)$ in the roles of $U_{n}\left(\cd\right)$
and $U\left(\cd\right)$, respectively, and $W_{n}\left(\cd\right)$
and $B\left(\cd\right)$ in the roles of $Y_{n}\left(\cd\right)$
and $Y\left(\cd\right)$, respectively, in Theorem 2.1 of \citet{hansen1992convergence}.

The proof of part (f) is analogous to that of part (c), and thus,
is omitted.

Part (g) holds by the same argument as given above for parts (a)--(c)
and (e), but with $Y_{t\left(s\right)}=\mu_{t\left(s\right)}+Y_{t\left(s\right)}^{0}+c_{t\left(s\right),t(s)-T_{0}}Y_{T_{0}}^{*}$
replaced by $\mu_{t\left(s\right)}+Y_{t\left(s\right)}^{0}$ in part
(a), which simplifies the proof because the initial condition $Y_{T_{0}}^{*}$
does not appear.

The subsequence version of Lemma \ref{lem:MT-asympdist_components_local=000020to=000020unity=000020case},
see Remark \ref{rem:MT-subseq=000020loc=000020unity=000020thm}, which
has $\{p_{n}\}_{n\geq1}$ in place of $\{n\}_{n\geq1},$ is proved
by replacing $n$ by $p_{n}$ and $h=h_{n}$ by $h_{p_{n}}$ throughout
the proof above. 
\end{proof}

\subsection{\protect\label{subsec:Proof-of-Thm-Asymptotic=000020Rho=000020Hat=000020Dist=000020Local=000020to=000020unity}Proof
of Theorem \ref{thm:MT-asym_rho_hat_local=000020to=000020unity=000020case}}

In this section, for notational simplicity in the proof, we assume
that $\sigma_{0}^{2}\left(\tau\right)=1.$ 
\begin{proof}[\textbf{Proof of Theorem \ref{thm:MT-asym_rho_hat_local=000020to=000020unity=000020case}}]
First, we prove Theorem \ref{thm:MT-asym_rho_hat_local=000020to=000020unity=000020case}
for the case where $r_{0}\in\left(0,\infty\right)$. To start, we
show that the denominator of \eqref{eq:MT-normalized=000020rho} divided
by $\sigma_{0}\left(\tau\right)$ converges in distribution to $\int_{0}^{1}I_{D,\psi}^{*2}\left(s\right)ds,$
where $\psi=r_{0}\kappa_{0}\left(\tau\right).$ 
\begin{align}
 & \left(nh\right)^{-2}\sum_{t=T_{1}}^{T_{2}}\left(Y_{t-1}-\ol Y_{nh,-1}\right)^{2}\nonumber \\
= & \left(nh\right)^{-2}\sum_{t=T_{1}}^{T_{2}}\left(Y_{t-1}\right)^{2}-\left(nh\right)^{-1}\left(\ol Y_{nh,-1}\right)^{2}\nonumber \\
= & \left(nh\right)^{-2}\sum_{t=T_{1}}^{T_{2}}\left(Y_{t-1}\right)^{2}-\left(\left(nh\right)^{-3/2}\sum_{t=T_{1}}^{T_{2}}Y_{t-1}\right)^{2}\nonumber \\
\dto & \int_{0}^{1}I_{\psi}^{*2}\left(s\right)ds-\left(\int_{0}^{1}I_{\psi}^{*}\left(s\right)ds\right)^{2}=\int_{0}^{1}I_{D,\psi}^{*2}\left(s\right)ds,\label{eq:denominator=000020of=000020rho_hat}
\end{align}
where the convergence holds by Lemma \ref{lem:MT-asympdist_components_local=000020to=000020unity=000020case}(b)
and (c) and the CMT.

Next, we show the numerator of \eqref{eq:MT-normalized=000020rho}
converges in distribution to $\int_{0}^{1}I_{D,\psi}^{*}\left(s\right)dB\left(s\right)$.
For any $t=T_{1},...,T_{2}$, we have 
\begin{align}
Y_{t}-\rho_{0,n}Y_{t-1} & =\left(\mu_{t}-\rho_{0,n}\mu_{t-1}\right)+\left(\rho_{t}Y_{t-1}^{*}+\sigma_{t}U_{t}-\rho_{0,n}Y_{t-1}^{*}\right)\nonumber \\
 & =\sigma_{t}U_{t}+\left(\mu_{t}-\rho_{0,n}\mu_{t-1}\right)+\left(\rho_{t}-\rho_{0,n}\right)Y_{t-1}^{0}+\left(\rho_{t}-\rho_{0,n}\right)c_{t-1,t-1-T_{0}}Y_{T_{0}}^{*},\label{eq:yt_dagger_decomp}
\end{align}
where the first equality holds by \eqref{eq:MT-tvp-model} and the
last equality holds by \eqref{eq:MT-Y_t^0=000020and=000020Y_t^star}.
Substituting \eqref{eq:yt_dagger_decomp} into the numerator of \eqref{eq:MT-normalized=000020rho}
gives 
\begin{align}
 & \left(nh\right)^{-1}\sum_{t=T_{1}}^{T_{2}}\left(Y_{t-1}-\ol Y_{nh,-1}\right)\left(Y_{t}-\rho_{0,n}Y_{t-1}\right)\nonumber \\
= & \left(nh\right)^{-1}\sum_{t=T_{1}}^{T_{2}}\left(Y_{t-1}-\ol Y_{nh,-1}\right)\sigma_{t}U_{t}\nonumber \\
 & +\left(nh\right)^{-1}\sum_{t=T_{1}}^{T_{2}}\left(Y_{t-1}-\ol Y_{nh,-1}\right)\left(\mu_{t}-\rho_{0,n}\mu_{t-1}\right)\nonumber \\
 & +\left(nh\right)^{-1}\sum_{t=T_{1}}^{T_{2}}\left(Y_{t-1}-\ol Y_{nh,-1}\right)\left(\rho_{t}-\rho_{0,n}\right)Y_{t-1}^{0}\nonumber \\
 & +\left(nh\right)^{-1}\sum_{t=T_{1}}^{T_{2}}\left(Y_{t-1}-\ol Y_{nh,-1}\right)\left(\rho_{t}-\rho_{0,n}\right)c_{t-1,t-1-T_{0}}Y_{T_{0}}^{*}\nonumber \\
=: & \ A_{1}+A_{2}+A_{3}+A_{4}.\label{eq:decompB}
\end{align}

For $A_{1}$, we have 
\begin{align}
A_{1} & =\left(nh\right)^{-1}\sum_{t=T_{1}}^{T_{2}}\left(Y_{t-1}-\ol Y_{nh,-1}\right)\sigma_{t}U_{t}\nonumber \\
 & =\left(nh\right)^{-1}\sum_{t=T_{1}}^{T_{2}}Y_{t-1}\sigma_{t}U_{t}-\left[\left(nh\right)^{-1/2}\ol Y_{nh,-1}\right]\left(nh\right)^{-1/2}\sum_{t=T_{1}}^{T_{2}}\sigma_{t}U_{t}\nonumber \\
 & \dto\int_{0}^{1}I_{\psi}^{*}\left(s\right)dB\left(s\right)-\int_{0}^{1}I_{\psi}^{*}\left(s\right)ds\int_{0}^{1}dB\left(s\right)=\int_{0}^{1}I_{D,\psi}^{*}\left(s\right)dB\left(s\right),\label{eq:B1_analysis}
\end{align}
where the convergence holds by Lemma \ref{lem:MT-asympdist_components_local=000020to=000020unity=000020case}(b),
(d), and (e).

For $A_{2}$, we have 
\begin{align}
\abs{A_{2}}^{2}= & \abs{\left(nh\right)^{-1}\sum_{t=T_{1}}^{T_{2}}\left(Y_{t-1}-\ol Y_{nh,-1}\right)\left(\mu_{t}-\rho_{0,n}\mu_{t-1}\right)}^{2}\nonumber \\
\leq & \left[\left(nh\right)^{-2}\sum_{t=T_{1}}^{T_{2}}\left(Y_{t-1}-\ol Y_{nh,-1}\right)^{2}\right]\left[\sum_{t=T_{1}}^{T_{2}}\left(\mu_{t}-\rho_{0,n}\mu_{t-1}\right)^{2}\right]\nonumber \\
\leq & \left[\left(nh\right)^{-2}\sum_{t=T_{1}}^{T_{2}}\left(Y_{t-1}-\ol Y_{nh,-1}\right)^{2}\right]\nonumber \\
 & \times\left[2\sum_{t=T_{1}}^{T_{2}}\left(\mu_{t}-\mu_{t-1}\right)^{2}+2\sum_{t=T_{1}}^{T_{2}}\left(\mu_{t-1}-\rho_{0,n}\mu_{t-1}\right)^{2}\right]\nonumber \\
\leq & \left[\left(nh\right)^{-2}\sum_{t=T_{1}}^{T_{2}}\left(Y_{t-1}-\ol Y_{nh,-1}\right)^{2}\right]\nonumber \\
 & \times2\left[nh\max_{t\in\left[T_{1},T_{2}\right]}\left(\mu_{t}-\mu_{t-1}\right)^{2}+nh\left(1-\rho_{0,n}\right)^{2}\max_{t\in\left[T_{1},T_{2}\right]}\mu_{t-1}^{2}\right]\nonumber \\
= & \ O_{p}\left(1\right)\left[nhO\left(n^{-2}\right)+nhO\left(\left(nh\right)^{-2}\right)\right]=o_{p}\left(1\right),\label{eq:B2_analysis}
\end{align}
where the first inequality holds by the Cauchy-Schwarz (CS) inequality,
the second inequality uses the fact that $\left(a+b\right)^{2}\leq2\left(a^{2}+b^{2}\right)$,
and the second last equality holds by \eqref{eq:denominator=000020of=000020rho_hat},
$\max_{t\in\left[T_{1},T_{2}\right]}\abs{\mu_{t}-\mu_{t-1}}\leq L_{2}/n$
by the Lipschitz condition on $\mu\left(\cd\right)$, $\abs{1-\rho_{0,n}}=O\left(\left(nh\right)^{-1}\right)$,
and $\max_{t\in\left[T_{1},T_{2}\right]}\mu_{t}^{2}$ is $O\left(1\right)$.

For $A_{3}$, we have 
\begin{align}
\abs{A_{3}}^{2}= & \abs{\left(nh\right)^{-1}\sum_{t=T_{1}}^{T_{2}}\left(Y_{t-1}-\ol Y_{nh,-1}\right)\left(\rho_{t}-\rho_{0,n}\right)Y_{t-1}^{0}}^{2}\nonumber \\
\leq & \left[\left(nh\right)^{-2}\sum_{t=T_{1}}^{T_{2}}\left(Y_{t-1}-\ol Y_{nh,-1}\right)^{2}\right]\left[\sum_{t=T_{1}}^{T_{2}}\left(\rho_{t}-\rho_{0,n}\right)^{2}\left(Y_{t-1}^{0}\right)^{2}\right]\nonumber \\
\leq & \left[\left(nh\right)^{-2}\sum_{t=T_{1}}^{T_{2}}\left(Y_{t-1}-\ol Y_{nh,-1}\right)^{2}\right]\left[\left(nh\right)^{2}\max_{t\in\left[T_{1},T_{2}\right]}\left(\rho_{t}-\rho_{0,n}\right)^{2}\right]\nonumber \\
 & \times\left[\left(nh\right)^{-2}\sum_{t=T_{1}}^{T_{2}}\left(Y_{t-1}^{0}\right)^{2}\right]\nonumber \\
= & \ O_{p}\left(1\right)O\left(\left(nh\right)^{2}n^{-2}\right)O_{p}\left(1\right)=o_{p}\left(1\right),\label{eq:B3_analysis}
\end{align}
where the first inequality holds by the CS inequality and the second
last equality holds by \eqref{eq:denominator=000020of=000020rho_hat},
Lemma \ref{lem:MT-Max=000020Intertemporal=000020Difference}(a), and
Lemma \ref{lem:MT-asympdist_components_local=000020to=000020unity=000020case}(f).
Note that Lemma \ref{lem:MT-Max=000020Intertemporal=000020Difference}(a)
implies $\max_{t\in\left[T_{1},T_{2}\right]}\left(\rho_{t}-\rho_{0,n}\right)^{2}=O\left(n^{-2}\right)$
because under $H_{0}$ and when $nh/b_{n}=O\left(1\right)$, 
\begin{equation}
\max_{t\in\left[T_{1},T_{2}\right]}\left(\rho_{t}-\rho_{0,n}\right)^{2}=\left(\max_{t\in\left[T_{1},T_{2}\right]}\abs{\rho_{t}-\rho_{n\tau}}\right)^{2}=\left(O\left(h/b_{n}\right)\left(b_{n}/nh\right)nh/b_{n}\right)^{2}=O\left(n^{-2}\right).\label{eq:=000020(rho_t-rho_0)^2=000020bound}
\end{equation}

For $A_{4}$, we have 
\begin{align}
\abs{A_{4}}^{2} & =\abs{\left(nh\right)^{-1}\sum_{t=T_{1}}^{T_{2}}\left(Y_{t-1}-\ol Y_{nh,-1}\right)\left(\rho_{t}-\rho_{0,n}\right)c_{t-1,t-1-T_{0}}Y_{T_{0}}^{*}}^{2}\nonumber \\
 & \leq\left[\left(nh\right)^{-2}\sum_{t=T_{1}}^{T_{2}}\left(Y_{t-1}-\ol Y_{nh,-1}\right)^{2}\right]\left[\sum_{t=T_{1}}^{T_{2}}\left(\rho_{t}-\rho_{0,n}\right)^{2}c_{t-1,t-1-T_{0}}^{2}Y_{T_{0}}^{*2}\right]\nonumber \\
 & \leq\left[\left(nh\right)^{-2}\sum_{t=T_{1}}^{T_{2}}\left(Y_{t-1}-\ol Y_{nh,-1}\right)^{2}\right]\left[nh\max_{t\in\left[T_{1},T_{2}\right]}\left(\rho_{t}-\rho_{0,n}\right)^{2}\right]Y_{T_{0}}^{*2}\nonumber \\
 & =O_{p}\left(1\right)O\left(nh\cd n^{-2}\right)O_{p}\left(n\right)=o_{p}\left(1\right),\label{eq:B4_analysis}
\end{align}
where the first inequality holds by the CS inequality, the second
inequality uses $\max_{t\in\left[T_{0},T_{2}\right]}c_{t,t-T_{0}}^{2}\leq1$,
and the second last equality holds by \eqref{eq:=000020bound=000020order=000020of=000020Y_T0*=000020Op(n^1/2)},
\eqref{eq:denominator=000020of=000020rho_hat}, and \eqref{eq:=000020(rho_t-rho_0)^2=000020bound}.

Therefore, we obtain 
\begin{align}
\left(nh\right)^{-1}\sum_{t=T_{1}}^{T_{2}}\left(Y_{t-1}-\ol Y_{nh,-1}\right)\left(Y_{t}-\rho_{0,n}Y_{t-1}\right) & =A_{1}+o_{p}\left(1\right)\dto\int_{0}^{1}I_{D,\psi}^{*}\left(s\right)dB\left(s\right)\label{eq:numerator_asym}
\end{align}
from \eqref{eq:B1_analysis}, \eqref{eq:B2_analysis}, \eqref{eq:B3_analysis},
and \eqref{eq:B4_analysis}.

Combining \eqref{eq:denominator=000020of=000020rho_hat} and \eqref{eq:numerator_asym},
we have 
\begin{equation}
nh\left(\widehat{\rho}_{n\tau}-\rho_{0,n}\right)\dto\left(\int_{0}^{1}I_{D,\psi}^{*2}\left(s\right)ds\right)^{-1}\int_{0}^{1}I_{D,\psi}^{*}\left(s\right)dB\left(s\right)\label{eq:norm_rho_hat_result}
\end{equation}
by the CMT.

$\ $

Next, for the t-statistic $T_{n}\left(\rho_{0,n}\right)$, we have
\begin{equation}
T_{n}\left(\rho_{0,n}\right)=\dfrac{\left(nh\right)^{1/2}\left(\widehat{\rho}_{n\tau}-\rho_{0,n}\right)}{\widehat{s}_{n\tau}}=\dfrac{nh\left(\widehat{\rho}_{n\tau}-\rho\right)}{\left(nh\widehat{s}_{n\tau}^{2}\right)^{1/2}}.\label{eq:T_n=000020def}
\end{equation}
Thus, by \eqref{eq:MT-sigma=000020ntau=000020def}, \eqref{eq:MT-t-stat=000020def},
\eqref{eq:denominator=000020of=000020rho_hat}, \eqref{eq:norm_rho_hat_result},
and the CMT, we only need to show 
\begin{equation}
\widehat{\s}_{n\tau}^{2}\coloneqq\left(nh\right)^{-1}\sum_{t=T_{1}}^{T_{2}}\left[Y_{t}-\ol Y_{nh}-\widehat{\rho}_{n\tau}\left(Y_{t-1}-\ol Y_{nh,-1}\right)\right]^{2}\pto1.\label{eq:sigma^2=000020numberator}
\end{equation}

First, we replace $\widehat{\rho}_{n\tau}$ with $\rho_{0,n}$ in
\eqref{eq:sigma^2=000020numberator}: 
\begin{align}
\widehat{\s}_{n\tau}^{2}= & \left(nh\right)^{-1}\sum_{t=T_{1}}^{T_{2}}\left[Y_{t}-\ol Y_{nh}-\widehat{\rho}_{n\tau}\left(Y_{t-1}-\ol Y_{nh,-1}\right)\right]^{2}\nonumber \\
= & \left(nh\right)^{-1}\sum_{t=T_{1}}^{T_{2}}\left[Y_{t}-\ol Y_{nh}-\rho_{0,n}\left(Y_{t-1}-\ol Y_{nh,-1}\right)\right]^{2}\nonumber \\
 & +\left(nh\right)^{-1}\sum_{t=T_{1}}^{T_{2}}\left[\left(\rho_{0,n}-\widehat{\rho}_{n\tau}\right)\left(Y_{t-1}-\ol Y_{nh,-1}\right)\right]^{2}\nonumber \\
 & +2\left(nh\right)^{-1}\sum_{t=T_{1}}^{T_{2}}\left[Y_{t}-\ol Y_{nh}-\rho_{0,n}\left(Y_{t-1}-\ol Y_{nh,-1}\right)\right]\left[\left(\rho_{0,n}-\widehat{\rho}_{n\tau}\right)\left(Y_{t-1}-\ol Y_{nh,-1}\right)\right]\nonumber \\
\eqqcolon & \ A_{5}+A_{6}+A_{7}.\label{eq:decomp=000020num=000020of=000020sigma=0000202}
\end{align}
We show that $A_{5}\pto1$ and $A_{6}\pto0$, which together imply
$A_{7}\pto0$ by the CS inequality.

For $A_{5}$, we have 
\begin{align}
 & \left(nh\right)^{-1}\sum_{t=T_{1}}^{T_{2}}\left[Y_{t}-\ol Y_{nh}-\rho_{0,n}\left(Y_{t-1}-\ol Y_{nh,-1}\right)\right]^{2}\nonumber \\
= & \left(nh\right)^{-1}\sum_{t=T_{1}}^{T_{2}}\left(Y_{t}-\rho_{0,n}Y_{t-1}\right)^{2}+\left(\ol Y_{nh}-\rho_{0,n}\ol Y_{nh,-1}\right)^{2}\nonumber \\
 & +2\left(nh\right)^{-1}\sum_{t=T_{1}}^{T_{2}}\left(\ol Y_{nh}-\rho_{0,n}\ol Y_{nh,-1}\right)\left(Y_{t}-\rho_{0,n}Y_{t-1}\right)\nonumber \\
\eqqcolon & \ A_{51}+A_{52}+A_{53}.\label{eq:A5=000020expand}
\end{align}
First, we show that $A_{51}$ converges in probability to 1. By \eqref{eq:yt_dagger_decomp},
we have 
\begin{align}
A_{51}= & \left(nh\right)^{-1}\sum_{t=T_{1}}^{T_{2}}\left[\begin{array}{c}
\sigma_{t}U_{t}+\left(\mu_{t}-\rho_{0,n}\mu_{t-1}\right)+\left(\rho_{t}-\rho_{0,n}\right)Y_{t-1}^{0}\\
+\left(\rho_{t}-\rho_{0,n}\right)c_{t-1,t-1-T_{0}}Y_{T_{0}}^{*}
\end{array}\right]^{2}\nonumber \\
= & \left(nh\right)^{-1}\sum_{t=T_{1}}^{T_{2}}\left(\sigma_{t}U_{t}\right)^{2}\nonumber \\
 & +\left(nh\right)^{-1}\sum_{t=T_{1}}^{T_{2}}\left[\left(\mu_{t}-\rho_{0,n}\mu_{t-1}\right)+\left(\rho_{t}-\rho_{0,n}\right)Y_{t-1}^{0}+\left(\rho_{t}-\rho_{0,n}\right)c_{t-1,t-1-T_{0}}Y_{T_{0}}^{*}\right]^{2}\nonumber \\
 & +2\left(nh\right)^{-1}\sum_{t=T_{1}}^{T_{2}}\sigma_{t}U_{t}\left[\left(\mu_{t}-\rho_{0,n}\mu_{t-1}\right)+\left(\rho_{t}-\rho_{0,n}\right)Y_{t-1}^{0}+\left(\rho_{t}-\rho_{0,n}\right)c_{t-1,t-1-T_{0}}Y_{T_{0}}^{*}\right]\nonumber \\
\eqqcolon & \ A_{511}+A_{512}+A_{513}.\label{eq:A_51}
\end{align}
For $A_{511}$, we have 
\begin{equation}
\left(nh\right)^{-1}\sum_{t=T_{1}}^{T_{2}}\left(\sigma_{t}U_{t}\right)^{2}=\left(nh\right)^{-1}\sum_{t=T_{1}}^{T_{2}}U_{t}^{2}+\left(nh\right)^{-1}\sum_{t=T_{1}}^{T_{2}}\left(\s_{t}^{2}-1\right)U_{t}^{2}.\label{eq:A511}
\end{equation}
By the weak law of large numbers, $\left(nh\right)^{-1}\sum_{t=T_{1}}^{T_{2}}U_{t}^{2}\pto1$
as $n\to\infty$. We also have 
\begin{equation}
\abs{\left(nh\right)^{-1}\sum_{t=T_{1}}^{T_{2}}\left(\s_{t}^{2}-1\right)U_{t}^{2}}\leq\max_{t\in\left[T_{1},T_{2}\right]}\abs{\s_{t}^{2}-1}\left(nh\right)^{-1}\sum_{t=T_{1}}^{T_{2}}U_{t}^{2}=o_{p}\left(1\right),\label{eq:A511=000020part=0000202}
\end{equation}
where the equality holds by \eqref{eq:sigma_t=000020converge=000020to=000020sigma_0}.
Therefore, 
\begin{equation}
A_{511}=\left(nh\right)^{-1}\sum_{t=T_{1}}^{T_{2}}\left(\sigma_{t}U_{t}\right)^{2}\pto1.\label{eq:a511=000020result}
\end{equation}
Next, we have $\sum_{t=T_{1}}^{T_{2}}\left(\mu_{t}-\rho_{0,n}\mu_{t-1}\right)^{2}=o\left(1\right)$,
$\sum_{t=T_{1}}^{T_{2}}\left(\rho_{t}-\rho_{0,n}\right)^{2}\left(Y_{t-1}^{0}\right)^{2}=o_{p}\left(1\right)$,
and $\sum_{t=T_{1}}^{T_{2}}\left(\rho_{t}-\rho_{0,n}\right)^{2}c_{t-1,t-1-T_{0}}^{2}Y_{T_{0}}^{*2}=o_{p}\left(1\right)$,
by \eqref{eq:B2_analysis}, \eqref{eq:B3_analysis}, and \eqref{eq:B4_analysis},
respectively. These results and the CS inequality yield $A_{512}\pto0$.
Finally, using the CS inequality again gives $A_{513}\pto0$. Therefore,
by the CMT we have $A_{51}\pto1$.

For $A_{52}$, by \eqref{eq:yt_dagger_decomp} we have 
\begin{align}
\left(A_{52}\right)^{1/2} & =\abs{\left(nh\right)^{-1}\sum_{t=T_{1}}^{T_{2}}\left(Y_{t}-\rho_{0,n}Y_{t-1}\right)}\nonumber \\
 & =\abs{\left(nh\right)^{-1}\sum_{t=T_{1}}^{T_{2}}\left[\begin{array}{c}
\sigma_{t}U_{t}+\left(\mu_{t}-\rho_{0,n}\mu_{t-1}\right)+\left(\rho_{t}-\rho_{0,n}\right)Y_{t-1}^{0}\\
+\left(\rho_{t}-\rho_{0,n}\right)c_{t-1,t-1-T_{0}}Y_{T_{0}}^{*}
\end{array}\right]}\pto0,\label{eq:A52}
\end{align}
where the convergence holds by Lemma \ref{lem:MT-asympdist_components_local=000020to=000020unity=000020case}(d),
the results stated after (\ref{eq:a511=000020result}), and the CS
inequality.

The convergence of $A_{53}\pto0$ follows from $A_{51}\pto1$, $A_{52}\pto0$,
and the CS inequality. Combining the results gives $A_{5}\pto1$.

For $A_{6}$, we have 
\begin{align}
A_{6}= & \left(nh\right)^{-1}\sum_{t=T_{1}}^{T_{2}}\left[\left(\rho_{0,n}-\widehat{\rho}_{n\tau}\right)\left(Y_{t-1}-\ol Y_{nh,-1}\right)\right]^{2}\nonumber \\
= & \left(\rho_{0,n}-\widehat{\rho}_{n\tau}\right)^{2}\times\left(nh\right)^{-1}\sum_{t=T_{1}}^{T_{2}}\left(Y_{t-1}-\ol Y_{nh,-1}\right)^{2}\nonumber \\
= & \ O_{p}\left(\left(nh\right)^{-2}\right)\times O_{p}\left(nh\right)=o_{p}\left(1\right),\label{eq:A6=000020local=000020to=000020unity}
\end{align}
where the third equality holds by \eqref{eq:denominator=000020of=000020rho_hat}
and \eqref{eq:norm_rho_hat_result}.

In conclusion, we have proved \eqref{eq:sigma^2=000020numberator},
which leads to 
\begin{equation}
T_{n}\left(\rho_{0,n}\right)\dto\left(\int_{0}^{1}I_{D,\psi}^{*2}\left(s\right)ds\right)^{-1/2}\int_{0}^{1}I_{D,\psi}^{*}\left(s\right)dB\left(s\right)\label{eq:proof=000020Tn=000020ltu=000020case}
\end{equation}
for the case where $r_{0}\in\left(0,\infty\right)$.

Now, we prove the results of Theorem \ref{thm:MT-asym_rho_hat_local=000020to=000020unity=000020case}
for the case where $r_{0}=0.$ The idea is to use the same proof as
given above for $r_{0}>0,$ except with $Y_{t-1}$ $(=\mu_{t-1}+Y_{t-1}^{0}+c_{t-1,t-1-T_{0}}Y_{T_{0}}^{\ast})$
split into the two pieces $\mu_{t-1}+Y_{t-1}^{0}$ and $c_{t-1,t-1-T_{0}}Y_{T_{0}}^{\ast}.$
To deal with the first component $\mu_{t-1}+Y_{t-1}^{0},$ we use
the argument given above for $r_{0}>0$ using the results of Lemma
\ref{lem:MT-asympdist_components_local=000020to=000020unity=000020case}(g),
which holds even when $r_{0}=0.$ Then, we show that the second component
$c_{t-1,t-1-T_{0}}Y_{T_{0}}^{\ast}$ has a negligible asymptotic effect
because $c_{t-1,t-1-T_{0}}$ is quite close to the constant $1$ since
$r_{0}=0.$ The reason it has a negligible asymptotic effect is that
the LS regression includes a constant term, and hence, $c_{t-1,t-1-T_{0}}Y_{T_{0}}^{\ast}$
only enters the LS estimator $\widehat{\rho}_{n\tau}$ through $(c_{t-1,t-1-T_{0}}-(nh)^{-1}\sum_{s=T_{1}}^{T_{2}}c_{s-1,s-1-T_{0}})Y_{T_{0}}^{\ast},$
which is close to $(1-1)Y_{T_{0}}^{\ast}.$ Combining this with Lemma
\ref{lem:MT-Order=000020of=000020Y_T0}(b), we show that its impact
is asymptotically negligible.

We have 
\begin{eqnarray}
\max_{t\in[T_{1},T_{2}]}|1-c_{t-1,t-1-T_{0}}|\hspace{-0.08in} & \leq & \hspace{-0.08in}\max_{t\in[T_{1},T_{2}]}|1-\rho_{n\tau}^{t-1-T_{0}}|+O(nh/b_{n})\nonumber \\
 & = & \hspace{-0.08in}1-\rho_{n\tau}^{nh}+O(nh/b_{n})\nonumber \\
 & = & \hspace{-0.08in}1-\exp\{-\kappa_{n}^{\ast}(\tau)nh/b_{n}\}+O(nh/b_{n})\nonumber \\
 & = & \hspace{-0.08in}1-(1-\kappa_{n}^{\ast}(\tau)(nh/b_{n})\exp\{\zeta_{n}\})+O(nh/b_{n})\nonumber \\
 & = & \hspace{-0.08in}O(nh/b_{n}),\label{r_0=000020=00003D00003D=0000200=000020pf=00002010}
\end{eqnarray}
where the inequality holds by Lemma \ref{lem:MT-Max=000020Intertemporal=000020Difference}(c)
and $O(nh^{2}/b_{n})=O(nh/b_{n}),$ the first equality uses $T_{2}-1-T_{0}=2\lfloor nh/2\rfloor,$
which we denote by $nh$ for simplicity, the second equality holds
by the definition of $\kappa_{n}^{\ast}(\tau)$ given in (\ref{Defn=000020of=000020cappa*_n}),
the third equality holds by a mean value expansion with $\zeta_{n}$
lying between $0$ and $-\kappa_{n}^{\ast}(\tau)nh/b_{n},$ and hence,
$|\zeta_{n}|\leq\kappa_{n}^{\ast}(\tau)nh/b_{n}=O(nh/b_{n})=o(1)$
using $\kappa_{n}^{\ast}(\tau)=O(1)$ (by $\kappa_{n}(\tau)\leq C_{4}<\infty$
by part (ii) of $\Lambda_{n}$ and Lemma \ref{lem:A.1=000020Kappa_n*=000020go=000020to=000020Kappa_n})
and $nh/b_{n}\rightarrow r_{0}=0,$ and the last equality holds by
$\kappa_{n}^{\ast}(\tau)=O(1)$ and $\zeta_{n}\rightarrow0.$

Equation (\ref{r_0=000020=00003D00003D=0000200=000020pf=00002010})
yields 
\begin{equation}
\max_{t\in[T_{1},T_{2}]}\left\vert c_{t-1,t-1-T_{0}}-\overline{c}_{nh}\right\vert =O(nh/b_{n}),\text{ where }\overline{c}_{nh}:=(nh)^{-1}\sum_{s=T_{1}}^{T_{2}}c_{s-1,s-1-T_{0}}.\label{r_0=000020=00003D00003D=0000200=000020pf=00002011}
\end{equation}

Now, consider the denominator of the normalized LS estimator given
in (\ref{eq:MT-normalized=000020rho}) and (\ref{eq:denominator=000020of=000020rho_hat}).
We have 
\begin{eqnarray}
 &  & \hspace{-0.08in}(nh)^{-2}\sum_{t=T_{1}}^{T_{2}}(Y_{t-1}-\overline{Y}_{nh,-1})^{2}\nonumber \\
 & = & \hspace{-0.08in}(nh)^{-2}\sum_{t=T_{1}}^{T_{2}}\left(\mu_{t-1}+Y_{t-1}^{0}-(\overline{\mu}_{nh,-1}+\overline{Y}_{nh,-1}^{0})+\left(c_{t-1,t-1-T_{0}}-\overline{c}_{nh}\right)Y_{T_{0}}^{\ast}\right)^{2}\nonumber \\
 & = & \hspace{-0.08in}(nh)^{-2}\sum_{t=T_{1}}^{T_{2}}\left(\mu_{t-1}+Y_{t-1}^{0}-(\overline{\mu}_{nh,-1}+\overline{Y}_{nh,-1}^{0})\right)^{2}+(nh)^{-1}\max_{t\in[T_{1},T_{2}]}\left(c_{t-1,t-1-T_{0}}-\overline{c}_{nh}\right)^{2}Y_{T_{0}}^{\ast2}\nonumber \\
 &  & +2O_{p}(1)(nh)^{-1/2}\max_{t\in[T_{1},T_{2}]}\left\vert c_{t-1,t-1-T_{0}}-\overline{c}_{nh}\right\vert \cdot|Y_{T_{0}}^{\ast}|\nonumber \\
 & = & \hspace{-0.08in}(nh)^{-2}\sum_{t=T_{1}}^{T_{2}}\left(\mu_{t-1}+Y_{t-1}^{0}-(\overline{\mu}_{nh,-1}+\overline{Y}_{nh,-1}^{0})\right)^{2}\nonumber \\
 &  & +(nh)^{-1}O((nh/b_{n})^{2})o_{p}\left(b_{n}^{2}/nh\right)+O_{p}(1)(nh)^{-1/2}O(nh/b_{n})o_{p}\left(b_{n}/(nh)^{1/2}\right)\nonumber \\
 & \rightarrow & \hspace{-0.13in}_{d}\text{ }\int_{0}^{1}I_{D,\psi}^{\ast2}(s)ds,\label{r_0=000020=00003D00003D=0000200=000020pf=00002012}
\end{eqnarray}
where $\overline{\mu}_{nh,-1}=(nh)^{-1}\sum_{t=T_{1}}^{T_{2}}\mu_{t-1},$
$\overline{Y}_{nh,-1}^{0}=(nh)^{-1}\sum_{t=T_{1}}^{T_{2}}Y_{t-1}^{0},$
the first equality uses (\ref{eq:MT-tvp-model}) and (\ref{eq:MT-Y_t^0=000020and=000020Y_t^star}),
the second equality uses $(nh)^{-3/2}\sum_{t=T_{1}}^{T_{2}}\left\vert \mu_{t-1}+Y_{t-1}^{0}-(\overline{\mu}_{nh,-1}+\overline{Y}_{nh,-1}^{0})\right\vert =O_{p}(1)$
by Lemma \ref{lem:MT-asympdist_components_local=000020to=000020unity=000020case}(g)(b)
(which refers to Lemma \ref{lem:MT-asympdist_components_local=000020to=000020unity=000020case}(b)
adjusted according to Lemma \ref{lem:MT-asympdist_components_local=000020to=000020unity=000020case}(g),
i.e., with $Y_{t-1}$ replaced by $\mu_{t-1}+Y_{t-1}^{0})$ and with
absolute values added to $\mu_{t-1}+Y_{t-1}^{0}$ (which does not
affect the argument), the third equality holds using (\ref{r_0=000020=00003D00003D=0000200=000020pf=00002011})
and Lemma \ref{lem:MT-Order=000020of=000020Y_T0}(b), and the convergence
holds because $(nh)^{-1}O((nh/b_{n})^{2})o_{p}\left(b_{n}^{2}/nh\right)=o_{p}(1),$
$O_{p}(1)(nh)^{-1/2}O(nh/b_{n})\times\allowbreak o_{p}\left(b_{n}/(nh)^{1/2}\right)=o_{p}(1),$
and the first summand on the lhs converges in distribution to $\int_{0}^{1}I_{D,\psi}^{\ast2}(s)ds$
by the same argument as in (\ref{eq:denominator=000020of=000020rho_hat})
but with $\mu_{t-1}+Y_{t-1}^{0}$ in place of $Y_{t-1}$ and using
Lemma \ref{lem:MT-asympdist_components_local=000020to=000020unity=000020case}(g),
which uses $\mu_{t-1}+Y_{t-1}^{0}$ and applies when $r_{0}=0.$

Next, for the numerator of the normalized LS estimator $\widehat{\rho}_{n\tau}$
given in (\ref{eq:MT-normalized=000020rho}) and (\ref{eq:decompB}),
we decompose $Y_{t-1}-\overline{Y}_{nh,-1}$ into $\mu_{t-1}+Y_{t-1}^{0}-(\overline{\mu}_{nh,-1}+\overline{Y}_{nh,-1}^{0})$
and $(c_{t-1,t-1-T_{0}}-\overline{c}_{nh})Y_{T_{0}}^{\ast}$ in each
of the summands $A_{1},...,A_{4}$ in (\ref{eq:decompB}). Thus, we
write 
\begin{eqnarray}
 &  & A_{1}\overset{}{=}(nh)^{-1}\sum_{t=T_{1}}^{T_{2}}(Y_{t-1}-\overline{Y}_{nh,-1})\sigma_{t}U_{t}\overset{}{=}A_{11}+A_{12},\text{ where}\nonumber \\
 &  & A_{11}\overset{}{:=}(nh)^{-1}\sum_{t=T_{1}}^{T_{2}}(\mu_{t-1}+Y_{t-1}^{0})\sigma_{t}U_{t}-(nh)^{-1/2}(\overline{\mu}_{nh,-1}+\overline{Y}_{nh,-1}^{0}))(nh)^{-1/2}\sum_{t=T_{1}}^{T_{2}}\sigma_{t}U_{t}\text{ and}\nonumber \\
 &  & A_{12}\overset{}{:=}(nh)^{-1}\sum_{t=T_{1}}^{T_{2}}(c_{t-1,t-1-T_{0}}-\overline{c}_{nh})\sigma_{t}U_{t}Y_{T_{0}}^{\ast}.\label{r_0=000020=00003D00003D=0000200=000020pf=00002013}
\end{eqnarray}
We have 
\begin{equation}
A_{11}\rightarrow_{d}\int_{0}^{1}I_{\psi}^{\ast}(s)dB(s)-\int_{0}^{1}I_{\psi}^{\ast}(s)ds\int_{0}^{1}dB(s)=\int_{0}^{1}I_{D,\psi}^{\ast}(s)dB(s),\label{r_0=000020=00003D00003D=0000200=000020pf=00002014}
\end{equation}
where the convergence holds by Lemma \ref{lem:MT-asympdist_components_local=000020to=000020unity=000020case}(g)(b)
(which refers to Lemma \ref{lem:MT-asympdist_components_local=000020to=000020unity=000020case}(b)
adjusted according to Lemma \ref{lem:MT-asympdist_components_local=000020to=000020unity=000020case}(g)),
Lemma \ref{lem:MT-asympdist_components_local=000020to=000020unity=000020case}(g)(e),
and Lemma \ref{lem:MT-asympdist_components_local=000020to=000020unity=000020case}(d).

For $A_{12},$ we have 
\begin{eqnarray}
 &  & \hspace{-0.08in}Var\left((nh)^{-1}\sum_{t=T_{1}}^{T_{2}}(c_{t-1,t-1-T_{0}}-\overline{c}_{nh})\sigma_{t}U_{t}\right)\overset{}{=}(nh)^{-2}\sum_{t=T_{1}}^{T_{2}}(c_{t-1,t-1-T_{0}}-\overline{c}_{nh})^{2}\sigma_{t}^{2}\nonumber \\
 & = & \hspace{-0.08in}O((nh)^{-1})\max_{t\in[T_{1},T_{2}]}(c_{t-1,t-1-T_{0}}-\overline{c}_{nh})^{2}=O((nh)^{-1})O((nh/b_{n})^{2})=O(nh/b_{n}^{2}),\nonumber \\
\label{r_0=000020=00003D00003D=0000200=000020pf=00002015}
\end{eqnarray}
where the third equality uses (\ref{r_0=000020=00003D00003D=0000200=000020pf=00002011}).
Hence, 
\begin{equation}
A_{12}=O_{p}\left(\frac{(nh)^{1/2}}{b_{n}}\right)Y_{T_{0}}^{\ast}=O_{p}\left(\frac{(nh)^{1/2}}{b_{n}}\frac{b_{n}}{nh^{1/2}}\right)=O_{p}\left(\frac{1}{n^{1/2}}\right)=o_{p}(1),\label{r_0=000020=00003D00003D=0000200=000020pf=00002016}
\end{equation}
where the second equality uses Lemma \ref{lem:MT-Order=000020of=000020Y_T0}(b).
Combining (\ref{r_0=000020=00003D00003D=0000200=000020pf=00002013})--(\ref{r_0=000020=00003D00003D=0000200=000020pf=00002016})
gives: $A_{1}\rightarrow_{d}\int_{0}^{1}I_{D,\psi}^{\ast}(s)dB(s)$
in the $r_{0}=0$ case, just as in the $r_{0}>0$ case considered
in (\ref{eq:B1_analysis}).

When $r_{0}=0,$ we have $A_{2}=o_{p}(1)$ and $A_{3}=o_{p}(1)$ (where
$A_{2}$ and $A_{3}$ are defined in (\ref{eq:decompB})) by the same
arguments as in (\ref{eq:B2_analysis}) and (\ref{eq:B3_analysis})
using $(nh)^{-2}\sum_{t=T_{1}}^{T_{2}}(Y_{t-1}-\overline{Y}_{nh,-1})^{2}=O_{p}(1)$
by (\ref{r_0=000020=00003D00003D=0000200=000020pf=00002012}).

When $r_{0}=0,$ for $A_{4}$ (defined in (\ref{eq:decompB})), we
have 
\begin{eqnarray}
|A_{4}|^{2}\hspace{-0.08in} & \leq & \hspace{-0.08in}\left[(nh)^{-2}\sum_{t=T_{1}}^{T_{2}}(Y_{t-1}-\overline{Y}_{nh,-1})^{2}\right]\left[nh\max_{t\in[T_{1},T_{2}]}(\rho_{t}-\rho_{0,n})^{2}\right]Y_{T_{0}}^{\ast2}\nonumber \\
 & = & \hspace{-0.08in}O_{p}(1)nhO(n^{-2})O_{p}(n)=O_{p}(h)=o_{p}(1),\label{r_0=000020=00003D00003D=0000200=000020pf=00002017}
\end{eqnarray}
where the inequality holds by the first three lines of (\ref{eq:B4_analysis})
and the first equality uses (\ref{eq:=000020(rho_t-rho_0)^2=000020bound}),
(\ref{r_0=000020=00003D00003D=0000200=000020pf=00002012}), and Lemma
\ref{lem:MT-Order=000020of=000020Y_T0}(a).

This completes the proof of (\ref{eq:numerator_asym}) concerning
the numerator of the normalized LS estimator in the $r_{0}=0$ case.
Combined with the result for the denominator in (\ref{r_0=000020=00003D00003D=0000200=000020pf=00002012}),
this establishes the result of (\ref{eq:norm_rho_hat_result}) for
the normalized LS estimator in the $r_{0}=0$ case.

For the t-statistic, as in (\ref{eq:T_n=000020def}) and (\ref{eq:sigma^2=000020numberator}),
it remains to show that $\widehat{\sigma}_{n\tau}^{2}\rightarrow_{p}1.$
For $r_{0}=0,$ this holds by the same argument as given in (\ref{eq:decomp=000020num=000020of=000020sigma=0000202})--(\ref{eq:A6=000020local=000020to=000020unity})
for the $r_{0}>0$ with the only change needed being that $(nh)^{-1}\sum_{t=T_{1}}^{T_{2}}(Y_{t-1}-\overline{Y}_{nh,-1})^{2}=O_{p}(nh)$
in the third equality of (\ref{eq:A6=000020local=000020to=000020unity})
by (\ref{r_0=000020=00003D00003D=0000200=000020pf=00002012}) when
$r_{0}=0,$ rather than by (\ref{eq:denominator=000020of=000020rho_hat}).

The subsequence versions of Lemma \ref{lem:MT-y0limit}, Lemma \ref{lem:MT-asympdist_components_local=000020to=000020unity=000020case},
and Theorem \ref{thm:MT-asym_rho_hat_local=000020to=000020unity=000020case},
see Remark \ref{rem:MT-subseq=000020loc=000020unity=000020thm}, which
have $\{p_{n}\}_{n\geq1}$ in place of $\{n\}_{n\geq1},$ are proved
by replacing $n$ by $p_{n}$ and $h=h_{n}$ by $h_{p_{n}}$ throughout
the proofs above. 
\end{proof}

\subsection{\protect\label{subsec:Proof-of-Lemma-Stationary-Components}Proof
of Lemma \ref{lem:MT-Asymptotic=000020Properties=000020of=000020the=000020Components=000020Stationary}}

In this section, for notational simplicity in the proof, we assume
that $\sigma_{0}^{2}\left(\tau\right)=1.$ 
\begin{proof}[\textbf{Proof of Lemma \ref{lem:MT-Asymptotic=000020Properties=000020of=000020the=000020Components=000020Stationary}$\left(\text{a}\right)$}]
To prove part (a), we let $\sum_{j=0}^{-1}c_{t,j}\sigma_{t-j}U_{t-j}=0$.
Then, by \eqref{eq:MT-decompY*}, we have 
\begin{align}
 & \ \E\left[\left(1-\rho_{0,n}^{2}\right)^{1/2}\left(nh\right)^{-1}\sum_{t=T_{1}}^{T_{2}}Y_{t-1}^{0}\right]^{2}\nonumber \\
= & \ \E\left[\left(1-\rho_{0,n}^{2}\right)^{1/2}\left(nh\right)^{-1}\sum_{t=T_{1}-1}^{T_{2}-1}\left(\sum_{j=0}^{t-T_{1}}c_{t,j}\sigma_{t-j}U_{t-j}\right)\right]^{2}\nonumber \\
= & \left(1-\rho_{0,n}^{2}\right)\left(nh\right)^{-2}\sum_{t=T_{1}}^{T_{2}-1}\E\left(\sum_{j=0}^{t-T_{1}}c_{t,j}\sigma_{t-j}U_{t-j}\right)^{2}\nonumber \\
 & +\left(1-\rho_{0,n}^{2}\right)\left(nh\right)^{-2}\sum_{t,s=T_{1}}^{T_{2}-1}\E\left(\sum_{i=0}^{t-T_{1}}c_{t,i}\sigma_{t-i}U_{t-i}\right)\left(\sum_{j=0}^{s-T_{1}}c_{s,j}\sigma_{s-j}U_{s-j}\right)\ind\left\{ t\neq s\right\} \nonumber \\
=: & \ A_{e1}+A_{e2},\label{eq:lemma5(a)_decomp-1}
\end{align}
where the first equality holds by \eqref{eq:MT-Y_t^0=000020and=000020Y_t^star}.
We show $\abs{A_{e1}}=o\left(1\right)$ and $\abs{A_{e2}}=o\left(1\right)$,
which establish part (a) by Markov's inequality.

For $\abs{A_{e1}}$, we have 
\begin{align}
\abs{A_{e1}} & =\left(1-\rho_{0,n}^{2}\right)\left(nh\right)^{-2}\sum_{t=T_{1}}^{T_{2}-1}\E\left(\sum_{j=0}^{t-T_{1}}c_{t,j}\sigma_{t-j}U_{t-j}\right)^{2}\nonumber \\
 & =\left(1-\rho_{0,n}^{2}\right)\left(nh\right)^{-2}\sum_{t=T_{1}}^{T_{2}-1}\sum_{j=0}^{t-T_{1}}\E U_{t-j}^{2}c_{t,j}^{2}\sigma_{t-j}^{2}\nonumber \\
 & \leq\left(1-\rho_{0,n}^{2}\right)\left(nh\right)^{-2}\sum_{t=T_{1}}^{T_{2}-1}\sum_{j=0}^{t-T_{1}}\ol{\rho}_{n}^{2j}\sigma_{t-j}^{2}\nonumber \\
 & \leq\max_{t\in\left[T_{1},T_{2}\right]}\sigma_{t}^{2}\left(1-\rho_{0,n}^{2}\right)\left(nh\right)^{-2}\sum_{t=T_{1}}^{T_{2}-1}\sum_{j=0}^{\infty}\ol{\rho}_{n}^{2j}\nonumber \\
 & =\ O\left(1\right)O\left(b_{n}^{-1}\right)\left(nh\right)^{-2}O\left(nhb_{n}\right)=o\left(1\right),\label{eq:R1=000020conv=000020to=0000200-1}
\end{align}
where the second equality uses the fact that $\left\{ U_{t}\right\} _{t=1}^{n}$
is a martingale difference sequence, the first inequality holds by
\eqref{eq:MT-c_tj=000020bounded-gto_1}, and the second last equality
holds by \eqref{eq:sigma_t=000020converge=000020to=000020sigma_0}
and \eqref{eq:MT-ps=000020rho^2t}.

For $\abs{A_{e2}}$, we have 
\begin{align}
\abs{A_{e2}}= & \abs{\begin{array}{c}
\left(1-\rho_{0,n}^{2}\right)\left(nh\right)^{-2}\sum_{t,s=T_{1}}^{T_{2}-1}\E\left(\sum_{i=0}^{t-T_{1}}c_{t,i}\sigma_{t-i}U_{t-i}\right)\\
\times\left(\sum_{j=0}^{s-T_{1}}c_{s,j}\sigma_{s-j}U_{s-j}\right)\ind\left\{ t\neq s\right\} 
\end{array}}\nonumber \\
= & \ 2\abs{\left(1-\rho_{0,n}^{2}\right)\left(nh\right)^{-2}\sum_{t=T_{1}+1}^{T_{2}-1}\sum_{s=T_{1}}^{t-1}\E\sum_{i=0}^{t-T_{1}}c_{t,i}\sigma_{t-i}U_{t-i}\sum_{j=0}^{s-T_{1}}c_{s,j}\sigma_{s-j}U_{s-j}}\nonumber \\
= & \ 2\abs{\left(1-\rho_{0,n}^{2}\right)\left(nh\right)^{-2}\sum_{t=T_{1}+1}^{T_{2}-1}\sum_{s=T_{1}}^{t-1}\sum_{j=0}^{s-T_{1}}\E U_{s-j}^{2}c_{t,j+\left(t-s\right)}c_{s,j}\sigma_{s-j}^{2}}\nonumber \\
\leq & \ 2\max_{t\in\left[T_{1},T_{2}\right]}\sigma_{t}^{2}\left(1-\rho_{0,n}^{2}\right)\left(nh\right)^{-2}\abs{\sum_{t=T_{1}+1}^{T_{2}-1}\sum_{s=T_{1}}^{t-1}\ol{\rho}_{n}^{t-s}\sum_{j=0}^{s-T_{1}}\ol{\rho}_{n}^{2j}}\nonumber \\
\leq & \ O\left(1\right)O\left(b_{n}^{-1}\right)\left(nh\right)^{-2}O\left(nhb_{n}^{2}\right)=O\left(b_{n}/nh\right)=o\left(1\right),\label{eq:Aa2=000020conv=000020to=0000200-1}
\end{align}
where the first equality uses the fact that $t$ and $s$ are symmetric,
the last inequality holds by (\ref{eq:sigma_t=000020converge=000020to=000020sigma_0}),
\eqref{eq:MT-ps=000020rho^2t}, and \eqref{eq:MT-ps=000020rho^(t-s)},
and the last equality holds by $b_{n}=o\left(nh\right)$.

Therefore, by Markov's inequality, we have 
\begin{equation}
\left(1-\rho_{0,n}^{2}\right)^{1/2}\left(nh\right)^{-1}\sum_{t=T_{1}}^{T_{2}}Y_{t-1}^{0}\pto0.\label{eq:lem=0000208=000020Y_t^0=000020part=000020a}
\end{equation}
\end{proof}
\begin{proof}[\textbf{Proof of Lemma \ref{lem:MT-Asymptotic=000020Properties=000020of=000020the=000020Components=000020Stationary}$\left(\text{b}\right)$}]
To prove part (b), we have 
\begin{align}
 & \E\left[\left(1-\rho_{0,n}^{2}\right)\left(nh\right)^{-1}\sum_{t=T_{1}}^{T_{2}}\left(Y_{t-1}^{0}\right)^{2}\right]\nonumber \\
= & \left(1-\rho_{0,n}^{2}\right)\left(nh\right)^{-1}\sum_{t=T_{1}-1}^{T_{2}-1}\E\left(\sum_{j=0}^{t-T_{1}}c_{t,j}\sigma_{t-j}U_{t-j}\right)^{2}\nonumber \\
= & \left(1-\rho_{0,n}^{2}\right)\left(nh\right)^{-1}\sum_{t=T_{1}}^{T_{2}-1}\sum_{j=0}^{t-T_{1}}c_{t,j}^{2}\sigma_{t-j}^{2}\nonumber \\
= & \left(1-\rho_{0,n}^{2}\right)\left(nh\right)^{-1}\sum_{t=T_{1}}^{T_{2}-1}\sum_{j=0}^{t-T_{1}}\rho_{0,n}^{2j}+\left(1-\rho_{0,n}^{2}\right)\left(nh\right)^{-1}\sum_{t=T_{1}}^{T_{2}-1}\sum_{j=0}^{t-T_{1}}\left(c_{t,j}^{2}-\rho_{0,n}^{2j}\right)\sigma_{t-j}^{2}\nonumber \\
 & +\left(1-\rho_{0,n}^{2}\right)\left(nh\right)^{-1}\sum_{t=T_{1}}^{T_{2}-1}\sum_{j=0}^{t-T_{1}}\rho_{0,n}^{2j}\left(\sigma_{t-j}^{2}-\sigma_{0}^{2}\left(\tau\right)\right)\nonumber \\
=: & \ A_{b1}+A_{b2}+A_{b3}.
\end{align}
We show $A_{b1}=1+o\left(1\right)$, $A_{b2}=o\left(1\right)$, and
$A_{b3}=o\left(1\right)$.

For $A_{b1}$, we have 
\begin{align}
A_{b1} & =\left(1-\rho_{0,n}^{2}\right)\left(nh\right)^{-1}\sum_{t=T_{1}}^{T_{2}-1}\sum_{j=0}^{t-T_{1}}\rho_{0,n}^{2j}=\left(1-\rho_{0,n}^{2}\right)\left(nh\right)^{-1}\sum_{t=T_{1}}^{T_{2}-1}\left(1-\rho_{0,n}^{2\left(t-T_{1}+1\right)}\right)\left(1-\rho_{0,n}^{2}\right)^{-1}\nonumber \\
 & =1-\left(nh\right)^{-1}\sum_{k=1}^{\lfloor nh\rfloor}\rho_{0,n}^{2k}=1+O\left(b_{n}/nh\right)=1+o\left(1\right),\label{eq:A_b1-1}
\end{align}
where the third equality uses the change of coordinates $k=t-T_{1}+1$
and the second last equality uses \eqref{eq:MT-ps=000020rho^2t}.

For $A_{b2}$, we have 
\begin{align}
\abs{A_{b2}}= & \left(1-\rho_{0,n}^{2}\right)\left(nh\right)^{-1}\abs{\sum_{t=T_{1}}^{T_{2}-1}\sum_{j=0}^{t-T_{1}}\left(c_{t,j}^{2}-\rho_{0,n}^{2j}\right)\sigma_{t-j}^{2}}\nonumber \\
\leq & \max_{t\in\left[T_{1},T_{2}\right]}\sigma_{t}^{2}\left(1-\rho_{0,n}^{2}\right)\left(nh\right)^{-1}\sum_{t=T_{1}}^{T_{2}-1}\sum_{j=0}^{t-T_{1}}\abs{c_{t,j}^{2}-\rho_{0,n}^{2j}}\nonumber \\
\leq & \ O\left(1\right)\left(1-\rho_{0,n}^{2}\right)\left(nh\right)^{-1}\sum_{t=T_{1}}^{T_{2}-1}\sum_{j=0}^{t-T_{1}}\abs{c_{t,j}-\rho_{0,n}^{j}}\nonumber \\
\leq & \ O\left(1\right)O\left(b_{n}^{-1}\right)\left(nh\right)^{-1}\sum_{t=T_{1}}^{T_{2}-1}\sum_{j=0}^{t-T_{1}}j\ol{\rho}_{n}^{j-1}L_{1}h/b_{n}=O\left(h\right)=o\left(1\right),\label{eq:A_b2-1}
\end{align}
where the second inequality holds by \eqref{eq:sigma_t^2=000020converge=000020to=000020sigma_0},
the last inequality uses \eqref{eq:MT-ctj=000020-=000020rho^j=000020bound-gto_1},
and the second last equality holds by \eqref{eq:MT-ps=000020t*rho^t}.

For $A_{b3}$, we have 
\begin{align}
\abs{A_{b3}} & =\left(1-\rho_{0,n}^{2}\right)\left(nh\right)^{-1}\abs{\sum_{t=T_{1}}^{T_{2}-1}\sum_{j=0}^{t-T_{1}}\rho_{0,n}^{2j}\left(\sigma_{t-j}^{2}-\sigma_{0}^{2}\left(\tau\right)\right)}\nonumber \\
 & \leq\max_{t\in\left[T_{1},T_{2}\right]}\abs{\sigma_{t}^{2}-\sigma_{0}^{2}\left(\tau\right)}\left(1-\rho_{0,n}^{2}\right)\left(nh\right)^{-1}\sum_{t=T_{1}}^{T_{2}-1}\sum_{j=0}^{t-T_{1}}\rho_{0,n}^{2j}\nonumber \\
 & =\ o\left(1\right)O\left(b_{n}^{-1}\right)\left(nh\right)^{-1}O\left(nhb_{n}\right)=o\left(1\right),\label{eq:A_b3-1}
\end{align}
where the second last equality uses \eqref{eq:sigma_t^2=000020converge=000020to=000020sigma_0}
and \eqref{eq:MT-ps=000020rho^2t}.

Combining \eqref{eq:A_b1-1}, \eqref{eq:A_b2-1}, and \eqref{eq:A_b3-1},
we obtain 
\begin{equation}
\E\left[\left(1-\rho_{0,n}^{2}\right)\left(nh\right)^{-1}\sum_{t=T_{1}}^{T_{2}}\left(Y_{t-1}^{0}\right)^{2}\right]=1+o\left(1\right).\label{eq:=000020expectation=000020of=000020sum=000020of=000020Yt-1=000020squared-1}
\end{equation}

Next, we show $\E\left[\left(1-\rho_{0,n}^{2}\right)\left(nh\right)^{-1}\sum_{t=T_{1}}^{T_{2}}\left(Y_{t-1}^{0}\right)^{2}\right]^{2}=1+o\left(1\right)$
by observing 
\begin{align}
 & \E\left[\left(1-\rho_{0,n}^{2}\right)\left(nh\right)^{-1}\sum_{t=T_{1}}^{T_{2}}\left(Y_{t-1}^{0}\right)^{2}\right]^{2}\nonumber \\
= & \left(1-\rho_{0,n}^{2}\right)^{2}\left(nh\right)^{-2}\E\left[\sum_{t_{1},t_{2}=T_{1}}^{T_{2}}\left(Y_{t_{1}-1}^{0}\right)^{2}\left(Y_{t_{2}-1}^{0}\right)^{2}\right]\nonumber \\
= & \left(1-\rho_{0,n}^{2}\right)^{2}\left(nh\right)^{-2}\E\left[\sum_{t_{1},t_{2}=T_{1}-1}^{T_{2}-1}\left(Y_{t_{1}}^{0}\right)^{2}\left(Y_{t_{2}}^{0}\right)^{2}\right]\nonumber \\
= & \left(1-\rho_{0,n}^{2}\right)^{2}\left(nh\right)^{-2}\sum_{t_{1},t_{2}=T_{1}}^{T_{2}-1}\E\left[\begin{array}{c}
\left(\sum_{i=0}^{t_{1}-T_{1}}c_{t_{1},i}\sigma_{t_{1}-i}U_{t_{1}-i}\right)^{2}\\
\times\left(\sum_{j=0}^{t_{2}-T_{1}}c_{t_{2},j}\sigma_{t_{2}-j}U_{t_{2}-j}\right)^{2}
\end{array}\right]\nonumber \\
= & \left(1-\rho_{0,n}^{2}\right)^{2}\left(nh\right)^{-2}\sum_{t_{1},t_{2}=T_{1}}^{T_{2}-1}\sum_{i_{1},i_{2}=0}^{t_{1}-T_{1}}\sum_{j_{1},j_{2}=0}^{t_{2}-T_{1}}\left[\begin{array}{c}
c_{t_{1},i_{1}}c_{t_{1},i_{2}}c_{t_{2},j_{1}}c_{t_{2},j_{2}}\sigma_{t_{1}-i_{1}}\sigma_{t_{1}-i_{2}}\sigma_{t_{2}-j_{1}}\sigma_{t_{2}-j_{2}}\\
\times\E\left(U_{t_{1}-i_{1}}U_{t_{1}-i_{2}}U_{t_{2}-j_{1}}U_{t_{2}-j_{2}}\right)
\end{array}\right],\label{eq:a.11.10-Ab3}
\end{align}
where the third equality holds by $Y_{T_{0}}^{0}=0$ and \eqref{eq:MT-Y_t^0=000020and=000020Y_t^star}.
All expectations in the last line are zero unless (i) all the indices
on the four innovation terms coincide or (ii) there are two groups
of two indices that each coincide or (iii) three larger indices coincide.

In case (i), we must have $i_{1}=i_{2}=i$, $j_{1}=j_{2}=j$ and $t_{1}-i_{1}=t_{2}-j_{1}$,
which implies 
\begin{align}
 & \ c_{t_{1},i_{1}}c_{t_{1},i_{2}}c_{t_{2},j_{1}}c_{t_{2},j_{2}}\E U_{t_{1}-i_{1}}U_{t_{1}-i_{2}}U_{t_{2}-j_{1}}U_{t_{2}-j_{2}}\sigma_{t_{1}-i_{1}}\sigma_{t_{1}-i_{2}}\sigma_{t_{2}-j_{1}}\sigma_{t_{2}-j_{2}}\nonumber \\
= & \ c_{t_{1},i}^{2}c_{t_{2},j}^{2}\E U_{t-i}^{4}\ind\left\{ t_{1}-i=t_{2}-j\right\} \sigma_{t_{1}-i}^{4}.\label{eq:case(i)-2}
\end{align}
Substituting \eqref{eq:case(i)-2} into the right-hand side of \eqref{eq:a.11.10-Ab3},
we have 
\begin{align}
 & \left(1-\rho_{0,n}^{2}\right)^{2}\left(nh\right)^{-2}\sum_{t_{1},t_{2}=T_{1}}^{T_{2}-1}\sum_{i=0}^{t_{1}-T_{1}}\sum_{j=0}^{t_{2}-T_{1}}c_{t_{1},i}^{2}c_{t_{2},j}^{2}\E U_{t-i}^{4}\ind\left\{ t_{1}-i=t_{2}-j\right\} \sigma_{t_{1}-i}^{4}\nonumber \\
\leq & \left(1-\rho_{0,n}^{2}\right)^{2}\max_{t\in\left[T_{1},T_{2}\right]}\sigma_{t}^{4}M\left(nh\right)^{-2}\sum_{t_{1},t_{2}=T_{1}}^{T_{2}-1}\sum_{i=0}^{t_{1}-T_{1}}c_{t_{1},i}^{2}\sum_{j=0}^{t_{2}-T_{1}}c_{t_{2},j}^{2}\ind\left\{ t_{1}-i=t_{2}-j\right\} \nonumber \\
\leq & \ O\left(1\right)\left(1-\rho_{0,n}^{2}\right)^{2}\left\{ 2\left(nh\right)^{-2}\sum_{t_{1}>t_{2}=T_{1}}^{T_{2}-1}\sum_{i=t_{1}-t_{2}}^{t_{1}-T_{1}}\ol{\rho}_{n}^{4i-2\left(t_{1}-t_{2}\right)}+\left(nh\right)^{-2}\sum_{t=T_{1}}^{T_{2}-1}\sum_{i=0}^{t-T_{1}}\ol{\rho}_{n}^{4i}\right\} \nonumber \\
= & \ O\left(1\right)\left(1-\rho_{0,n}^{2}\right)^{2}\left(nh\right)^{-2}\sum_{t_{1}>t_{2}=T_{1}}^{T_{2}-1}\ol{\rho}_{n}^{2\left(t_{1}-t_{2}\right)}\sum_{k=0}^{t_{2}-T_{1}}\ol{\rho}_{n}^{4k}+O\left(1\right)\left(1-\rho_{0,n}^{2}\right)^{2}\left(nh\right)^{-2}\sum_{t=T_{1}}^{T_{2}-1}\sum_{i=0}^{t-T_{1}}\ol{\rho}_{n}^{4i}\nonumber \\
\leq & \ O\left(1\right)\left(1-\rho_{0,n}^{2}\right)^{2}\left(nh\right)^{-2}\sum_{t_{1}=T_{1}+1}^{T_{2}-1}\sum_{l=1}^{t_{1}-T_{1}}\ol{\rho}_{n}^{2l}\sum_{k=0}^{\infty}\ol{\rho}_{n}^{4k}+O\left(1\right)\left(1-\rho_{0,n}^{2}\right)^{2}\left(nh\right)^{-2}\sum_{t=T_{1}}^{T_{2}-1}\sum_{i=0}^{\infty}\ol{\rho}_{n}^{4i}\nonumber \\
= & \ O\left(\left(nh\right)^{-1}\right)=o\left(1\right),\label{eq:case(i)-1}
\end{align}
where the first inequality holds by $\E U_{t}^{4}<M$, the second
inequality uses \eqref{eq:sigma_t=000020converge=000020to=000020sigma_0}
and \eqref{eq:MT-c_tj=000020bounded-gto_1}, the first equality uses
the change of coordinates $k=i-\left(t_{1}-t_{2}\right)$, and the
second last equality holds by $\sum_{k=0}^{\infty}\ol{\rho}_{n}^{4k}=\left(1-\ol{\rho}_{n}^{4}\right)^{-1}=O\left(b_{n}\right)$
and \eqref{eq:MT-ps=000020rho^2t}.

In case (ii), we must have (ii1) ($i_{1}=i_{2}=i$ and $j_{1}=j_{2}=j$
and $t_{1}-i\neq t_{2}-j$) or (ii2) ($t_{1}-i_{1}=t_{2}-j_{1}$ and
$t_{1}-i_{2}=t_{2}-j_{2}$ and $i_{1}\neq i_{2}$) or (ii3) ($t_{1}-i_{1}=t_{2}-j_{2}$
and $t_{1}-i_{2}=t_{2}-j_{1}$ and $i_{1}\neq i_{2}$).

In case (ii1), we have 
\begin{align}
 & \ c_{t_{1},i_{1}}c_{t_{1},i_{2}}c_{t_{2},j_{1}}c_{t_{2},j_{2}}\E U_{t_{1}-i_{1}}U_{t_{1}-i_{2}}U_{t_{2}-j_{1}}U_{t_{2}-j_{2}}\sigma_{t_{1}-i_{1}}\sigma_{t_{1}-i_{2}}\sigma_{t_{2}-j_{1}}\sigma_{t_{2}-j_{2}}\nonumber \\
= & \ c_{t_{1},i}^{2}c_{t_{2},j}^{2}\E U_{t_{1}-i}^{2}U_{t_{2}-j}^{2}\ind\left\{ t_{1}-i\neq t_{2}-j\right\} \sigma_{t_{1}-i}^{2}\sigma_{t_{2}-j}^{2}.\label{eq:case(ii)-1-A.11.13}
\end{align}
Substituting \eqref{eq:case(ii)-1-A.11.13} into the right-hand side
of \eqref{eq:a.11.10-Ab3}, we have 
\begin{align}
 & \left(1-\rho_{0,n}^{2}\right)^{2}\left(nh\right)^{-2}\sum_{t_{1},t_{2}=T_{1}}^{T_{2}-1}\sum_{i=0}^{t_{1}-T_{1}}\sum_{j=0}^{t_{2}-T_{1}}c_{t_{1},i}^{2}c_{t_{2},j}^{2}\sigma_{t_{1}-i}^{2}\sigma_{t_{2}-j}^{2}\E\left(U_{t_{1}-i}^{2}U_{t_{2}-j}^{2}\right)\ind\left\{ t_{1}-i\neq t_{2}-j\right\} \nonumber \\
= & \left(1-\rho_{0,n}^{2}\right)^{2}\left(nh\right)^{-2}\sum_{t_{1},t_{2}=T_{1}}^{T_{2}-1}\sum_{i=0}^{t_{1}-T_{1}}\sum_{j=0}^{t_{2}-T_{1}}c_{t_{1},i}^{2}c_{t_{2},j}^{2}\sigma_{t_{1}-i}^{2}\sigma_{t_{2}-j}^{2}\ind\left\{ t_{1}-i\neq t_{2}-j\right\} \nonumber \\
= & \left(1-\rho_{0,n}^{2}\right)^{2}\left(nh\right)^{-2}\sum_{t_{1},t_{2}=T_{1}}^{T_{2}-1}\sum_{i=0}^{t_{1}-T_{1}}\sum_{j=0}^{t_{2}-T_{1}}c_{t_{1},i}^{2}c_{t_{2},j}^{2}\sigma_{t_{1}-i}^{2}\sigma_{t_{2}-j}^{2}\nonumber \\
 & -\left(1-\rho_{0,n}^{2}\right)^{2}\left(nh\right)^{-2}\sum_{t_{1},t_{2}=T_{1}}^{T_{2}-1}\sum_{i=0}^{t_{1}-T_{1}}\sum_{j=0}^{t_{2}-T_{1}}c_{t_{1},i}^{2}c_{t_{2},j}^{2}\sigma_{t_{1}-i}^{2}\sigma_{t_{2}-j}^{2}\ind\left\{ t_{1}-i=t_{2}-j\right\} \nonumber \\
= & \left(1-\rho_{0,n}^{2}\right)^{2}\left(nh\right)^{-2}\sum_{t_{1},t_{2}=T_{1}}^{T_{2}-1}\sum_{i=0}^{t_{1}-T_{1}}\sum_{j=0}^{t_{2}-T_{1}}\rho_{0,n}^{2\left(i+j\right)}\nonumber \\
 & +\left(1-\rho_{0,n}^{2}\right)^{2}\left(nh\right)^{-2}\sum_{t_{1},t_{2}=T_{1}}^{T_{2}-1}\sum_{i=0}^{t_{1}-T_{1}}\sum_{j=0}^{t_{2}-T_{1}}\left(c_{t_{1},i}^{2}c_{t_{2},j}^{2}-\rho_{0,n}^{2\left(i+j\right)}\right)\nonumber \\
 & +\left(1-\rho_{0,n}^{2}\right)^{2}\left(nh\right)^{-2}\sum_{t_{1},t_{2}=T_{1}}^{T_{2}-1}\sum_{i=0}^{t_{1}-T_{1}}\sum_{j=0}^{t_{2}-T_{1}}c_{t_{1},i}^{2}c_{t_{2},j}^{2}\left(\sigma_{t_{1}-i}^{2}\sigma_{t_{2}-j}^{2}-\sigma_{0}^{4}\left(\tau\right)\right)-o(1)\nonumber \\
=: & A_{b4}+A_{b5}+A_{b6}+o\left(1\right),\label{eq:A_b=000020decomp-1}
\end{align}
where the first equality holds by $\E\left[\rest{U_{t}^{2}}\mathscr{G}_{t-1}\right]=1$
a.s. and the third equality holds by \eqref{eq:case(i)-1}.

Because $A_{b4}=A_{b1}^{2}$, we obtain $A_{b4}=1+o\left(1\right)$
from \eqref{eq:A_b1-1}. Thus, for case (ii1), we only need to show
$A_{b5}=o\left(1\right)$ and $A_{b6}=o\left(1\right)$.

For $A_{b5}$, we have 
\begin{align}
\abs{A_{b5}}\leq & \left(1-\rho_{0,n}^{2}\right)^{2}\left(nh\right)^{-2}\sum_{t_{1},t_{2}=T_{1}}^{T_{2}-1}\sum_{i=0}^{t_{1}-T_{1}}\sum_{j=0}^{t_{2}-T_{1}}\left[c_{t_{2},j}^{2}\abs{c_{t_{1},i}^{2}-\rho_{0,n}^{2i}}+\rho_{0,n}^{2i}\abs{c_{t_{2},j}^{2}-\rho_{0,n}^{2j}}\right]\nonumber \\
\leq & \left(1-\rho_{0,n}^{2}\right)^{2}\left(nh\right)^{-2}\sum_{t_{1},t_{2}=T_{1}}^{T_{2}-1}\sum_{i=0}^{t_{1}-T_{1}}\sum_{j=0}^{t_{2}-T_{1}}\left(\ol{\rho}_{n}^{2j}i\ol{\rho}_{n}^{i-1}+\ol{\rho}_{n}^{2i}j\ol{\rho}_{n}^{j-1}\right)2L_{1}h/b_{n}\nonumber \\
= & \left(1-\rho_{0,n}^{2}\right)^{2}\left(4L_{1}h/b_{n}\right)\left(nh\right)^{-2}\sum_{t_{1},t_{2}=T_{1}}^{T_{2}-1}\sum_{i=0}^{t_{1}-T_{1}}i\ol{\rho}_{n}^{i-1}\sum_{j=0}^{t_{2}-T_{1}}\ol{\rho}_{n}^{2j}\nonumber \\
= & \ O\left(b_{n}^{-2}\left(h/b_{n}\right)\left(nh\right)^{-2}\left(nh\right)^{2}b_{n}^{2}b_{n}\right)=O\left(h\right)=o\left(1\right),\label{eq:A_b5-1}
\end{align}
where the first inequality uses the triangle inequality, the second
inequality holds by \eqref{eq:MT-c_tj=000020bounded-gto_1} and \eqref{eq:MT-ctj=000020-=000020rho^j=000020bound-gto_1},
the first equality uses the fact that $t_{1}$ and $t_{2}$ are symmetric,
and the second equality holds by \eqref{eq:MT-ps=000020rho^2t} and
\eqref{eq:MT-ps=000020t*rho^t}.

For $A_{b6}$, we have 
\begin{align}
\abs{A_{b6}} & \leq\left(1-\rho_{0,n}^{2}\right)^{2}\left(nh\right)^{-2}\sum_{t_{1},t_{2}=T_{1}}^{T_{2}-1}\sum_{i=0}^{t_{1}-T_{1}}\sum_{j=0}^{t_{2}-T_{1}}c_{t_{1},i}^{2}c_{t_{2},j}^{2}\abs{\sigma_{t_{1}-i}^{2}\sigma_{t_{2}-j}^{2}-\sigma_{0}^{4}\left(\tau\right)}\nonumber \\
 & \leq O\left(b_{n}^{-2}\right)\max_{t_{1},t_{2}\in\left[T_{1},T_{2}\right]}\abs{\sigma_{t_{1}}^{2}\sigma_{t_{2}}^{2}-\sigma_{0}^{4}\left(\tau\right)}\left(nh\right)^{-2}\sum_{t_{1},t_{2}=T_{1}}^{T_{2}-1}\sum_{i=0}^{t_{1}-T_{1}}\sum_{j=0}^{t_{2}-T_{1}}\ol{\rho}_{n}^{2\left(i+j\right)}=O\left(h\right),\label{eq:A_b6-1}
\end{align}
where the second inequality uses \eqref{eq:MT-c_tj=000020bounded-gto_1}
and the equality holds by 
\begin{equation}
\sum_{t_{1},t_{2}=T_{1}}^{T_{2}-1}\sum_{i=0}^{t_{1}-T_{1}}\sum_{j=0}^{t_{2}-T_{1}}\ol{\rho}_{n}^{2\left(i+j\right)}\leq\left(nh\right)^{2}\left(1-\ol{\rho}_{n}^{2}\right)^{-2}=O\left(\left(nhb_{n}\right)^{2}\right)
\end{equation}
and 
\begin{align}
 & \max_{t,s\in\left[T_{1},T_{2}\right]}\abs{\sigma_{t}^{2}\sigma_{s}^{2}-\sigma_{0}^{4}\left(\tau\right)}\nonumber \\
\leq & \max_{s\in\left[T_{1},T_{2}\right]}\abs{\sigma_{s}^{2}-\sigma_{0}^{2}\left(\tau\right)}\max_{t\in\left[T_{1},T_{2}\right]}\sigma_{t}^{2}+\max_{t\in\left[T_{1},T_{2}\right]}\abs{\sigma_{t}^{2}-\sigma_{0}^{2}\left(\tau\right)}=O\left(h\right).
\end{align}
This completes case (ii1).

Since case (ii2) and (ii3) are symmetric, we only prove the result
for case (ii2) and show it is $o\left(1\right)$. Observe that when
$t_{1}-i_{1}=t_{2}-j_{1}$, $t_{1}-i_{2}=t_{2}-j_{2}$, and $i_{1}\neq i_{2}$,
\begin{align}
 & \ c_{t_{1},i_{1}}c_{t_{1},i_{2}}c_{t_{2},j_{1}}c_{t_{2},j_{2}}\E U_{t_{1}-i_{1}}U_{t_{1}-i_{2}}U_{t_{2}-j_{1}}U_{t_{2}-j_{2}}\sigma_{t_{1}-i_{1}}\sigma_{t_{1}-i_{2}}\sigma_{t_{2}-j_{1}}\sigma_{t_{2}-j_{2}}\nonumber \\
= & \ c_{t_{1},i_{1}}c_{t_{1},i_{2}}c_{t_{2},i_{1}-\left(t_{1}-t_{2}\right)}c_{t_{2},i_{2}-\left(t_{1}-t_{2}\right)}\E U_{t_{1}-i_{1}}^{2}U_{t_{1}-i_{2}}^{2}\ind\left\{ i_{1}\neq i_{2}\right\} \sigma_{t_{1}-i_{1}}^{2}\sigma_{t_{1}-i_{2}}^{2}.\label{eq:case(ii)-2}
\end{align}
Substituting \eqref{eq:case(ii)-2} into the right-hand side of \eqref{eq:a.11.10-Ab3},
we have

\begin{align}
 & \left(1-\rho_{0,n}^{2}\right)^{2}\left(nh\right)^{-2}\sum_{t_{1},t_{2}=T_{1}}^{T_{2}-1}\sum_{i_{1},i_{2}=0}^{t_{1}-T_{1}}\sum_{j_{1},j_{2}=0}^{t_{2}-T_{1}}\left[\begin{array}{c}
c_{t_{1},i_{1}}c_{t_{1},i_{2}}c_{t_{2},i_{1}-\left(t_{1}-t_{2}\right)}c_{t_{2},i_{2}-\left(t_{1}-t_{2}\right)}\ind\left\{ i_{1}\neq i_{2}\right\} \\
\times\E U_{t_{1}-i_{1}}^{2}U_{t_{1}-i_{2}}^{2}\sigma_{t_{1}-i_{1}}^{2}\sigma_{t_{1}-i_{2}}^{2}
\end{array}\right]\nonumber \\
\leq & \ 2\left(1-\rho_{0,n}^{2}\right)^{2}\left(nh\right)^{-2}\sum_{t_{1}>t_{2}=T_{1}}^{T_{2}-1}\sum_{i_{1},i_{2}=t_{1}-t_{2}}^{t_{1}-T_{1}}c_{t_{1},i_{1}}c_{t_{1},i_{2}}c_{t_{2},t_{2}-\left(t_{1}-i_{1}\right)}c_{t_{2},t_{2}-\left(t_{1}-i_{2}\right)}\sigma_{t_{1}-i_{1}}^{2}\sigma_{t_{1}-i_{2}}^{2}\nonumber \\
 & +\left(1-\rho_{0,n}^{2}\right)^{2}\left(nh\right)^{-2}\sum_{t=T_{1}}^{T_{2}-1}\sum_{i_{1},i_{2}=0}^{t-T_{1}}c_{t,i_{1}}^{2}c_{t,i_{2}}^{2}\sigma_{t-i_{1}}^{2}\sigma_{t-i_{2}}^{2}\nonumber \\
\leq & \left(\max_{t\in\left[T_{1},T_{2}\right]}\sigma_{t}\right)^{4}2\left(1-\rho_{0,n}^{2}\right)^{2}\left(nh\right)^{-2}\sum_{t_{1}>t_{2}=T_{1}}^{T_{2}-1}\sum_{i_{1},i_{2}=t_{1}-t_{2}}^{t_{1}-T_{1}}\ol{\rho}_{n}^{i_{1}+i_{2}+t_{2}-\left(t_{1}-i_{1}\right)+t_{2}-\left(t_{1}-i_{2}\right)}\nonumber \\
 & +\left(\max_{t\in\left[T_{1},T_{2}\right]}\sigma_{t}\right)^{4}\left(1-\rho_{0,n}^{2}\right)^{2}\left(nh\right)^{-2}\sum_{t=T_{1}}^{T_{2}-1}\sum_{i_{1},i_{2}=0}^{\infty}\ol{\rho}_{n}^{2\left(i_{1}+i_{2}\right)}\nonumber \\
= & \ O\left(1\right)\left(1-\rho_{0,n}^{2}\right)^{2}\left(nh\right)^{-2}\sum_{t_{1}>t_{2}=T_{1}}^{T_{2}-1}\ol{\rho}_{n}^{2\left(t_{2}-t_{1}\right)}\sum_{i_{1},i_{2}=t_{1}-t_{2}}^{t_{1}-T_{1}}\ol{\rho}_{n}^{2\left(i_{1}+i_{2}\right)}+O\left(\left(nh\right)^{-1}\right)\nonumber \\
= & \ O\left(1\right)\left(1-\rho_{0,n}^{2}\right)^{2}\left(nh\right)^{-2}\sum_{t_{1}>t_{2}=T_{1}}^{T_{2}-1}\ol{\rho}_{n}^{2\left(t_{1}-t_{2}\right)}\sum_{l_{1},l_{2}=0}^{t_{2}-T_{1}}\ol{\rho}_{n}^{2\left(l_{1}+l_{2}\right)}+O\left(\left(nh\right)^{-1}\right)\nonumber \\
\leq & \ O\left(1\right)\left(1-\rho_{0,n}^{2}\right)^{2}\left(nh\right)^{-2}\sum_{t_{1}>t_{2}=T_{1}}^{T_{2}-1}\ol{\rho}_{n}^{2\left(t_{1}-t_{2}\right)}\left(1-\ol{\rho}_{n}^{2}\right)^{-2}+O\left(\left(nh\right)^{-1}\right)\nonumber \\
= & \ O\left(1\right)O\left(b_{n}^{-2}\right)\left(nh\right)^{-2}O\left(nhb_{n}\right)O\left(b_{n}^{2}\right)+O\left(\left(nh\right)^{-1}\right)=O\left(\left(nh/b_{n}\right)^{-1}\right)=o\left(1\right),\label{eq:case=000020(ii2)-1}
\end{align}
where the second equality uses the changes of coordinates $l_{1}=i_{1}-\left(t_{1}-t_{2}\right)$
and $l_{2}=i_{2}-\left(t_{1}-t_{2}\right)$, the last inequality holds
by \eqref{eq:MT-ps=000020rho^2t}, the second last equality holds
by \eqref{eq:MT-ps=000020rho^(t-s)}, and the last equality uses $b_{n}/nh=o\left(1\right)$.
This completes case (ii2).

In case (iii) we must have (iii1) ($t_{1}-i_{1}=t_{1}-i_{2}=t_{2}-j_{1}>t_{2}-j_{2}$)
or (iii2) ($t_{1}-i_{1}=t_{1}-i_{2}=t_{2}-j_{2}>t_{2}-j_{1}$) or
(iii3) ($t_{1}-i_{1}=t_{2}-j_{2}=t_{2}-j_{1}>t_{1}-i_{2}$) or (iii4)
($t_{1}-i_{2}=t_{2}-j_{2}=t_{2}-j_{1}>t_{1}-i_{1}$).

Now, we prove the desired result for case (iii1). Note that in this
case, it must be true that $i_{1}=i_{2}=i$ and $j_{2}>j_{1}=i-\left(t_{1}-t_{2}\right)$,
which implies 
\begin{align}
 & \ c_{t_{1},i_{1}}c_{t_{1},i_{2}}c_{t_{2},j_{1}}c_{t_{2},j_{2}}\E U_{t_{1}-i_{1}}U_{t_{1}-i_{2}}U_{t_{2}-j_{1}}U_{t_{2}-j_{2}}\sigma_{t_{1}-i_{1}}\sigma_{t_{1}-i_{2}}\sigma_{t_{2}-j_{1}}\sigma_{t_{2}-j_{2}}\nonumber \\
= & \ c_{t_{1},i}^{2}c_{t_{2},i-\left(t_{1}-t_{2}\right)}c_{t_{2},j}\E U_{t_{1}-i}^{3}U_{t_{2}-j}\ind\left\{ j>i-\left(t_{1}-t_{2}\right)\right\} \sigma_{t_{1}-i}^{3}\sigma_{t_{2}-j}.\label{eq:case(iii)-1-RHS}
\end{align}
Substituting \eqref{eq:case(iii)-1-RHS} into the right-hand side
of \eqref{eq:a.11.10-Ab3}, we have

\begin{align}
 & \abs{\left(1-\rho_{0,n}^{2}\right)^{2}\left(nh\right)^{-2}\sum_{t_{1},t_{2}=T_{1}}^{T_{2}-1}\sum_{i_{1},i_{2}=0}^{t_{1}-T_{1}}\sum_{j_{1},j_{2}=0}^{t_{2}-T_{1}}\left[\begin{array}{c}
c_{t_{1},i}^{2}c_{t_{2},i-\left(t_{1}-t_{2}\right)}c_{t_{2},j}\E U_{t_{1}-i}^{3}U_{t_{2}-j}\\
\times\ind\left\{ j>i-\left(t_{1}-t_{2}\right)\right\} \sigma_{t_{1}-i}^{3}\sigma_{t_{2}-j}
\end{array}\right]}\nonumber \\
\leq & \left(\max_{t\in\left[T_{1},T_{2}\right]}\sigma_{t}\right)^{4}2\left(1-\rho_{0,n}^{2}\right)^{2}\left(nh\right)^{-2}\sum_{t_{1}>t_{2}=T_{1}}^{T_{2}-1}\sum_{i=t_{1}-t_{2}}^{t_{1}-T_{1}}\sum_{j=i-\left(t_{1}-t_{2}\right)+1}^{t_{2}-T_{1}}\begin{array}{c}
\ol{\rho}_{n}^{3i+j-\left(t_{1}-t_{2}\right)}\end{array}\abs{\E U_{t_{1}-i}^{3}U_{t_{2}-j}}\nonumber \\
 & +\left(\max_{t\in\left[T_{1},T_{2}\right]}\sigma_{t}\right)^{4}\left(1-\rho_{0,n}^{2}\right)^{2}\left(nh\right)^{-2}\sum_{t=T_{1}}^{T_{2}-1}\sum_{i=0}^{t-T_{1}}\sum_{j=i+1}^{t-T_{1}}\begin{array}{c}
\ol{\rho}_{n}^{3i+j}\end{array}\abs{\E U_{t-i}^{3}U_{t-j}}\nonumber \\
\leq & \ O\left(b_{n}^{-2}\right)\left[\left(nh\right)^{-2}\sum_{t_{1}>t_{2}=T_{1}}^{T_{2}-1}\sum_{i=t_{1}-t_{2}}^{t_{1}-T_{1}}\sum_{j=i-\left(t_{1}-t_{2}\right)+1}^{t_{2}-T_{1}}\ol{\rho}_{n}^{3i+j-\left(t_{1}-t_{2}\right)}+\left(nh\right)^{-2}\sum_{t=T_{1}}^{T_{2}-1}\sum_{i=0}^{\infty}\sum_{j=0}^{\infty}\begin{array}{c}
\ol{\rho}_{n}^{3i+j}\end{array}\right]\nonumber \\
= & \ O\left(b_{n}^{-2}\right)\left[\left(nh\right)^{-2}\sum_{t_{1}>t_{2}=T_{1}}^{T_{2}-1}\sum_{i=t_{1}-t_{2}}^{t_{1}-T_{1}}\ol{\rho}_{n}^{4i-2\left(t_{1}-t_{2}\right)}\sum_{l=1}^{t_{1}-T_{1}-i}\ol{\rho}_{n}^{l}+O\left(b_{n}^{2}/nh\right)\right]\nonumber \\
= & \ O\left(b_{n}^{-2}\right)\left[\left(nh\right)^{-2}\sum_{t_{1}>t_{2}=T_{1}}^{T_{2}-1}\ol{\rho}_{n}^{2\left(t_{1}-t_{2}\right)}O\left(b_{n}^{2}\right)+O\left(b_{n}^{2}/nh\right)\right]=O\left(b_{n}/nh\right)=o\left(1\right),\label{eq:case(iii)-1}
\end{align}
where the second inequality holds by H�lder's inequality and part
(iv) of $\Lambda_{n}$: $\abs{\E U_{t}^{3}U_{t-j}}\leq\E\abs{U_{t}^{3}U_{t-j}}\leq\left(\E U_{t}^{4}\right)^{3/4}\left(\E U_{t-j}^{4}\right)^{1/4}<M$
for $j>0$, and \eqref{eq:sigma_t=000020converge=000020to=000020sigma_0},
the first equality holds by the change of coordinates $l=j-i+\left(t_{1}-t_{2}\right)-1$
and \eqref{eq:MT-ps=000020rho^t}, the second equality uses the change
of variables $k=i-\left(t_{1}-t_{2}\right)$, \eqref{eq:MT-ps=000020rho^t},
and \eqref{eq:MT-ps=000020rho^2t}, and the second last equality holds
by \eqref{eq:MT-ps=000020rho^(t-s)}.

The proofs for cases (iii2)-(iii4) are analogous to case (iii1) and
thus are omitted.

Combining cases (i)--(iii), we have 
\begin{equation}
\E\left[\left(1-\rho_{0,n}^{2}\right)\left(nh\right)^{-1}\sum_{t=T_{1}}^{T_{2}}\left(Y_{t-1}^{0}\right)^{2}\right]^{2}=1+o\left(1\right).\label{eq:case=000020i-ii-1}
\end{equation}
Note that by Markov's inequality, for any random variable $X_{n}$
\begin{equation}
\P\left(\abs{X_{n}-1}>\varepsilon\right)\leq\dfrac{\E\left(X_{n}-1\right)^{2}}{\varepsilon^{2}}=\dfrac{\E X_{n}^{2}-2\E X_{n}+1}{\varepsilon^{2}}.\label{eq:markov=000020argument}
\end{equation}
Let $X_{n}:=\left(1-\rho_{0,n}^{2}\right)\left(nh\right)^{-1}\sum_{t=T_{1}}^{T_{2}}\left(Y_{t-1}^{0}\right)^{2}$.
Then, substituting \eqref{eq:=000020expectation=000020of=000020sum=000020of=000020Yt-1=000020squared-1}
and \eqref{eq:case=000020i-ii-1} into \eqref{eq:markov=000020argument},
we have

\begin{equation}
\left(1-\rho_{0,n}^{2}\right)\left(nh\right)^{-1}\sum_{t=T_{1}}^{T_{2}}\left(Y_{t-1}^{0}\right)^{2}\pto1.\label{eq:=000020lem=0000205=000020part=000020b-1}
\end{equation}
\end{proof}
\begin{proof}[\textbf{Proof of Lemma \ref{lem:MT-Asymptotic=000020Properties=000020of=000020the=000020Components=000020Stationary}$\left(\text{c}\right)$}]
To prove part (c), by a central limit theorem for a triangular array
of martingale differences as in Corollary 3.1 in \citet*{hall1980martingale},
it is sufficient to establish the Lindeberg condition (i) $\sum_{t=T_{1}}^{T_{2}}\E\left[\rest{\zeta_{t}^{2}\ind\left\{ \abs{\zeta_{t}}>\delta\right\} }\mathscr{G}_{t-1}\right]\pto0$
for any $\d>0$ and (ii) $\sum_{t=T_{1}}^{T_{2}}\E\left(\rest{\zeta_{t}^{2}}\mathscr{G}_{t-1}\right)\pto1$,
where $\zeta_{t}:=\left(nh\right)^{-1/2}\left(1-\rho_{0,n}^{2}\right)^{1/2}Y_{t-1}^{0}\sigma_{t}U_{t}$
for $t=T_{1},...,T_{2}$. To prove (i), by Markov's inequality, it
is enough to show that $\sum_{t=T_{1}}^{T_{2}}\E\left[\zeta_{t}^{2}\ind\left\{ \abs{\zeta_{t}}>\delta\right\} \right]\to0$
for any $\d>0$. By $\sum_{t=T_{1}}^{T_{2}}\E\left[\zeta_{t}^{2}\ind\left\{ \abs{\zeta_{t}}>\delta\right\} \right]\leq\sum_{t=T_{1}}^{T_{2}}\E\left[\zeta_{t}^{4}\right]/\d^{2}$,
it is then sufficient to show $\sum_{t=T_{1}}^{T_{2}}\E\left[\zeta_{t}^{4}\right]=o\left(1\right)$,
which is true because 
\begin{align}
\sum_{t=T_{1}}^{T_{2}}\E\left[\zeta_{t}^{4}\right]= & \left(1-\rho_{0,n}^{2}\right)^{2}\left(nh\right)^{-2}\sum_{t=T_{1}}^{T_{2}}\E\left(Y_{t-1}^{0}U_{t}\sigma_{t}\right)^{4}\nonumber \\
= & \left(1-\rho_{0,n}^{2}\right)^{2}\left(nh\right)^{-2}\sum_{t=T_{1}}^{T_{2}}\E\left[\left(Y_{t-1}^{0}\right)^{4}\E\left(\rest{U_{t}^{4}}\mathscr{G}_{t-1}\right)\right]\sigma_{t}^{4}\nonumber \\
\leq & \ O\left(b_{n}^{-2}\right)M\max_{t\in\left[T_{1},T_{2}\right]}\sigma_{t}^{4}\left(nh\right)^{-2}\sum_{t=T_{1}}^{T_{2}-1}\E\left(Y_{t}^{0}\right)^{4}\nonumber \\
= & \ O\left(b_{n}^{-2}\right)\left(nh\right)^{-2}\sum_{t=T_{1}}^{T_{2}-1}\E\left[\sum_{j=0}^{t-T_{1}}c_{t,j}\sigma_{t-j}U_{t-j}\right]^{4}\nonumber \\
\leq & \ O\left(b_{n}^{-2}\right)\left(\max_{t\in\left[T_{1},T_{2}\right]}\sigma_{t}\right)^{4}\left(nh\right)^{-2}\sum_{t=T_{1}}^{T_{2}-1}\sum_{j_{1},j_{2},j_{3},j_{4}=0}^{t-T_{1}}\abs{\begin{array}{c}
c_{t,j_{1}}c_{t,j_{2}}c_{t,j_{3}}c_{t,j_{4}}\\
\times\E\left(U_{t-j_{1}}U_{t-j_{2}}U_{t-j_{3}}U_{t-j_{4}}\right)
\end{array}}\nonumber \\
\leq & \ O\left(b_{n}^{-2}\right)\left(nh\right)^{-2}\sum_{t=T_{1}}^{T_{2}-1}\left[\sum_{j=0}^{t-T_{1}}\ol{\rho}_{n}^{4j}M+3\sum_{i,j=0}^{t-T_{1}}\ol{\rho}_{n}^{2\left(i+j\right)}\ind\left\{ i\neq j\right\} +4\sum_{i=0}^{t-T_{1}}\sum_{j=i+1}^{t-T_{1}}\ol{\rho}_{n}^{3i+j}M^{1/2}\right]\nonumber \\
= & \ O\left(b_{n}^{-2}\right)\left(nh\right)^{-2}nh\left(O\left(b_{n}\right)+O\left(b_{n}^{2}\right)+O\left(b_{n}^{2}\right)\right)=o\left(1\right),\label{eq:lindeberg=000020cond=000020i-1}
\end{align}
where the second equality uses the law of iterated expectations, the
first inequality uses $\E\left[\rest{U_{t}^{4}}\mathscr{G}_{t-1}\right]<M$
a.s. and \eqref{eq:MT-ps=000020rho^2t}, the third equality uses \eqref{eq:sigma_t=000020converge=000020to=000020sigma_0},
and the last inequality holds by dividing the sum into three cases
of (i) all four indices on the four innovation terms coincide, (ii)
two pairs of two indices coincide, and (iii) three larger indices
coincide and \eqref{eq:MT-c_tj=000020bounded-gto_1}.

To prove (ii), by part (b) we have $\left(1-\rho_{0,n}^{2}\right)\left(nh\right)^{-1}\sum_{t=T_{1}}^{T_{2}}\left(Y_{t-1}^{0}\right)^{2}=1+o_{p}\left(1\right)$.
Thus, 
\begin{align}
\sum_{t=T_{1}}^{T_{2}}\E\left(\rest{\zeta_{t}^{2}}\mathscr{G}_{t-1}\right)= & \left(1-\rho_{0,n}^{2}\right)\left(nh\right)^{-1}\sum_{t=T_{1}}^{T_{2}}\left(Y_{t-1}^{0}\right)^{2}\E\left[\rest{U_{t}^{2}}\mathscr{G}_{t-1}\right]\sigma_{t}^{2}\nonumber \\
= & \left(1-\rho_{0,n}^{2}\right)\left(nh\right)^{-1}\sum_{t=T_{1}}^{T_{2}}\left(Y_{t-1}^{0}\right)^{2}\nonumber \\
 & +\left(1-\rho_{0,n}^{2}\right)\left(nh\right)^{-1}\sum_{t=T_{1}}^{T_{2}}\left(Y_{t-1}^{0}\right)^{2}\left(\sigma_{t}^{2}-\sigma_{0}^{2}\left(\tau\right)\right)\nonumber \\
= & 1+o_{p}\left(1\right)+o_{p}\left(1\right)\pto1,\label{eq:=000020clt=000020cond=000020ii-1}
\end{align}
where the second equality holds by $\E\left[\rest{U_{t}^{2}}\mathscr{G}_{t-1}\right]=1\ \text{a.s.}$
and the last equality holds by part (b) and 
\begin{align}
 & \abs{\left(1-\rho_{0,n}^{2}\right)\left(nh\right)^{-1}\sum_{t=T_{1}}^{T_{2}}\left(Y_{t-1}^{0}\right)^{2}\left(\sigma_{t}^{2}-\sigma_{0}^{2}\left(\tau\right)\right)}\nonumber \\
 & \leq\max_{t\in\left[T_{1},T_{2}\right]}\abs{\sigma_{t}^{2}-\sigma_{0}^{2}\left(\tau\right)}\abs{\left(1-\rho_{0,n}^{2}\right)\left(nh\right)^{-1}\sum_{t=T_{1}}^{T_{2}}\left(Y_{t-1}^{0}\right)^{2}}\nonumber \\
 & =\ o\left(1\right)O_{p}\left(1\right)=o_{p}\left(1\right).\label{eq:part(ii)-1}
\end{align}
\end{proof}

\subsection{\protect\label{subsec:Proof-of-Lemma-denom-stationary}Proof of Lemma
\ref{lem:MT-asympdist_dist_stationary-denom}}

In this section, for notational simplicity in the proof, we assume
that $\sigma_{0}^{2}\left(\tau\right)=1.$ 
\begin{proof}[\textbf{Proof of Lemma \ref{lem:MT-asympdist_dist_stationary-denom}$\left(\text{a}\right)$}]
First, we prove part (a). We have 
\begin{align}
 & \left(1-\rho_{0,n}^{2}\right)\left(nh\right)^{-1}\sum_{t=T_{1}}^{T_{2}}\left(Y_{t-1}-\ol{\mu}_{nh,-1}\right)^{2}\nonumber \\
= & \left(1-\rho_{0,n}^{2}\right)\left(nh\right)^{-1}\sum_{t=T_{1}}^{T_{2}}\left[Y_{t-1}^{0}+\left(\mu_{t-1}-\ol{\mu}_{nh,-1}\right)+c_{t-1,t-1-T_{0}}Y_{T_{0}}^{*}\right]^{2}\nonumber \\
= & \left(1-\rho_{0,n}^{2}\right)\left(nh\right)^{-1}\sum_{t=T_{1}}^{T_{2}}\left(Y_{t-1}^{0}\right)^{2}\nonumber \\
 & +\left(1-\rho_{0,n}^{2}\right)\nonumber \\
 & \times\left(nh\right)^{-1}\sum_{t=T_{1}}^{T_{2}}\left[\begin{array}{c}
\left(\mu_{t-1}-\ol{\mu}_{nh,-1}\right)^{2}+c_{t-1,t-1-T_{0}}^{2}Y_{T_{0}}^{*2}+2\left(\mu_{t-1}-\ol{\mu}_{nh,-1}\right)Y_{t-1}^{0}\\
+2\left(\mu_{t-1}-\ol{\mu}_{nh,-1}\right)c_{t-1,t-1-T_{0}}Y_{T_{0}}^{*}+2Y_{t-1}^{0}c_{t-1,t-1-T_{0}}Y_{T_{0}}^{*}
\end{array}\right]\nonumber \\
=: & \left(1-\rho_{0,n}^{2}\right)\left(nh\right)^{-1}\sum_{t=T_{1}}^{T_{2}}\left(Y_{t-1}^{0}\right)^{2}+\sum_{i=1}^{5}A_{ai}.\label{eq:denominator_rest-1}
\end{align}
By Lemma \ref{lem:MT-Asymptotic=000020Properties=000020of=000020the=000020Components=000020Stationary}(b),
$\left(1-\rho_{0,n}^{2}\right)\left(nh\right)^{-1}\sum_{t=T_{1}}^{T_{2}}\left(Y_{t-1}^{0}\right)^{2}\pto1$.
Hence, we need to show $\sum_{i=1}^{5}A_{ai}$ in \eqref{eq:denominator_rest-1}
converges in probability to 0. By the CS inequality, we only need
to show $A_{a1}\gto0$ and $A_{a2}\pto0$.

For $A_{a1}$, we have 
\begin{align}
A_{a1}=\left(1-\rho_{0,n}^{2}\right)\left(nh\right)^{-1}\sum_{t=T_{1}}^{T_{2}}\left(\mu_{t-1}-\ol{\mu}_{nh,-1}\right)^{2} & \leq O\left(b_{n}^{-1}\right)\max_{t,s\in\left[T_{0},T_{2}\right]}\left(\mu_{t}-\mu_{s}\right)^{2}=o\left(1\right),\label{eq:den_A_11-1}
\end{align}
where the last equality holds by Lemma \ref{lem:MT-Max=000020Intertemporal=000020Difference}(d)
and the triangle inequality.

For $A_{a2},$ we have 
\begin{eqnarray}
A_{a2}\hspace{-0.08in} & = & \hspace{-0.08in}(1-\rho_{0,n}^{2})(nh)^{-1}\sum_{t=T_{1}}^{T_{2}}c_{t-1,t-1-T_{0}}^{2}Y_{T_{0}}^{\ast2}=O(b_{n}^{-1})(nh)^{-1}\sum_{t=T_{1}}^{T_{2}}\overline{\rho}_{n}^{2(t-1-T_{0})}O_{p}(b_{n})\nonumber \\
 & = & \hspace{-0.08in}O(b_{n}^{-1})(nh)^{-1}O(b_{n})O_{p}(b_{n})=O_{p}(b_{n}/nh)=o_{p}(1),
\end{eqnarray}
where the second equality uses $1-\rho_{0,n}=-\kappa_{n}(\tau)/b_{n},$
$\kappa_{n}(\tau)=O(1)$ (by part (ii) of $\Lambda_{n}),$ (\ref{eq:MT-c_tj=000020bounded-gto_1}),
and Lemma \ref{lem:MT-Order=000020of=000020Y_T0}(c), which applies
because the lemma assumes that $nh/b_{n}\rightarrow r_{0}=\infty,$
the third equality uses (\ref{eq:MT-ps=000020rho^2t}), and the last
equality holds because $nh/b_{n}\rightarrow r_{0}=\infty.$ This completes
the proof of part (a). 
\end{proof}
\begin{proof}[\textbf{Proof of Lemma \ref{lem:MT-asympdist_dist_stationary-denom}$\left(\text{b}\right)$}]
To prove part (b), it suffices to show that 
\begin{equation}
A_{b}:=(1-\rho_{0,n}^{2})^{1/2}(\overline{Y}_{nh,-1}-\overline{\mu}_{nh,-1})=o_{p}(1).
\end{equation}
We have 
\begin{eqnarray}
A_{b}\hspace{-0.08in} & = & \hspace{-0.08in}A_{b1}+A_{b2},\text{ where}\\
A_{b1}\hspace{-0.1in} & : & \hspace{-0.14in}=\text{ }(1-\rho_{0,n}^{2})^{1/2}(nh)^{-1}\sum_{t=T_{1}}^{T_{2}}Y_{t-1}^{0}\text{ and }A_{b2}:=(1-\rho_{0,n}^{2})^{1/2}(nh)^{-1}\sum_{t=T_{1}}^{T_{2}}c_{t-1,t-1-T_{0}}Y_{T_{0}}^{\ast}.\nonumber 
\end{eqnarray}
By Lemma \ref{lem:MT-Asymptotic=000020Properties=000020of=000020the=000020Components=000020Stationary}(a),
$A_{b1}=o_{p}(1).$ In addition, we have 
\begin{equation}
|A_{b2}|\leq(1-\rho_{0,n}^{2})^{1/2}(nh)^{-1}\sum_{t=T_{1}}^{T_{2}}\overline{\rho}_{n}^{t-1-T_{0}}|Y_{T_{0}}^{\ast}|\leq O(b_{n}^{-1/2})(nh)^{-1}O(b_{n})O_{p}(b_{n}^{1/2})=o_{p}(1),
\end{equation}
where the first inequality uses (\ref{eq:MT-c_tj=000020bounded-gto_1}),
the second inequality uses $\sum_{t=T_{1}}^{T_{2}}\overline{\rho}_{n}^{t-1-T_{0}}\leq(1-\overline{\rho}_{n})^{-1}=O(b_{n})$
(by (\ref{eq:MT-ps=000020rho^t})) and Lemma \ref{lem:MT-Order=000020of=000020Y_T0}(c),
which applies because $nh/b_{n}\rightarrow r_{0}=\infty,$ which is
an assumption of the lemma, and the equality holds because $b_{n}/nh\rightarrow0.$
Hence, $A_{b}=o_{p}(1)$ and part (b) is established. 
\end{proof}

\subsection{\protect\label{subsec:Proof-of-Lemma-numerator-stationary}Proof
of Lemma \ref{lem:MT-asympdist_dist_stationary-numerator=000020}}

In this section, for notational simplicity in the proof, we assume
that $\sigma_{0}^{2}\left(\tau\right)=1.$ 
\begin{proof}[\textbf{Proof of Lemma \ref{lem:MT-asympdist_dist_stationary-numerator=000020}$\left(\text{a}\right)$}]
To prove part (a), we express 
\begin{equation}
Y_{t-1}-\ol{\mu}_{nh,-1}=Y_{t-1}^{0}+\left(\mu_{t-1}-\ol{\mu}_{nh,-1}\right)+c_{t-1,t-1-T_{0}}Y_{T_{0}}^{*}\label{eq:decomp=000020first=000020part-1}
\end{equation}
and 
\begin{align}
 & \ Y_{t}-\ol{\mu}_{nh}-\rho_{0,n}\left(Y_{t-1}-\ol{\mu}_{nh,-1}\right)\nonumber \\
= & \ \sigma_{t}U_{t}+\left(\rho_{t}-\rho_{0,n}\right)Y_{t-1}^{0}+\left(\mu_{t}-\ol{\mu}_{nh}\right)-\rho_{0,n}\left(\mu_{t-1}-\ol{\mu}_{nh,-1}\right)+\left(\rho_{t}-\rho_{0,n}\right)c_{t-1,t-1-T_{0}}Y_{T_{0}}^{*}.\label{eq:5.10.2}
\end{align}

Substituting \eqref{eq:decomp=000020first=000020part-1} and \eqref{eq:5.10.2}
into the left-hand side of part (a), we have 
\begin{align}
 & \left(1-\rho_{0,n}^{2}\right)^{1/2}\left(nh\right)^{-1/2}\sum_{t=T_{1}}^{T_{2}}\left(Y_{t-1}-\ol{\mu}_{nh,-1}\right)\left[Y_{t}-\ol{\mu}_{nh}-\rho_{0,n}\left(Y_{t-1}-\ol{\mu}_{nh,-1}\right)\right]\nonumber \\
= & \left(1-\rho_{0,n}^{2}\right)^{1/2}\left(nh\right)^{-1/2}\sum_{t=T_{1}}^{T_{2}}Y_{t-1}^{0}\sigma_{t}U_{t}\nonumber \\
 & +\left(1-\rho_{0,n}^{2}\right)^{1/2}\left(nh\right)^{-1/2}\nonumber \\
 & \times\sum_{t=T_{1}}^{T_{2}}\left[\begin{array}{c}
\left(\rho_{t}-\rho_{0,n}\right)\left(Y_{t-1}^{0}\right)^{2}+Y_{t-1}^{0}\left(\mu_{t}-\ol{\mu}_{nh}\right)-Y_{t-1}^{0}\rho_{0,n}\left(\mu_{t-1}-\ol{\mu}_{nh,-1}\right)\\
+Y_{t-1}^{0}\left(\rho_{t}-\rho_{0,n}\right)c_{t-1,t-1-T_{0}}Y_{T_{0}}^{*}+\left(\mu_{t-1}-\ol{\mu}_{nh,-1}\right)\sigma_{t}U_{t}\\
+\left(\mu_{t-1}-\ol{\mu}_{nh,-1}\right)\left(\rho_{t}-\rho_{0,n}\right)Y_{t-1}^{0}+\left(\mu_{t-1}-\ol{\mu}_{nh,-1}\right)\left(\mu_{t}-\ol{\mu}_{nh}\right)\\
-\rho_{0,n}\left(\mu_{t-1}-\ol{\mu}_{nh,-1}\right)^{2}+\left(\mu_{t-1}-\ol{\mu}_{nh,-1}\right)\left(\rho_{t}-\rho_{0,n}\right)c_{t-1,t-1-T_{0}}Y_{T_{0}}^{*}\\
+c_{t-1,t-1-T_{0}}Y_{T_{0}}^{*}\sigma_{t}U_{t}+c_{t-1,t-1-T_{0}}Y_{T_{0}}^{*}\left(\rho_{t}-\rho_{0,n}\right)Y_{t-1}^{0}\\
+c_{t-1,t-1-T_{0}}Y_{T_{0}}^{*}\left(\mu_{t}-\ol{\mu}_{nh}\right)-c_{t-1,t-1-T_{0}}Y_{T_{0}}^{*}\rho_{0,n}\left(\mu_{t-1}-\ol{\mu}_{nh,-1}\right)\\
+c_{t-1,t-1-T_{0}}^{2}\left(\rho_{t}-\rho_{0,n}\right)Y_{T_{0}}^{*2}
\end{array}\right]\nonumber \\
=: & \left(1-\rho_{0,n}^{2}\right)^{1/2}\left(nh\right)^{-1/2}\sum_{t=T_{1}}^{T_{2}}Y_{t-1}^{0}\sigma_{t}U_{t}+\sum_{i=1}^{14}A_{ci},\label{eq:decomp=000020A3=000020stationary=000020numerator-1}
\end{align}
where $A_{ci}$ is the $i^{th}$ term in the second last line of \eqref{eq:decomp=000020A3=000020stationary=000020numerator-1}.
For example, 
\begin{equation}
A_{c1}:=\left(1-\rho_{0,n}^{2}\right)^{1/2}\left(nh\right)^{-1/2}\sum_{t=T_{1}}^{T_{2}}\left(\rho_{t}-\rho_{0,n}\right)\left(Y_{t-1}^{0}\right)^{2}
\end{equation}
and 
\begin{equation}
A_{c2}:=\left(1-\rho_{0,n}^{2}\right)^{1/2}\left(nh\right)^{-1/2}\sum_{t=T_{1}}^{T_{2}}Y_{t-1}^{0}\left(\mu_{t}-\ol{\mu}_{nh}\right).
\end{equation}

By Lemma \ref{lem:MT-Asymptotic=000020Properties=000020of=000020the=000020Components=000020Stationary}(c),
we have 
\begin{equation}
\left(1-\rho_{0,n}^{2}\right)^{1/2}\left(nh\right)^{-1/2}\sum_{t=T_{1}}^{T_{2}}Y_{t-1}^{0}\sigma_{t}U_{t}\dto N\left(0,1\right).\label{eq:=000020dominate=000020term=000020numerator-1}
\end{equation}
Therefore, we need to show $\sum_{i=1}^{14}A_{ci}$ converges in probability
to 0. We examine each of its components one by one.

First, we consider $A_{c1}.$ Let $Y_{t}^{c}$ be the constant parameter
version of $Y_{t}^{0}$ based on $\rho_{0,n}$ and $\sigma_{n\tau}$
and with $T_{2}-T_{1}$ lagged innovations (which does not depend
on $t),$ rather than $t-T_{1}$ lags. That is, 
\begin{equation}
Y_{t}^{c}:=\sum_{j=0}^{T_{2}-T_{1}}\rho_{0,n}^{j}\sigma_{n\tau}U_{t-j}\text{ for }t\in\left[T_{1},T_{2}\right].\label{Defn=000020of=000020Y^c_t}
\end{equation}
Here, the superscript $c$ stands for ``constant parameter.''

We decompose $A_{c1}$ into three terms: 
\begin{eqnarray}
A_{c1}\hspace{-0.08in} & = & \hspace{-0.08in}A_{c1,1}+A_{c1,2}+A_{c1,3},\text{ where}\nonumber \\
A_{c1,1}\hspace{-0.08in} & : & \hspace{-0.12in}=\text{ }(1-\rho_{0,n}^{2})^{1/2}(nh)^{-1/2}\sum_{t=T_{1}}^{T_{2}}(\rho_{t}-\rho_{0,n})\E(Y_{t-1}^{c})^{2},\nonumber \\
A_{c1,2}\hspace{-0.08in} & : & \hspace{-0.12in}=\text{ }(1-\rho_{0,n}^{2})^{1/2}(nh)^{-1/2}\sum_{t=T_{1}}^{T_{2}}(\rho_{t}-\rho_{0,n})\left((Y_{t-1}^{0})^{2}-(Y_{t-1}^{c})^{2}\right),\text{ and}\nonumber \\
A_{c1,3}\hspace{-0.08in} & : & \hspace{-0.12in}=\text{ }(1-\rho_{0,n}^{2})^{1/2}(nh)^{-1/2}\sum_{t=T_{1}}^{T_{2}}(\rho_{t}-\rho_{0,n})\left((Y_{t-1}^{c})^{2}-\E(Y_{t-1}^{c})^{2}\right).\label{Decomposition=000020of=000020A_c1}
\end{eqnarray}
We show that $A_{c1,b}=o_{p}(1)$ for $b=1,2,3$, provided Assumption
\ref{assu:MT-stationary-nh^5=000020->=0000200} holds, which yields
$A_{c1}=o_{p}(1).$

Now, we consider $A_{c1,1}.$ We have 
\begin{equation}
\E(Y_{t-1}^{c})^{2}=\sigma_{n\tau}^{2}\sum_{j=0}^{T_{2}-T_{1}}\sum_{k=0}^{T_{2}-T_{1}}\rho_{0,n}^{j}\rho_{0,n}^{k}\E U_{t-j-1}U_{t-k-1}=\sigma_{n\tau}^{2}\sum_{j=0}^{T_{2}-T_{1}}\rho_{0,n}^{2j},\label{Pf=000020re=000020A_c1,1=000020eqn=0000201}
\end{equation}
which does not depend on $t,$ and hence, can be taken out of the
sum over $t$ in the definition of $A_{c1,1}.$

By definition, 
\begin{equation}
\rho_{t}-\rho_{0,n}:=\rho_{0,n}(t/n)-\rho_{0,n}(\tau)=\kappa_{0,n}(\tau)/b_{n}-\kappa_{0,n}(t/n)/b_{n}\label{Pf=000020re=000020A_c1,1=000020eqn=0000202}
\end{equation}
using part (ii) in the definition of $\Lambda_{n}.$ Using the condition
in $\Lambda_{n}$ that $\kappa\left(\cd\right)$ is a twice continuously
differentiable function , by a two-term Taylor expansion of $\kappa_{0,n}(t/n)$
around $\tau,$ we obtain 
\begin{equation}
\kappa_{0,n}(t/n)-\kappa_{0,n}(\tau)=\kappa_{0,n}^{\prime}(\tau)(t/n-\tau)+\kappa_{0,n}^{\prime\prime}(\widetilde{\tau}_{n,t})(t/n-\tau)^{2},\label{Pf=000020re=000020A_c1,1=000020eqn=0000203}
\end{equation}
where $\widetilde{\tau}_{n,t}$ lies between $t/n$ and $\tau,$ and
hence, lies between $T_{1}/n$ and $T_{2}/n$ for $t\in\left[T_{1},T_{2}\right].$

Let $\alpha_{n}:=\lfloor nh/2\rfloor.$ Using $T_{1}:=\lfloor n\tau\rfloor-\lfloor nh/2\rfloor$
and $T_{2}:=\lfloor n\tau\rfloor+\lfloor nh/2\rfloor$ by \eqref{eq:MT-timeperiod},
we have 
\begin{equation}
\sum_{t=T_{1}}^{T_{2}}(t-n\tau)=\sum_{t=\lfloor n\tau\rfloor+1}^{\lfloor n\tau\rfloor+\alpha_{n}}(t-n\tau)+\sum_{t=\lfloor n\tau\rfloor-\alpha_{n}}^{\lfloor n\tau\rfloor-1}(t-n\tau)+\lfloor n\tau\rfloor-n\tau.\label{Pf=000020re=000020A_c1,1=000020eqn=0000203a}
\end{equation}
In addition, we have 
\begin{eqnarray}
\sum_{t=\lfloor n\tau\rfloor+1}^{\lfloor n\tau\rfloor+\alpha_{n}}(t-n\tau)\hspace{-0.08in} & = & \hspace{-0.08in}\sum_{s=1}^{\alpha_{n}}(s+\lfloor n\tau\rfloor-n\tau)\text{ and}\nonumber \\
\sum_{t=\lfloor n\tau\rfloor-\alpha_{n}}^{\lfloor n\tau\rfloor-1}(t-n\tau)\hspace{-0.08in} & = & \hspace{-0.08in}\sum_{s=1}^{\alpha_{n}}(-s+\lfloor n\tau\rfloor-n\tau),\label{Pf=000020re=000020A_c1,1=000020eqn=0000203b}
\end{eqnarray}
where the first line uses a change of variables with $s=t-\lfloor n\tau\rfloor$
and the second line uses a change of variables with $s=-t+\lfloor n\tau\rfloor.$
Combining \eqref{Pf=000020re=000020A_c1,1=000020eqn=0000203a} and
\eqref{Pf=000020re=000020A_c1,1=000020eqn=0000203b} gives 
\begin{equation}
\sum_{t=T_{1}}^{T_{2}}(t/n-\tau)=2n^{-1}\sum_{s=1}^{\alpha_{n}}(\lfloor n\tau\rfloor-n\tau)+n^{-1}(\lfloor n\tau\rfloor-n\tau)=O(n^{-1}\alpha_{n})=O(h).\label{Pf=000020re=000020A_c1,1=000020eqn=0000203c}
\end{equation}

Using (\ref{Decomposition=000020of=000020A_c1})--(\ref{Pf=000020re=000020A_c1,1=000020eqn=0000203})
and (\ref{Pf=000020re=000020A_c1,1=000020eqn=0000203c}), we get 
\begin{eqnarray}
\abs{A_{c1,1}}\hspace{-0.08in} & = & \hspace{-0.08in}\left\vert \E(Y_{t-1}^{c})^{2}b_{n}^{-1}(1-\rho_{0,n}^{2})^{1/2}(nh)^{-1/2}\kappa_{0,n}^{\prime}(\tau)\sum_{t=T_{1}}^{T_{2}}(t/n-\tau)\right.\nonumber \\
 &  & \left.+\E(Y_{t-1}^{c})^{2}b_{n}^{-1}(1-\rho_{0,n}^{2})^{1/2}(nh)^{-1/2}\sum_{t=T_{1}}^{T_{2}}\kappa_{0,n}^{\prime\prime}(\widetilde{\tau}_{n,t})(t/n-\tau)^{2}\right\vert \nonumber \\
 & \leq & \hspace{-0.08in}C\E(Y_{t-1}^{c})^{2}b_{n}^{-1}(1-\rho_{0,n}^{2})^{1/2}(nh)^{-1/2}\left(O(h)+\sum_{t=T_{1}}^{T_{2}}(t/n-\tau)^{2}\right)\nonumber \\
 & = & \hspace{-0.08in}O(b_{n}^{-1}(1-\rho_{0,n}^{2})^{-1/2}(nh)^{-1/2}\left(h+nh\cdot h^{2}\right))\label{Pf=000020re=000020A_c1,1=000020eqn=0000204}\\
 & = & \hspace{-0.08in}O(b_{n}^{-1}b_{n}^{1/2}\left((h/n)^{1/2}+(nh^{5})^{1/2}\right))=O((h/n)^{1/2}+(nh^{5})^{1/2})=o(1),\nonumber 
\end{eqnarray}
for some finite constant $C$ for which $\abs{\kappa_{0,n}^{\prime}(\tau)}\leq C$
and $\sup_{t,n}\abs{\kappa_{0,n}^{\prime\prime}(\widetilde{\tau}_{n,t})}\leq C,$
where the inequality uses these inequalities and (\ref{Pf=000020re=000020A_c1,1=000020eqn=0000203c}),
the second equality uses $\E(Y_{t-1}^{c})^{2}=O((1-\rho_{0,n}^{2})^{-1})$
by (\ref{Pf=000020re=000020A_c1,1=000020eqn=0000201}) \textbf{\ }and
$\abs{t/n-\tau}\leq h$ because $\abs{t-n\tau}\leq nh/2$ for $t\in\left[T_{1},T_{2}\right],$
and the third equality uses $(1-\rho_{0,n}^{2})^{-1/2}=O(b_{n}^{1/2})$
(by \eqref{eq:MT-ps=000020rho^2t} with $\rho_{0,n}$ in place of
$\ol{\rho}_{n}$), the fourth equality holds because $b_{n}^{-1/2}=O(1),$
and the last equality holds by $h\rightarrow0$ (by Assumption \ref{assu:MT-h})
and $nh^{5}\rightarrow0$ (by Assumption \ref{assu:MT-stationary-nh^5=000020->=0000200}).

Next, we consider $A_{c1,2}$. We have 
\begin{eqnarray}
|A_{c1,2}|\hspace{-0.08in} & \leq & \hspace{-0.08in}(1-\rho_{0,n}^{2})^{1/2}\max_{t\in[T_{1},T_{2}]}|\rho_{t}-\rho_{0,n}|(nh)^{-1/2}\sum_{t=T_{1}}^{T_{2}}|Y_{t-1}^{0}+Y_{t-1}^{c}|\cdot|Y_{t-1}^{0}-Y_{t-1}^{c}|\nonumber \\
 & = & \hspace{-0.08in}O(b_{n}^{-1/2})O(h/b_{n})(nh)^{-1/2}\sum_{t=T_{1}}^{T_{2}}|Y_{t-1}^{0}+Y_{t-1}^{c}|\cdot|Y_{t-1}^{0}-Y_{t-1}^{c}|,\label{Pf=000020re=000020A_c1,2=000020eqn=0000201}
\end{eqnarray}
where the equality uses Lemma \ref{lem:MT-Max=000020Intertemporal=000020Difference}(a).

We have 
\begin{equation}
Y_{t-1}^{0}-Y_{t-1}^{c}=\sum_{j=0}^{t-1-T_{1}}(c_{t-1,j}\sigma_{t-j-1}-\rho_{0,n}^{j}\sigma_{n\tau})U_{t-j-1}-\sum_{j=t-T_{1}}^{T_{2}-T_{1}}\rho_{0,n}^{j}\sigma_{n\tau}U_{t-j-1}.\label{Pf=000020re=000020A_c1,2=000020eqn=0000202}
\end{equation}

Combining (\ref{Pf=000020re=000020A_c1,2=000020eqn=0000201}) and
(\ref{Pf=000020re=000020A_c1,2=000020eqn=0000202}) gives 
\begin{eqnarray}
\E|A_{c1,2}|\hspace{-0.08in} & \leq & \hspace{-0.08in}(b_{n}^{-1/2})O(h/b_{n})(nh)^{-1/2}\sum_{t=T_{1}}^{T_{2}}2\max\{(\E(Y_{t-1}^{0})^{2})^{1/2},(\E(Y_{t-1}^{c})^{2})^{1/2}\}\nonumber \\
 &  & \times2\max\left\{ \left(\E\left(\sum_{j=0}^{t-1-T_{1}}(c_{t-1,j}\sigma_{t-j-1}-\rho_{0,n}^{j}\sigma_{n\tau})U_{t-j-1}\right)^{2}\right)^{1/2},\right.\nonumber \\
 &  & \left.\left(\E\left(\sum_{j=t-T_{1}}^{T_{2}-T_{1}}\rho_{0,n}^{j}\sigma_{n\tau}U_{t-j-1}\right)^{2}\right)^{1/2}\right\} \label{Pf=000020re=000020A_c1,2=000020eqn=0000203}
\end{eqnarray}
using the triangle and CS inequalities.

We have $(\E(Y_{t-1}^{c})^{2})^{1/2}=O(b_{n}^{1/2})$ uniformly over
$t\in[T_{1},T_{2}]$ by the discussion following (\ref{Pf=000020re=000020A_c1,1=000020eqn=0000204}).
In addition, 
\begin{eqnarray}
\E(Y_{t-1}^{0})^{2}\hspace{-0.08in} & = & \hspace{-0.08in}\E\left(\sum_{j=0}^{t-1-T_{1}}c_{t-1,j}\sigma_{t-j-1}U_{t-j-1}\right)^{2}=\sum_{j=0}^{t-1-T_{1}}c_{t-1,j}^{2}\sigma_{t-1-j}^{2}\nonumber \\
 & \leq & \hspace{-0.08in}C_{3,U}\sum_{j=0}^{\infty}\overline{\rho}_{n}^{2}=C_{3,U}(1-\overline{\rho}_{n}^{2})^{-1}=O(b_{n})\label{Pf=000020re=000020A_c1,2=000020eqn=0000204}
\end{eqnarray}
uniformly over $t\in{[}T_{1},T_{2}],$ where the inequality uses \eqref{eq:MT-c_tj=000020bounded-gto_1}
and the bound $C_{3,U}$ on $\sigma^{2}(\cdot)$ in part (i) of $\Lambda_{n}$
and the last equality uses \eqref{eq:MT-ps=000020rho^2t}. So, $(\E(Y_{t-1}^{0})^{2})^{1/2}=O(b_{n}^{1/2}).$

Next, we have 
\begin{equation}
|c_{t-1,j}\sigma_{t-j-1}-\rho_{0,n}^{j}\sigma_{n\tau}|\leq\sigma_{t-j-1}|c_{t-1,j}-\rho_{0,n}^{j}|+|\rho_{0,n}|^{j}|\sigma_{t-j-1}-\sigma_{n\tau}|\leq j\overline{\rho}_{n}^{j-1}O(h/b_{n})+\overline{\rho}_{n}^{j}O(h),\label{Pf=000020re=000020A_c1,2=000020eqn=0000205}
\end{equation}
where the second inequality uses \eqref{eq:MT-ctj=000020-=000020rho^j=000020bound-gto_1},
a uniform bound on $\sigma_{t-j-1}$ across $t,j,$ and $\max_{t\in[T_{1},T_{2}]}|\sigma_{t-j-1}\allowbreak-\sigma_{n\tau}|=O(h)$
by Lemma \ref{lem:MT-Max=000020Intertemporal=000020Difference}(b)
and $\sigma_{t-j-1}-\sigma_{n\tau}=(\sigma_{t-j-1}^{2}-\sigma_{n\tau}^{2})/(\sigma_{t-j-1}+\sigma_{n\tau}).$
In consequence, we obtain 
\begin{eqnarray}
 &  & \hspace{-0.08in}\E\left(\sum_{j=0}^{t-1-T_{1}}(c_{t-1,j}\sigma_{t-j-1}-\rho_{0,n}^{j}\sigma_{n\tau})U_{t-j-1}\right)^{2}\overset{}{=}\sum_{j=0}^{t-1-T_{1}}(c_{t-1,j}\sigma_{t-j-1}-\rho_{0,n}^{j}\sigma_{n\tau})^{2}\E U_{t-j-1}^{2}\nonumber \\
 & = & \hspace{-0.08in}O\left(\sum_{j=0}^{\infty}j^{2}\overline{\rho}_{n}^{2j-2}h^{2}/b_{n}^{2}\right)+O\left(\sum_{j=0}^{\infty}\overline{\rho}_{n}^{2j}h^{2}\right)=O(b_{n}^{3}h^{2}/b_{n}^{2})+O(b_{n}h^{2})=O(b_{n}h^{2}),\nonumber \\
\label{Pf=000020re=000020A_c1,2=000020eqn=0000206}
\end{eqnarray}
where the second equality uses (\ref{Pf=000020re=000020A_c1,2=000020eqn=0000205})
and $\E U_{t}^{2}=1$ for all $t,$ and the third equality uses \eqref{eq:MT-ps=000020rho^2t}
and \eqref{eq:MT-ps=000020t^2*rho^2t}.

In addition, 
\begin{eqnarray}
\E\left(\sum_{j=t-T_{1}}^{T_{2}-T_{1}}\rho_{0,n}^{j}\sigma_{n\tau}U_{t-j-1}\right)^{2}\hspace{-0.08in} & = & \hspace{-0.08in}\sum_{j=t-T_{1}}^{T_{2}-T_{1}}\rho_{0,n}^{2j}\sigma_{n\tau}^{2}\E U_{t-j-1}^{2}\leq\sigma_{n\tau}^{2}\rho_{0,n}^{2(t-T_{1})}\sum_{j=0}^{\infty}\rho_{0,n}^{2j}\nonumber \\
 & = & \hspace{-0.08in}\sigma_{n\tau}^{2}\rho_{0,n}^{2(t-T_{1})}(1-\rho_{0,n}^{2})^{-1}=\rho_{0,n}^{2(t-T_{1})}O(b_{n}).\label{Pf=000020re=000020A_c1,2=000020eqn=0000207}
\end{eqnarray}
And so, 
\begin{eqnarray}
 &  & \hspace{-0.08in}\sum_{t=T_{1}}^{T_{2}}\left(\E\left(\sum_{j=t-T_{1}}^{T_{2}-T_{1}}\rho_{0,n}^{j}\sigma_{n\tau}U_{t-j-1}\right)^{2}\right)^{1/2}\overset{}{\leq}\sum_{t=T_{1}}^{T_{2}}|\rho_{0,n}|^{t-T_{1}}O(b_{n}^{1/2})\nonumber \\
 & \leq & \hspace{-0.08in}O(b_{n}^{1/2})\sum_{t=0}^{\infty}|\rho_{0,n}|^{t}=O(b_{n}^{1/2})(1-\overline{\rho}_{n})^{-1}=O(b_{n}^{1/2})O(b_{n})=O(b_{n}^{3/2}).\label{Pf=000020re=000020A_c1,2=000020eqn=0000208}
\end{eqnarray}
Combining (\ref{Pf=000020re=000020A_c1,2=000020eqn=0000203}), (\ref{Pf=000020re=000020A_c1,2=000020eqn=0000204}),
(\ref{Pf=000020re=000020A_c1,2=000020eqn=0000206}), and (\ref{Pf=000020re=000020A_c1,2=000020eqn=0000208}),
gives 
\begin{eqnarray}
\E|A_{c1,2}|\hspace{-0.08in} & \leq & \hspace{-0.08in}O(b_{n}^{-1/2})O(h/b_{n})(nh)^{-1/2}nhO(b_{n}^{1/2})O(b_{n}^{1/2}h)\nonumber \\
 &  & +O(b_{n}^{-1/2})O(h/b_{n})(nh)^{-1/2}O(b_{n}^{1/2})O(b_{n}^{3/2})\label{Pf=000020re=000020A_c1,2=000020eqn=0000209}\\
 & = & \hspace{-0.08in}O(n^{1/2}h^{5/2}b_{n}^{-1/2})+O(n^{-1/2}h^{1/2}b_{n}^{1/2})=O((nh^{5})^{1/2})+O((b_{n}/nh)^{1/2}h)=o(1),\nonumber 
\end{eqnarray}
where the last equality uses $nh^{5}\rightarrow0$ by Assumption \ref{assu:MT-stationary-nh^5=000020->=0000200},
$h\rightarrow0$ by Assumption \ref{assu:MT-h}, and $b_{n}/nh\rightarrow0$
in the ``stationary'' case. By Markov's inequality, this gives $A_{c1,2}=o_{p}(1).$

Now, we consider $A_{c1,3}.$ Below we reuse calculations in \eqref{eq:a.11.10-Ab3}--\eqref{eq:case(iii)-1},
which bound the term 
\begin{eqnarray}
 &  & \E S_{n}^{2}\overset{}{:=}\E\left[(1-\rho_{0,n}^{2})(nh)^{-1}\sum_{t=T_{1}}^{T_{2}}(Y_{t-1}^{0})^{2}\right]^{2}\overset{}{=}(1-\rho_{0,n}^{2})^{2}(nh)^{-2}\E\left[\sum_{t=T_{1}}^{T_{2}}(Y_{t-1}^{0})^{2}\right]^{2},\text{ whereas }\nonumber \\
 &  & \E A_{c1,3}^{2}\overset{}{=}(1-\rho_{0,n}^{2})(nh)^{-1}\E\left[\sum_{t=T_{1}}^{T_{2}}(\rho_{t}-\rho_{0,n})\left((Y_{t-1}^{c})^{2}-\E(Y_{t-1}^{c})^{2}\right)\right]^{2}.\label{Pf=000020re_A_c1,3=000020eqn=0000201}
\end{eqnarray}
The terms $\E S_{n}^{2}$ and $\E A_{c1,3}^{2}$ are similar, but
differ as follows. The quantities $(1-\rho_{0,n}^{2})^{2}=O(b_{n}^{-2})$
and $(nh)^{-2}$ appear in $\E S_{n}^{2},$ whereas $(1-\rho_{0,n}^{2})=O(b_{n}^{-1})$
and $(nh)^{-1}$ appear in $\E A_{c1,3}^{2}.$ The quantity $\max_{t\in[T_{1},T_{2}]}|\rho_{t}-\rho_{0,n}|^{2}=O(h^{2}/b_{n}^{2})$
(by Lemma \ref{lem:MT-Max=000020Intertemporal=000020Difference}(a))
appears in the bound we obtain on $\E A_{c1,3}^{2},$ but does not
appear in the bound on $\E S_{n}^{2}.$ The difference between $Y_{t-1}^{0},$
which appears in $S_{n},$ and $Y_{t-1}^{c},$ which appears in $\E A_{c1,3}^{2},$
is not important because the same bounds can be employed with either
one. The term $S_{n}$ is based on summands $(Y_{t-1}^{0})^{2},$
which do not have mean 0, whereas $A_{c1,3}$ is based on mean zero
summands $(Y_{t-1}^{c})^{2}-\E(Y_{t-1}^{c})^{2},$ which is an important
difference.

Analogously to \eqref{eq:a.11.10-Ab3}, we can write 
\begin{eqnarray}
 &  & \hspace{-0.08in}\E\left[\sum_{t=T_{1}}^{T_{2}}(Y_{t-1}^{c})^{2}-\E(Y_{t-1}^{c})^{2}\right]^{2}\nonumber \\
 & = & \hspace{-0.08in}\E\sum_{t_{1},t_{2}=T_{1}-1}^{T_{2}-1}\left(\left(\sum_{i=0}^{t_{1}-T_{1}}\rho_{0,n}^{i}\sigma_{n\tau}U_{t_{1}-i}\right)^{2}-\E(Y_{t_{1}}^{c})^{2}\right)\left(\left(\sum_{j=0}^{t_{2}-T_{1}}\rho_{0,n}^{j}\sigma_{n\tau}U_{t_{2}-j}\right)^{2}-\E(Y_{t_{2}}^{c})^{2}\right)\nonumber \\
 & = & \hspace{-0.08in}\sigma_{n\tau}^{4}\E\sum_{t_{1},t_{2}=T_{1}-1}^{T_{2}-1}\left(\sum_{i_{1},i_{2}=0}^{t_{1}-T_{1}}\rho_{0,n}^{i_{1}}\rho_{0,n}^{i_{2}}(U_{t_{1}-i_{1}}U_{t_{1}-i_{2}}-\E U_{t_{1}-i_{1}}U_{t_{1}-i_{2}})\right)\nonumber \\
 &  & \times\left(\sum_{j_{1},j_{2}=0}^{t_{2}-T_{1}}\rho_{0,n}^{j_{1}}\rho_{0,n}^{j_{2}}(U_{t_{2}-j_{1}}U_{t_{2}-j_{2}}-\E U_{t_{2}-j_{1}}U_{t_{2}-j_{2}})\right)\nonumber \\
 & = & \hspace{-0.08in}\sigma_{n\tau}^{4}\sum_{t_{1},t_{2}=T_{1}-1}^{T_{2}-1}\sum_{i_{1},i_{2}=0}^{t_{1}-T_{1}}\sum_{j_{1},j_{2}=0}^{t_{2}-T_{1}}\rho_{0,n}^{i_{1}}\rho_{0,n}^{i_{2}}\rho_{0,n}^{j_{1}}\rho_{0,n}^{j_{2}}\nonumber \\
 &  & \times\E(U_{t_{1}-i_{1}}U_{t_{1}-i_{2}}-\E U_{t_{1}-i_{1}}U_{t_{1}-i_{2}})(U_{t_{2}-j_{1}}U_{t_{2}-j_{2}}-\E U_{t_{2}-j_{1}}U_{t_{2}-j_{2}})\nonumber \\
 & = & \hspace{-0.08in}\sigma_{n\tau}^{4}\sum_{t_{1},t_{2}=T_{1}-1}^{T_{2}-1}\sum_{i_{1},i_{2}=0}^{t_{1}-T_{1}}\sum_{j_{1},j_{2}=0}^{t_{2}-T_{1}}\rho_{0,n}^{i_{1}}\rho_{0,n}^{i_{2}}\rho_{0,n}^{j_{1}}\rho_{0,n}^{j_{2}}\nonumber \\
 &  & \times(\E U_{t_{1}-i_{1}}U_{t_{1}-i_{2}}U_{t_{2}-j_{1}}U_{t_{2}-j_{2}}-\E U_{t_{1}-i_{1}}U_{t_{1}-i_{2}}\cdot\E U_{t_{2}-j_{1}}U_{t_{2}-j_{2}}).\label{Pf=000020re_A_c1,3=000020eqn=0000202}
\end{eqnarray}
The term $\E A_{c1,3}^{2}$ equals the rhs of (\ref{Pf=000020re_A_c1,3=000020eqn=0000202})
multiplied by $(1-\rho_{0,n}^{2})(nh)^{-1}$ and with $(\rho_{t_{1}+1}-\rho_{0,n})(\rho_{t_{2}+1}-\rho_{0,n})$
inserted after the three summation signs.

As discussed following \eqref{eq:a.11.10-Ab3}, the expectations $\E U_{t_{1}-i_{1}}U_{t_{1}-i_{2}}U_{t_{2}-j_{1}}U_{t_{2}-j_{2}}$
in the last line of \eqref{eq:a.11.10-Ab3} are all zero except when
the indices fall in cases (i), (ii), or (iii), which are defined there.
Furthermore, case (ii) is subdivided into cases (ii1), (ii2), and
(ii3) just above \eqref{eq:case(ii)-1-A.11.13}. The difference between
the expectations on the rhs of (\ref{Pf=000020re_A_c1,3=000020eqn=0000202})
and that of \eqref{eq:a.11.10-Ab3} is that the former has $\eta_{t_{1}t_{2}i_{1}i_{2}j_{1}j_{2}}:=\E U_{t_{1}-i_{1}}U_{t_{1}-i_{2}}\cdot\E U_{t_{2}-j_{1}}U_{t_{2}-j_{2}}$
subtracted off, whereas the latter does not. The quantity $\eta_{t_{1}t_{2}i_{1}i_{2}j_{1}j_{2}}$
is non-zero iff $i_{1}=i_{2}$ and $j_{1}=j_{2}.$ The case $i_{1}=i_{2}$
and $j_{1}=j_{2}$ with $t_{1}-i_{1}=t_{2}-j_{1}$ is case (i). The
case $i_{1}=i_{2}$ and $j_{1}=j_{2}$ with $t_{1}-i_{1}\neq t_{2}-j_{1}$
is case (ii1). Hence, the expectations on the rhs of (\ref{Pf=000020re_A_c1,3=000020eqn=0000202})
are all zero except when the indices fall in cases (i), (ii), and
(iii), just as in \eqref{eq:a.11.10-Ab3}. In addition, in case (ii1),
the expectations on the rhs of (\ref{Pf=000020re_A_c1,3=000020eqn=0000202})
are zero, because $\E U_{t_{1}-i_{1}}U_{t_{1}-i_{2}}U_{t_{2}-j_{1}}U_{t_{2}-j_{2}}=\E U_{t_{1}-i_{1}}^{2}U_{t_{2}-j_{1}}^{2}=\E U_{t_{1}-i_{1}}^{2}\E U_{t_{2}-j_{1}}^{2}=\eta_{t_{1}t_{2}i_{1}i_{2}j_{1}j_{2}},$
where the second equality holds because $t_{1}-i_{1}\neq t_{2}-j_{1}$
in case (ii1). We conclude that the expectations on the rhs of (\ref{Pf=000020re_A_c1,3=000020eqn=0000202})
are non-zero only in cases (i), (ii2), (ii3), and (iii).

Now, we use the calculations in \eqref{eq:case(i)-2}--\eqref{eq:case(iii)-1}
to bound the terms in (\ref{Pf=000020re_A_c1,3=000020eqn=0000202})
when the indices fall in cases (i), (ii2), (ii3), and (iii).

For case (i), using \eqref{eq:case(i)-2} and \eqref{eq:case(i)-1},
we have: the sum over the indices in case (i) on the rhs of (\ref{Pf=000020re_A_c1,3=000020eqn=0000202})
is 
\begin{equation}
O(1)\sum_{t_{1}=T_{1}-1}^{T_{2}-1}\sum_{\ell=1}^{t_{1}-T_{1}}\overline{\rho}_{n}^{2\ell}\sum_{k=0}^{\infty}\overline{\rho}_{n}^{4k}+O(1)\sum_{t_{1}=T_{1}}^{T_{2}-1}\sum_{i=0}^{\infty}\overline{\rho}_{n}^{4i}=O(nhb_{n}^{2}),\label{Pf=000020re_A_c1,3=000020eqn=0000203}
\end{equation}
where the last equality uses \eqref{eq:MT-ps=000020rho^2t} and $\sum_{k=0}^{\infty}\overline{\rho}_{n}^{4k}=O(b_{n}).$
As noted above, $\E A_{c1,3}^{2}$ equals the rhs of (\ref{Pf=000020re_A_c1,3=000020eqn=0000202})
multiplied by $(1-\rho_{0,n}^{2})(nh)^{-1}=O(b_{n}^{-1})(nh)^{-1}$
and with $(\rho_{t_{1}+1}-\rho_{0,n})(\rho_{t_{2}+1}-\rho_{0,n})$
inserted after the three summands, where $\max_{t\in[T_{1},T_{2}]}|\rho_{t}-\rho_{0,n}|^{2}=O(h^{2}/b_{n}^{2})$
by Lemma \ref{lem:MT-Max=000020Intertemporal=000020Difference}(a).
In consequence, a bound on the sum of the terms in $\E A_{c1,3}^{2}$
that correspond to indices in case (i) is 
\begin{equation}
O(b_{n}^{-1})(nh)^{-1}O(h^{2}/b_{n}^{2})O(nhb_{n}^{2})=O(b_{n}^{-1}h^{2})=o(1).\label{Pf=000020re_A_c1,3=000020eqn=0000204}
\end{equation}

For case (ii2), using \eqref{eq:case(ii)-2} and \eqref{eq:case=000020(ii2)-1},
we have: the sum over the indices in case (ii2) on the rhs of (\ref{Pf=000020re_A_c1,3=000020eqn=0000202})
is 
\begin{equation}
O(1)\sum_{t_{1}>t_{2}=T_{1}}^{T_{2}-1}\overline{\rho}_{n}^{2(t_{1}-t_{2})}(1-\overline{\rho}^{2})^{-2}+O(1)\sum_{t=T_{1}}^{T_{2}-1}\sum_{i_{1},i_{2}=0}^{\infty}\overline{\rho}_{n}^{2(i_{1}-i_{2})}=O(nhb_{n}^{3})+O(nhb_{n}^{2})=O(nhb_{n}^{3}),\label{Pf=000020re_A_c1,3=000020eqn=0000205}
\end{equation}
where the first equality uses \eqref{eq:MT-ps=000020rho^2t}. Since
the bound in (\ref{Pf=000020re_A_c1,3=000020eqn=0000205}) is larger
than that in (\ref{Pf=000020re_A_c1,3=000020eqn=0000203}) by the
factor $b_{n},$ (\ref{Pf=000020re_A_c1,3=000020eqn=0000204}) implies
that the sum of the terms in $\E A_{c1,3}^{2}$ that correspond to
indices in case (ii2) is $O(h^{2})=o(1).$ Cases (ii2) and (ii3) are
symmetric. So, the same result holds for case (ii3).

For case (iii), using \eqref{eq:case(iii)-1-RHS} and \eqref{eq:case(iii)-1},
we have: the sum over the indices in case (iii) on the rhs of (\ref{Pf=000020re_A_c1,3=000020eqn=0000202})
is 
\begin{equation}
\sum_{t_{1}>t_{2}=T_{1}}^{T_{2}-1}\overline{\rho}_{n}^{2(t_{1}-t_{2})}O(b_{n}^{2})+O(nhb_{n}^{2})=O(nhb_{n}^{3})+O(nhb_{n}^{2})=O(nhb_{n}^{3}).\label{Pf=000020re_A_c1,3=000020eqn=0000206}
\end{equation}
Since the bounds in (\ref{Pf=000020re_A_c1,3=000020eqn=0000205})
and (\ref{Pf=000020re_A_c1,3=000020eqn=0000206}) are the same, the
sum of the terms in $\E A_{c1,3}^{2}$ that correspond to indices
in case (iii) is $O(h^{2})=o(1).$

To conclude, we have $\E A_{c1,3}^{2}=o(1).$ Hence, by Markov's inequality,
$A_{c1,3}=o_{p}(1).$ This concludes the proof that 
\begin{equation}
A_{c1}=o_{p}(1).\label{eq:Ac1_gto_0}
\end{equation}

For $A_{c2}$, we have 
\begin{align}
\E A_{c2}^{2}= & \ \E\left[\left(1-\rho_{0,n}^{2}\right)^{1/2}\left(nh\right)^{-1/2}\sum_{t=T_{1}}^{T_{2}}\left(\mu_{t}-\ol{\mu}_{nh}\right)Y_{t-1}^{0}\right]^{2}\nonumber \\
= & \left(1-\rho_{0,n}^{2}\right)\left(nh\right)^{-1}\sum_{t,s=T_{1}}^{T_{2}}\left(\mu_{t}-\ol{\mu}_{nh}\right)\left(\mu_{s}-\ol{\mu}_{nh}\right)\E\left(Y_{t-1}^{0}Y_{s-1}^{0}\right)\nonumber \\
= & \ O\left(b_{n}^{-1}\right)\left(nh\right)^{-1}\sum_{t,s=T_{1}}^{T_{2}}\left(\mu_{t}-\ol{\mu}_{nh}\right)\left(\mu_{s}-\ol{\mu}_{nh}\right)\sum_{i=0}^{t-T_{1}}\sum_{j=0}^{s-T_{1}}c_{t-1,i}c_{s-1,j}\sigma_{t-1-i}\sigma_{s-1-j}\nonumber \\
 & \ \times\E\left(U_{t-1-i}U_{s-1-j}\right)\nonumber \\
\leq & \ O\left(b_{n}^{-1}\right)\left(nh\right)^{-1}\left(\max_{t,s\in\left[T_{1},T_{2}\right]}\abs{\mu_{t}-\mu_{s}}\right)^{2}\left(2\sum_{t>s=T_{1}}^{T_{2}}\sum_{i=t-s}^{t-T_{1}}\ol{\rho}_{n}^{2i-\left(t-s\right)}+\sum_{t=T_{1}}^{T_{2}}\sum_{i=0}^{t-T_{1}}\ol{\rho}_{n}^{2i}\right)\nonumber \\
= & \ O\left(b_{n}^{-1}\right)\left(nh\right)^{-1}O\left(h^{2}/b_{n}^{2}\right)\left(O\left(nhb_{n}^{2}\right)+O\left(nhb_{n}\right)\right)=O\left(h^{2}/b_{n}\right)=o\left(1\right),\label{eq:A_C2=000020pto=0000200=000020lem=0000203.3}
\end{align}
where the inequality holds by dividing the case into $t=s$ and $t\neq s$
and \eqref{eq:MT-c_tj=000020bounded-gto_1} and the fourth equality
uses the change of coordinates $l=i-\left(t-s\right)$ and $k=t-s$,
\eqref{eq:MT-ps=000020rho^2t}, \eqref{eq:MT-ps=000020rho^(t-s)},
Lemma \ref{lem:MT-Max=000020Intertemporal=000020Difference}(d) and
the triangle inequality. Then, we have $A_{c2}\pto0$ by Markov's
inequality.

The proof of $A_{c3}\pto0$ is the same as that of $A_{c2}\pto0$
and is omitted.

For $A_{c4}$, we have 
\begin{align}
\E A_{c4}^{2}= & \ \E\left[\left(1-\rho_{0,n}^{2}\right)^{1/2}\left(nh\right)^{-1/2}\sum_{t=T_{1}}^{T_{2}}\left(\rho_{t}-\rho_{0,n}\right)c_{t-1,t-1-T_{0}}Y_{t-1}^{0}Y_{T_{0}}^{*}\right]^{2}\nonumber \\
= & \ O\left(b_{n}^{-1}\right)\left(nh\right)^{-1}\sum_{t,s=T_{1}}^{T_{2}}\left(\rho_{t}-\rho_{0,n}\right)\left(\rho_{s}-\rho_{0,n}\right)c_{t-1,t-1-T_{0}}c_{s-1,s-1-T_{0}}\E Y_{t-1}^{0}Y_{s-1}^{0}Y_{T_{0}}^{*2}\nonumber \\
\leq & \ O\left(b_{n}^{-1}\right)\left(nh\right)^{-1}\sum_{t,s=T_{1}+1}^{T_{2}}\left[\begin{array}{c}
\left(\rho_{t}-\rho_{0,n}\right)\left(\rho_{s}-\rho_{0,n}\right)c_{t-1,t-1-T_{0}}c_{s-1,s-1-T_{0}}\\
\times\sum_{i=0}^{t-1-T_{1}}\sum_{j=0}^{s-1-T_{1}}c_{t-1,i}c_{s-1,j}\E U_{t-1-i}U_{s-1-j}Y_{T_{0}}^{*2}
\end{array}\right]C_{3,U}\nonumber \\
= & \ O\left(b_{n}^{-1}\right)\left(nh\right)^{-1}2\sum_{t>s=T_{1}+1}^{T_{2}}\left[\begin{array}{c}
\left(\rho_{t}-\rho_{0,n}\right)\left(\rho_{s}-\rho_{0,n}\right)c_{t-1,t-1-T_{0}}c_{s-1,s-1-T_{0}}\\
\times\sum_{i=t-s}^{t-1-T_{1}}c_{t-1,i}c_{s-1,i-\left(t-s\right)}\E U_{t-1-i}^{2}Y_{T_{0}}^{*2}
\end{array}\right]\nonumber \\
 & +O\left(b_{n}^{-1}\right)\left(nh\right)^{-1}\sum_{t=T_{1}+1}^{T_{2}}\left[\left(\rho_{t}-\rho_{0,n}\right)^{2}c_{t-1,t-1-T_{0}}^{2}\sum_{i=0}^{t-1-T_{1}}c_{t-1,i}^{2}\E U_{t-1-i}^{2}Y_{T_{0}}^{*2}\right]\nonumber \\
\leq & \ O\left(b_{n}^{-1}\right)\left(nh\right)^{-1}O\left(h^{2}/b_{n}^{2}\right)\sum_{t>s=T_{1}+1}^{T_{2}}\left[\ol{\rho}_{n}^{t-1-T_{0}+s-1-T_{0}}\sum_{i=t-s}^{t-1-T_{1}}\ol{\rho}_{n}^{2i-\left(t-s\right)}\right]O\left(n\right)\nonumber \\
 & +O\left(b_{n}^{-1}\right)\left(nh\right)^{-1}O\left(h^{2}/b_{n}^{2}\right)\sum_{t=T_{1}+1}^{T_{2}}\left[\ol{\rho}_{n}^{2\left(t-1-T_{0}\right)}\sum_{i=0}^{t-1-T_{1}}\ol{\rho}_{n}^{2i}\right]O\left(n\right)\nonumber \\
= & \ O\left(b_{n}^{-1}\right)\left(nh\right)^{-1}O\left(h^{2}/b_{n}^{2}\right)O\left(b_{n}^{3}\right)O\left(n\right)+O\left(b_{n}^{-1}\right)\left(nh\right)^{-1}O\left(h^{2}/b_{n}^{2}\right)O\left(b_{n}^{2}\right)O\left(n\right)\nonumber \\
= & \ O\left(h\right)=o\left(1\right),\label{eq:A_c4=000020pto=0000200}
\end{align}
where the first inequality holds by $\max_{t\in\left[T_{1},T_{2}\right]}\sigma_{t}^{2}\leq C_{3,U}$,
the second inequality holds by Lemma \ref{lem:MT-Max=000020Intertemporal=000020Difference}(a),
\eqref{eq:=000020bound=000020order=000020of=000020E(Y_T0*)^2=000020O(n)=000020A.6.2},
\eqref{eq:MT-c_tj=000020bounded-gto_1}, and $\E\left[\rest{U_{t}^{2}}\mathscr{G}_{t-1}\right]=1\ \text{a.s.}$,
the third equality uses the law of iterated expectations (LIE) and
part (iv) of $\Lambda_{n}$, and the fourth equality holds by the
change of variables $l=2i-\left(t-s\right)$ and \eqref{eq:MT-ps=000020rho^2t}.
Then, by Markov's inequality, we have $A_{c4}\pto0$.

For $A_{c5}$, we have 
\begin{align}
\E A_{c5}^{2}= & \ \E\left[\left(1-\rho_{0,n}^{2}\right)^{1/2}\left(nh\right)^{-1/2}\sum_{t=T_{1}}^{T_{2}}\left(\mu_{t-1}-\ol{\mu}_{nh,-1}\right)\sigma_{t}U_{t}\right]^{2}\nonumber \\
= & \left(1-\rho_{0,n}^{2}\right)\left(nh\right)^{-1}\sum_{t=T_{1}}^{T_{2}}\left(\mu_{t-1}-\ol{\mu}_{nh,-1}\right)^{2}\sigma_{t}^{2}\E U_{t}^{2}\nonumber \\
\leq & \ O\left(b_{n}^{-1}\right)\left(nh\right)^{-1}nhO\left(h^{2}/b_{n}^{2}\right)C_{3,U}=o\left(1\right),\label{eq:A_C5=000020pto=0000200=000020lem=0000203.3}
\end{align}
where the inequality holds by 
\begin{equation}
\max_{t\in\left[T_{1},T_{2}\right]}\left(\mu_{t-1}-\ol{\mu}_{nh,-1}\right)^{2}\leq\max_{t,s\in\left[T_{0},T_{2}\right]}\left(\mu_{t}-\mu_{s}\right)^{2}=O\left(h^{2}/b_{n}^{2}\right)
\end{equation}
using Lemma \ref{lem:MT-Max=000020Intertemporal=000020Difference}(d)
and the triangle inequality. By Markov's inequality, we obtain $A_{c5}\pto0$.

The proof of 
\begin{equation}
A_{c6}:=\left(1-\rho_{0,n}^{2}\right)^{1/2}\left(nh\right)^{-1/2}\sum_{t=T_{1}}^{T_{2}}\left(\mu_{t-1}-\ol{\mu}_{nh,-1}\right)\left(\rho_{t}-\rho_{0,n}\right)Y_{t-1}^{0}\pto0\label{eq:A_c6=000020pto=0000200}
\end{equation}
is quite similar to that of $A_{c2}\gto0$ given in \eqref{eq:A_C2=000020pto=0000200=000020lem=0000203.3},
and hence, is omitted.

The proofs of 
\begin{equation}
\abs{A_{c7}}:=\abs{\left(1-\rho_{0,n}^{2}\right)^{1/2}\left(nh\right)^{-1/2}\sum_{t=T_{1}}^{T_{2}}\left(\mu_{t-1}-\ol{\mu}_{nh,-1}\right)\left(\mu_{t}-\ol{\mu}_{nh}\right)}\gto0\label{eq:A_c7-def}
\end{equation}
and 
\begin{equation}
\abs{A_{c8}}:=\abs{\left(1-\rho_{0,n}^{2}\right)^{1/2}\left(nh\right)^{-1/2}\sum_{t=T_{1}}^{T_{2}}\rho_{0,n}\left(\mu_{t-1}-\ol{\mu}_{nh,-1}\right)^{2}}\gto0\label{eq:A_c8-def}
\end{equation}
are identical, thus we only prove the result for $A_{c8}$. An application
of the triangle inequality gives 
\begin{equation}
\abs{A_{c8}}\leq\left(1-\rho_{0,n}^{2}\right)^{1/2}\left(nh\right)^{-1/2}O\left(nh\right)\max_{t,s\in\left[T_{0},T_{2}\right]}\left(\mu_{t}-\mu_{s}\right)^{2}=O\left(\left(nh^{5}/b_{n}^{5}\right)^{1/2}\right)=o\left(1\right),\label{eq:A_c8=000020pto=0000200}
\end{equation}
where the second last equality holds by Lemma \ref{lem:MT-Max=000020Intertemporal=000020Difference}(d)
and the triangle inequality and the last equality holds by Assumption
\ref{assu:MT-stationary-nh^5=000020->=0000200}.

For $A_{c9}$, we have 
\begin{align}
\abs{A_{c9}}= & \abs{\left(1-\rho_{0,n}^{2}\right)^{1/2}\left(nh\right)^{-1/2}\sum_{t=T_{1}}^{T_{2}}\left(\mu_{t-1}-\ol{\mu}_{nh,-1}\right)\left(\rho_{t}-\rho_{0,n}\right)c_{t-1,t-1-T_{0}}Y_{T_{0}}^{*}}\nonumber \\
\leq & \ O\left(b_{n}^{-1/2}\right)\left(nh\right)^{-1/2}\max_{t,s\in\left[T_{0},T_{2}\right]}\abs{\mu_{t}-\mu_{s}}\max_{t\in\left[T_{1},T_{2}\right]}\abs{\rho_{t}-\rho_{0,n}}\sum_{t=T_{1}}^{T_{2}}\ol{\rho}_{n}^{t-1-T_{0}}\abs{Y_{T_{0}}^{*}}\nonumber \\
= & \ O\left(b_{n}^{-1/2}\right)\left(nh\right)^{-1/2}O\left(h\right)O\left(h/b_{n}\right)O\left(b_{n}\right)O_{p}\left(n^{1/2}\right)=O_{p}\left(h^{3/2}b_{n}^{-1/2}\right),\label{eq:A_c9=000020pto=0000200=000020lem=0000203.3}
\end{align}
where the second last equality holds by Lemma \ref{lem:MT-Max=000020Intertemporal=000020Difference}(a)
and (d), \eqref{eq:=000020bound=000020order=000020of=000020Y_T0*=000020Op(n^1/2)},
and \eqref{eq:MT-ps=000020rho^t}.

For $A_{c10},$ we have 
\begin{eqnarray}
EA_{c10}^{2}\hspace{-0.08in} & = & \hspace{-0.08in}E\left[(1-\rho_{0,n}^{2})^{1/2}(nh)^{-1/2}\sum_{t=T_{1}}^{T_{2}}c_{t-1,t-1-T_{0}}Y_{T_{0}}^{\ast}\sigma_{t}U_{t}\right]^{2}\nonumber \\
 & = & \hspace{-0.08in}(1-\rho_{0,n}^{2})(nh)^{-1}\sum_{t=T_{1}}^{T_{2}}c_{t-1,t-1-T_{0}}^{2}\sigma_{t}^{2}EU_{t}^{2}Y_{T_{0}}^{\ast2}\nonumber \\
 & \leq & \hspace{-0.08in}(1-\rho_{0,n}^{2})(nh)^{-1}\sum_{t=T_{1}}^{T_{2}}\overline{\rho}_{n}^{2(t-1-T_{0})}\sigma_{t}^{2}EY_{T_{0}}^{\ast2}\nonumber \\
 & \leq & \hspace{-0.08in}O(b_{n}^{-1})(nh)^{-1}O(b_{n})O(b_{n})=O(b_{n}/nh)=o(1),\label{A_c10=000020eqn}
\end{eqnarray}
where the second equality uses the martingale difference properties
of $\{U_{t}\}_{t\leq n},$ the first inequality uses (\ref{eq:MT-c_tj=000020bounded-gto_1})
and $E(U_{t}^{2}|Y_{T_{0}}^{\ast2})=1$ a.s. by part (iv) of $\Lambda_{n},$
the second inequality uses Lemma \ref{lem:MT-Order=000020of=000020Y_T0}(c),
which applies because $nh/b_{n}\rightarrow r_{0}=\infty$ is an assumption
of the lemma, and (\ref{eq:MT-ps=000020rho^t}), and the last equality
holds because $b_{n}/nh\rightarrow0.$ Equation (\ref{A_c10=000020eqn})
and Markov's inequality give $A_{c10}=o_{p}(1).$

For $A_{c11},$ we have: $A_{c11}=A_{c4}=o_{p}(1)$ by (\ref{eq:A_c4=000020pto=0000200}).

For $A_{c12},$ we have 
\begin{eqnarray}
|A_{c12}|\hspace{-0.08in} & = & \hspace{-0.08in}\left\vert (1-\rho_{0,n}^{2})^{1/2}(nh)^{-1/2}\sum_{t=T_{1}}^{T_{2}}c_{t-1,t-1-T_{0}}(\mu_{t-1}-\overline{\mu}_{nh,-1})Y_{T_{0}}^{\ast}\right\vert \nonumber \\
 & \leq & \hspace{-0.08in}(1-\rho_{0,n}^{2})^{1/2}(nh)^{-1/2}\sum_{t=T_{1}}^{T_{2}}\overline{\rho}_{n}^{t-1-T_{0}}\max_{t\in[T_{1},T_{2}]}|\mu_{t-1}-\overline{\mu}_{nh,-1}|\cdot|Y_{T_{0}}^{\ast}|\nonumber \\
 & = & \hspace{-0.08in}O(b_{n}^{-1/2})(nh)^{-1/2}O(b_{n})O(h/b_{n})O_{p}(b_{n}^{1/2})=O_{p}((nh)^{-1/2}h)=o_{p}(1),\nonumber \\
\end{eqnarray}
where the inequality uses (\ref{eq:MT-c_tj=000020bounded-gto_1}),
the second equality uses (\ref{eq:MT-ps=000020rho^t}), Lemma \ref{lem:MT-Max=000020Intertemporal=000020Difference}(d),
and Lemma \ref{lem:MT-Order=000020of=000020Y_T0}(c).

We have $|A_{c13}|=|\rho_{0,n}A_{c12}|=o_{p}(1)$ (since $A_{c12}=o_{p}(1)$
and $|\rho_{0,n}|\leq1).$

For $A_{c14},$ we have 
\begin{eqnarray}
|A_{c14}|\hspace{-0.08in} & = & \hspace{-0.08in}\left\vert (1-\rho_{0,n}^{2})^{1/2}(nh)^{-1/2}\sum_{t=T_{1}}^{T_{2}}c_{t-1,t-1-T_{0}}^{2}(\rho_{t}-\rho_{0,n})Y_{T_{0}}^{\ast2}\right\vert \nonumber \\
 & \leq & \hspace{-0.08in}(1-\rho_{0,n}^{2})^{1/2}(nh)^{-1/2}\sum_{t=T_{1}}^{T_{2}}\overline{\rho}_{n}^{2(t-1-T_{0})}\max_{t\in[T_{1},T_{2}]}|\rho_{t}-\rho_{0,n}|Y_{T_{0}}^{\ast2}\nonumber \\
 & = & \hspace{-0.08in}O(b_{n}^{-1/2})(nh)^{-1/2}O(b_{n})O(h/b_{n})O_{p}(b_{n})\nonumber \\
 & = & \hspace{-0.08in}O_{p}((b_{n}/nh)^{1/2}h)=o_{p}(1),\label{A_c14=000020pf}
\end{eqnarray}
where the inequality uses (\ref{eq:MT-c_tj=000020bounded-gto_1}),
the second equality uses (\ref{eq:MT-ps=000020rho^t}), Lemma \ref{lem:MT-Max=000020Intertemporal=000020Difference}(a),
and Lemma \ref{lem:MT-Order=000020of=000020Y_T0}(c).

By \eqref{eq:Ac1_gto_0}--(\ref{A_c14=000020pf}), we obtain 
\begin{equation}
\sum_{i=1}^{14}A_{ci}\pto0.\label{eq:=000020ac1+ac19=000020pto=0000200-1}
\end{equation}

Combining \eqref{eq:decomp=000020A3=000020stationary=000020numerator-1},
\eqref{eq:=000020dominate=000020term=000020numerator-1}, and \eqref{eq:=000020ac1+ac19=000020pto=0000200-1},
we have 
\begin{align}
\left(1-\rho_{0,n}^{2}\right)^{1/2}\left(nh\right)^{-1/2}\sum_{t=T_{1}}^{T_{2}}\left(Y_{t-1}-\ol{\mu}_{nh,-1}\right) & \left[\begin{array}{c}
Y_{t}-\ol{\mu}_{nh}-\rho_{0,n}\left(Y_{t-1}-\ol{\mu}_{nh,-1}\right)\end{array}\right]\dto N\left(0,1\right),\label{eq:part(c)-stationary=000020case-1}
\end{align}
as desired. 
\end{proof}
\begin{proof}[\textbf{Proof of Lemma \ref{lem:MT-asympdist_dist_stationary-numerator=000020}$\left(\text{b}\right)$}]
To prove part (b), by Lemma \ref{lem:MT-asympdist_dist_stationary-denom}(b),
we have 
\begin{equation}
\left(1-\rho_{0,n}^{2}\right)^{1/2}\left(\ol Y_{nh,-1}-\ol{\mu}_{nh,-1}\right)=o_{p}\left(1\right).\label{eq:=000020part=000020b=000020result-1}
\end{equation}
Thus, we only need to show 
\begin{equation}
\left(nh\right)^{-1/2}\sum_{t=T_{1}}^{T_{2}}\left[Y_{t}-\ol{\mu}_{nh}-\rho_{0,n}\left(Y_{t-1}-\ol{\mu}_{nh,-1}\right)\right]=O_{p}\left(1\right).\label{eq:part(d)-=000020goal-1}
\end{equation}
This is done by examining each component of \eqref{eq:part(d)-=000020goal-1}
and showing it is $O_{p}\left(1\right)$. Specifically, by \eqref{eq:5.10.2},
we expand 
\begin{align}
 & \left(nh\right)^{-1/2}\sum_{t=T_{1}}^{T_{2}}\left[Y_{t}-\ol{\mu}_{nh}-\rho_{0,n}\left(Y_{t-1}-\ol{\mu}_{nh,-1}\right)\right]\nonumber \\
= & \left(nh\right)^{-1/2}\sum_{t=T_{1}}^{T_{2}}\left[\begin{array}{c}
\sigma_{t}U_{t}+\left(\rho_{t}-\rho_{0,n}\right)Y_{t-1}^{0}+\left(\mu_{t}-\ol{\mu}_{nh}\right)\\
-\rho_{0,n}\left(\mu_{t-1}-\ol{\mu}_{nh,-1}\right)+\left(\rho_{t}-\rho_{0,n}\right)c_{t-1,t-1-T_{0}}Y_{T_{0}}^{*}
\end{array}\right]=:\sum_{i=1}^{5}A_{di}.\label{eq:d1-5-1}
\end{align}
We observe that$A_{d1}=O_{p}\left(1\right)$ by the central limit
theorem for martingale difference sequences. The proof of $A_{d2}=o_{p}\left(1\right)$
is almost the same as that of $A_{c2}$ except that (i) $A_{d2}$
does not have the $\left(1-\rho_{0,n}^{2}\right)^{1/2}$ multiplicand,
which is $O\left(b_{n}^{-1/2}\right)$ and (ii) $A_{d2}$ has $\rho_{t}-\rho_{0,n}$
in place of $\mu_{t}-\ol{\mu}_{nh}$, both of which have the same
order of $O\left(h/b_{n}\right)$ of maximum intertemporal difference
on $\left[T_{1},T_{2}\right]$ by Lemma \ref{lem:MT-Max=000020Intertemporal=000020Difference}(a)
and (d). Thus, following an argument identical to that in \eqref{eq:A_C2=000020pto=0000200=000020lem=0000203.3},
we have 
\begin{equation}
A_{d2}=O\left(b_{n}^{1/2}\right)O_{p}\left(hb_{n}^{-1/2}\right)=O_{p}\left(h\right)=o_{p}\left(1\right).\label{eq:Ad2}
\end{equation}
By the definitions of $A_{d3}$ and $A_{d4}$, we get $A_{d3}=A_{d4}=0$.
Finally, we derive $A_{d5}=o_{p}\left(1\right)$ by observing 
\begin{align}
\abs{A_{d5}} & \leq\max_{t\in\left[T_{1},T_{2}\right]}\abs{\rho_{t}-\rho_{0,n}}\left(nh\right)^{-1/2}\sum_{t=T_{1}}^{T_{2}}\begin{array}{c}
\ol{\rho}_{n}^{t-T_{1}}\abs{Y_{T_{0}}^{*}}\end{array}\nonumber \\
 & =O\left(h/b_{n}\right)\left(nh\right)^{-1/2}O\left(b_{n}\right)O_{p}\left(n^{1/2}\right)=O_{p}\left(h^{1/2}\right)=o_{p}\left(1\right),\label{eq:Ad5=000020pto=0000200}
\end{align}
where the first equality holds by Lemma \ref{lem:MT-Max=000020Intertemporal=000020Difference}(a),
\eqref{eq:=000020bound=000020order=000020of=000020Y_T0*=000020Op(n^1/2)},
and \eqref{eq:MT-ps=000020rho^t}.

Thus, $\sum_{i=1}^{5}A_{di}=o_{p}\left(1\right)$, (\ref{eq:part(d)-=000020goal-1})
holds, and the proof is complete. 
\end{proof}

\subsection{\protect\label{subsec:Proof-of-Theorem-2-Stationary}Proof of Theorem
\ref{thm:MT-asympdist_dist_rho^hat-stationary}}

In this section, for notational simplicity in the proof, we assume
that $\sigma_{0}^{2}\left(\tau\right)=1.$ 
\begin{proof}[\textbf{Proof of Theorem \ref{thm:MT-asympdist_dist_rho^hat-stationary}}]
Recall from \eqref{eq:MT-stationary=000020-=000020rho=000020hat-1}
that 
\begin{align}
 & \left(1-\rho_{0,n}^{2}\right)^{-1/2}\left(nh\right)^{1/2}\left(\widehat{\rho}_{n\tau}-\rho_{0,n}\right)\nonumber \\
= & \dfrac{\left(1-\rho_{0,n}^{2}\right)^{1/2}\left(nh\right)^{-1/2}\sum_{t=T_{1}}^{T_{2}}\left(Y_{t-1}-\ol Y_{nh,-1}\right)\left(Y_{t}-\rho_{0,n}Y_{t-1}\right)}{\left(1-\rho_{0,n}^{2}\right)\left(nh\right)^{-1}\sum_{t=T_{1}}^{T_{2}}\left(Y_{t-1}-\ol Y_{nh,-1}\right)^{2}}.\label{eq:=000020rho_hat_stat_normalized-1}
\end{align}
We analyze the denominator and numerator of \eqref{eq:=000020rho_hat_stat_normalized-1}
separately.

First, for the denominator of \eqref{eq:=000020rho_hat_stat_normalized-1},
we have 
\begin{align}
 & \left(1-\rho_{0,n}^{2}\right)\left(nh\right)^{-1}\sum_{t=T_{1}}^{T_{2}}\left(Y_{t-1}-\ol Y_{nh,-1}\right)^{2}\nonumber \\
= & \left(1-\rho_{0,n}^{2}\right)\left(nh\right)^{-1}\sum_{t=T_{1}}^{T_{2}}\left(Y_{t-1}-\ol{\mu}_{nh,-1}\right)^{2}-\left(1-\rho_{0,n}^{2}\right)\left(\ol Y_{nh,-1}-\ol{\mu}_{nh,-1}\right)^{2}\nonumber \\
=: & \ A_{f1}+A_{f2}.\label{eq:denominator=000020decomp-1}
\end{align}
By Lemma \ref{lem:MT-asympdist_dist_stationary-denom}(a) and (b),
we have $A_{f1}\pto1$ and $A_{f2}\pto0$, respectively. Therefore,
we have 
\begin{equation}
\left(1-\rho_{0,n}^{2}\right)\left(nh\right)^{-1}\sum_{t=T_{1}}^{T_{2}}\left(Y_{t-1}-\ol Y_{nh,-1}\right)^{2}\pto1.\label{eq:=000020denominator=000020stat=000020norm-1}
\end{equation}

Next, for the numerator of \eqref{eq:=000020rho_hat_stat_normalized-1},
we have 
\begin{align}
 & \left(1-\rho_{0,n}^{2}\right)^{1/2}\left(nh\right)^{-1/2}\sum_{t=T_{1}}^{T_{2}}\left(Y_{t-1}-\ol Y_{nh,-1}\right)\left(Y_{t}-\rho_{0,n}Y_{t-1}\right)\nonumber \\
= & \left(1-\rho_{0,n}^{2}\right)^{1/2}\left(nh\right)^{-1/2}\nonumber \\
 & \times\sum_{t=T_{1}}^{T_{2}}\left[\begin{array}{c}
Y_{t-1}-\ol{\mu}_{nh,-1}-\left(\ol Y_{nh,-1}-\ol{\mu}_{nh,-1}\right)\end{array}\right]\left[\begin{array}{c}
Y_{t}-\ol{\mu}_{nh}-\rho_{0,n}\left(Y_{t-1}-\ol{\mu}_{nh,-1}\right)\end{array}\right]\nonumber \\
= & \left(1-\rho_{0,n}^{2}\right)^{1/2}\left(nh\right)^{-1/2}\sum_{t=T_{1}}^{T_{2}}\left(Y_{t-1}-\ol{\mu}_{nh,-1}\right)\left[Y_{t}-\ol{\mu}_{nh}-\rho_{0,n}\left(Y_{t-1}-\ol{\mu}_{nh,-1}\right)\right]\nonumber \\
 & +\left(1-\rho_{0,n}^{2}\right)^{1/2}\left(nh\right)^{-1/2}\sum_{t=T_{1}}^{T_{2}}\left(\ol Y_{nh,-1}-\ol{\mu}_{nh,-1}\right)\left[Y_{t}-\ol{\mu}_{nh}-\rho_{0,n}\left(Y_{t-1}-\ol{\mu}_{nh,-1}\right)\right]\nonumber \\
=: & A_{f3}+A_{f4}.
\end{align}
By Lemma \ref{lem:MT-asympdist_dist_stationary-numerator=000020}(a)
and (b), we have $A_{f3}\dto N\left(0,1\right)$ and $A_{f4}\pto0$,
respectively. Therefore, we obtain 
\begin{equation}
\left(1-\rho_{0,n}^{2}\right)^{1/2}\left(nh\right)^{-1/2}\sum_{t=T_{1}}^{T_{2}}\left(Y_{t-1}-\ol Y_{nh,-1}\right)\left(Y_{t}-\rho_{0,n}Y_{t-1}\right)\dto N\left(0,1\right).\label{eq:numerator-stationary=000020case-1}
\end{equation}

Combining \eqref{eq:=000020denominator=000020stat=000020norm-1} and
\eqref{eq:numerator-stationary=000020case-1}, we have 
\begin{equation}
\left(1-\rho_{0,n}^{2}\right)^{-1/2}\left(nh\right)^{1/2}\left(\widehat{\rho}_{n\tau}-\rho_{0,n}\right)\dto N\left(0,1\right).\label{eq:stationary-limit-distribution-1}
\end{equation}

For the t-statistic $T_{n}\left(\rho_{0,n}\right)$, by \eqref{eq:MT-sigma=000020ntau=000020def},
\eqref{eq:MT-t-stat=000020def}, \eqref{eq:stationary-limit-distribution-1},
and the CMT, we only need to show 
\begin{equation}
\left(1-\rho_{0,n}^{2}\right)^{-1}\widehat{s}_{n\tau}^{2}=\dfrac{\left(nh\right)^{-1}\sum_{t=T_{1}}^{T_{2}}\left[Y_{t}-\ol Y_{nh}-\widehat{\rho}_{n\tau}\left(Y_{t-1}-\ol Y_{nh,-1}\right)\right]^{2}}{\left(1-\rho_{0,n}^{2}\right)\left(nh\right)^{-1}\sum_{t=T_{1}}^{T_{2}}\left(Y_{t-1}-\ol Y_{nh,-1}\right)^{2}}\pto1.\label{eq:denom=000020of=000020stat=000020case}
\end{equation}
Equation \eqref{eq:=000020denominator=000020stat=000020norm-1} shows
that the denominator converges in probability to one. For the numerator
of \eqref{eq:denom=000020of=000020stat=000020case}, we have 
\begin{align}
 & \left(nh\right)^{-1}\sum_{t=T_{1}}^{T_{2}}\left[Y_{t}-\ol Y_{nh}-\widehat{\rho}_{n\tau}\left(Y_{t-1}-\ol Y_{nh,-1}\right)\right]^{2}\nonumber \\
= & \left(nh\right)^{-1}\sum_{t=T_{1}}^{T_{2}}\left[Y_{t}-\ol Y_{nh}-\rho_{0,n}\left(Y_{t-1}-\ol Y_{nh,-1}\right)\right]^{2}\nonumber \\
 & +\left(nh\right)^{-1}\sum_{t=T_{1}}^{T_{2}}\left[\left(\widehat{\rho}_{n\tau}-\rho_{0,n}\right)\left(Y_{t-1}-\ol Y_{nh,-1}\right)\right]^{2}\nonumber \\
 & +2\left(nh\right)^{-1}\sum_{t=T_{1}}^{T_{2}}\left\{ \begin{array}{c}
\left(\widehat{\rho}_{n\tau}-\rho_{0,n}\right)\left(Y_{t-1}-\ol Y_{nh,-1}\right)\\
\times\left[Y_{t}-\ol Y_{nh}-\rho_{0,n}\left(Y_{t-1}-\ol Y_{nh,-1}\right)\right]
\end{array}\right\} \nonumber \\
\eqqcolon & A_{f5}+A_{f6}+A_{f7}.\label{eq:decomp}
\end{align}
We show $A_{f5}\pto1$ and $A_{f6}\pto0$, which imply $A_{f7}\pto0$
by the CS inequality.

For $A_{f5}$, we have 
\begin{align}
 & \left(nh\right)^{-1}\sum_{t=T_{1}}^{T_{2}}\left[Y_{t}-\ol Y_{nh}-\rho_{0,n}\left(Y_{t-1}-\ol Y_{nh,-1}\right)\right]^{2}\nonumber \\
= & \left(nh\right)^{-1}\sum_{t=T_{1}}^{T_{2}}\left[Y_{t}-\ol{\mu}_{nh}-\rho_{0,n}\left(Y_{t-1}-\ol{\mu}_{nh,-1}\right)\right]^{2}\nonumber \\
 & +\left[\ol Y_{nh}-\ol{\mu}_{nh}-\rho_{0,n}\left(\ol Y_{nh,-1}-\ol{\mu}_{nh,-1}\right)\right]^{2}\nonumber \\
 & +2\left(nh\right)^{-1}\sum_{t=T_{1}}^{T_{2}}\left\{ \begin{array}{c}
\left[Y_{t}-\ol{\mu}_{nh}-\rho_{0,n}\left(Y_{t-1}-\ol{\mu}_{nh,-1}\right)\right]\\
\times\left[\ol Y_{nh}-\ol{\mu}_{nh}-\rho_{0,n}\left(Y_{t-1}-\ol{\mu}_{nh,-1}\right)\right]
\end{array}\right\} \nonumber \\
\eqqcolon & A_{f51}+A_{f52}+A_{f53}.\label{eq:A5}
\end{align}
We show $A_{f51}\pto1$ and $A_{f52}\pto0$, which imply $A_{f53}\pto0$.
Recall that by \eqref{eq:5.10.2}, we have 
\begin{align}
A_{f51}= & \left(nh\right)^{-1}\sum_{t=T_{1}}^{T_{2}}\left[\begin{array}{c}
\sigma_{t}U_{t}+\left(\rho_{t}-\rho_{0,n}\right)Y_{t-1}^{0}+\left(\mu_{t}-\ol{\mu}_{nh}\right)\\
-\rho_{0,n}\left(\mu_{t-1}-\ol{\mu}_{nh,-1}\right)+\left(\rho_{t}-\rho_{0,n}\right)c_{t-1,t-1-T_{0}}Y_{T_{0}}^{*}
\end{array}\right]^{2}\nonumber \\
= & \left(nh\right)^{-1}\sum_{t=T_{1}}^{T_{2}}\left(\sigma_{t}U_{t}\right)^{2}\nonumber \\
 & +\left(nh\right)^{-1}\sum_{t=T_{1}}^{T_{2}}\left[\begin{array}{c}
\left(\rho_{t}-\rho_{0,n}\right)Y_{t-1}^{0}+\left(\mu_{t}-\ol{\mu}_{nh}\right)\\
-\rho_{0,n}\left(\mu_{t-1}-\ol{\mu}_{nh,-1}\right)+\left(\rho_{t}-\rho_{0,n}\right)c_{t-1,t-1-T_{0}}Y_{T_{0}}^{*}
\end{array}\right]^{2}\nonumber \\
 & +2\left(nh\right)^{-1}\sum_{t=T_{1}}^{T_{2}}\sigma_{t}U_{t}\left[\begin{array}{c}
\left(\rho_{t}-\rho_{0,n}\right)Y_{t-1}^{0}+\left(\mu_{t}-\ol{\mu}_{nh}\right)\\
-\rho_{0,n}\left(\mu_{t-1}-\ol{\mu}_{nh,-1}\right)+\left(\rho_{t}-\rho_{0,n}\right)c_{t-1,t-1-T_{0}}Y_{T_{0}}^{*}
\end{array}\right]\nonumber \\
\eqqcolon & A_{f511}+A_{f512}+A_{f513}.\label{eq:A51}
\end{align}

By \eqref{eq:a511=000020result}, we have $A_{f511}\pto1$. For $A_{f512}$,
we have 
\begin{align}
 & \left(nh\right)^{-1}\sum_{t=T_{1}}^{T_{2}}\left[\left(\rho_{t}-\rho_{0,n}\right)Y_{t-1}^{0}\right]^{2}\nonumber \\
\leq & \max_{t\in\left[T_{1},T_{2}\right]}\left(\rho_{t}-\rho_{0,n}\right)^{2}\left(1-\rho_{0,n}^{2}\right)^{-1}\left(1-\rho_{0,n}^{2}\right)\left(nh\right)^{-1}\sum_{t=T_{1}}^{T_{2}}\left(Y_{t-1}^{0}\right)^{2}\nonumber \\
= & \ O\left(\left(h/b_{n}\right)^{2}\right)O\left(b_{n}\right)O_{p}\left(1\right)=o_{p}\left(1\right),\label{eq:SM-Af5121}
\end{align}
where the first equality holds by Lemmas \ref{lem:MT-Max=000020Intertemporal=000020Difference}(a)
and \ref{lem:MT-Asymptotic=000020Properties=000020of=000020the=000020Components=000020Stationary}(b),
and the fact that $\left(1-\rho_{0,n}^{2}\right)^{-1}=O\left(b_{n}\right)$.
Additionally, by Lemmas \ref{lem:MT-Max=000020Intertemporal=000020Difference}(a),
\ref{lem:MT-Max=000020Intertemporal=000020Difference}(d), and \ref{lem:MT-Order=000020of=000020Y_T0}(a),
and \eqref{eq:MT-ps=000020rho^2t}, we have 
\begin{equation}
\left(nh\right)^{-1}\sum_{t=T_{1}}^{T_{2}}\left(\mu_{t}-\ol{\mu}_{nh}\right)^{2}\pto0,\label{eq:SM-Af5122}
\end{equation}
and 
\begin{align}
 & \left(nh\right)^{-1}\sum_{t=T_{1}}^{T_{2}}\left[\left(\rho_{t}-\rho_{0,n}\right)c_{t-1,t-1-T_{0}}Y_{T_{0}}^{*}\right]^{2}\nonumber \\
\leq & \left(nh\right)^{-1}\max_{t\in\left[T_{1},T_{2}\right]}\left(\rho_{t}-\rho_{0,n}\right)^{2}\sum_{t=T_{1}}^{T_{2}}\ol{\rho}_{n}^{2\left(t-T_{1}\right)}O_{p}\left(n\right)=O_{p}\left(\left(nh\right)^{-1}\left(h/b_{n}\right)^{2}nb_{n}\right)=o_{p}\left(1\right).\label{eq:SM-Af5124}
\end{align}
Similarly to (\ref{eq:SM-Af5122}), we have 
\begin{equation}
\left(nh\right)^{-1}\sum_{t=T_{1}}^{T_{2}}\left[\rho_{0,n}\left(\mu_{t-1}-\ol{\mu}_{nh,-1}\right)\right]^{2}\pto0\label{eq:SM-Af5123}
\end{equation}
since $\abs{\rho_{0,n}}\leq1$. By (\ref{eq:SM-Af5121})--(\ref{eq:SM-Af5123})
and the CS inequality, we have $A_{f512}\pto0$. Then, by $A_{f511}\pto1$,
$A_{f512}\pto0$, and the CS inequality, we have $A_{f513}\pto0$.
Thus, $A_{f51}\pto1$.

Next, for $A_{f52}$ we have 
\begin{align}
 & \ol Y_{nh}-\ol{\mu}_{nh}-\rho_{0,n}\left(\ol Y_{nh,-1}-\ol{\mu}_{nh,-1}\right)\nonumber \\
= & \left(nh\right)^{-1}\sum_{t=T_{1}}^{T_{2}}\left[\sigma_{t}U_{t}+\left(\rho_{t}-\rho_{0,n}\right)Y_{t-1}^{0}+\left(\rho_{t}-\rho_{0,n}\right)c_{t-1,t-1-T_{0}}Y_{T_{0}}^{*}\right]\nonumber \\
\eqqcolon & A_{f521}+A_{f522}+A_{f523}.\label{eq:A52=000020expand}
\end{align}

By the weak law of large numbers, $A_{f521}\pto0$. For $A_{f522}$,
we have 
\begin{align}
A_{f522}^{2}= & \left\{ \left(nh\right)^{-1}\sum_{t=T_{1}}^{T_{2}}\left[\left(\rho_{t}-\rho_{0,n}\right)Y_{t-1}^{0}\right]\right\} ^{2}\nonumber \\
\leq & \max_{t\in\left[T_{1},T_{2}\right]}\left(\rho_{t}-\rho_{0,n}\right)^{2}\left(1-\rho_{0,n}^{2}\right)^{-1}\left(1-\rho_{0,n}^{2}\right)\left(nh\right)^{-1}\sum_{t=T_{1}}^{T_{2}}\left(Y_{t-1}^{0}\right)^{2}\nonumber \\
= & \ O\left(h^{2}/b_{n}^{2}\right)O\left(b_{n}\right)O_{p}\left(1\right)=o_{p}\left(1\right),\label{eq:A522}
\end{align}
where the inequality holds by the CS inequality and the second last
equality holds by Lemmas \ref{lem:MT-Max=000020Intertemporal=000020Difference}(a)
and \ref{lem:MT-Asymptotic=000020Properties=000020of=000020the=000020Components=000020Stationary}(b).

For $A_{f523}$, we have 
\begin{align}
\abs{A_{f523}}= & \abs{\left(nh\right)^{-1}\sum_{t=T_{1}}^{T_{2}}\left[\left(\rho_{t}-\rho_{0,n}\right)c_{t-1,t-1-T_{0}}Y_{T_{0}}^{*}\right]}\nonumber \\
\leq & \left(nh\right)^{-1}O\left(h/b_{n}\right)O\left(b_{n}\right)O_{p}\left(n^{1/2}\right)=o_{p}\left(1\right),\label{eq:A523}
\end{align}
where the inequality holds by Lemmas \ref{lem:MT-Max=000020Intertemporal=000020Difference}(a)
and \ref{lem:MT-Order=000020of=000020Y_T0}(a), \eqref{eq:MT-c_tj=000020bounded-gto_1},
and \eqref{eq:MT-ps=000020rho^t}. Combining the results, we have
$A_{f5}\pto1$.

For $A_{f6}$, we use \eqref{eq:stationary-limit-distribution-1}
and obtain 
\begin{align}
 & \left(nh\right)^{-1}\sum_{t=T_{1}}^{T_{2}}\left[\left(\widehat{\rho}_{n\tau}-\rho_{0,n}\right)\left(Y_{t-1}-\ol Y_{nh,-1}\right)\right]^{2}\nonumber \\
= & \left[\left(1-\rho_{0,n}^{2}\right)^{-1}nh\left(\widehat{\rho}_{n\tau}-\rho_{0,n}\right)^{2}\right]\left[\left(1-\rho_{0,n}^{2}\right)\left(nh\right)^{-1}\sum_{t=T_{1}}^{T_{2}}\left(Y_{t-1}-\ol Y_{nh,-1}\right)^{2}\right]/\left(nh\right)\nonumber \\
= & \ O_{p}\left(1\right)O_{p}\left(1\right)/\left(nh\right)=o_{p}\left(1\right),\label{eq:A6}
\end{align}
where the second equality holds by (\ref{eq:=000020denominator=000020stat=000020norm-1})
and \eqref{eq:stationary-limit-distribution-1}.

Therefore, we have shown that the quantity in (\ref{eq:decomp}) equals
$1+o_{p}\left(1\right)$. This and (\ref{eq:=000020denominator=000020stat=000020norm-1})
establish (\ref{eq:denom=000020of=000020stat=000020case}) and the
proof of the result for the t-statistic is complete.

The subsequence versions of Lemmas \ref{lem:MT-Asymptotic=000020Properties=000020of=000020the=000020Components=000020Stationary}--\ref{lem:MT-asympdist_dist_stationary-numerator=000020}
and Theorem \ref{thm:MT-asympdist_dist_rho^hat-stationary}, see Remark
\ref{rem:MT-subseq=000020staty=000020case}, which have $\{p_{n}\}_{n\geq1}$
in place of $\{n\}_{n\geq1},$ are proved by replacing $n$ by $p_{n}$
and $h=h_{n}$ by $h_{p_{n}}$ throughout the proofs above. 
\end{proof}

\subsection{\protect\label{subsec:Proof-First_data_dep}Proof of Lemma \ref{lem:MT-First=000020Data_Dependent=000020Bandwidth=000020Lem}}
\begin{proof}[Proof of Lemma \ref{lem:MT-First=000020Data_Dependent=000020Bandwidth=000020Lem}.]
\textbf{ }By the definition of $\widehat{h}_{opt},$ $L_{n}(\widehat{h}_{opt})/L_{n}(\widehat{h})\leq1$
$\forall n\geq1.$ So, the result of the lemma follows from $L_{n}(\widehat{h}_{opt})/L_{n}(\widehat{h})\geq1+o_{p}(1).$
We have 
\begin{eqnarray}
\frac{L_{n}(\widehat{h}_{opt})}{L_{n}(\widehat{h})}\hspace{-0.08in} & = & \hspace{-0.08in}\frac{L_{n}(\widehat{h}_{opt})-2C_{n}(\widehat{h}_{opt})}{L_{n}(\widehat{h})}+2\frac{C_{n}(\widehat{h}_{opt})}{L_{n}(\widehat{h})}\nonumber \\
 & \geq & \hspace{-0.08in}\frac{L_{n}(\widehat{h})-2C_{n}(\widehat{h})}{L_{n}(\widehat{h})}+2\frac{C_{n}(\widehat{h}_{opt})}{L_{n}(\widehat{h})}\nonumber \\
 & = & \hspace{-0.08in}1-2\frac{C_{n}(\widehat{h})}{L_{n}(\widehat{h})}+2\frac{C_{n}(\widehat{h}_{opt})}{L_{n}(\widehat{h})}\nonumber \\
 & = & \hspace{-0.08in}1+o_{p}(1),\label{Bandwidth=000020Lem=000020eqn=0000201}
\end{eqnarray}
where the inequality holds because $\widehat{h}$ minimizes $L_{n}(h)-2C_{n}(h)$
over $\mathcal{H}_{n}$ and the last equality holds using Assumption
\ref{assu:MT-Asymp_h1} and $|\frac{C_{n}(\widehat{h}_{opt})}{L_{n}(\widehat{h})}|\leq|\frac{C_{n}(\widehat{h}_{opt})}{L_{n}(\widehat{h}_{opt})}|=o_{p}(1)$
since $\widehat{h}_{opt}$ minimizes $L_{n}(h)$ and using Assumption
\ref{assu:MT-Asymp_h1} again. 
\end{proof}

\subsection{\protect\label{subsec:Proof-of-Lemma_data_depend_2}Proof of Lemma
\ref{lem:MT-Second=000020Data_Dependent=000020Bandwidth=000020Lem}}
\begin{proof}[Proof of Lemma \ref{lem:MT-Second=000020Data_Dependent=000020Bandwidth=000020Lem}.]
\textbf{ }We have: $\E C_{2n}(h)=0$ because (i) $\E(U_{t}|\mathcal{G}_{t-1})=0$
a.s. by the definition of $\Lambda_{n}$ which applies by Assumption
\ref{assu:MT-Assu_h2}(a) and (ii) $Y_{t-1}$ and $\widehat{\rho}_{t-1}(h)$
are functions of $(U_{t-1},...,U_{1},Y_{0}^{\ast})$ provided $t>nh$
and these variables are in $\mathcal{G}_{t-1}$ by the definition
of $\Lambda_{n}.$ Let $n_{\ast}:=n-nh.$ Next, we have 
\begin{eqnarray}
Var\left(C_{2n}(h)\right)\hspace{-0.08in} & = & \hspace{-0.08in}\E\left(n_{\ast}^{-1}\sum_{t=nh_{\max}+1}^{n}\sigma_{t}U_{t}(\widehat{\mu}_{t-1}(h)-\mu_{t}+Y_{t-1}(\widehat{\rho}_{t-1}(h)-\rho_{t}))\right)^{2}\nonumber \\
 & = & \hspace{-0.08in}n_{\ast}^{-2}\sum_{t=nh_{\max}+1}^{n}\sigma_{t}^{2}\E U_{t}^{2}(\widehat{\mu}_{t-1}(h)-\mu_{t}+Y_{t-1}(\widehat{\rho}_{t-1}(h)-\rho_{t}))^{2}\nonumber \\
 & = & \hspace{-0.08in}n_{\ast}^{-2}\sum_{t=nh_{\max}+1}^{n}\sigma_{t}^{2}\E(\widehat{\mu}_{t-1}(h)-\mu_{t}+Y_{t-1}^{2}(\widehat{\rho}_{t-1}(h)-\rho_{t}))^{2}\nonumber \\
 & \leq & \hspace{-0.08in}C_{3U}\cdot n_{\ast}^{-1}\E L_{2n}(h),\label{Bandwidth=000020Lem=000020eqn=0000202}
\end{eqnarray}
where the first equality holds because $\E C_{2n}(h)=0,$ the second
equality holds because, if $t>s$ (and $t>nh_{\max}),$ $\E U_{t}(\widehat{\mu}_{t-1}(h)-\mu_{t}+Y_{t-1}(\widehat{\rho}_{t-1}(h)-\rho_{t}))U_{s}(\widehat{\mu}_{s-1}(h)-\mu_{s}+Y_{s-1}(\widehat{\rho}_{s-1}(h)-\rho_{s}))=0$
by (i) since $(\widehat{\mu}_{t-1}(h),Y_{t-1},\widehat{\rho}_{t-1}(h),U_{s},\widehat{\mu}_{s-1}(h),Y_{s-1},\widehat{\rho}_{s-1}(h))$
are functions of $(U_{t-1},...,U_{1},Y_{0}^{\ast})$ and these variables
are in $\mathcal{G}_{t-1},$ and analogously if $t<s,$ the third
equality holds because $\E(U_{i}^{2}|\mathcal{G}_{t-1})=1$ a.s. by
the definition of $\Lambda_{n},$ and the inequality holds by the
bound $C_{3U}$ on the variance function $\sigma^{2}(\cdot)$ by the
definition of $\Lambda_{n},$ which implies that $\sigma_{t}^{2}\leq C_{3U}$
and the definition of $L_{2n}(h).$\footnote{One could consider $\E C_{2n}\left(h\right)^{m}$ for some even number
$m>2,$ rather than the variance of $C_{2n}(h).$ With i.i.d. observations
$\{Y_{i}\}_{i\leq n},$ this would yield a bound that decreases to
zero faster as a function of $n_{\ast}$ than $n_{\ast}^{-1},$ which
appears in (\ref{Bandwidth=000020Lem=000020eqn=0000202}) for the
case of $m=2.$ However, in the present model, a faster rate is not
obtained for $m>2$ because the summands in the $m$-fold sum are
zero only when the largest index of $U_{a}U_{b}\cdots U_{t}$ is unique,
not when any index is unique, as occurs with i.i.d. summands.}

For any positive constant $K,$ 
\begin{eqnarray}
P\left(\sup_{h\in\mathcal{H}_{n}}\left\vert \frac{C_{2n}(h)}{\xi_{n}^{1/2}Var^{1/2}(C_{2n}(h))}\right\vert >K\right)\hspace{-0.08in} & \leq & \hspace{-0.08in}\sum_{h\in\mathcal{H}_{n}}P\left(\left\vert \frac{C_{2n}(h)}{\xi_{n}^{1/2}Var^{1/2}(C_{2n}(h))}\right\vert >K\right)\nonumber \\
 & \leq & \hspace{-0.08in}\sum_{h\in\mathcal{H}_{n}}\frac{\E C_{2n}(h)^{2}}{\xi_{n}Var(C_{2n}(h))K^{2}}=\frac{1}{K^{2}}\label{Bandwidth=000020Lem=000020eqn=0000203}
\end{eqnarray}
for all $n\geq1,$ where the second inequality holds by Markov's inequality.
In consequence, 
\begin{equation}
O_{p}(1)=\sup_{h\in\mathcal{H}_{n}}\left\vert \frac{C_{2n}(h)}{\xi_{n}^{1/2}Var^{1/2}(C_{2n}(h))}\right\vert \geq\sup_{h\in\mathcal{H}_{n}}\left\vert \frac{n_{\ast}^{1/2}C_{2n}(h)}{C_{3U}^{1/2}\xi_{n}^{1/2}(\E L_{2n}(h))^{1/2}}\right\vert ,\label{Bandwidth=000020Lem=000020eqn=0000204}
\end{equation}
where the equality holds by \eqref{Bandwidth=000020Lem=000020eqn=0000203}
and the inequality holds by \eqref{Bandwidth=000020Lem=000020eqn=0000204}.We
have 
\begin{eqnarray}
\sup_{h\in\mathcal{H}_{n}}\left\vert \frac{C_{2n}(h)}{L_{2n}(h)}\right\vert \hspace{-0.08in} & = & \hspace{-0.08in}\sup_{h\in\mathcal{H}_{n}}\left\vert \frac{n_{\ast}^{1/2}C_{2n}(h)}{C_{3U}^{1/2}\xi_{n}^{1/2}(\E L_{2n}(h))^{1/2}}\frac{C_{3U}^{1/2}\xi_{n}^{1/2}(\E L_{2n}(h))^{1/2}}{n_{\ast}^{1/2}L_{2n}(h)}\right\vert \nonumber \\
 & = & \hspace{-0.08in}O_{p}(1)\cdot\sup_{h\in\mathcal{H}_{n}}\frac{\xi_{n}^{1/2}(\E L_{2n}(h))^{1/2}}{n_{\ast}^{1/2}L_{2n}(h)}.\label{Bandwidth=000020Lem=000020eqn=0000205}
\end{eqnarray}
The right-hand side expression is $o_{p}(1)$ iff 
\begin{equation}
\inf_{h\in\mathcal{H}_{n}}\frac{n_{\ast}^{1/2}L_{2n}(h)}{\xi_{n}^{1/2}(R_{2n}(h))^{1/2}}\rightarrow_{p}\infty\text{ iff }\inf_{h\in\mathcal{H}_{n}}\frac{n_{\ast}^{1/2}(R_{2n}(h))^{1/2}}{\xi_{n}^{1/2}}\rightarrow\infty\text{ iff }\frac{n\inf_{h\in\mathcal{H}_{n}}R_{2n}(h)}{\xi_{n}}\rightarrow\infty,\label{Bandwidth=000020Lem=000020eqn=0000206}
\end{equation}
where the first ``iff'' uses $R_{2n}(h)=\E L_{2n}(h),$ the second
``iff'' uses Assumption \ref{assu:MT-Assu_h2}(c), the third ``iff''
uses Assumption \ref{assu:MT-Assu_h2}(d), and the last condition
holds by Assumption \ref{assu:MT-Assu_h2}(e). In consequence, $\sup_{h\in\mathcal{H}_{n}}\left\vert \frac{C_{2n}(h)}{L_{2n}(h)}\right\vert =o_{p}(1).$
This, combined with Assumption \ref{assu:MT-Assu_h2}(b) and $L_{2n}(h)\geq0$
$\forall h,$ verifies Assumption \ref{assu:MT-Asymp_h1}. 
\end{proof}

\subsection{\protect\label{subsec:Pf=000020NEW=000020Thm=000020A.1}Proof of
Theorem \ref{thm:ThmA.1_As2star}}

To show that Theorem \ref{thm:MT-asympdist_dist_rho^hat-stationary}
holds with Assumption \ref{assu:MT-stationary-nh^5=000020->=0000200}
replaced by Assumption \ref{assu:assumption2*}, the proof of Theorem
\ref{thm:ThmA.1_As2star} uses the following extension of Lemma \ref{lem:MT-Max=000020Intertemporal=000020Difference}
that improves its bounds in the case where $\ell_{n}\rightarrow0$
as $n\rightarrow\infty.$
\begin{lem}
\noindent\label{lem:A.3_New_7.1}Under Assumptions \ref{assu:MT-h}
and \ref{assu:MT-kappa_n_sigma_n=000020mu_n=000020convergence},
\begin{enumerate}
\item $\max_{t\in\lbrack T_{1},T_{2}]}|\rho_{t}-\rho_{n\tau}|=O(\ell_{n}h/b_{n}),$
\item $\max_{t\in\lbrack T_{1},T_{2}]}|\sigma_{t}^{2}-\sigma_{n\tau}^{2}|=O(\ell_{n}h),$
\item $\max_{t\in\lbrack T_{1},T_{2}]}|c_{t,j}-\rho_{n\tau}^{j}|=O(\ell_{n}nh^{2}/b_{n}),$
and
\item $\max_{t\in\lbrack T_{1},T_{2}]}|\mu_{t}-\mu_{n\tau}|=O(\ell_{n}h/b_{n}).$
\end{enumerate}
\end{lem}
\begin{proof}[\textbf{Proof of Lemma} \ref{lem:A.3_New_7.1}]
\noindent\textbf{ }Part (a) holds by the proof of Lemma \ref{lem:MT-Max=000020Intertemporal=000020Difference}(a)
by replacing the Lipschitz bound $L_{4}$ in (\ref{eq:kappanapprox_a.5.8})
by $\ell_{n},$ which implies that the rhs bound in (\ref{eq:kappanapprox_a.5.8})
becomes $O(\ell_{n}h).$ In turn, the rhs bound in (\ref{eq:rho_t-rho_ntau_A.5.9})
becomes $O(\ell_{n}h/b_{n}).$ Parts (b)-(d) then hold by the same
argument as in the proof of Lemma \ref{lem:MT-Max=000020Intertemporal=000020Difference}
with the additional term $\ell_{n}$ appearing in each of the error
bounds. 
\end{proof}
\begin{proof}[\textbf{Proof of Theorem} \ref{thm:ThmA.1_As2star}]
\textbf{} First, we show that Theorem \ref{thm:MT-asym_rho_hat_local=000020to=000020unity=000020case}
holds with Assumption \ref{assu:MT-stationary-nh^5=000020->=0000200}
replaced by Assumption \ref{assu:assumption2*}. Assumption \ref{assu:MT-stationary-nh^5=000020->=0000200}
enters the proof of Theorem \ref{thm:MT-asym_rho_hat_local=000020to=000020unity=000020case}
only through its application of Lemma \ref{lem:MT-Order=000020of=000020Y_T0}(b),
which relies on Assumption \ref{assu:MT-stationary-nh^5=000020->=0000200}.
In turn, Assumption \ref{assu:MT-stationary-nh^5=000020->=0000200}
enters the proof of Lemma \ref{lem:MT-Order=000020of=000020Y_T0}(b)
only through its use in (\ref{Pf=000020Lem=0000207.2(b)=000020eqn=00002011})
to show that $h^{1/2}\ln(n)=o(1).$ Since the latter holds under Assumption
\ref{assu:assumption2*}(ii), this completes the proof for Theorem
\ref{thm:MT-asym_rho_hat_local=000020to=000020unity=000020case} under
Assumption \ref{assu:assumption2*}.

Next, we show that Theorem \ref{thm:MT-asympdist_dist_rho^hat-stationary}
holds with Assumption \ref{assu:MT-stationary-nh^5=000020->=0000200}
replaced by Assumption \ref{assu:assumption2*}. Assumption \ref{assu:MT-stationary-nh^5=000020->=0000200}
enters the proof of Theorem \ref{thm:MT-asympdist_dist_rho^hat-stationary}
only through its application of Lemmas \ref{lem:MT-asympdist_dist_stationary-denom}
and \ref{lem:MT-asympdist_dist_stationary-numerator=000020}(a), which
both use Assumption \ref{assu:MT-stationary-nh^5=000020->=0000200}.
Assumption \ref{assu:MT-stationary-nh^5=000020->=0000200} enters
the proof of Lemma \ref{lem:MT-asympdist_dist_stationary-denom} only
through its application of Lemma \ref{lem:MT-Order=000020of=000020Y_T0}(c),
which relies on Assumption \ref{assu:MT-stationary-nh^5=000020->=0000200}.
In turn, Assumption \ref{assu:MT-stationary-nh^5=000020->=0000200}
enters the proof of Lemma \ref{lem:MT-Order=000020of=000020Y_T0}(c)
only through its use in (\ref{Pf=000020Lem=0000207.2(c)=000020eqn=00002015})
to show that $h\ln(n)=o(1).$ The latter holds under Assumption \ref{assu:assumption2*}(ii).

Assumption \ref{assu:MT-stationary-nh^5=000020->=0000200} is used
in the proof of Lemma \ref{lem:MT-asympdist_dist_stationary-numerator=000020}(a)
in equations (\ref{Pf=000020re=000020A_c1,1=000020eqn=0000204}),
(\ref{Pf=000020re=000020A_c1,2=000020eqn=0000209}), and (\ref{eq:A_c8=000020pto=0000200})
and because the proof applies Lemma \ref{lem:MT-Order=000020of=000020Y_T0}(c)
(which we have just shown to hold under Assumption \ref{assu:assumption2*}(ii)).
To verify (\ref{Pf=000020re=000020A_c1,1=000020eqn=0000204}) using
Assumption \ref{assu:assumption2*} in place of Assumption \ref{assu:MT-stationary-nh^5=000020->=0000200},
we bound $|\kappa_{0,n}^{\prime\prime}(\widetilde{\tau}_{n,t})|$
by $\ell_{n}$ in the fourth last line of (\ref{Pf=000020re=000020A_c1,1=000020eqn=0000204})
(which is why $\ell_{n}$ is defined to bound the absolute value of
second derivative of $\kappa_{n}(\cdot)=\kappa_{0,n}(\cdot)$ over
the interval $I_{\tau,\varepsilon_{2}}).$ This leads to $\ell_{n}\sum_{t=T_{1}}^{T_{2}}(t/n-\tau)^{2}$
appearing in place of $\sum_{t=T_{1}}^{T_{2}}(t/n-\tau)^{2}$ in the
third last line of (\ref{Pf=000020re=000020A_c1,1=000020eqn=0000204}),
which in turn leads to $\ell_{n}nh\cdot h^{2}$ appearing in place
of $nh\cdot h^{2}$ in the second last line of (\ref{Pf=000020re=000020A_c1,1=000020eqn=0000204}).
The latter leads to $(n\ell_{n}^{2}h^{5})^{1/2}$ appearing in place
of $(nh^{5})^{1/2}$ in two places in the last line of (\ref{Pf=000020re=000020A_c1,1=000020eqn=0000204}).
Since $(n\ell_{n}^{2}h^{5})^{1/2}=o(1)$ by Assumption \ref{assu:assumption2*}(i),
(\ref{Pf=000020re=000020A_c1,1=000020eqn=0000204}) holds under Assumption
\ref{assu:assumption2*}.

To verify (\ref{Pf=000020re=000020A_c1,2=000020eqn=0000209}) under
Assumption \ref{assu:assumption2*}, we employ Lemma \ref{lem:A.3_New_7.1}(b)
and (c) to yield the rhs of (\ref{Pf=000020re=000020A_c1,2=000020eqn=0000205})
to be $j\overline{\rho}_{n}^{j-1}O(\ell_{n}h/b_{n})+\overline{\rho}_{n}^{j}O(\ell_{n}h)$
rather than $j\overline{\rho}_{n}^{j-1}O(h/b_{n})+\overline{\rho}_{n}^{j}O(h)$
(which is why $\ell_{n}$ is defined to bound the Lipschitz constants
for $\kappa_{n}(\cdot)$ and $\sigma_{n}^{2}(\cdot)$ over the interval
$I_{\tau,\varepsilon_{2}}).$ In consequence, each of the terms on
the last line of (\ref{Pf=000020re=000020A_c1,2=000020eqn=0000206})
gets multiplied by $\ell_{n}^{2},$ and so, the rhs of (\ref{Pf=000020re=000020A_c1,2=000020eqn=0000206})
becomes $O(\ell_{n}^{2}b_{n}h^{2}).$ In turn, this causes $O(\ell_{n}b_{n}^{1/2}h)$
to appear in place of $O(b_{n}^{1/2}h)$ at the end of the first line
of (\ref{Pf=000020re=000020A_c1,2=000020eqn=0000209}). And this causes
$O(n^{1/2}\ell_{n}h^{5/2}b_{n}^{-1/2})$ and $O((n\ell_{n}^{2}h^{5})^{1/2})$
to appear in place of $O(n^{1/2}h^{5/2}b_{n}^{-1/2})$ and $O((nh^{5})^{1/2}),$
respectively, in the last line of (\ref{Pf=000020re=000020A_c1,2=000020eqn=0000209}).
Since $O((n\ell_{n}^{2}h^{5})^{1/2})=o(1)$ under Assumption \ref{assu:assumption2*}(i),
(\ref{Pf=000020re=000020A_c1,2=000020eqn=0000209}) holds under Assumption
\ref{assu:assumption2*}.

To verify (\ref{eq:A_c8=000020pto=0000200}) under Assumption \ref{assu:assumption2*},
we employ Lemma \ref{lem:A.3_New_7.1}(d) to yield the bound $O((n\ell_{n}^{2}h^{5}/b_{n}^{5})^{1/2})$
rather than the bound $O((nh^{5}/b_{n}^{5})^{1/2})$ in (\ref{eq:A_c8=000020pto=0000200}).
Since $O((n\ell_{n}^{2}h^{5}/b_{n}^{5})^{1/2})=o(1)$ under Assumption
\ref{assu:assumption2*}(i) (which is why $\ell_{n}$ is defined to
bound the Lipschitz constant for $\eta_{n}(\cdot)$ over the interval
$I_{\tau,\varepsilon_{2}}),$ (\ref{eq:A_c8=000020pto=0000200}) holds
under Assumption \ref{assu:assumption2*}. This completes the proof.
\end{proof}

\section{Additional Simulation Results\protect\label{sec:Additional-Simulation-Results}}

For a discussion of the results given in Figures \ref{fig:Sim_CP_AL_MAD_SM1}--\ref{fig:Sim_CP_AL_MAD_SM4}
below, see Section \ref{sec:MT-Monte-Carlo-Simulations} of the paper. 
\begin{center}
\global\long\def\thefigure{SM.\arabic{figure}}%
 \captionsetup[subfigure]{position=top,font=scriptsize,singlelinecheck=off,justification=raggedright}
\setcounter{figure}{0} 
\begin{figure}[H]
\vspace{-2.5em}

\begin{centering}
\subfloat[sin 1.00-0.90-1.00, time-varying $\mu$ and $\sigma$]{\includegraphics[scale=0.36]{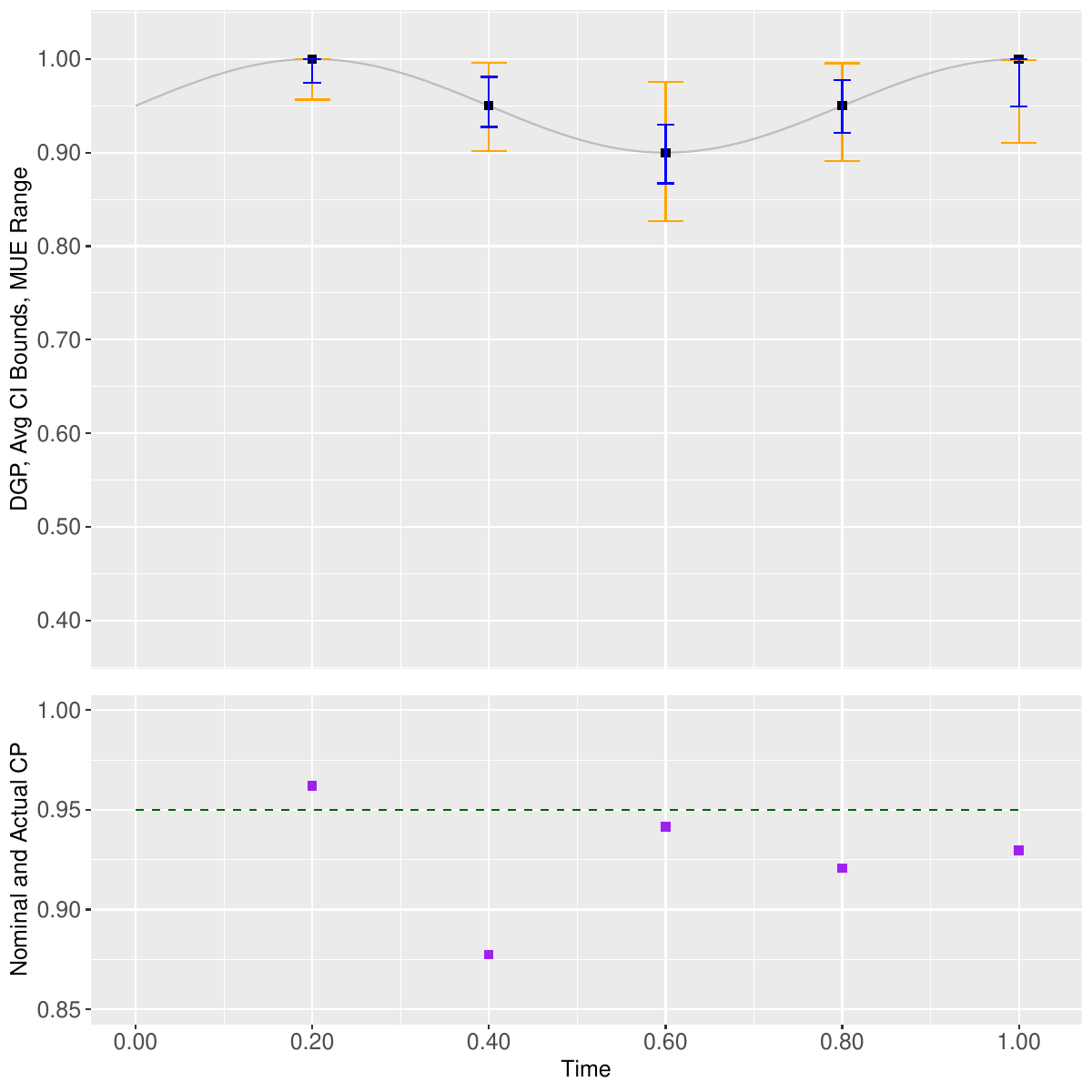}

}\quad{}\subfloat[sin 0.90-1.00-1.00, time-varying $\mu$ and $\sigma$]{\includegraphics[scale=0.36]{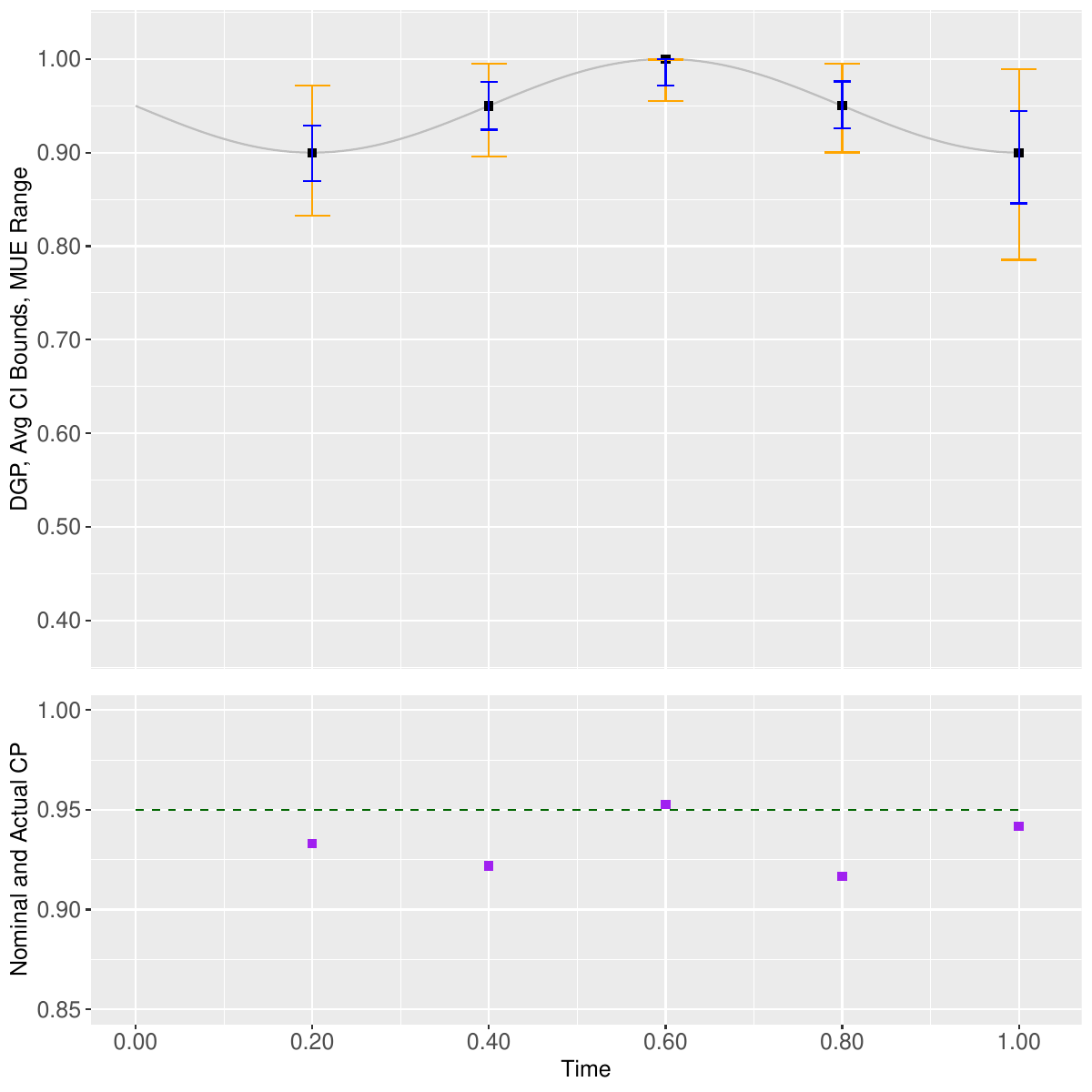}

}
\par\end{centering}
\vspace{-0.65em}

\begin{centering}
\subfloat[sin 1.00-0.80-1.00, time-varying $\mu$ and $\sigma$]{\includegraphics[scale=0.36]{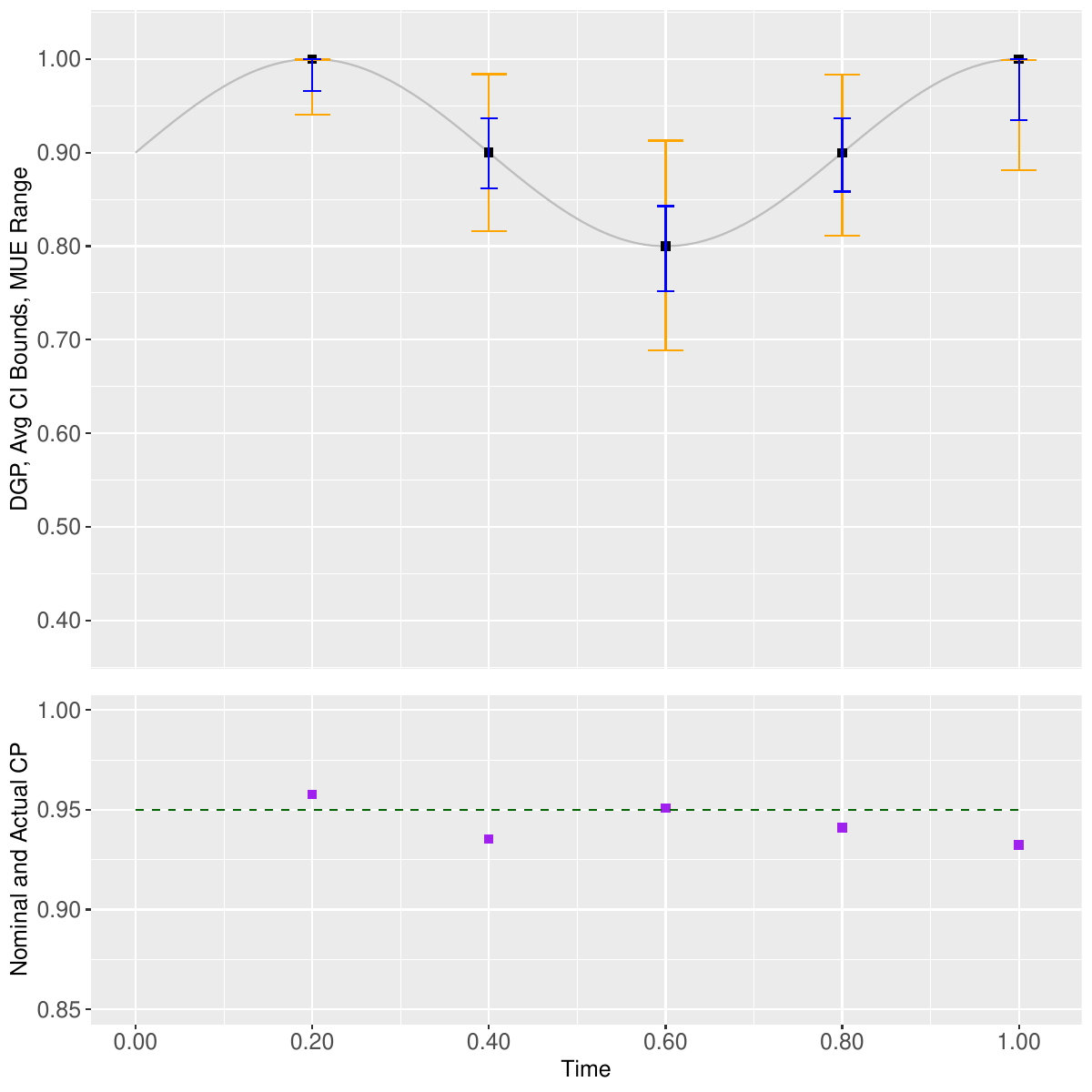}

}\quad{}\subfloat[sin 0.80-1.00-0.80, time-varying $\mu$ and $\sigma$]{\includegraphics[scale=0.36]{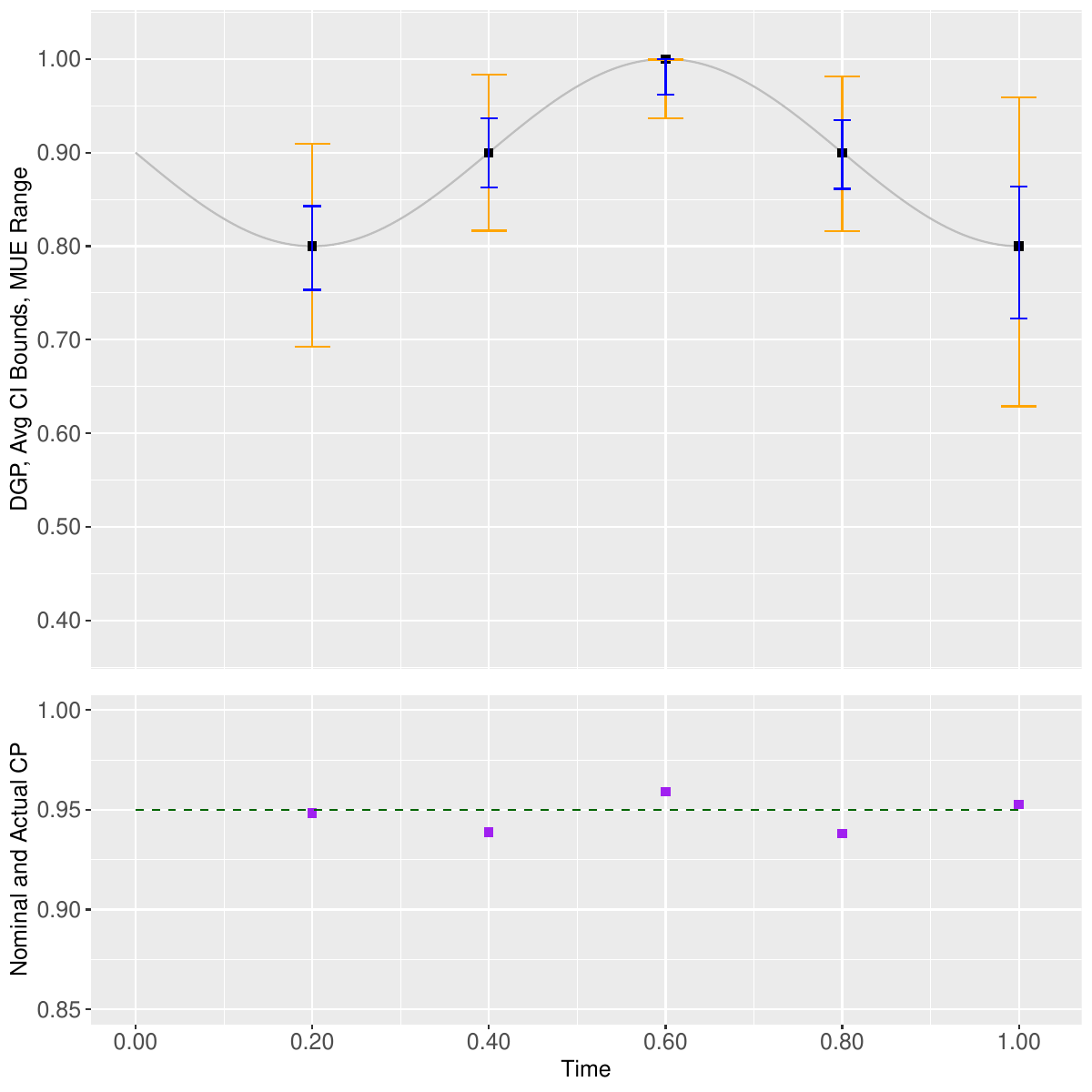}

}
\par\end{centering}
\vspace{-0.65em}

\begin{centering}
\subfloat[sin 1.00-0.60-1.00, time-varying $\mu$ and $\sigma$]{\includegraphics[scale=0.36]{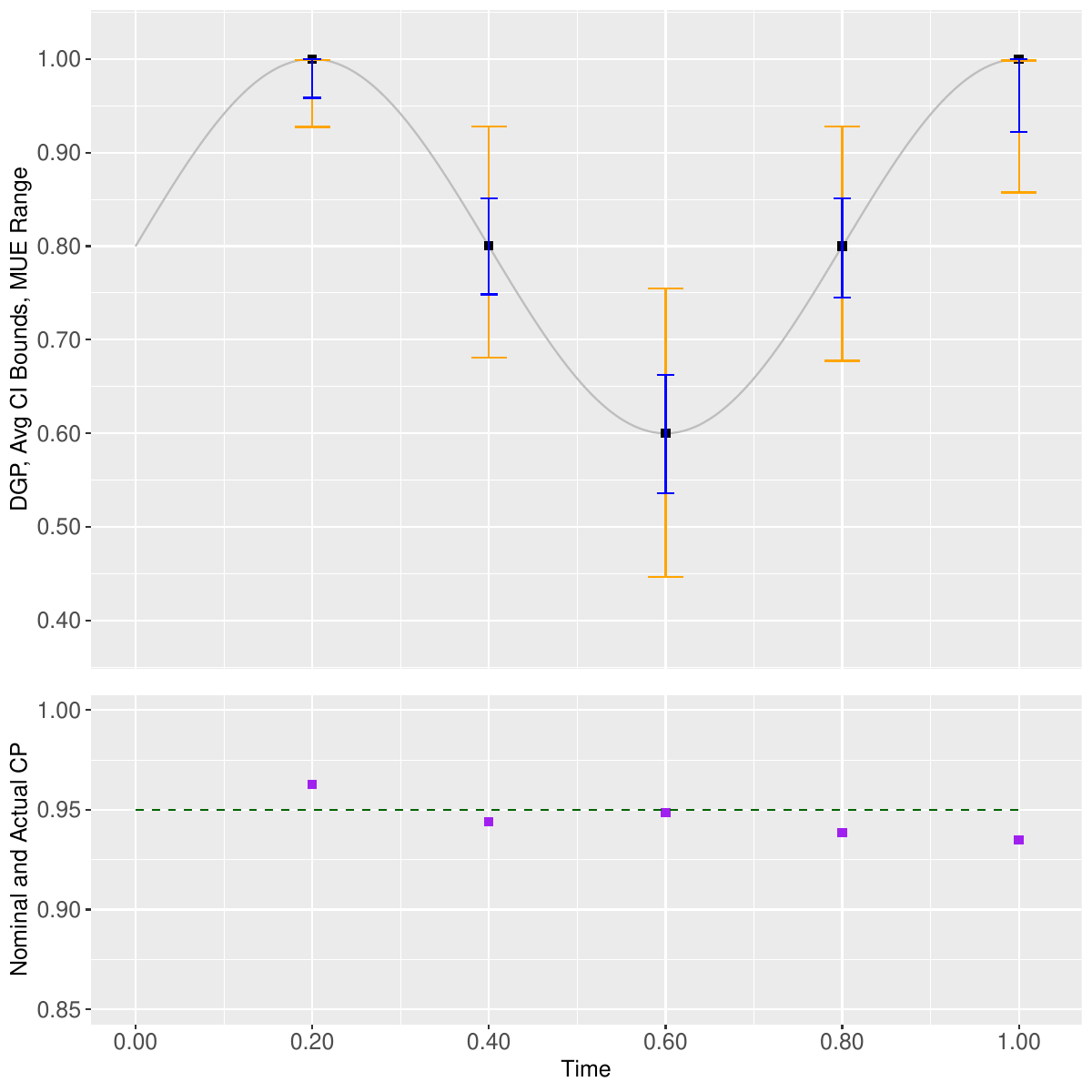}

}\quad{}\subfloat[sin 0.60-1.00-0.60, time-varying $\mu$ and $\sigma$]{\includegraphics[scale=0.36]{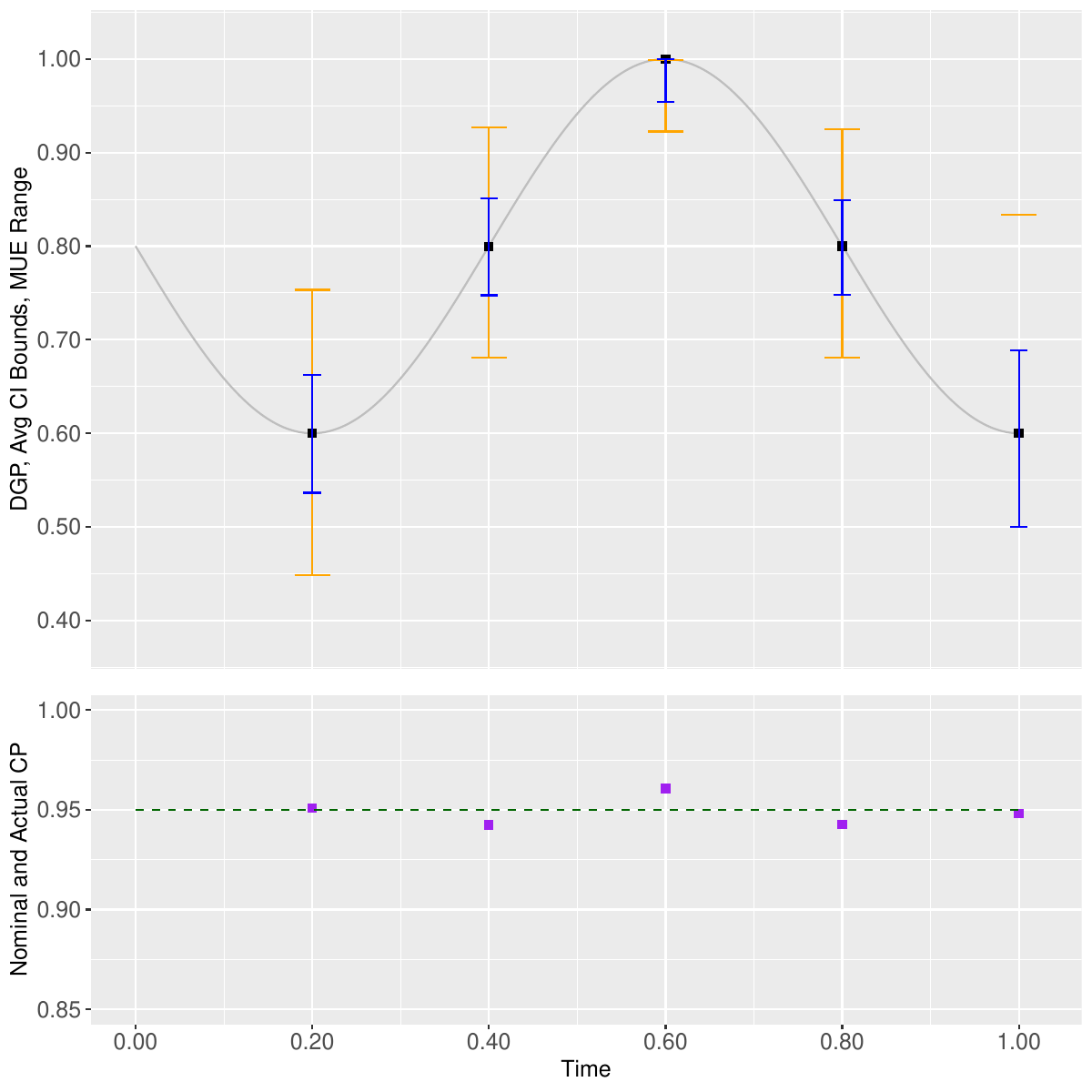}

}
\par\end{centering}
\vspace{-0.75em}

\caption{\protect\label{fig:Sim_CP_AL_MAD_SM1}CP's and AL's of CI's for $\rho\left(\tau\right)$
and MAD's of the MUE of $\rho\left(\tau\right)$}
\end{figure}
\begin{figure}[H]
\vspace{-2.5em}

\begin{centering}
\subfloat[linear 0.90-1.00, time-varying $\mu$ and $\sigma$]{\includegraphics[scale=0.36]{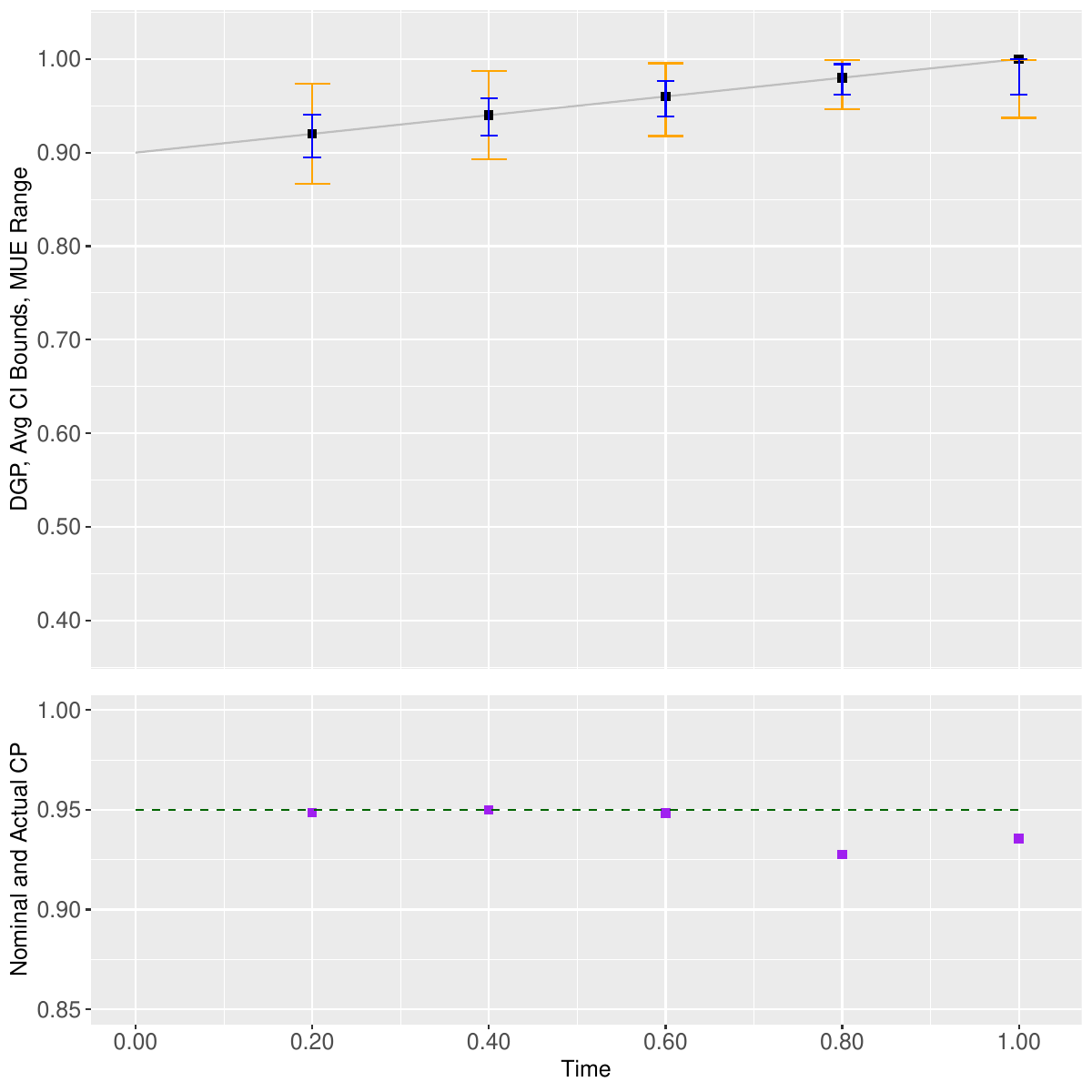}

}\quad{}\subfloat[linear 1.00-0.90, time-varying $\mu$ and $\sigma$]{\includegraphics[scale=0.36]{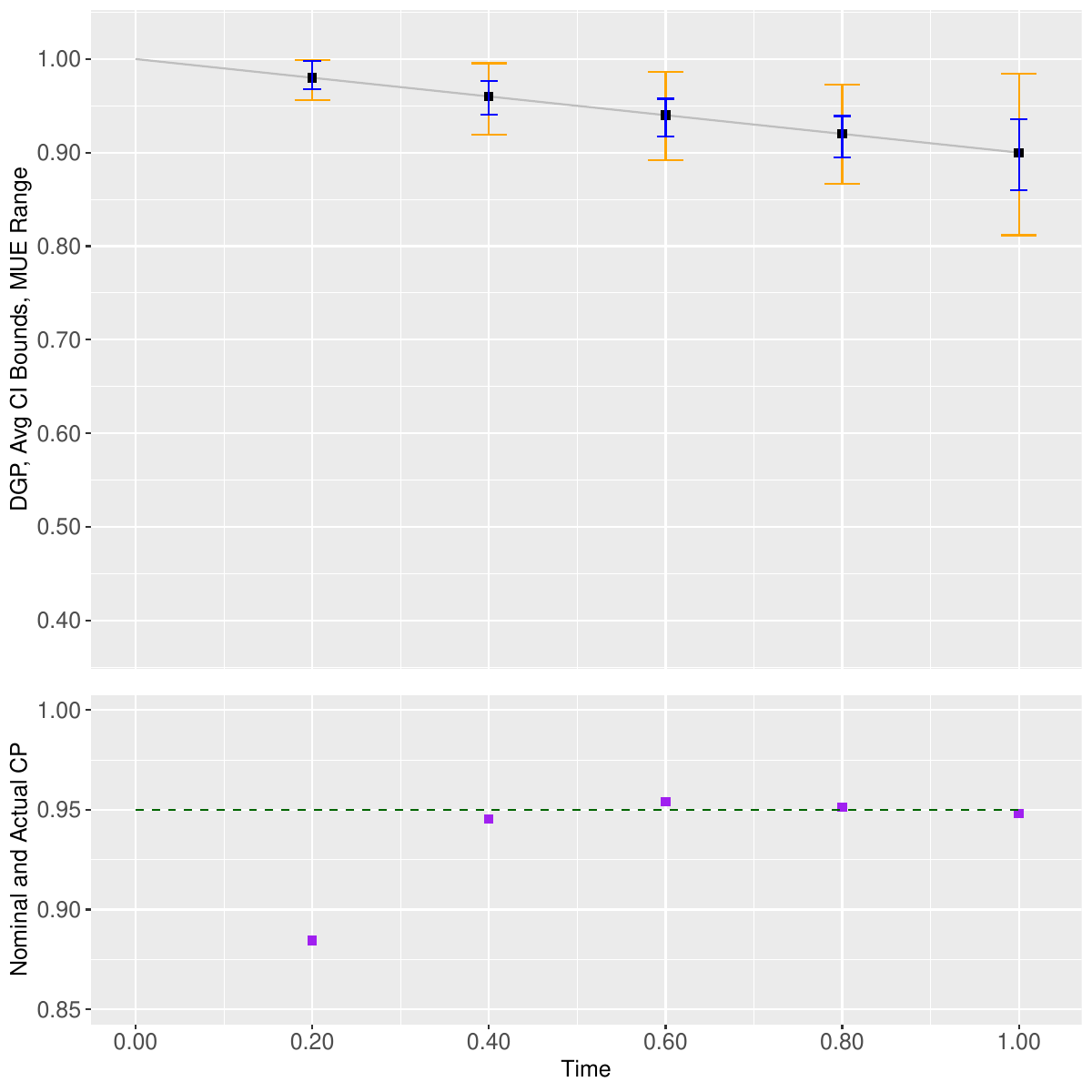}

}
\par\end{centering}
\vspace{-0.65em}

\begin{centering}
\subfloat[linear 0.60-0.90, time-varying $\mu$ and $\sigma$]{\includegraphics[scale=0.36]{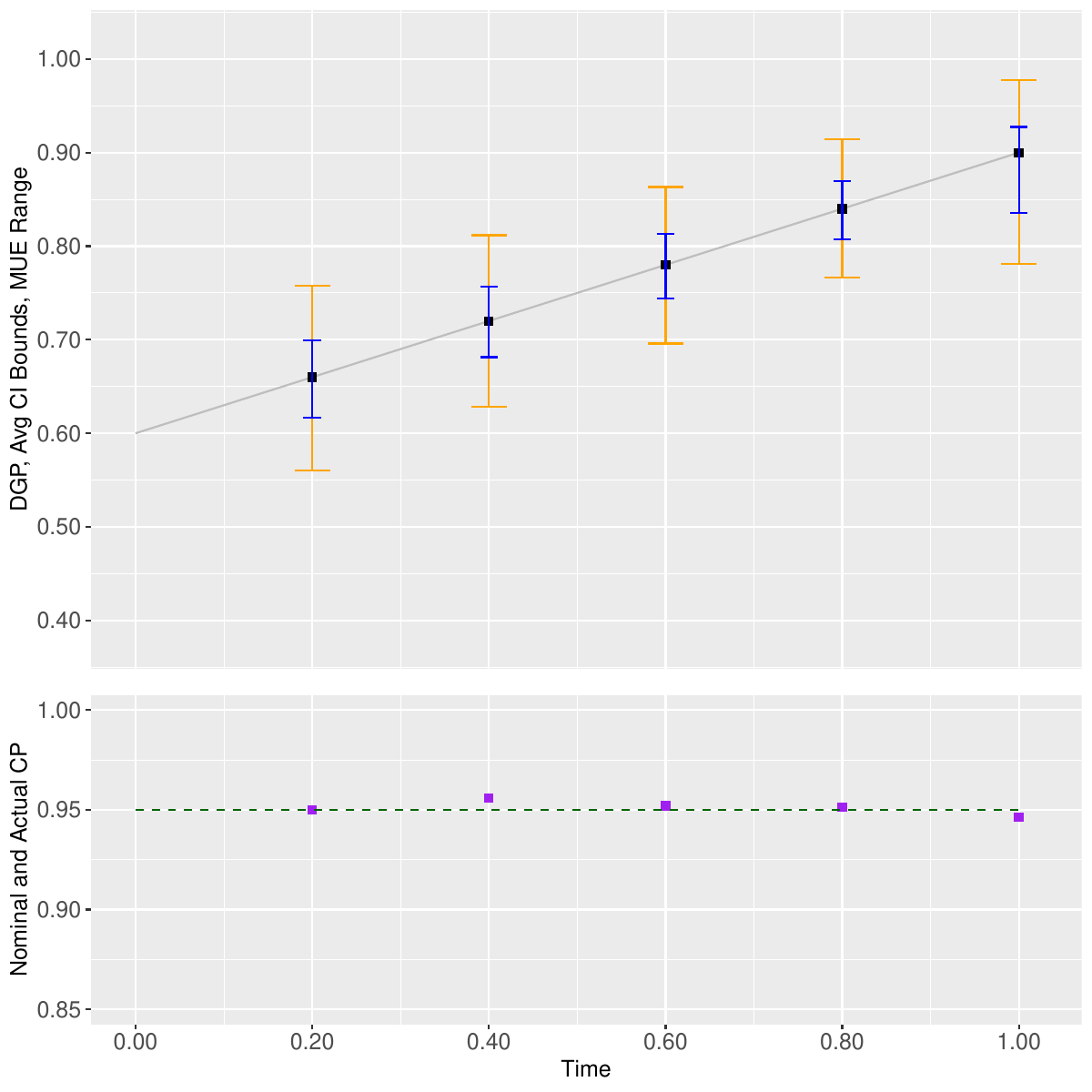}

}\quad{}\subfloat[linear 0.90-0.60, time-varying $\mu$ and $\sigma$]{\includegraphics[scale=0.36]{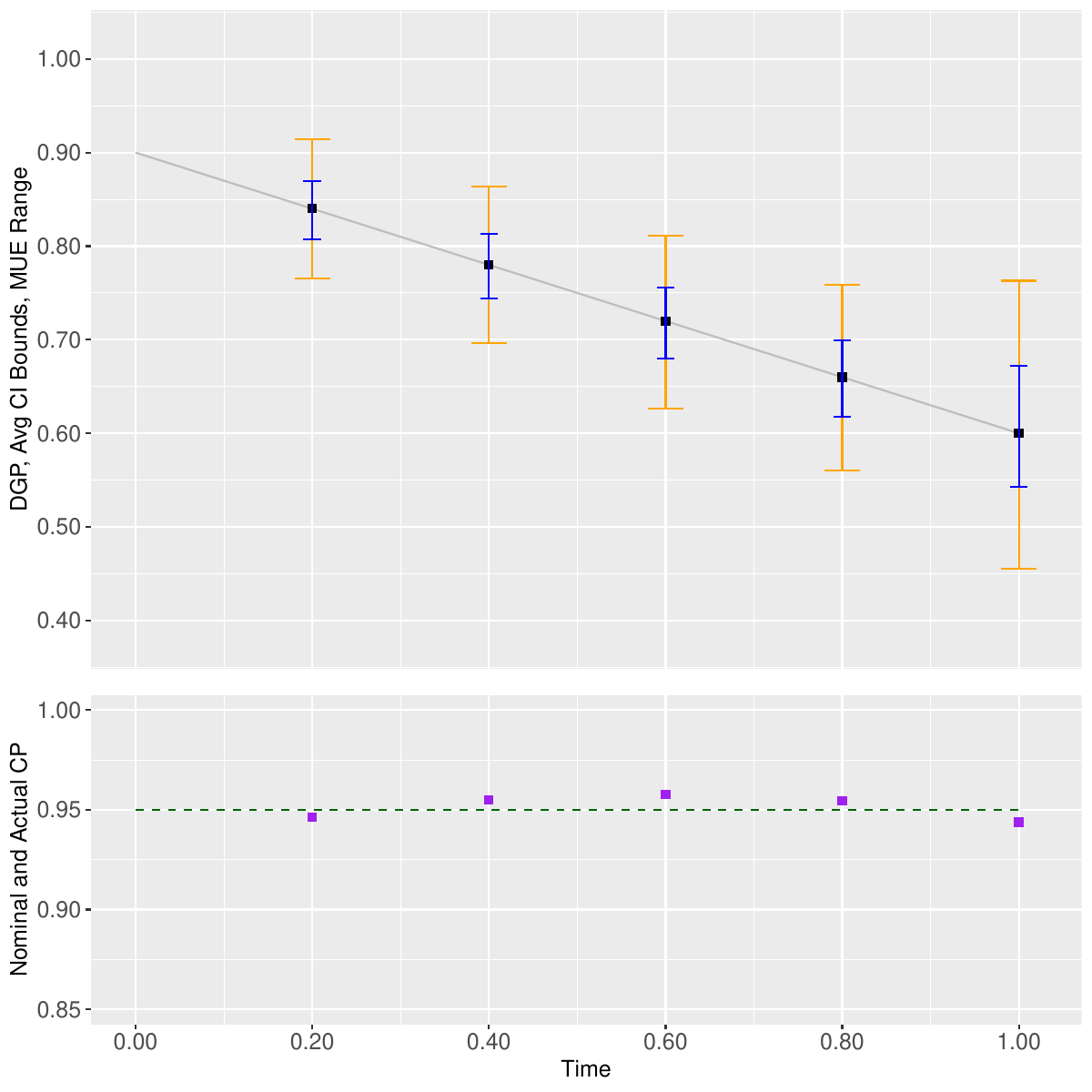}

}
\par\end{centering}
\vspace{-0.65em}

\begin{centering}
\subfloat[flat-lin 0.99-0.90, constant $\mu$ and $\sigma$]{\includegraphics[scale=0.36]{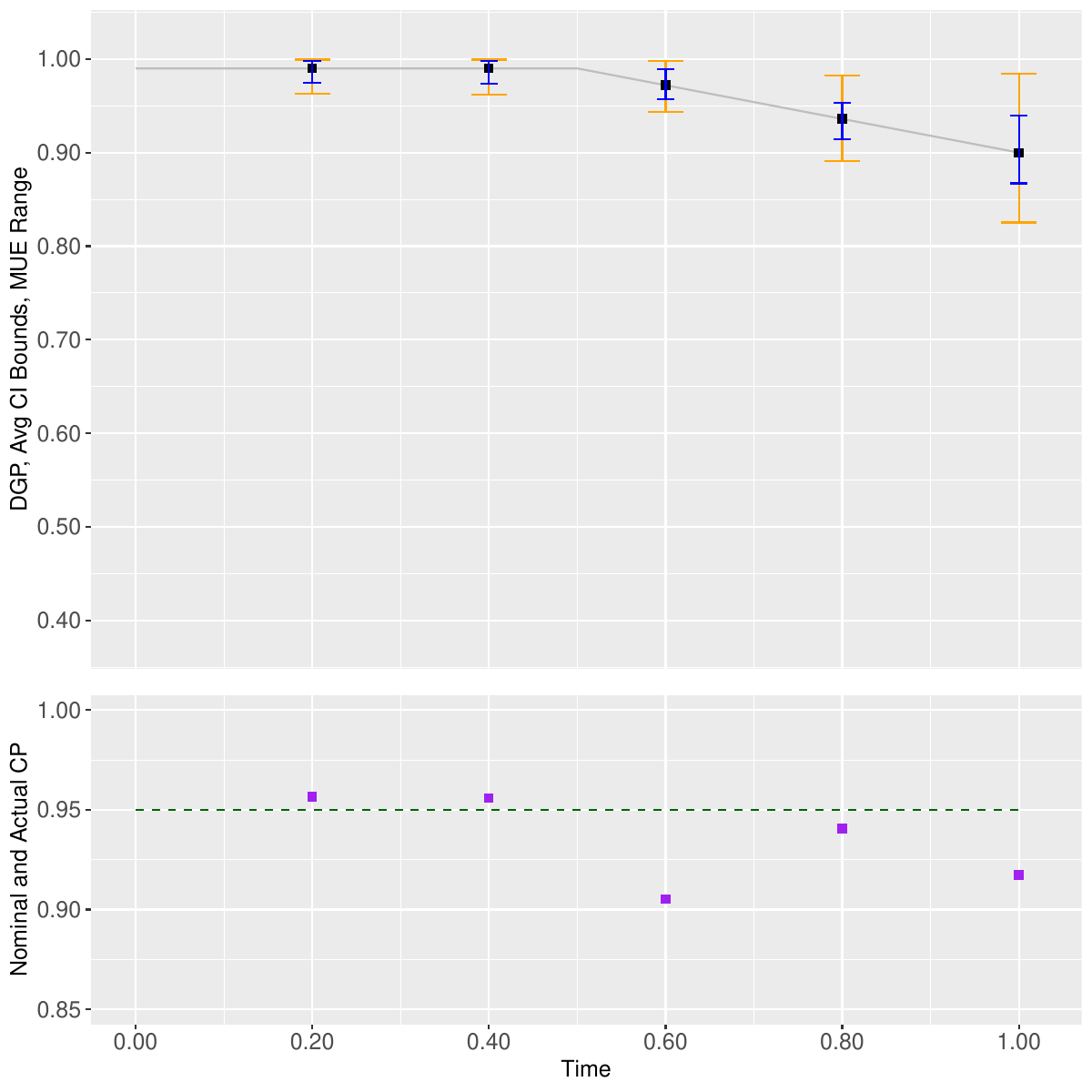}

}\quad{}\subfloat[flat-lin 0.99-0.90, time-varying $\mu$ and $\sigma$]{\includegraphics[scale=0.36]{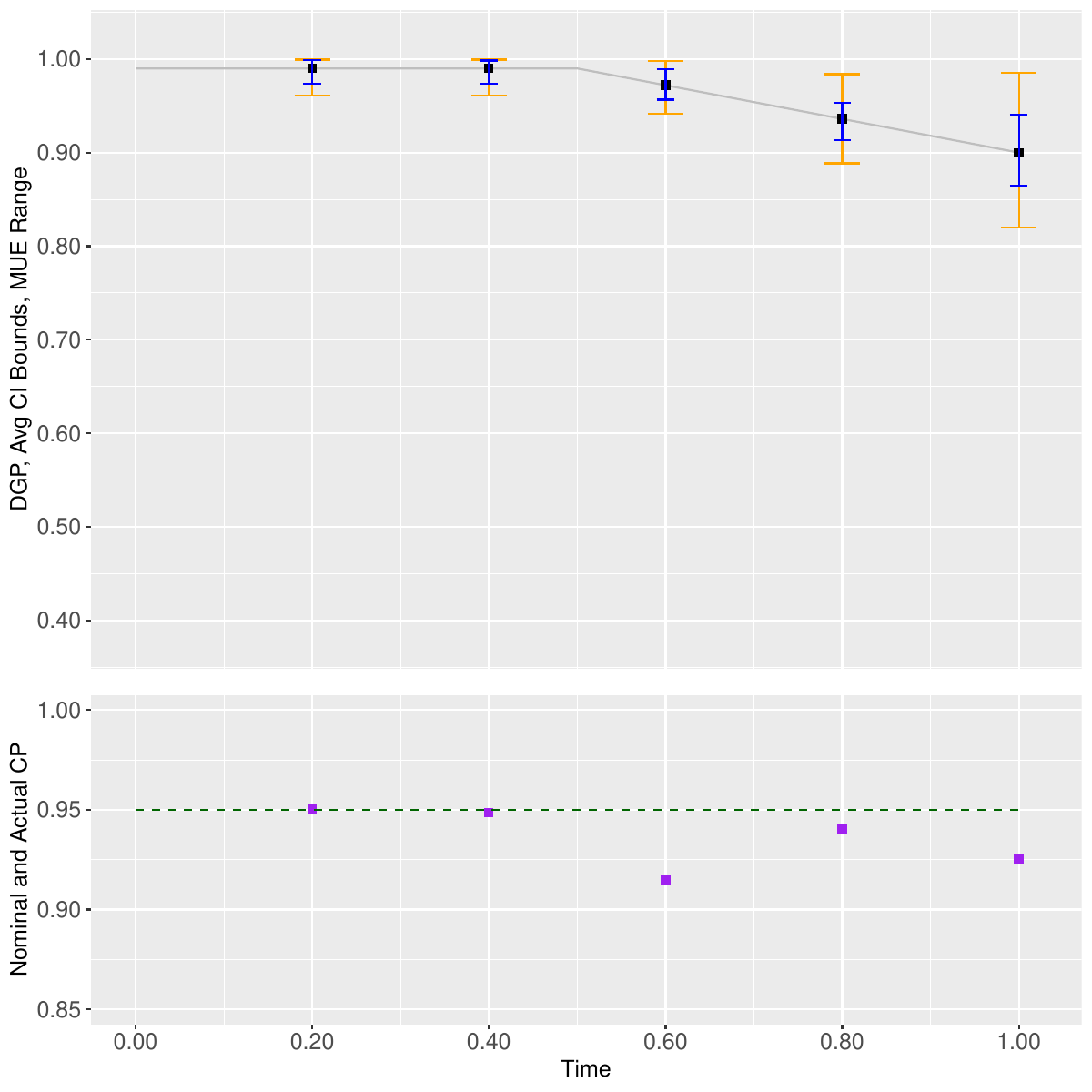}

}
\par\end{centering}
\vspace{-0.75em}

\caption{\protect\label{fig:Sim_CP_AL_MAD_SM2}CP's and AL's of CI's for $\rho\left(\tau\right)$
and MAD's of the MUE of $\rho\left(\tau\right)$}
\end{figure}
\begin{figure}[H]
\vspace{-2.5em}

\begin{centering}
\subfloat[flat-lin 0.99-0.80, constant $\mu$ and $\sigma$]{\includegraphics[scale=0.36]{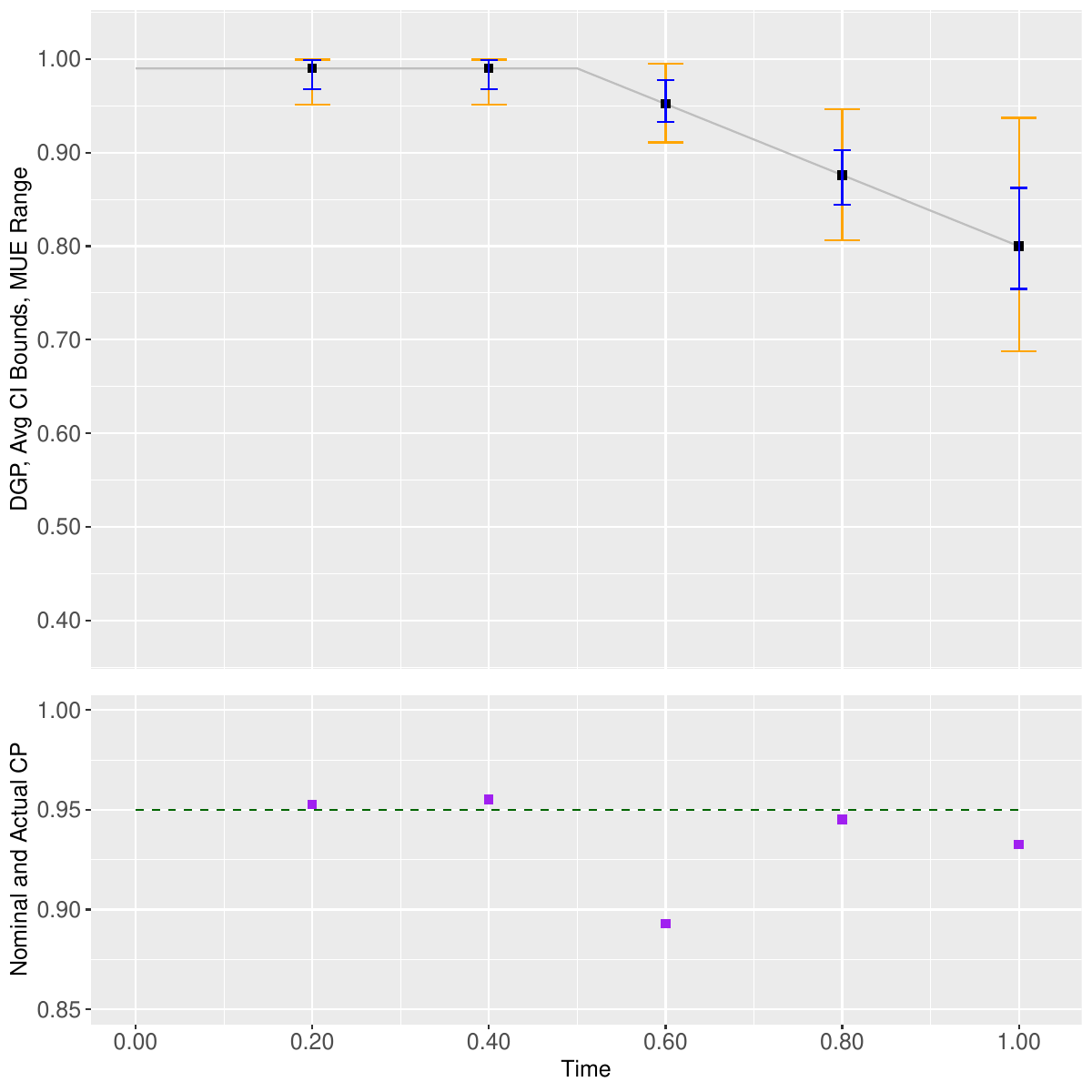}

}\quad{}\subfloat[flat-lin 0.99-0.80, time-varying $\mu$ and $\sigma$]{\includegraphics[scale=0.36]{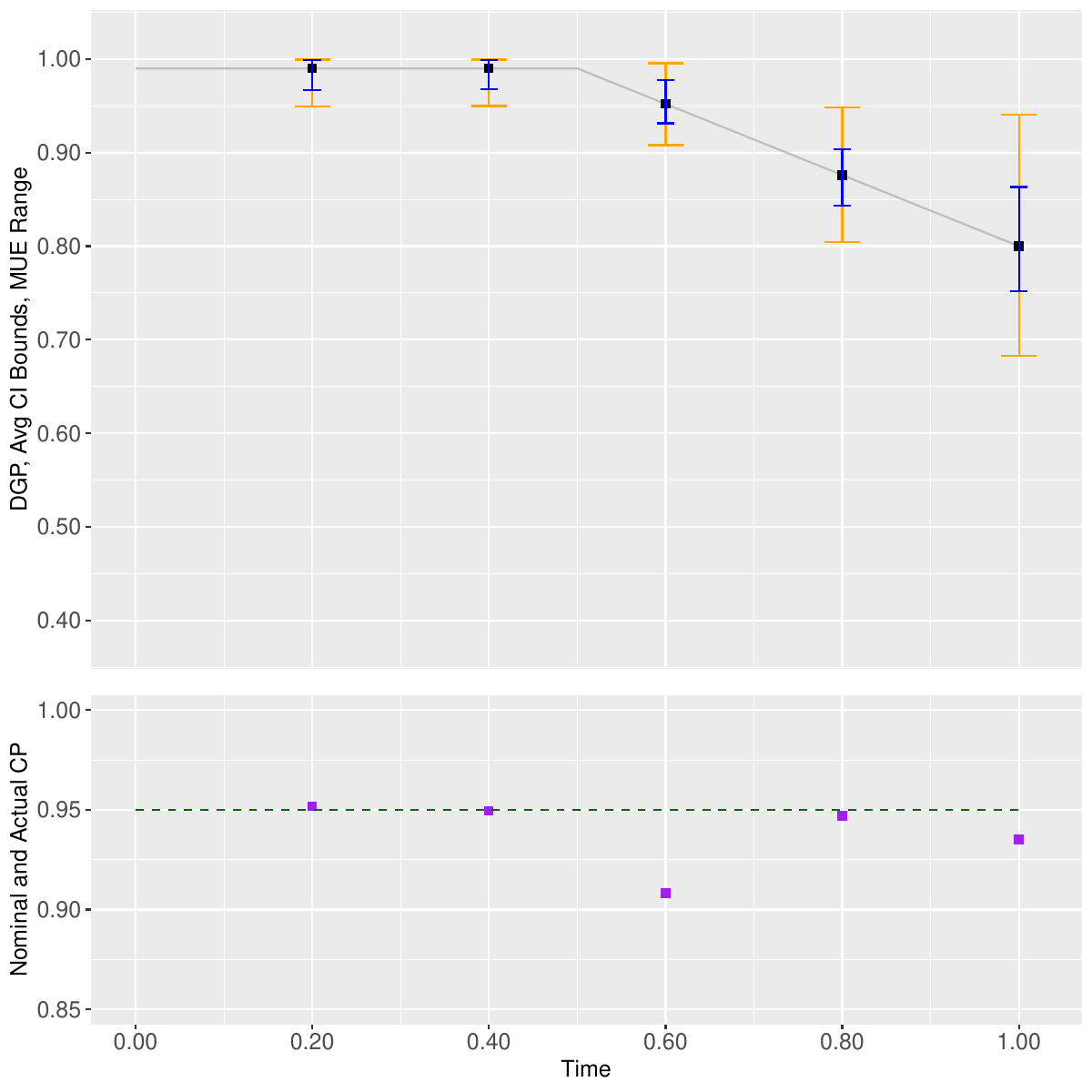}

}
\par\end{centering}
\vspace{-0.65em}

\begin{centering}
\subfloat[flat 0.75, constant $\mu$ and $\sigma$]{\includegraphics[scale=0.36]{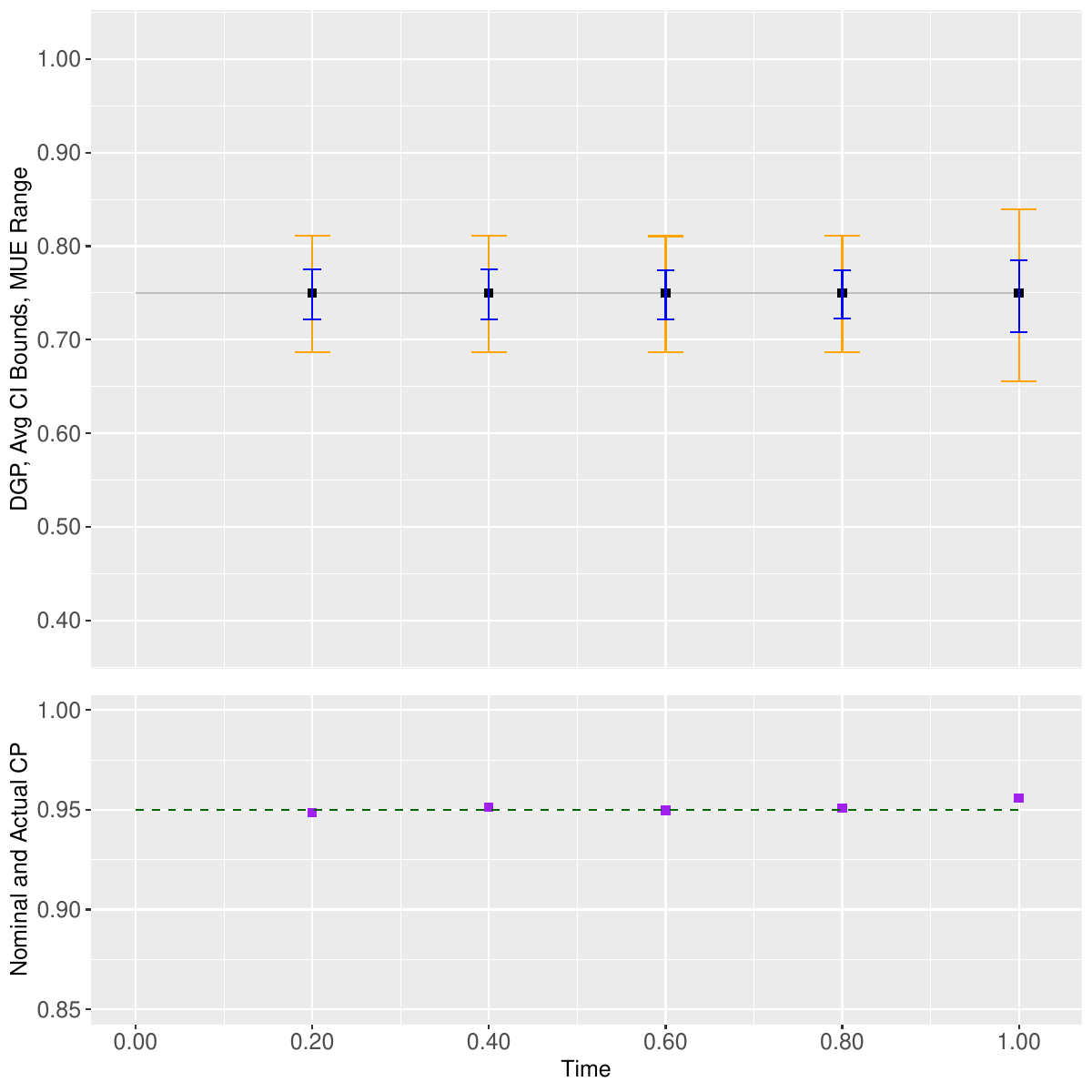}

}\quad{}\subfloat[flat 0.75, time-varying $\mu$ and $\sigma$]{\includegraphics[scale=0.36]{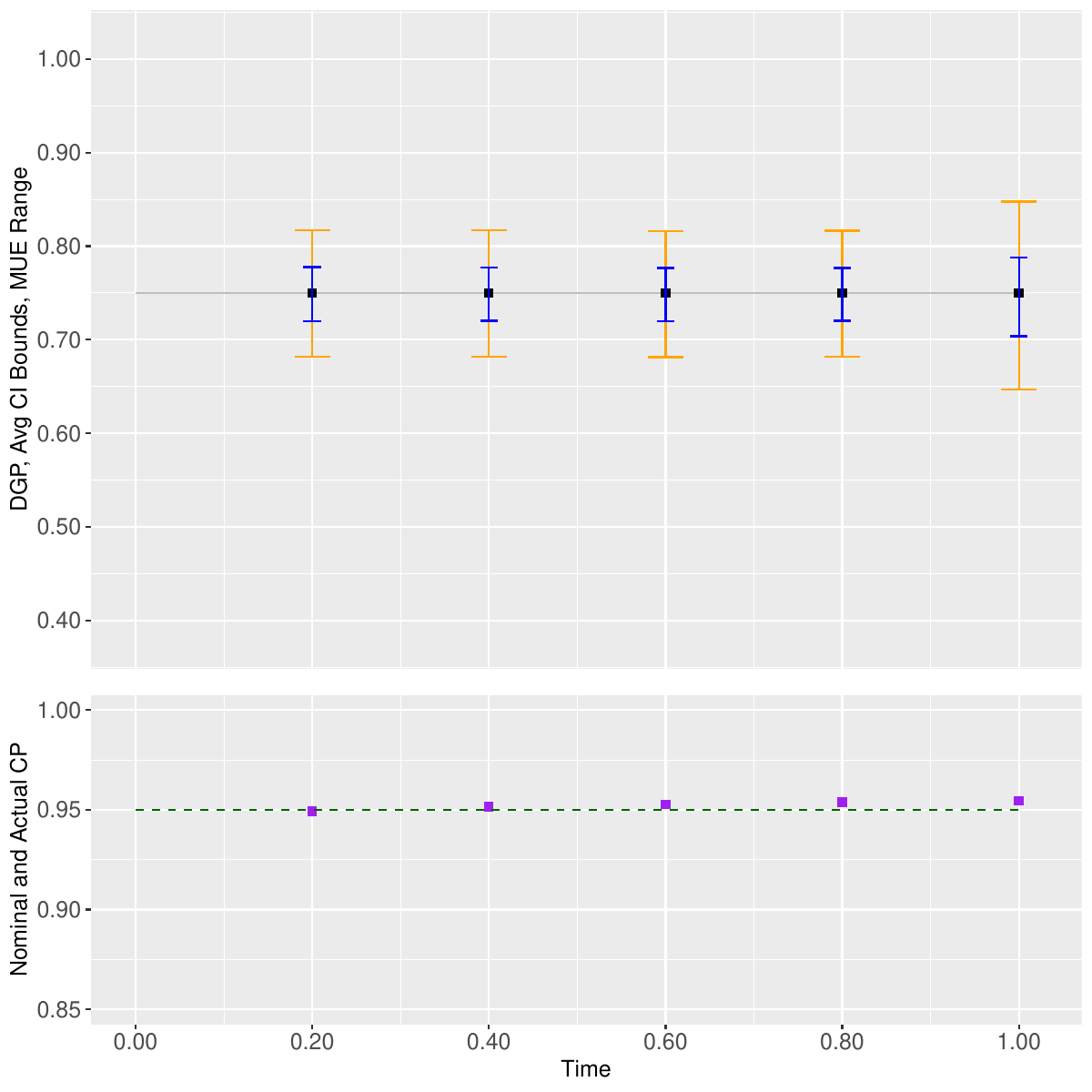}

}
\par\end{centering}
\vspace{-0.65em}

\begin{centering}
\subfloat[kinked 1.00-0.80-1.00, constant $\mu$ and $\sigma$]{\includegraphics[scale=0.36]{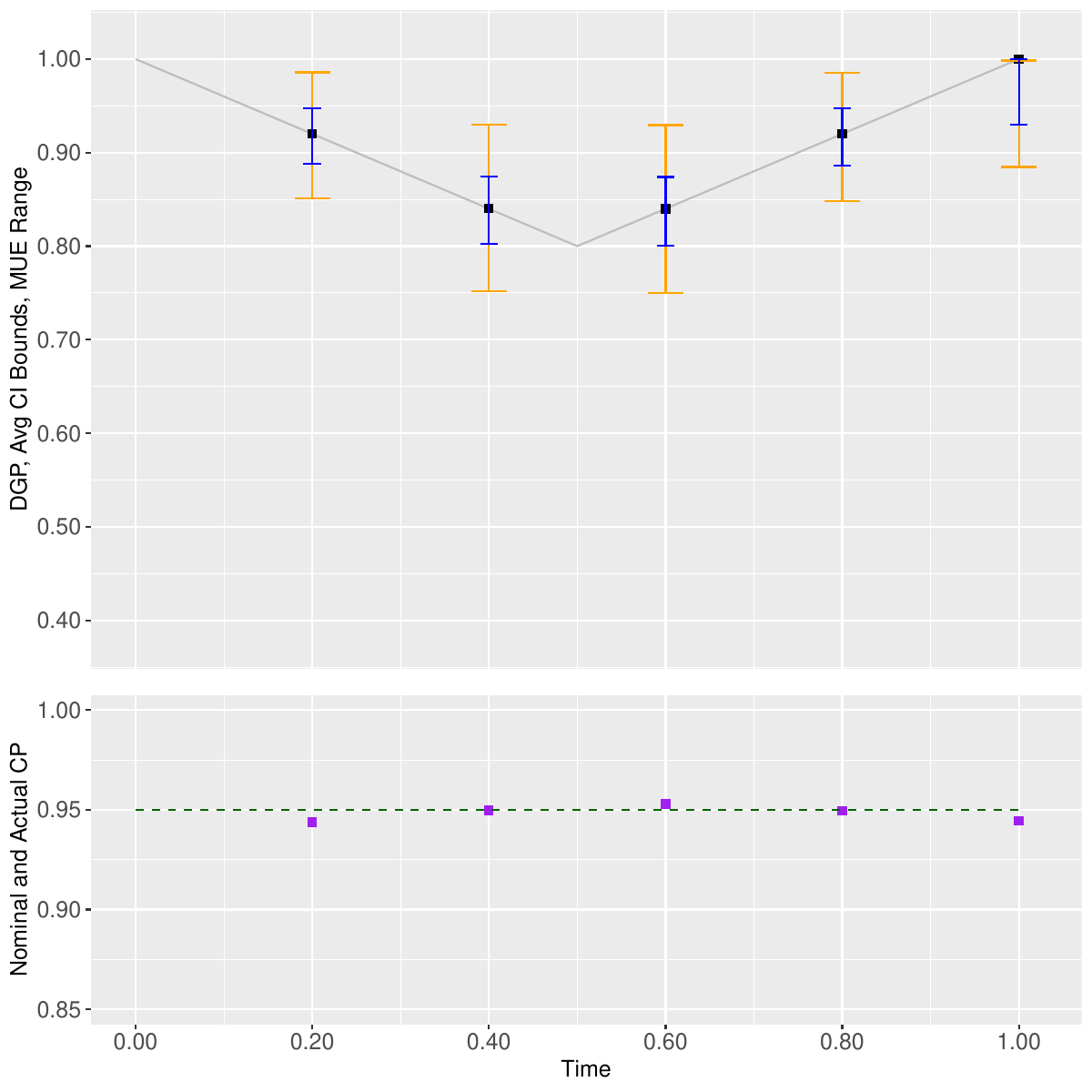}

}\quad{}\subfloat[kinked 1.00-0.80-1.00, time-varying $\mu$ and $\sigma$]{\includegraphics[scale=0.36]{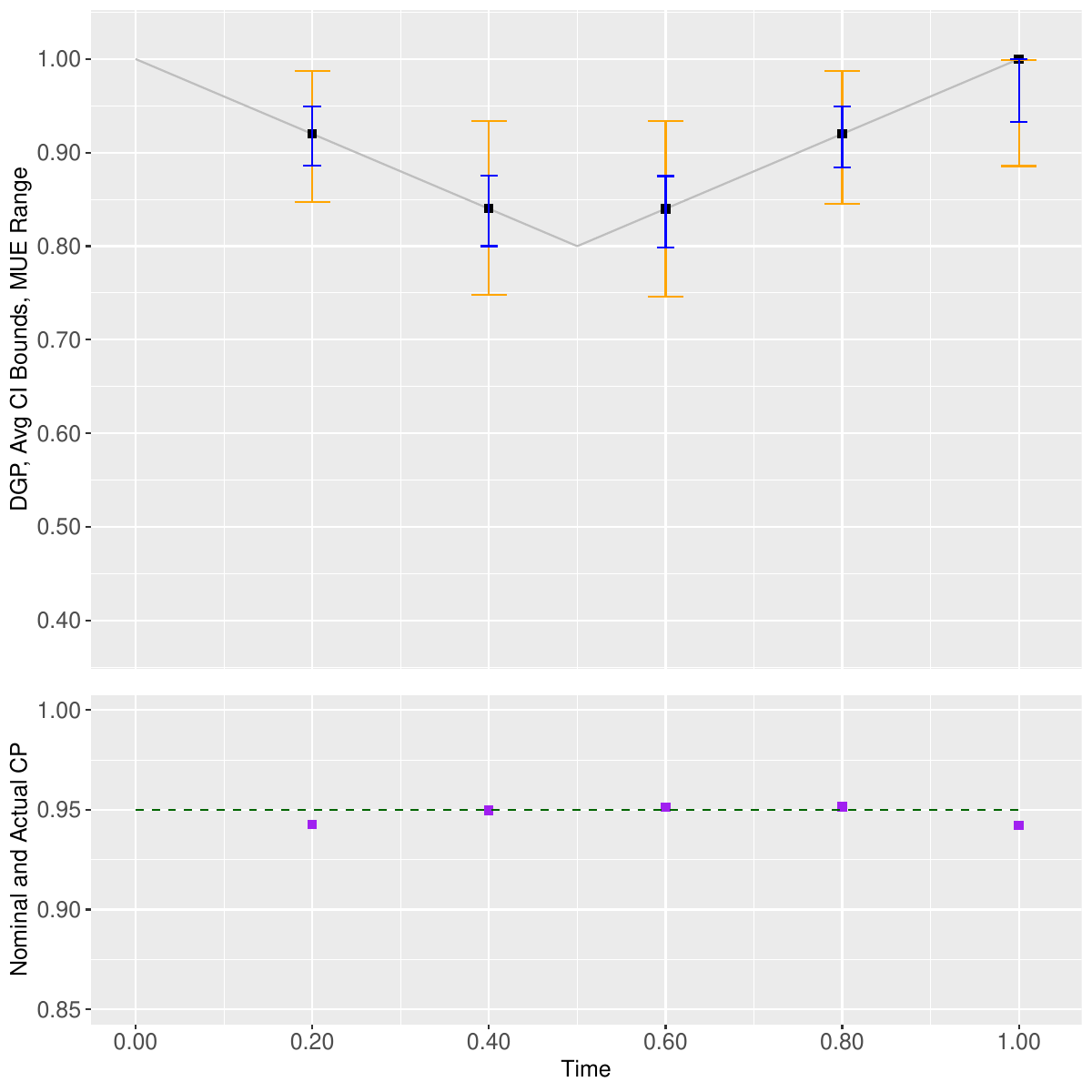}

}
\par\end{centering}
\vspace{-0.75em}

\caption{\protect\label{fig:Sim_CP_AL_MAD_SM3}CP's and AL's of CI's for $\rho\left(\tau\right)$
and MAD's of the MUE of $\rho\left(\tau\right)$}
\end{figure}
\begin{figure}[H]
\vspace{-2.5em}

\begin{centering}
\subfloat[kinked 0.80-1.00-0.80, constant $\mu$ and $\sigma$]{\includegraphics[scale=0.36]{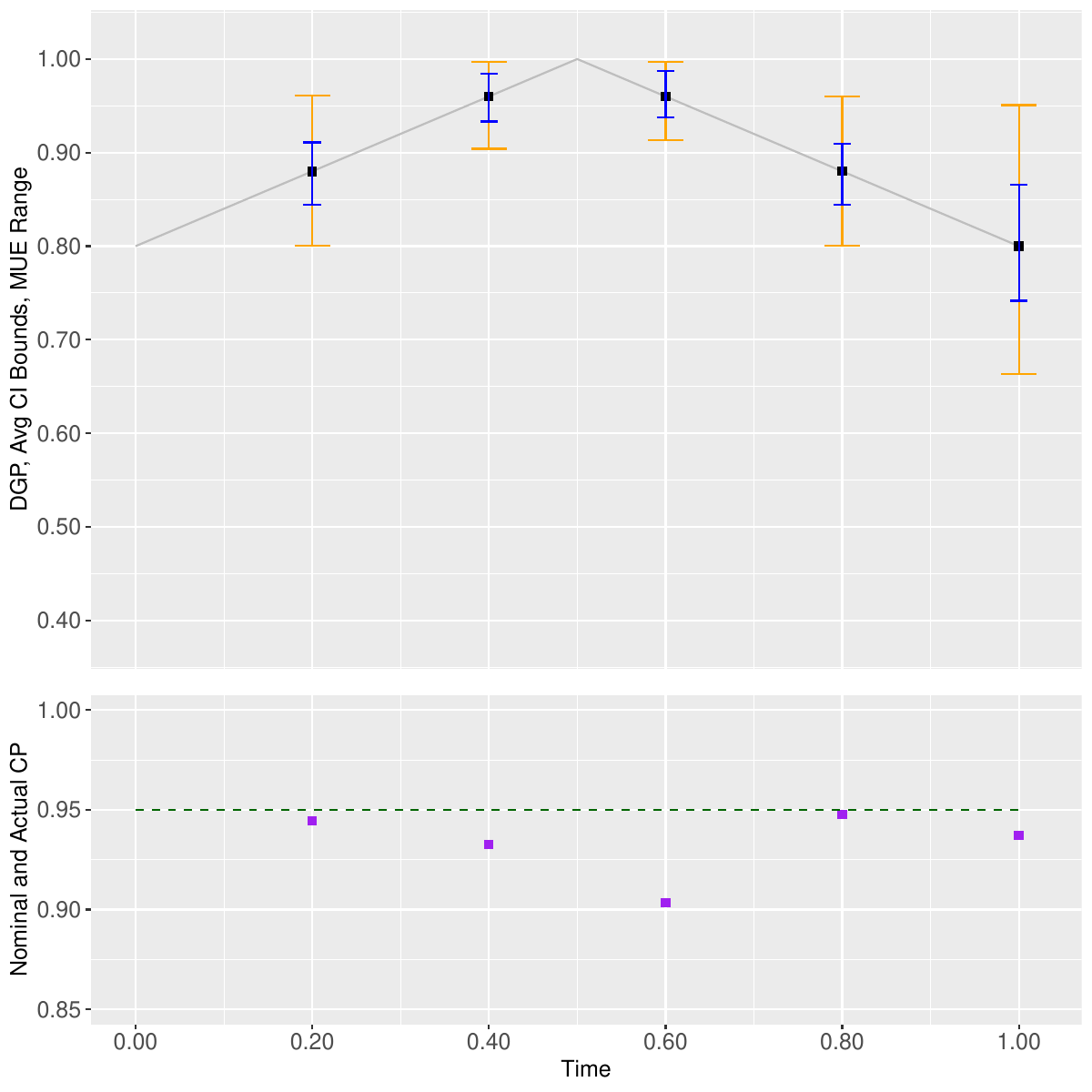}

}\quad{}\subfloat[kinked 0.80-1.00-0.80, time-varying $\mu$ and $\sigma$]{\includegraphics[scale=0.36]{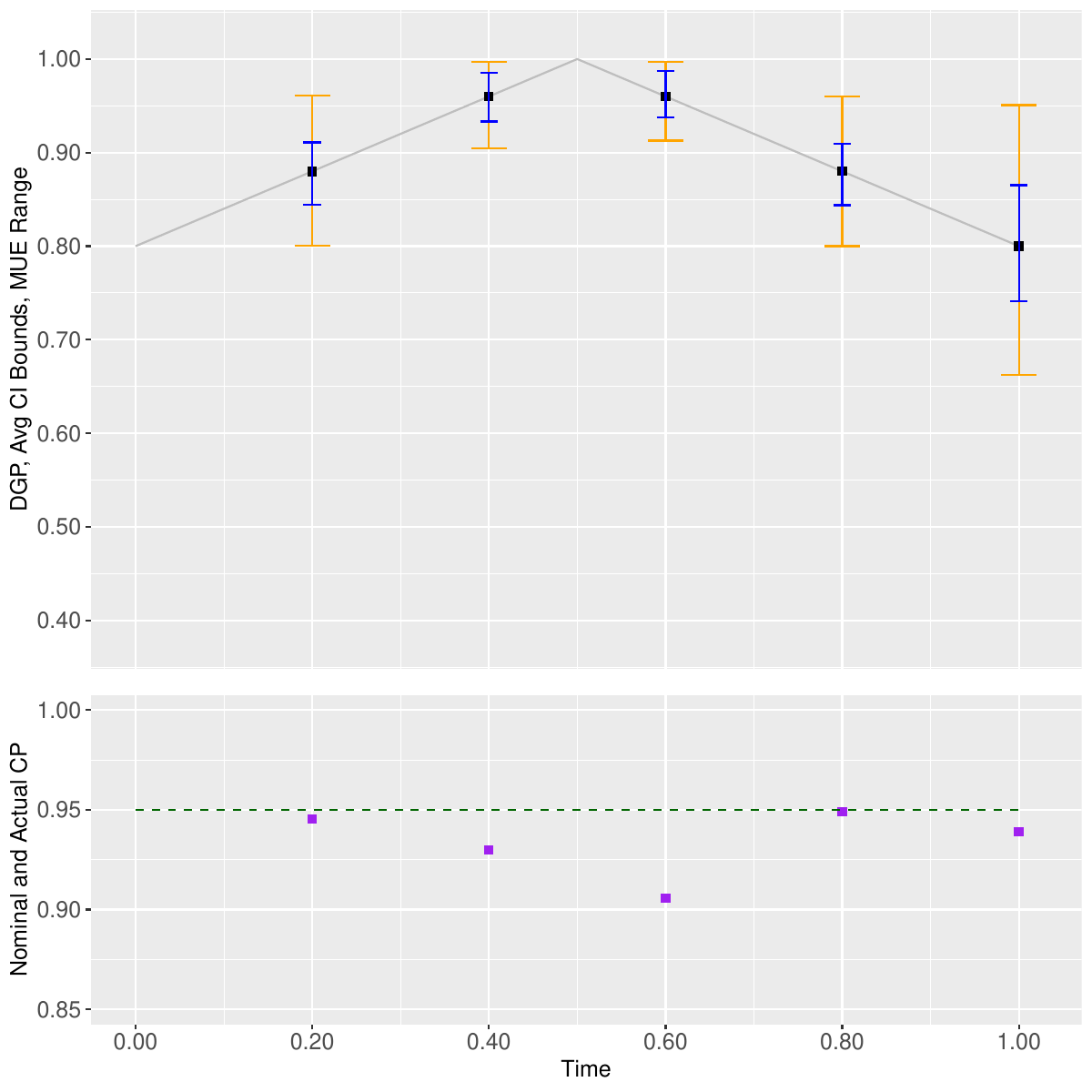}

}
\par\end{centering}
\vspace{-0.65em}

\begin{centering}
\subfloat[kinked 1.00-0.60-1.00, constant $\mu$ and $\sigma$]{\includegraphics[scale=0.36]{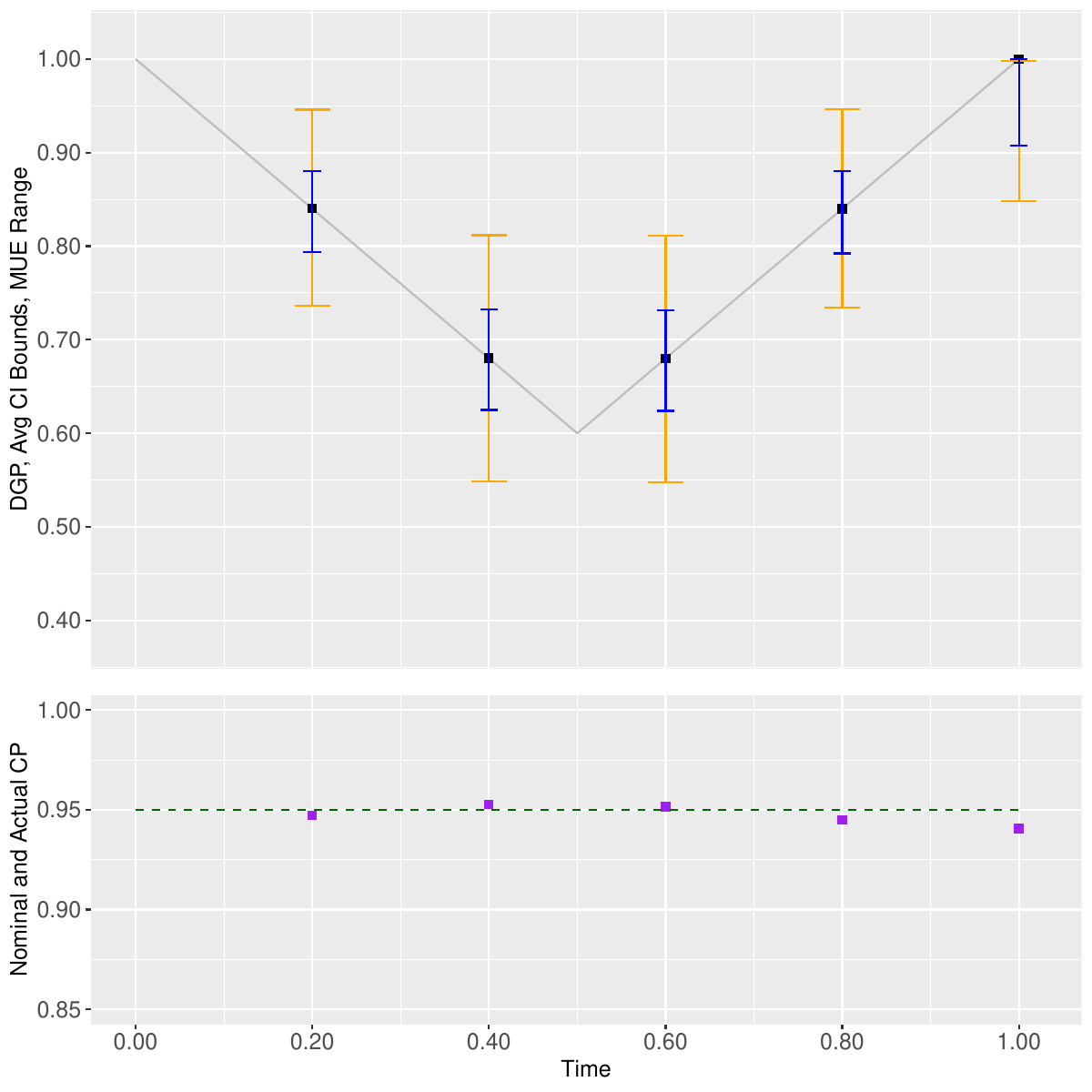}

}\quad{}\subfloat[kinked 1.00-0.60-1.00, time-varying $\mu$ and $\sigma$]{\includegraphics[scale=0.36]{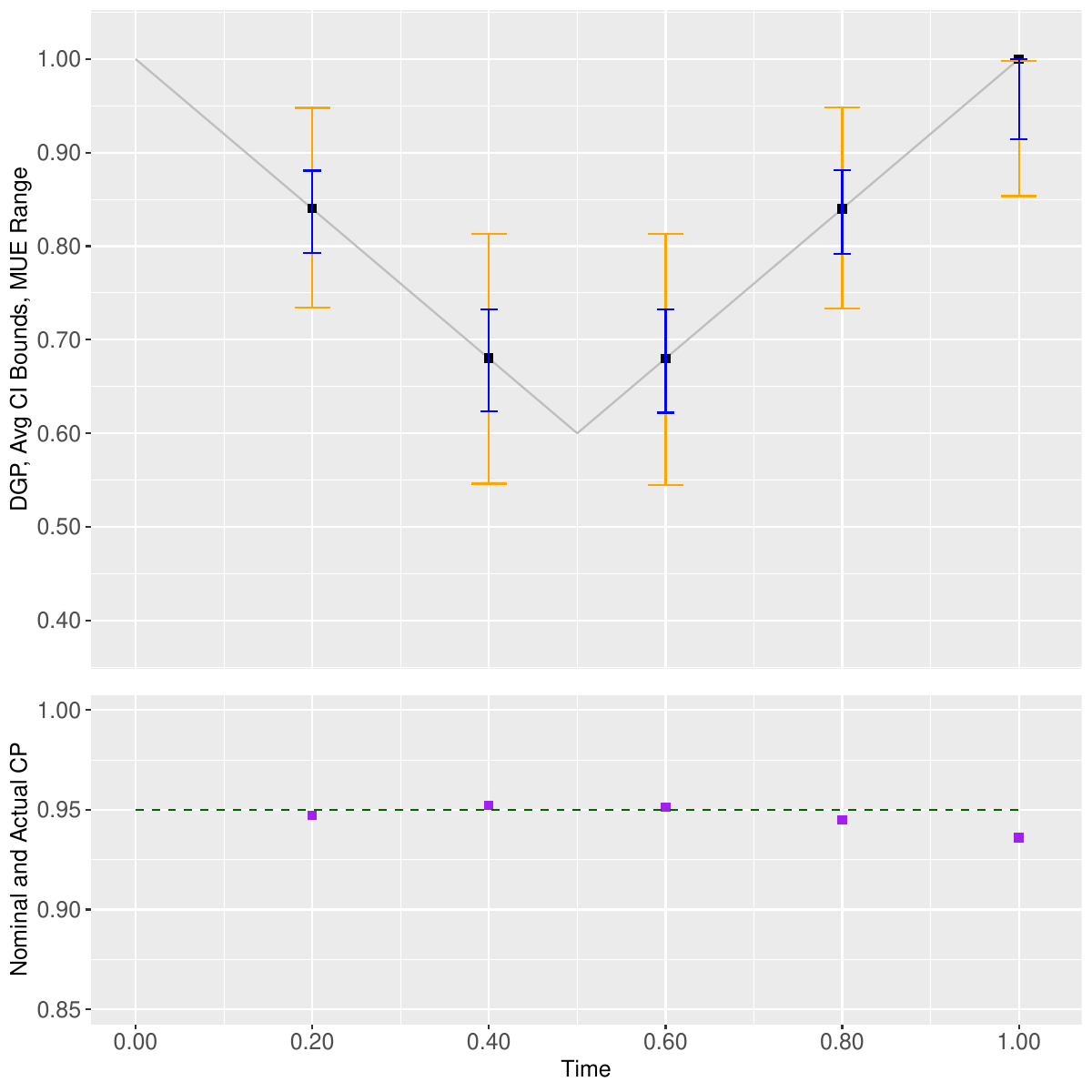}

}
\par\end{centering}
\vspace{-0.65em}

\begin{centering}
\subfloat[kinked 0.60-1.00-0.60, constant $\mu$ and $\sigma$]{\includegraphics[scale=0.36]{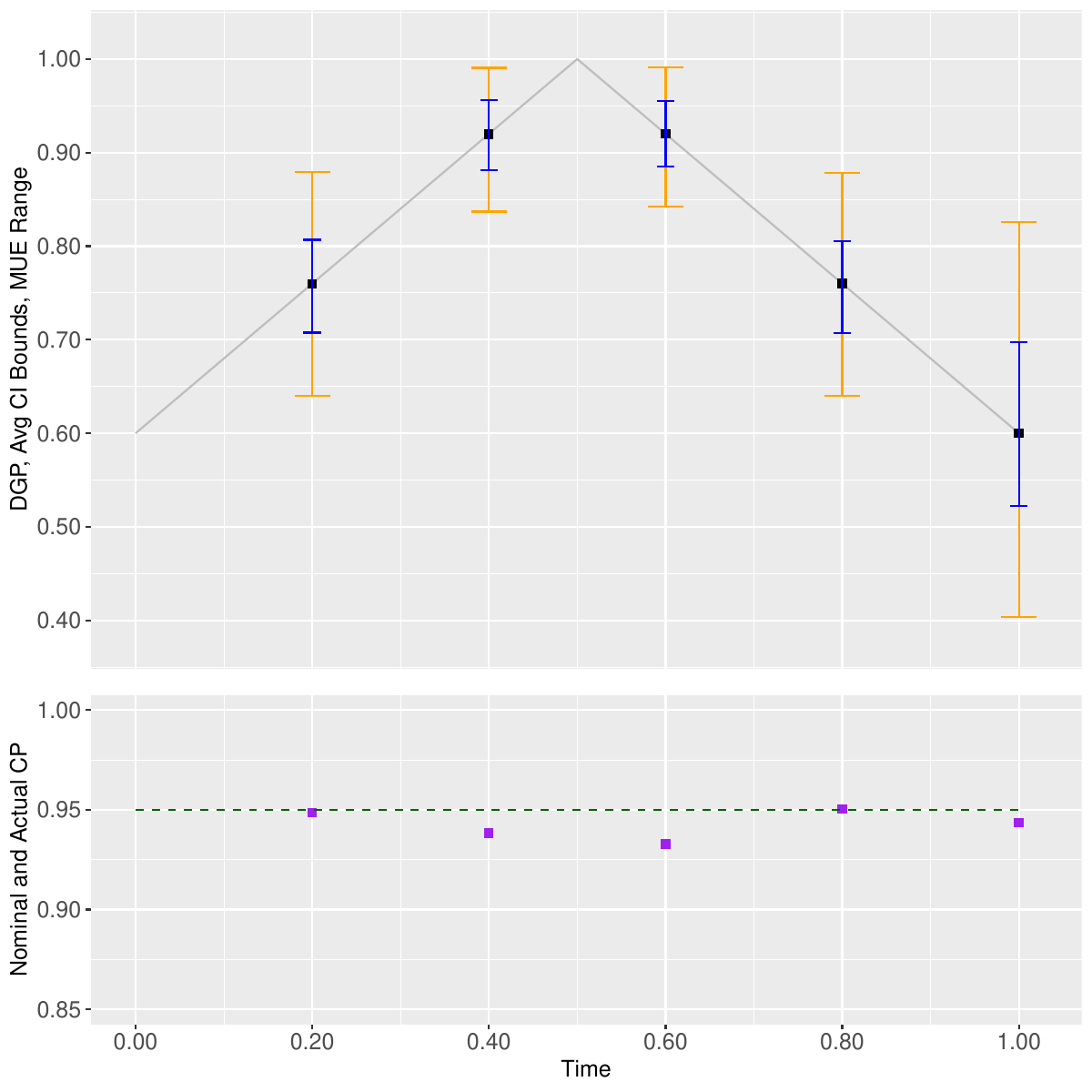}

}\quad{}\subfloat[kinked 0.60-1.00-0.60, time-varying $\mu$ and $\sigma$]{\includegraphics[scale=0.36]{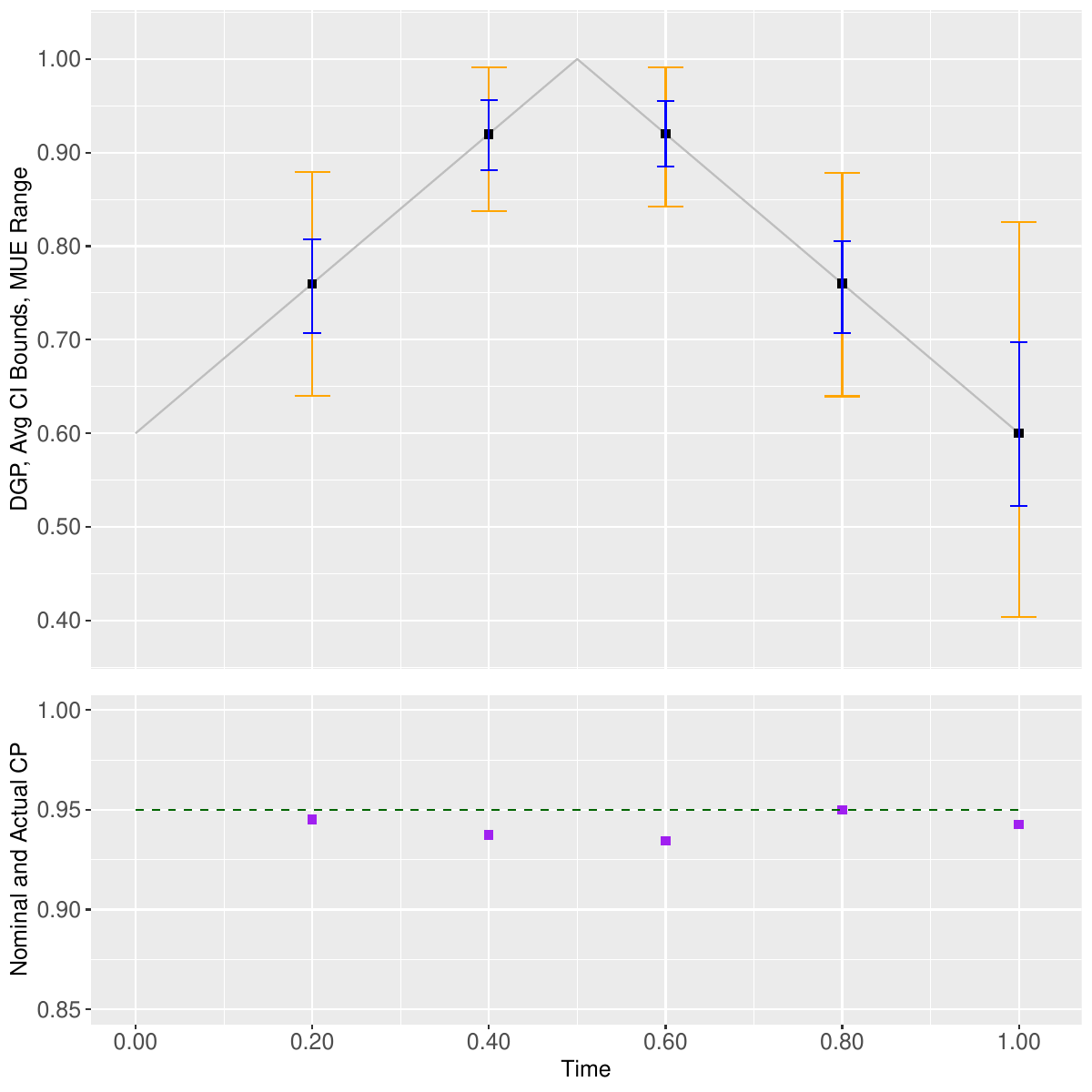}

}
\par\end{centering}
\vspace{-0.75em}

\caption{\protect\label{fig:Sim_CP_AL_MAD_SM4}CP's and AL's of CI's for $\rho\left(\tau\right)$
and MAD's of the MUE of $\rho\left(\tau\right)$}
\end{figure}
\par\end{center}

\section{Extension to TVP-AR(p) Models\protect\label{sec:Extension-to-TVP-AR(p)}}

\numberwithin{equation}{section}Here we discuss how the methods
introduced in the paper for a TVP-AR(1) model can be extended to a
TVP-AR(p) model. We combine the approach discussed above with a method
for constant parameter AR(p) models that is similar to, but somewhat
different from, methods that have been considered in the literature
to date. It is similar to Hansen's \citeyearpar{hansen1999grid} grid
bootstrap, also see \citet[Sec. 7]{mikusheva2007uniform}, but uses
asymptotic critical values, rather than bootstrap critical values,
which eases computation considerably in the time-varying case because
the tabulated quantiles for the AR(1) case can be utilized (with an
adjustment of the $\psi_{nh}$ value that is employed). Asymptotic
results for $p>1$ are beyond the scope of this paper and are not
provided here.

Consider the following TVP-AR(p) model written in augmented Dickey-Fuller
(ADF) form: 
\begin{align}
Y_{t} & =\mu_{t}+Y_{t}^{\ast}\ \text{and}\nonumber \\
Y_{t}^{\ast} & =\rho_{t}Y_{t-1}^{\ast}+\sum_{j=1}^{p-1}\beta_{jt}\Delta Y_{t-j}^{\ast}+\sigma_{t}U_{t},\ \text{for\ }t=1,...,n,
\end{align}
where $\Delta Y_{t-j}^{\ast}=Y_{t-j}^{\ast}-Y_{t-j-1}^{\ast}$ for
$j=1,...,p-1.$ Here, $\mu_{t},$ $\rho_{t},$ $\sigma_{t},$ and
$U_{t}$ are as in Section \ref{sec:MT-Model=000020Setup}. The coefficients
$\beta_{jt}$ are possibly time varying and satisfy analogous properties
to those of $\mu_{t}.$ The parameter $\rho_{t}$ is the sum of the
$p$ AR coefficients. It is the parameter of interest because it is
a suitable measure of the persistence of the time series, see \citet[Sec. 2.2]{andrewschen1994approximately}
for a discussion. As in Section \ref{sec:MT-Model=000020Setup}, $\rho_{t}:=\rho\left(t/n\right)$
and, for $\tau\in\left(0,1\right),$ we consider estimation and inference
concerning $\rho\left(\tau\right).$

To construct a CI for $\rho\left(\tau\right)$ in the AR(p) model,
we proceed as follows. First, consider the regression of $Y_{t}$
on a constant, $Y_{t-1},$ $\Delta Y_{t-1},...,\Delta Y_{t-p+1}$
for $t=T_{1},...,T_{2},$ where $\Delta Y_{s}:=Y_{s}-Y_{s-1}.$ For
arbitrary $\rho_{0}\in(-1,1],$ let $T_{n}\left(\rho_{0},p\right)$
be the t-statistic for testing the null hypothesis that the coefficient
on the regressor $Y_{t-1}$ in this regression equals $\rho_{0}.$
Second, compute $\widehat{\beta}(\rho_{0})\in R^{p-1}$ from the regression
of $Y_{t}-\rho_{0}Y_{t-1}$ on a constant, $\Delta Y_{t-1},...,\Delta Y_{t-p+1}$
for $t=T_{1},...,T_{2}.$ Third, one computes 
\begin{equation}
\psi_{nh,\rho_{0}}^{p}:=\frac{-nh\ln(\rho_{0})}{\widehat{\lambda}(\rho_{0})}\ \text{for\ }\rho_{0}>0\text{\ and \ensuremath{\psi_{nh,\rho_{0}}^{p}}:= \ensuremath{\infty\ \ \text{for\ }} \ensuremath{\rho_{0}\leq}0, where }\widehat{\lambda}(\rho_{0}):=1-\sum_{j=1}^{p-1}\widehat{\beta}_{j}(\rho_{0}).
\end{equation}
A nominal $1-\alpha$ equal-tailed two-sided CI for $\rho(\tau)$
is given by the formula in (\ref{eq:MT-1-sided=000020CIs}) with $T_{n}\left(\rho_{0},p\right)$
in place of $T_{n}\left(\rho_{0}\right)$ and with $\psi_{nh,\rho_{0}}^{p}$
in place of $\psi_{nh,\rho_{0}}$ in the critical values. A median-unbiased
interval estimator $\widetilde{\rho}_{n\tau}$ of $\rho\left(\tau\right)$
is defined as in Section \ref{subsec:MT-Median-Unbiased-Interval-Est}
with the same changes. For the motivation behind the definition of
$\psi_{nh}^{p}$ above, see \citet{hansen1995rethinking,hansen1999grid}.

Note that the only computational difference between the above CI's
for $\rho$ in the TVP-AR(1) and TVP-AR(p) models is that the latter
requires the computation of $\widehat{\beta}(\rho_{0})$ for a fine
grid of $\rho_{0}$ values and each value $\tau$ of interest. In
contrast, if one replaces the critical values $c_{\psi_{nh}^{p}}(\alpha/2)$
and $c_{\psi_{nh}^{p}}(1-\alpha/2)$ by the $\alpha/2$ and $1-\alpha/2$
quantiles of a bootstrap test statistic, e.g., as in Hansen's \citeyearpar{hansen1999grid}
grid bootstrap, then one needs to simulate these quantiles for a fine
grid of $\rho_{0}$ values for each $\tau$ of interest, which is
computationally quite expensive for reasonable choices of the number
of simulation repetitions.

Some empirical applications with $p=6$ and $12$ are reported in
Section \ref{sec:addl-Empirical-Results} below.

\section{Additional Empirical Results\protect\label{sec:addl-Empirical-Results}}

In this section, we present results for some additional time series
in the IFS dataset and some in the FRED dataset. Some of these series
require a TVP-AR(p) model for $p>1$.

As noted in the Introduction, and described in Section \ref{sec:Extension-to-TVP-AR(p)}
of the Supplemental Material, the methods introduced above for the
TVP-AR(1) model can be extended to TVP-AR(p) models with $p>1$. In
the TVP-AR(p) model, the parameter $\rho_{t}$ is the sum of the autoregressive
coefficients at time $t$, or equivalently, the coefficient at time
$t$ on the lagged $Y_{t}$ value in the augmented Dickey-Fuller representation
of the model. This coefficient is a suitable measure of the persistence
of the time series at time $t$, e.g., see \citet[Sec. 2.2]{andrewschen1994approximately}.

For each time series, we estimate a TVP-AR(p) model with $p=1,6,12$
and examine the degree of autocorrelation of the corresponding residuals
by computing Ljung-Box tests with six lags of the residuals. For each
time series, we present the results from the TVP-AR(p) model with
the smallest value of $p$ for which the null hypothesis of no autocorrelation
is not rejected at the 5\% level. When $p\in\left\{ 6,12\right\} $,
the MUE's and CI's are for the time-varying autoregressive parameter
corresponding to the lagged dependent variable in ADF form. We group
the results based on the selected $p$ in the figures.

$\ $

First, we consider additional time series from the IFS dataset, which
include real exchange rate series for Norway, Canada, and Japan, interest
rate series for Australia, Canada, and the US, and inflation series
for Switzerland. The definition of real exchange rates and inflation
are the same as described in Section \ref{subsec:MT-EMP_Inflation}
and \ref{subsec:MT-EMP_Real-Exchange-Rate}, respectively. For the
interest rate series, we use the monthly interbank interest rate,
which is a key monetary tool for central banks to achieve their policy
goals. More details about the length, time period, and frequency of
the additional IFS series can be found in Tables \ref{tab:ACF_Test_AR1}--\ref{tab:ACF_Test_AR12}. 

Figure \ref{fig:EMP_AR6_SM1} presents the MUE's and 90\% CI's of
$\rho\left(t\right)$ for the additional real exchange rate series.
We fit a TVP-AR(1) model for Norway and TVP-AR(6) model for Canada
and Japan based on the Ljung-Box tests results. Across all three countries,
the MUE's are close to one with reasonably tight 90\% CI's. The selected
$n\widehat{h}_{us}$ values are quite large, consistent with the parameter
estimates that show little time variation. The results echo the empirical
findings of high real exchange rate persistence for the developed
countries presented in Section \ref{subsec:MT-EMP_Real-Exchange-Rate}. 

Figure \ref{fig:EMP_AR6_SM2} shows that the MUE's of $\rho\left(t\right)$
for the interest rate series from a TVP-AR(6) model are close to one
with a moderate degree of variation over time. Most notable are the
estimates of $\rho\left(t\right)$ for Canada and the US during the
period around 2012 when the MUE's drop to as low as .7. The 90\% CI's
are fairly tight. In comparison, the constant parameter MUE's from
an AR(6) model are uniformly one for the three series. 

Figure \ref{fig:EMP_AR12_SM}(a)--(b) summarizes the results for
estimating a TVP-AR(12) model for the Switzerland inflation series.
The MUE's of $\rho\left(t\right)$ are quite volatile over time, ranging
between -.6 and 1 in Figure \ref{fig:EMP_AR12_SM}(a). This is different
from the constant parameter estimate which is close to .9, highlighting
the importance of allowing for possible time variation in the autoregressive
parameters in these models.

$\ $

Second, we consider the FRED series. We have a total of eight time
series for the US, including the 10 year bond yield, average wages
for the manufacturing sector, industrial production, real GDP per
capita, real GNP, real GNP per capita, S\&P 500 index, and the unemployment
rate. We provide details about the length, time period, and frequency
of the FRED series in Tables \ref{tab:ACF_Test_AR6}--\ref{tab:ACF_Test_AR12}.

Figures \ref{fig:EMP_AR6_SM3}--\ref{fig:EMP_AR6_SM4} show the results
from estimating a TVP-AR(6) model for the US FRED series for which
the null hypothesis of no autocorrelation is not rejected at the 5\%
level. In Figure \ref{fig:EMP_AR6_SM3}(e) and (f), there are some
variation in the MUE's of $\rho\left(t\right)$ for the US unemployment
rate series, however the magnitude is small. For all other series
in Figures \ref{fig:EMP_AR6_SM3}--\ref{fig:EMP_AR6_SM4}, the MUE's
of $\rho\left(t\right)$ are uniformly one or very close to one over
time and almost the same as constant parameter estimates. All of the
90\% CI's are tight with a length smaller than .02. The selected $n\widehat{h}_{us}$
values are large, in line with the parameter estimates that show little
time variation. Hence, the methods proposed in the paper deliver a
constant parameter unit root, or near unit root, model in circumstances
in which such a model is appropriate.

Figure \ref{fig:EMP_AR12_SM}(c)--(f) provides the results for fitting
a TVP-AR(12) model to the time series on the S\&P 500 index and the
US industrial production. For the S\&P 500 index series, the MUE's
of $\rho\left(t\right)$ are one and similar to the constant parameter
estimates. The results are consistent with predictions from a random
walk hypothesis for the US stock markets. For the US industrial production
series, the MUE's of $\rho\left(t\right)$ in Figure \ref{fig:EMP_AR12_SM}(e)
are close to 1 for most of the time except before 1925 and after 2010.
It may be caused by a boundary effect.

\global\long\def\thetable{SM.\arabic{table}}%
 
\begin{table}[H]
\begin{centering}
\caption{\protect\label{tab:ACF_Test_AR1}Autocorrelation Test Results for
Residuals from Estimating TVP-AR(1) Model}
\par\end{centering}
\centering{}{\footnotesize{}}%
\begin{tabular}{lccccc>{\centering}m{2.5cm}}
\toprule 
 & {\footnotesize{}From} & {\footnotesize{}To} & {\footnotesize{}Frequency} & {\footnotesize{}$n$} & {\footnotesize{}$n\widehat{h}_{us}$} & {\footnotesize{}Ljung-Box Test}{\footnotesize\par}

{\footnotesize{}p-Value}\tabularnewline
\midrule 
{\footnotesize{}US Inflation} & {\footnotesize{}1/2/1955} & {\footnotesize{}1/10/2022} & {\footnotesize{}monthly} & {\footnotesize{}813} & {\footnotesize 125} & {\footnotesize .27}\tabularnewline
{\footnotesize{}US Inflation} & {\footnotesize{}1/2/1955} & {\footnotesize{}1/10/2022} & {\footnotesize{}monthly} & {\footnotesize{}813} & {\footnotesize 188} & {\footnotesize .28}\tabularnewline
{\footnotesize{}Canada Inflation} & {\footnotesize{}1/2/1955} & {\footnotesize{}1/10/2022} & {\footnotesize{}monthly} & {\footnotesize{}813} & {\footnotesize 125} & {\footnotesize .38}\tabularnewline
{\footnotesize{}Canada Inflation} & {\footnotesize{}1/2/1955} & {\footnotesize{}1/10/2022} & {\footnotesize{}monthly} & {\footnotesize{}813} & {\footnotesize 188} & {\footnotesize .22}\tabularnewline
{\footnotesize{}Germany Inflation} & {\footnotesize{}1/2/1955} & {\footnotesize{}1/10/2022} & {\footnotesize{}monthly} & {\footnotesize{}813} & {\footnotesize 125} & {\footnotesize .76}\tabularnewline
{\footnotesize{}Germany Inflation} & {\footnotesize{}1/2/1955} & {\footnotesize{}1/10/2022} & {\footnotesize{}monthly} & {\footnotesize{}813} & {\footnotesize 188} & {\footnotesize .81}\tabularnewline
{\footnotesize{}UK Real Exchange Rate} & {\footnotesize{}1/1/1957} & {\footnotesize{}1/8/2022} & {\footnotesize{}monthly} & {\footnotesize{}788} & {\footnotesize 823} & {\footnotesize .52}\tabularnewline
{\footnotesize{}UK Real Exchange Rate} & {\footnotesize{}1/1/1957} & {\footnotesize{}1/8/2022} & {\footnotesize{}monthly} & {\footnotesize{}788} & {\footnotesize 1,234} & {\footnotesize .51}\tabularnewline
{\footnotesize{}Sweden Real Exchange Rate} & {\footnotesize{}1/1/1957} & {\footnotesize{}1/8/2022} & {\footnotesize{}monthly} & {\footnotesize{}788} & {\footnotesize 823} & {\footnotesize .28}\tabularnewline
{\footnotesize{}Sweden Real Exchange Rate} & {\footnotesize{}1/1/1957} & {\footnotesize{}1/8/2022} & {\footnotesize{}monthly} & {\footnotesize{}788} & {\footnotesize 1,234} & {\footnotesize .28}\tabularnewline
{\footnotesize{}Switzerland Real Exchange Rate} & {\footnotesize{}1/1/1957} & {\footnotesize{}1/8/2022} & {\footnotesize{}monthly} & {\footnotesize{}788} & {\footnotesize 393} & {\footnotesize .56}\tabularnewline
{\footnotesize{}Switzerland Real Exchange Rate} & {\footnotesize{}1/1/1957} & {\footnotesize{}1/8/2022} & {\footnotesize{}monthly} & {\footnotesize{}788} & {\footnotesize 590} & {\footnotesize .54}\tabularnewline
{\footnotesize{}Norway Real Exchange Rate} & {\footnotesize{}1/1/1957} & {\footnotesize{}1/8/2022} & {\footnotesize{}monthly} & {\footnotesize{}788} & {\footnotesize 823} & {\footnotesize .37}\tabularnewline
{\footnotesize{}Norway Real Exchange Rate} & {\footnotesize{}1/1/1957} & {\footnotesize{}1/8/2022} & {\footnotesize{}monthly} & {\footnotesize{}788} & {\footnotesize 1,234} & {\footnotesize .38}\tabularnewline
\bottomrule
\end{tabular}
\end{table}
\begin{table}[H]
\begin{centering}
\caption{\protect\label{tab:ACF_Test_AR6}Autocorrelation Test Results for
Residuals from Estimating TVP-AR(6) Model}
\par\end{centering}
\centering{}{\footnotesize{}}%
\begin{tabular}{lccccc>{\centering}m{2.5cm}}
\toprule 
 & {\footnotesize{}From} & {\footnotesize{}To} & {\footnotesize{}Frequency} & {\footnotesize{}$n$} & {\footnotesize{}$n\widehat{h}_{us}$} & {\footnotesize{}Ljung-Box Test}{\footnotesize\par}

{\footnotesize{}p-Value}\tabularnewline
\midrule 
{\footnotesize{}Canada Real Exchange Rate} & {\footnotesize{}1/1/1957} & {\footnotesize{}1/8/2022} & {\footnotesize{}Monthly} & {\footnotesize{}788} & {\footnotesize 823} & {\footnotesize 1.00}\tabularnewline
{\footnotesize{}Canada Real Exchange Rate} & {\footnotesize{}1/1/1957} & {\footnotesize{}1/8/2022} & {\footnotesize{}Monthly} & {\footnotesize{}788} & {\footnotesize 1,234} & {\footnotesize 1.00}\tabularnewline
{\footnotesize{}Japan Real Exchange Rate} & {\footnotesize{}1/1/1957} & {\footnotesize{}1/8/2022} & {\footnotesize{}Monthly} & {\footnotesize{}788} & {\footnotesize 823} & {\footnotesize .81}\tabularnewline
{\footnotesize{}Japan Real Exchange Rate} & {\footnotesize{}1/1/1957} & {\footnotesize{}1/8/2022} & {\footnotesize{}Monthly} & {\footnotesize{}788} & {\footnotesize 1,234} & {\footnotesize .84}\tabularnewline
{\footnotesize{}Australia Interest Rate} & {\footnotesize{}1/5/1976} & {\footnotesize{}1/4/2017} & {\footnotesize{}Monthly} & {\footnotesize{}492} & {\footnotesize 160} & {\footnotesize .97}\tabularnewline
{\footnotesize{}Australia Interest Rate} & {\footnotesize{}1/5/1976} & {\footnotesize{}1/4/2017} & {\footnotesize{}Monthly} & {\footnotesize{}492} & {\footnotesize 240} & {\footnotesize .98}\tabularnewline
{\footnotesize{}Canada Interest Rate} & {\footnotesize{}1/5/1976} & {\footnotesize{}1/4/2017} & {\footnotesize{}Monthly} & {\footnotesize{}492} & {\footnotesize 95} & {\footnotesize .95}\tabularnewline
{\footnotesize{}Canada Interest Rate} & {\footnotesize{}1/5/1976} & {\footnotesize{}1/4/2017} & {\footnotesize{}Monthly} & {\footnotesize{}492} & {\footnotesize 143} & {\footnotesize .97}\tabularnewline
{\footnotesize{}US Interest Rate} & {\footnotesize{}1/5/1976} & {\footnotesize{}1/4/2017} & {\footnotesize{}Monthly} & {\footnotesize{}492} & {\footnotesize 103} & {\footnotesize .98}\tabularnewline
{\footnotesize{}US Interest Rate} & {\footnotesize{}1/5/1976} & {\footnotesize{}1/4/2017} & {\footnotesize{}Monthly} & {\footnotesize{}492} & {\footnotesize 155} & {\footnotesize .97}\tabularnewline
{\footnotesize{}US 10yr Bond Yield} & {\footnotesize{}2/1/1962} & {\footnotesize{}20/1/2023} & {\footnotesize{}Daily} & {\footnotesize{}15248} & {\footnotesize 16,006} & {\footnotesize 1.00}\tabularnewline
{\footnotesize{}US 10yr Bond Yield} & {\footnotesize{}2/1/1962} & {\footnotesize{}20/1/2023} & {\footnotesize{}Daily} & {\footnotesize{}15248} & {\footnotesize 24,009} & {\footnotesize 1.00}\tabularnewline
{\footnotesize{}US Average Wages Manufacturing} & {\footnotesize{}1/1/1939} & {\footnotesize{}1/12/2022} & {\footnotesize{}Monthly} & {\footnotesize{}1008} & {\footnotesize 152} & {\footnotesize .12}\tabularnewline
{\footnotesize{}US Average Wages Manufacturing} & {\footnotesize{}1/1/1939} & {\footnotesize{}1/12/2022} & {\footnotesize{}Monthly} & {\footnotesize{}1008} & {\footnotesize 228} & {\footnotesize .27}\tabularnewline
{\footnotesize{}US Unemployment Rate} & {\footnotesize{}1/1/1948} & {\footnotesize{}1/12/2022} & {\footnotesize{}Monthly} & {\footnotesize{}900} & {\footnotesize 900} & {\footnotesize .98}\tabularnewline
{\footnotesize{}US Unemployment Rate} & {\footnotesize{}1/1/1948} & {\footnotesize{}1/12/2022} & {\footnotesize{}Monthly} & {\footnotesize{}900} & {\footnotesize 1,350} & {\footnotesize .99}\tabularnewline
{\footnotesize{}US Real GDP Per Capita} & {\footnotesize{}1/1/1947} & {\footnotesize{}1/7/2022} & {\footnotesize{}Quarterly} & {\footnotesize{}303} & {\footnotesize 317} & {\footnotesize .97}\tabularnewline
{\footnotesize{}US Real GDP Per Capita} & {\footnotesize{}1/1/1947} & {\footnotesize{}1/7/2022} & {\footnotesize{}Quarterly} & {\footnotesize{}303} & {\footnotesize 476} & {\footnotesize .98}\tabularnewline
{\footnotesize{}US Real GNP} & {\footnotesize{}1/1/1947} & {\footnotesize{}1/7/2022} & {\footnotesize{}Quarterly} & {\footnotesize{}303} & {\footnotesize 317} & {\footnotesize 1.00}\tabularnewline
{\footnotesize{}US Real GNP} & {\footnotesize{}1/1/1947} & {\footnotesize{}1/7/2022} & {\footnotesize{}Quarterly} & {\footnotesize{}303} & {\footnotesize 476} & {\footnotesize .48}\tabularnewline
{\footnotesize{}US Real GNP Per Capita} & {\footnotesize{}1/1/1947} & {\footnotesize{}1/7/2022} & {\footnotesize{}Quarterly} & {\footnotesize{}303} & {\footnotesize 317} & {\footnotesize .96}\tabularnewline
{\footnotesize{}US Real GNP Per Capita} & {\footnotesize{}1/1/1947} & {\footnotesize{}1/7/2022} & {\footnotesize{}Quarterly} & {\footnotesize{}303} & {\footnotesize 476} & {\footnotesize .99}\tabularnewline
\bottomrule
\end{tabular}
\end{table}
\begin{table}[H]
\begin{centering}
\caption{\protect\label{tab:ACF_Test_AR12}Autocorrelation Test Results for
Residuals from Estimating TVP-AR(12) Model}
\par\end{centering}
\centering{}{\footnotesize{}}%
\begin{tabular}{cccccc>{\centering}m{2.5cm}}
\toprule 
 & {\footnotesize{}From} & {\footnotesize{}To} & {\footnotesize{}Frequency} & {\footnotesize{}n} & {\footnotesize{}nh} & {\footnotesize{}Ljung-Box Test}{\footnotesize\par}

{\footnotesize{}p-Value}\tabularnewline
\midrule 
{\footnotesize{}Switzerland Inflation} & {\footnotesize{}1/2/1955} & {\footnotesize{}1/10/2022} & {\footnotesize{}Monthly} & {\footnotesize{}813} & {\footnotesize 125} & {\footnotesize .96}\tabularnewline
{\footnotesize{}Switzerland Inflation} & {\footnotesize{}1/2/1955} & {\footnotesize{}1/10/2022} & {\footnotesize{}Monthly} & {\footnotesize{}813} & {\footnotesize 188} & {\footnotesize .63}\tabularnewline
{\footnotesize{}S\&P 500 Index} & {\footnotesize{}24/1/2013} & {\footnotesize{}23/1/2023} & {\footnotesize{}Daily} & {\footnotesize{}2517} & {\footnotesize 2,518} & {\footnotesize 1.00}\tabularnewline
{\footnotesize{}S\&P 500 Index} & {\footnotesize{}24/1/2013} & {\footnotesize{}23/1/2023} & {\footnotesize{}Daily} & {\footnotesize{}2517} & {\footnotesize 3,777} & {\footnotesize 1.00}\tabularnewline
{\footnotesize{}US Industrial Production} & {\footnotesize{}1/1/1919} & {\footnotesize{}1/12/2022} & {\footnotesize{}Monthly} & {\footnotesize{}1248} & {\footnotesize 368} & {\footnotesize 1.00}\tabularnewline
{\footnotesize{}US Industrial Production} & {\footnotesize{}1/1/1919} & {\footnotesize{}1/12/2022} & {\footnotesize{}Monthly} & {\footnotesize{}1248} & {\footnotesize 551} & {\footnotesize .97}\tabularnewline
\bottomrule
\end{tabular}
\end{table}
\global\long\def\thefigure{SM.\arabic{figure}}%
\captionsetup[subfigure]{position=top,font=scriptsize,singlelinecheck=off,justification=raggedright}
\begin{figure}[H]
\begin{centering}
\subfloat[Norway Real Exchange Rate, $n\widehat{h}_{us}=$ 823]{\noindent\includegraphics[scale=0.36]{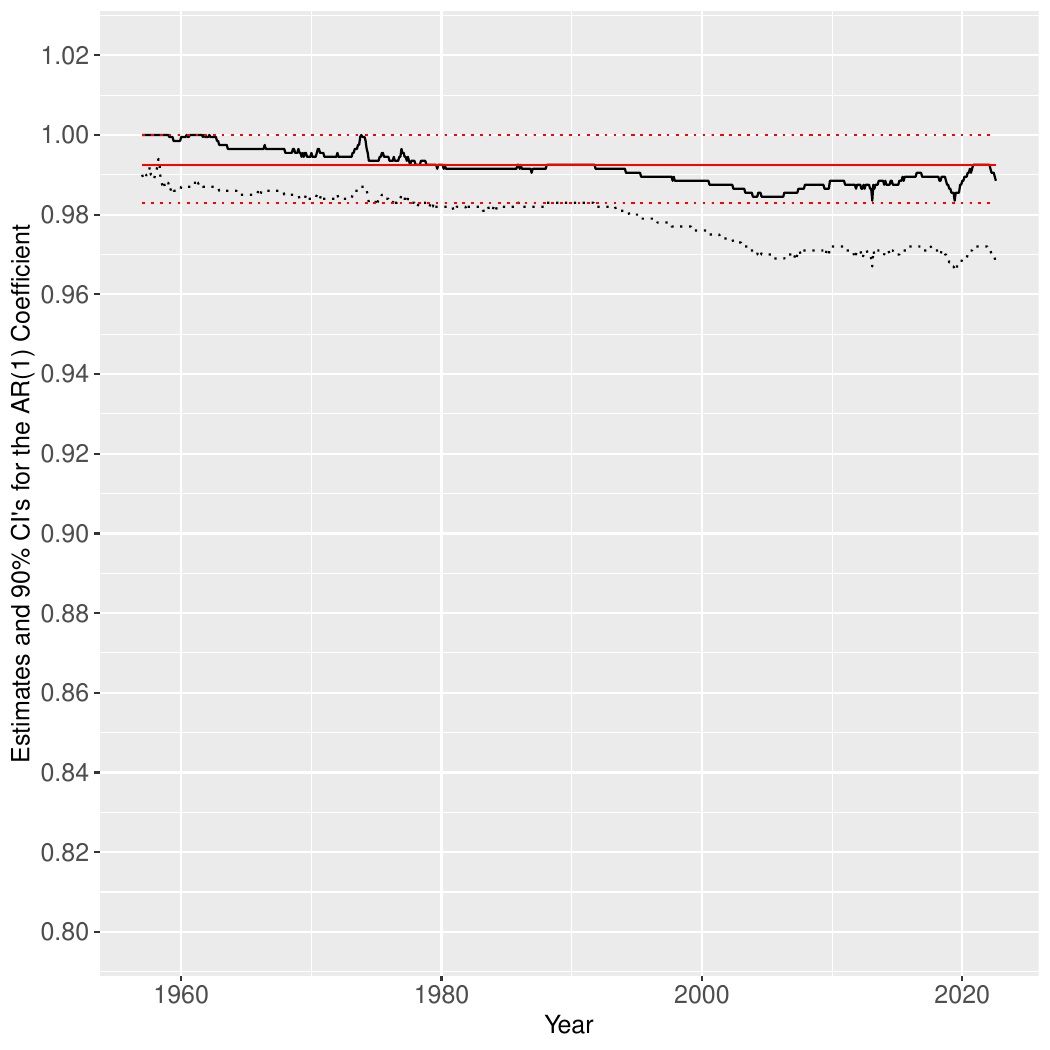}

}\quad{}\subfloat[Norway Real Exchange Rate, $1.5n\widehat{h}_{us}=$ 1,234]{\includegraphics[scale=0.36]{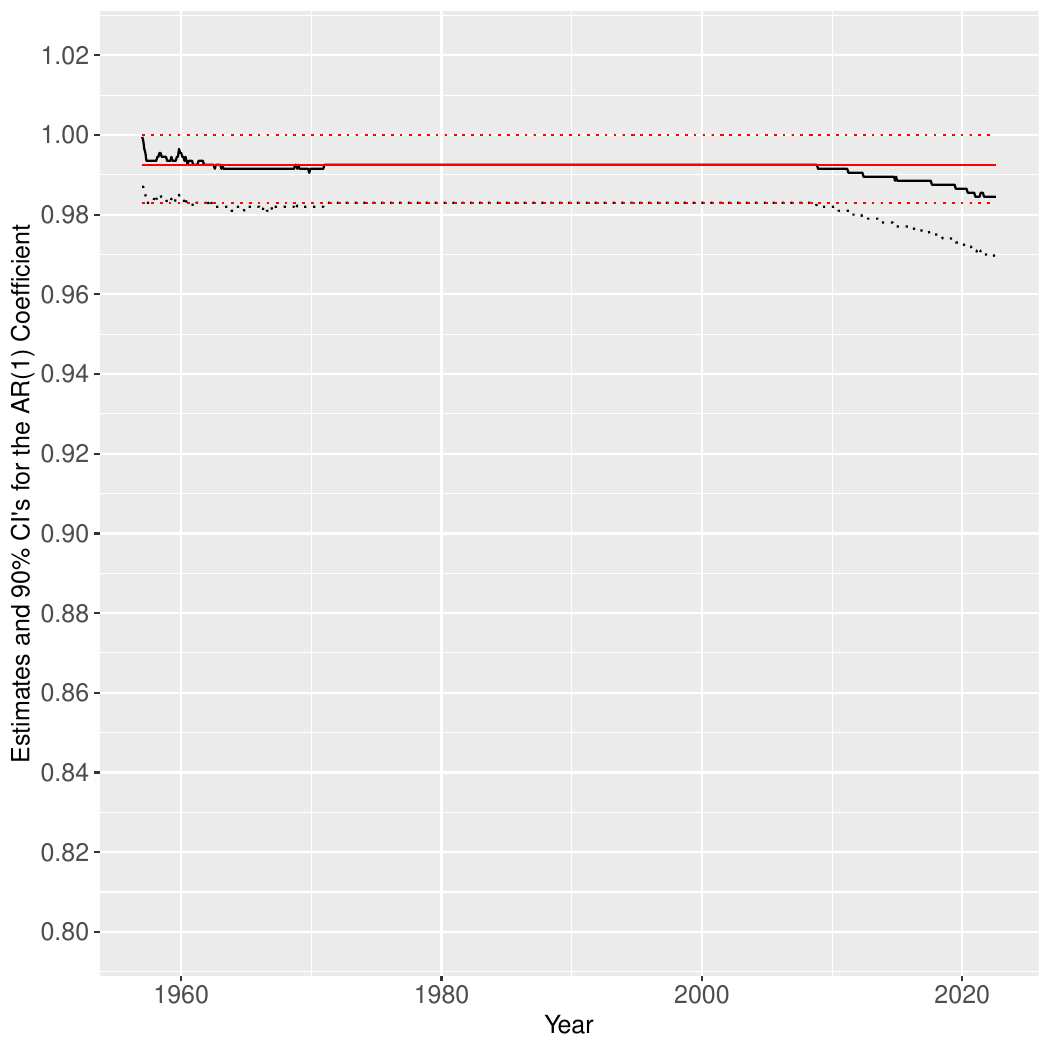}

}
\par\end{centering}
\begin{centering}
\subfloat[Canada Real Exchange Rate, $n\widehat{h}_{us}=$ 823]{\includegraphics[scale=0.36]{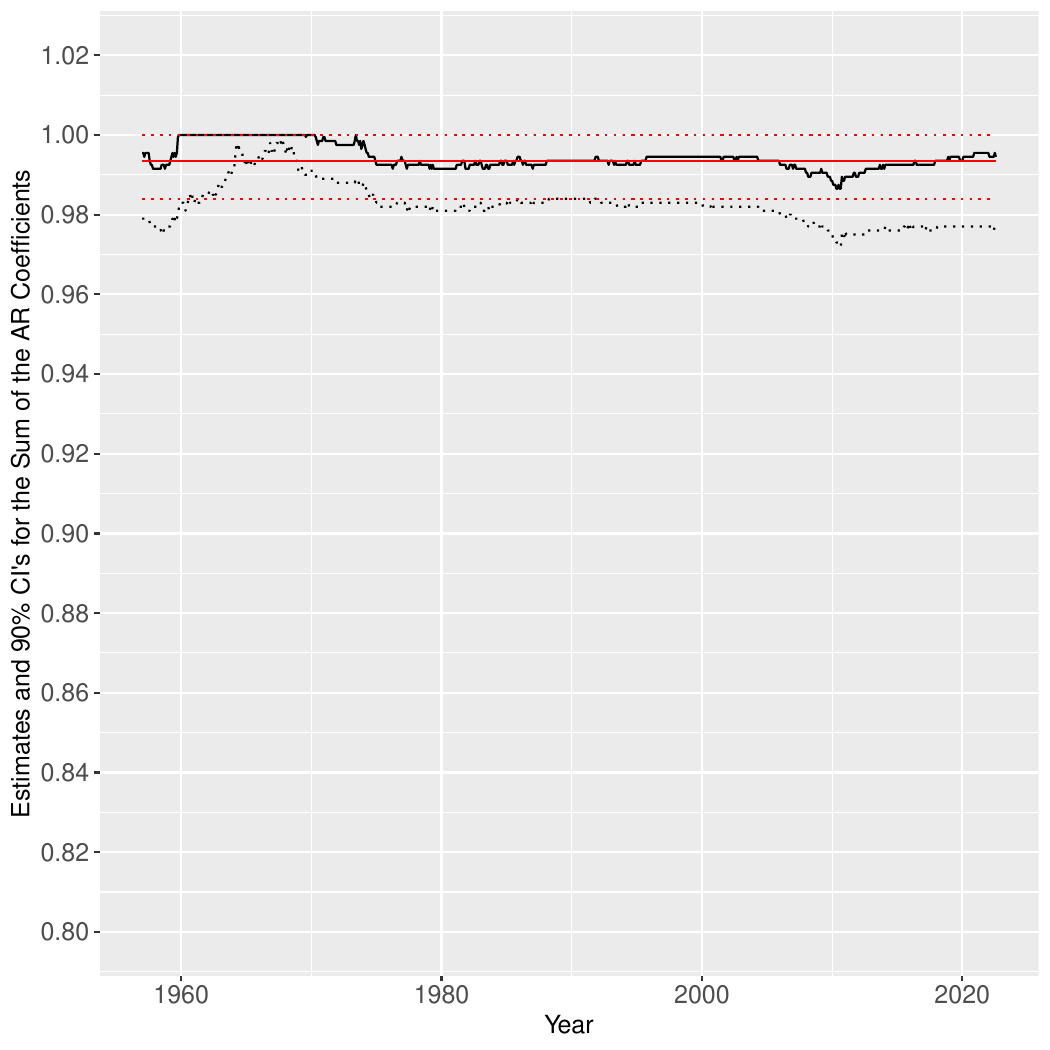}

}\quad{}\subfloat[Canada Real Exchange Rate, $1.5n\widehat{h}_{us}=$ 1,234]{\includegraphics[scale=0.36]{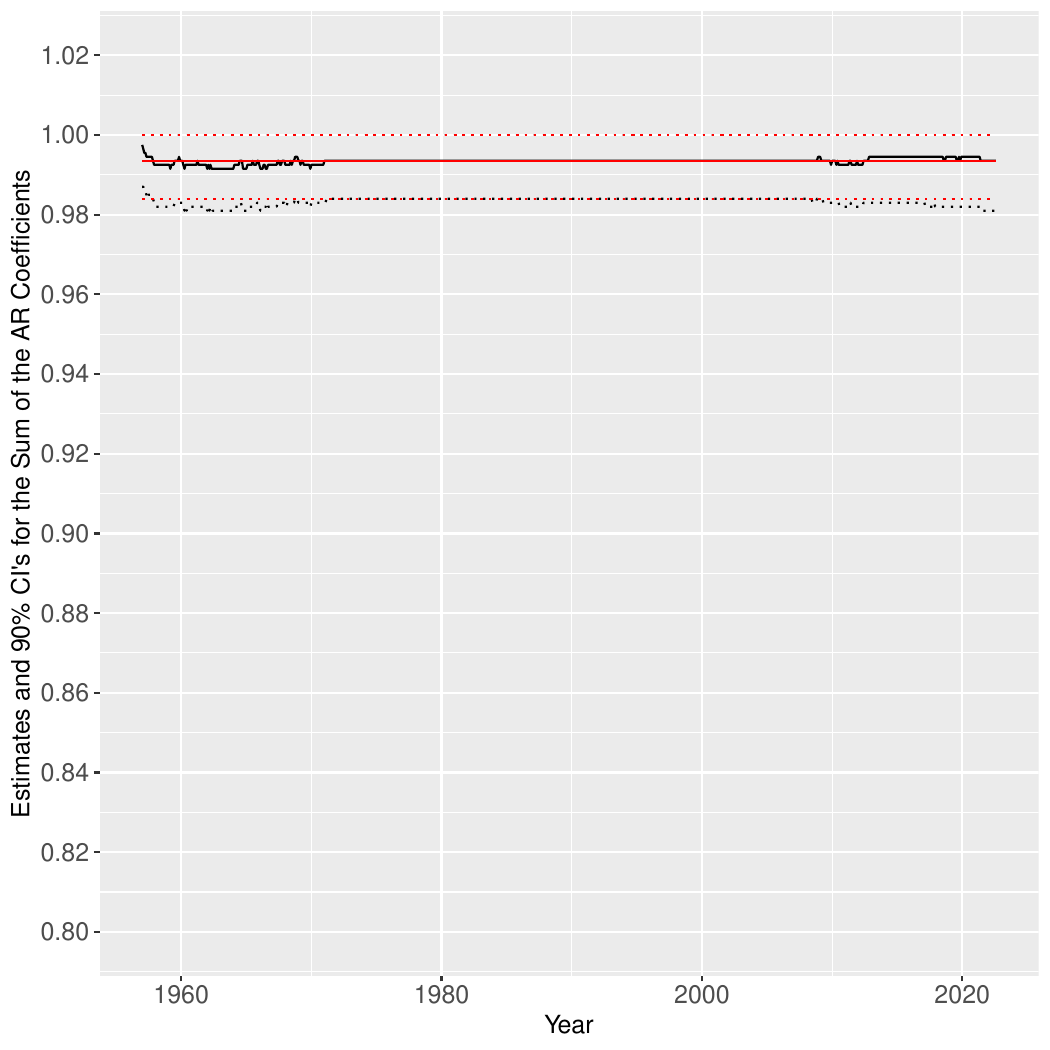}

}
\par\end{centering}
\begin{centering}
\subfloat[Japan Real Exchange Rate, $n\widehat{h}_{us}=$ 823]{\noindent\includegraphics[scale=0.36]{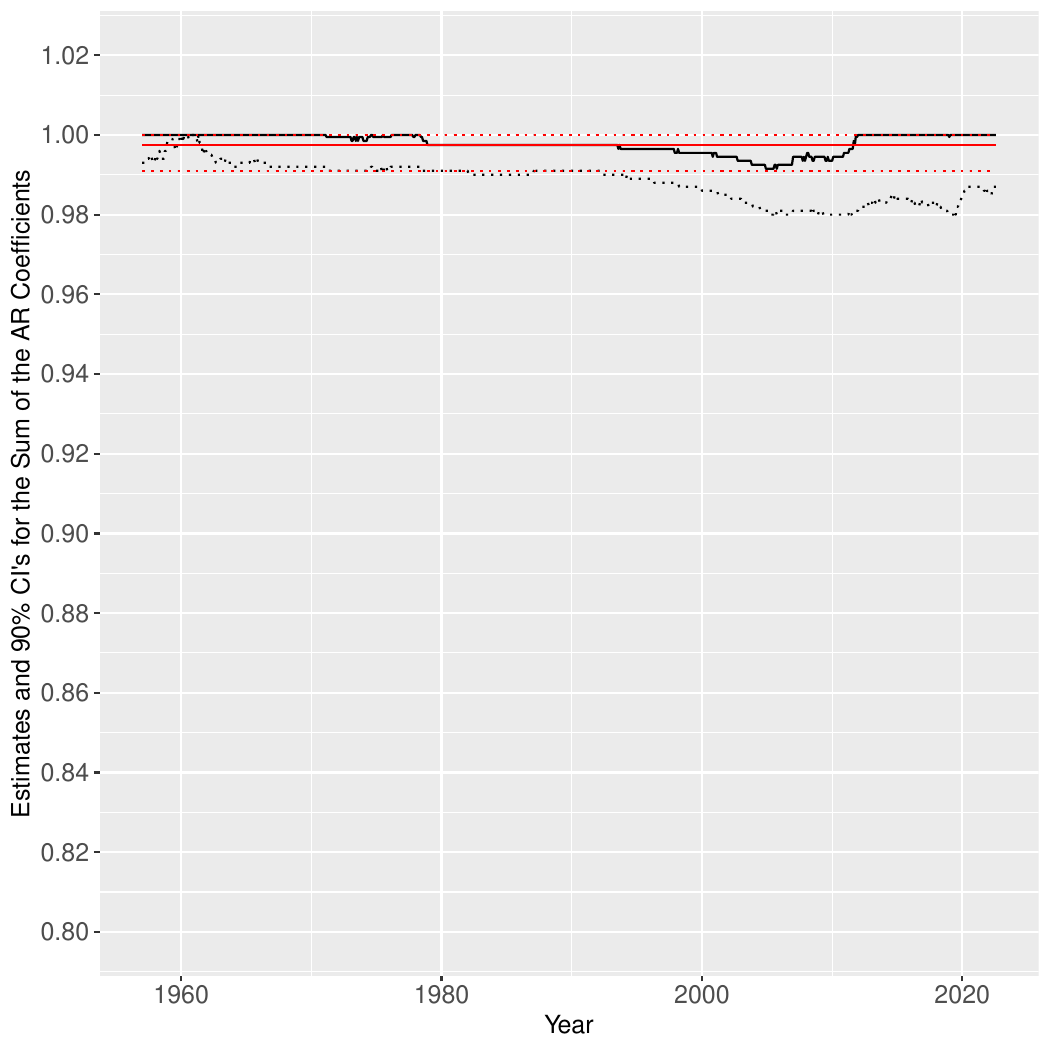}

}\quad{}\subfloat[Japan Real Exchange Rate, $1.5n\widehat{h}_{us}=$ 1,234]{\includegraphics[scale=0.36]{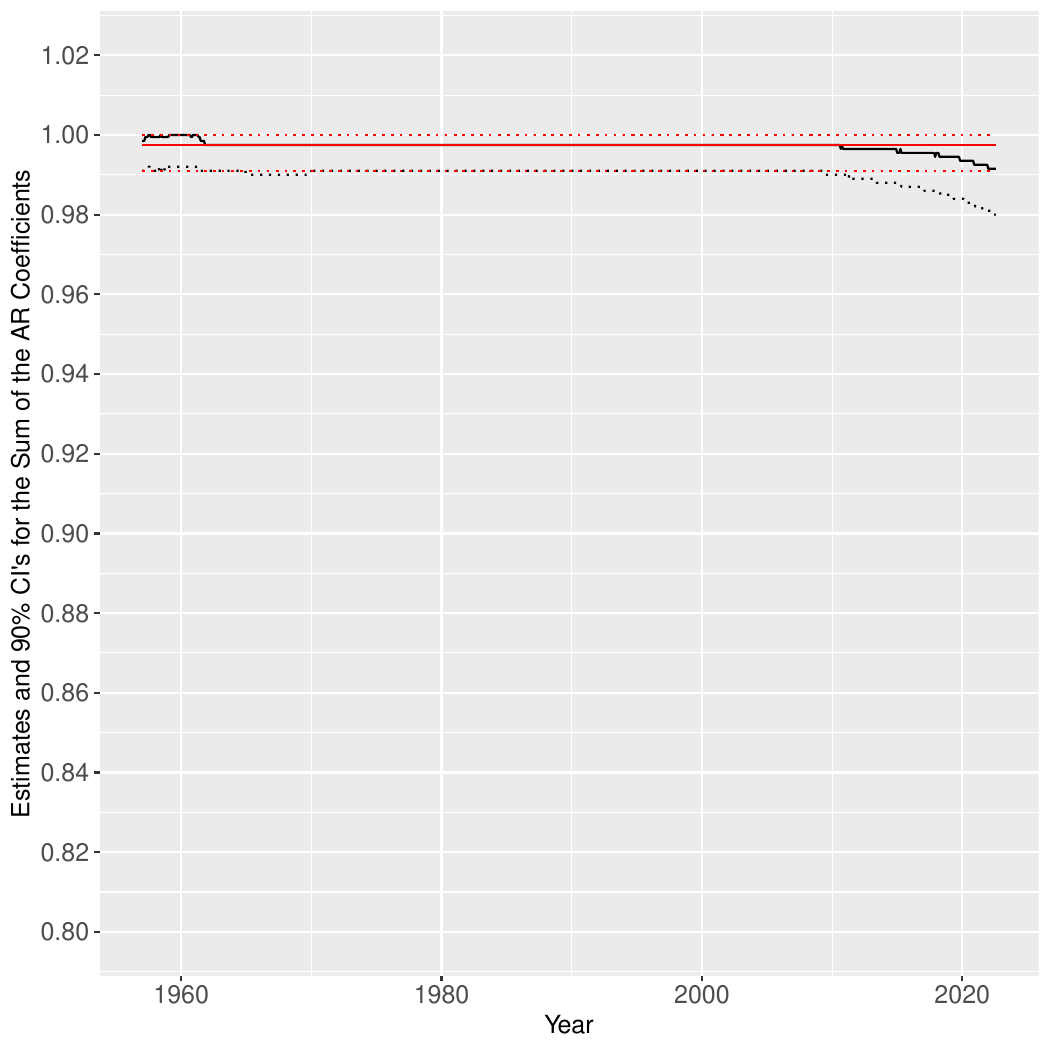}

}
\par\end{centering}
\caption{Estimates and 90\% CI's for the Sum of the AR Coefficients in TVP-AR
Models: Norway Real Exchange Rate (TVP-AR(1)), Canada and Japan Real
Exchange Rate (TVP-AR(6)) \protect\label{fig:EMP_AR6_SM1}}
\end{figure}
\begin{figure}[H]
\begin{centering}
\subfloat[Australia Interest Rate, $n\widehat{h}_{us}=$ 160]{\noindent\includegraphics[scale=0.36]{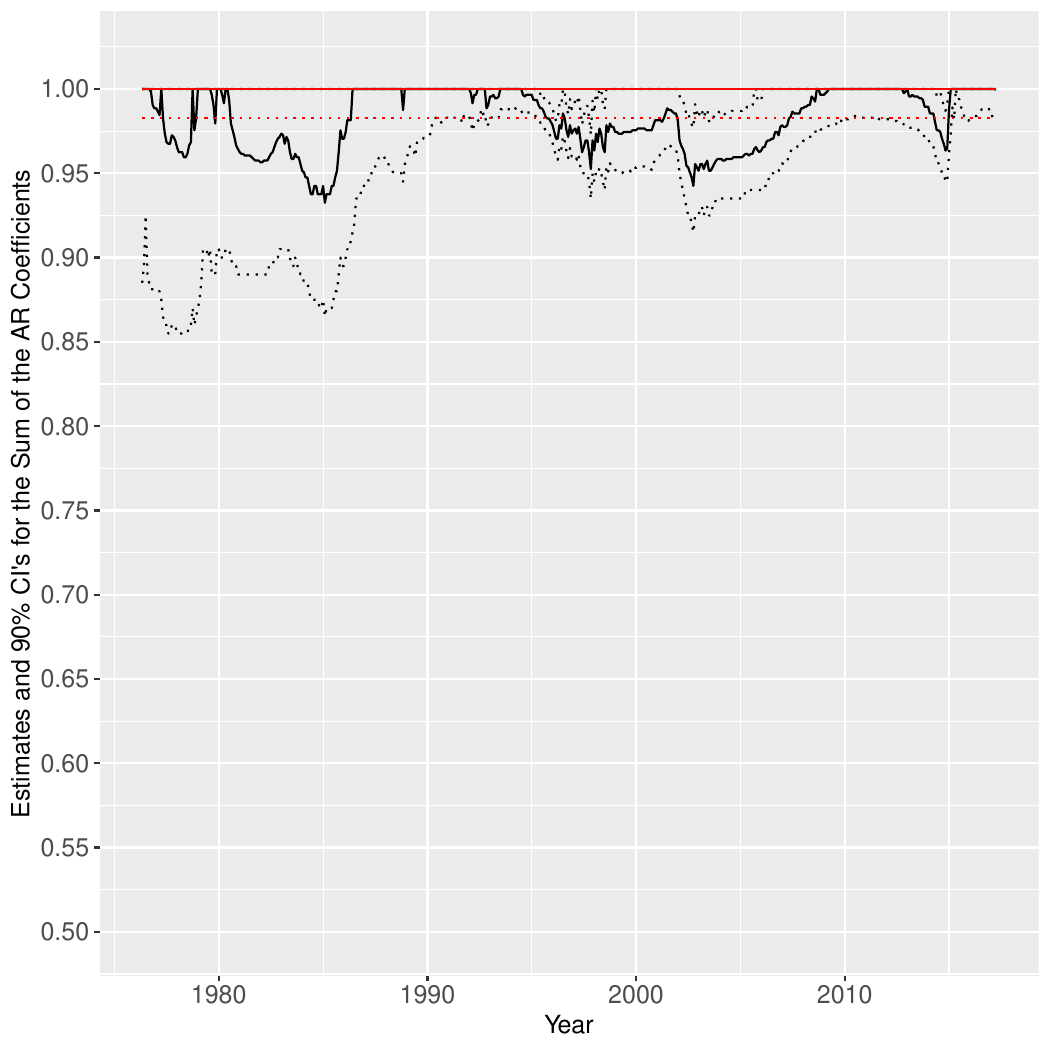}

}\quad{}\subfloat[Australia Interest Rate, $1.5n\widehat{h}_{us}=$ 240]{\includegraphics[scale=0.36]{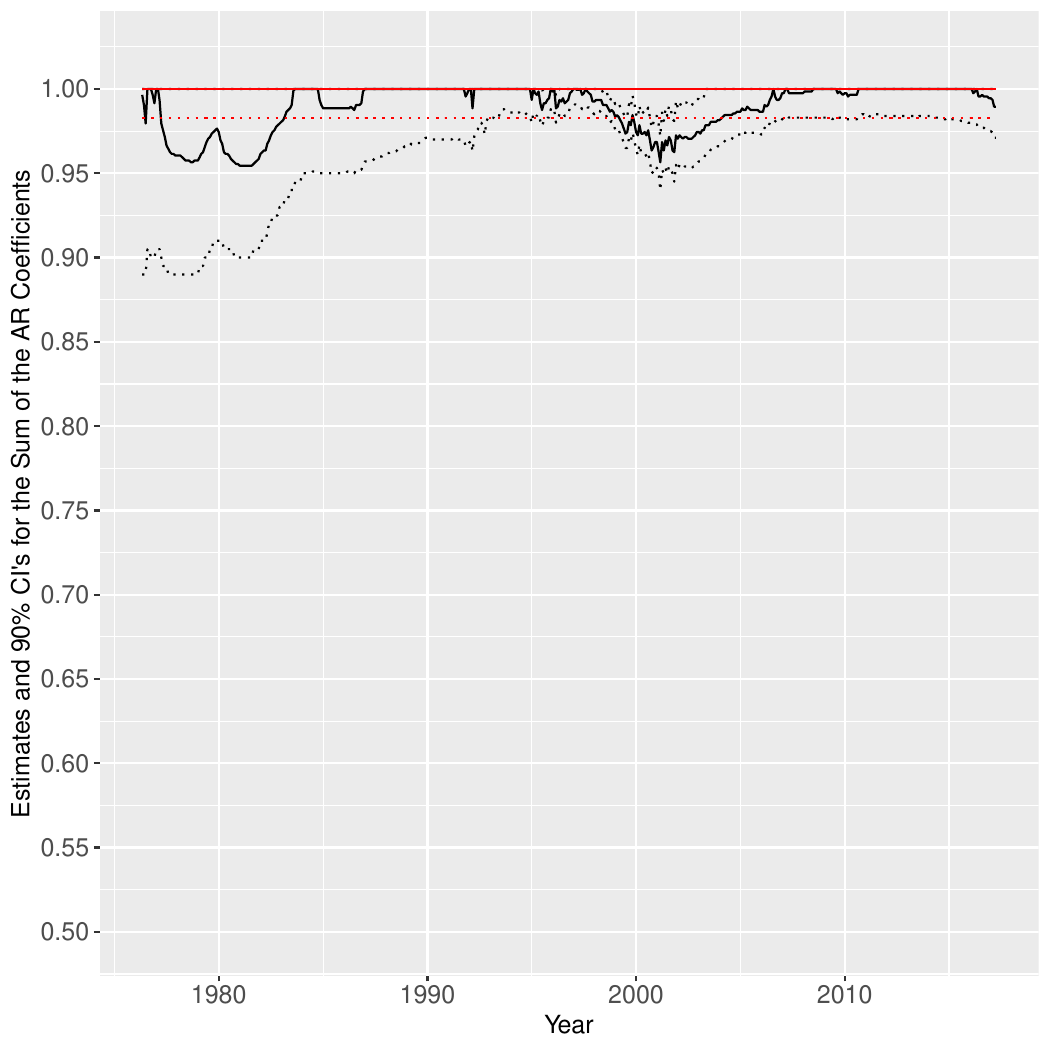}

}
\par\end{centering}
\begin{centering}
\subfloat[Canada Interest Rate, $n\widehat{h}_{us}=$ 95]{\includegraphics[scale=0.36]{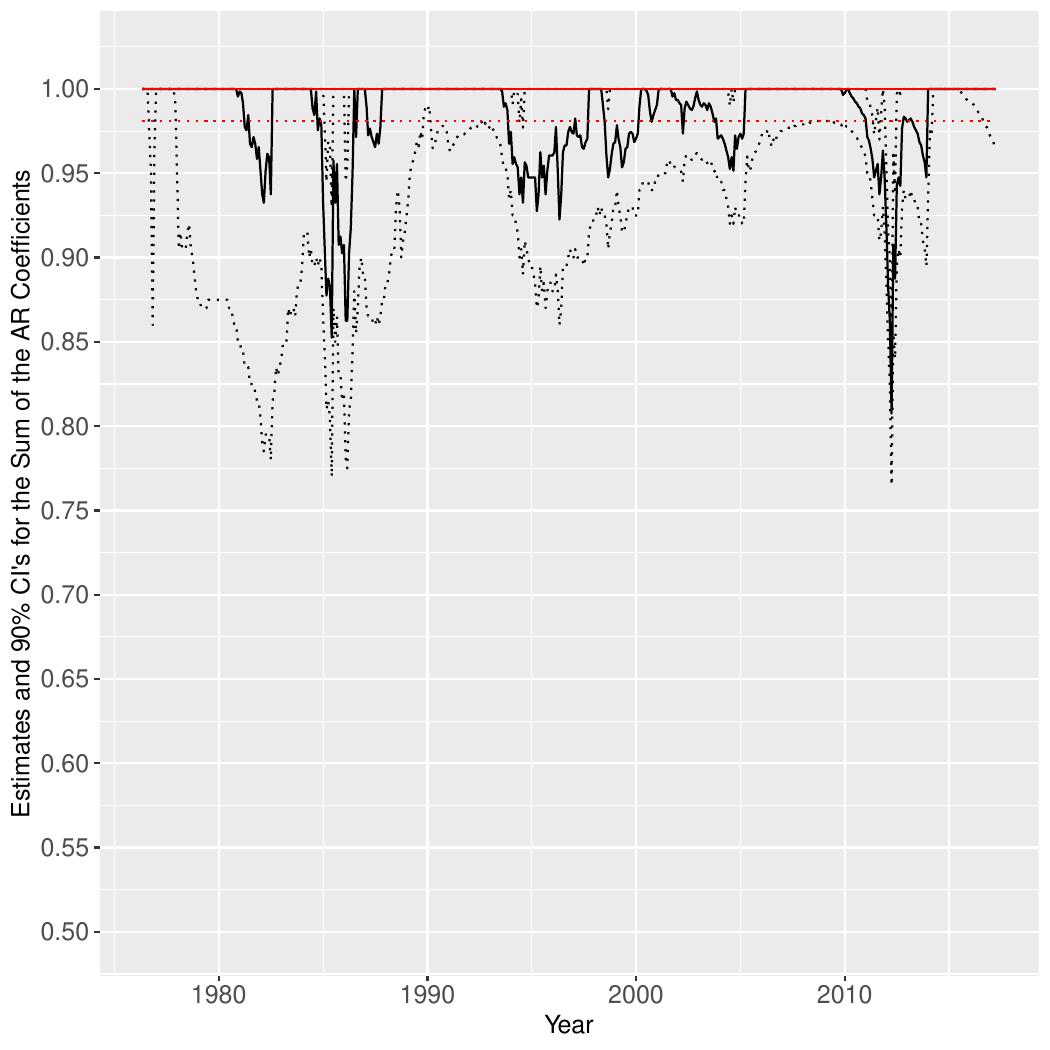}

}\quad{}\subfloat[Canada Interest Rate, $1.5n\widehat{h}_{us}=$ 143]{\includegraphics[scale=0.36]{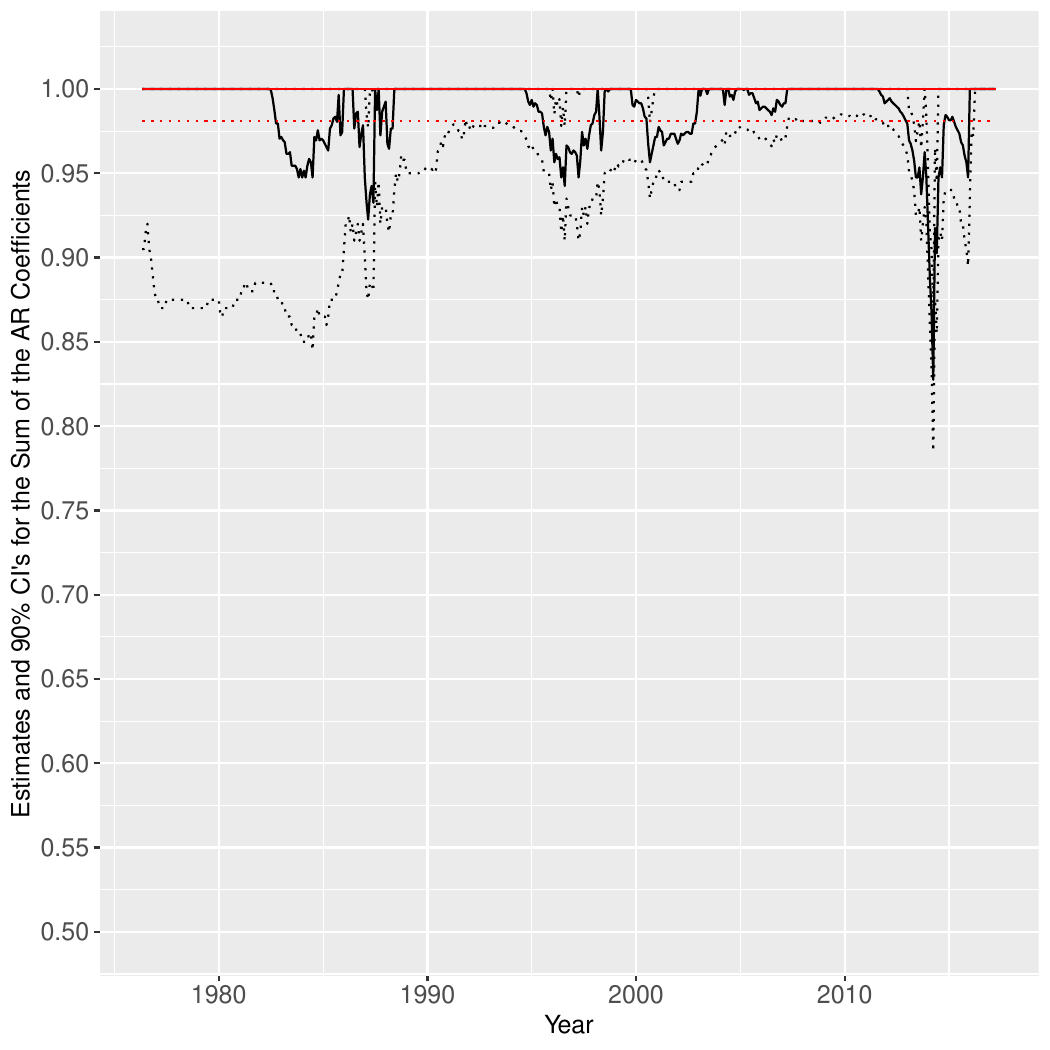}

}
\par\end{centering}
\begin{centering}
\subfloat[US Interest Rate, $n\widehat{h}_{us}=$ 103]{\includegraphics[scale=0.36]{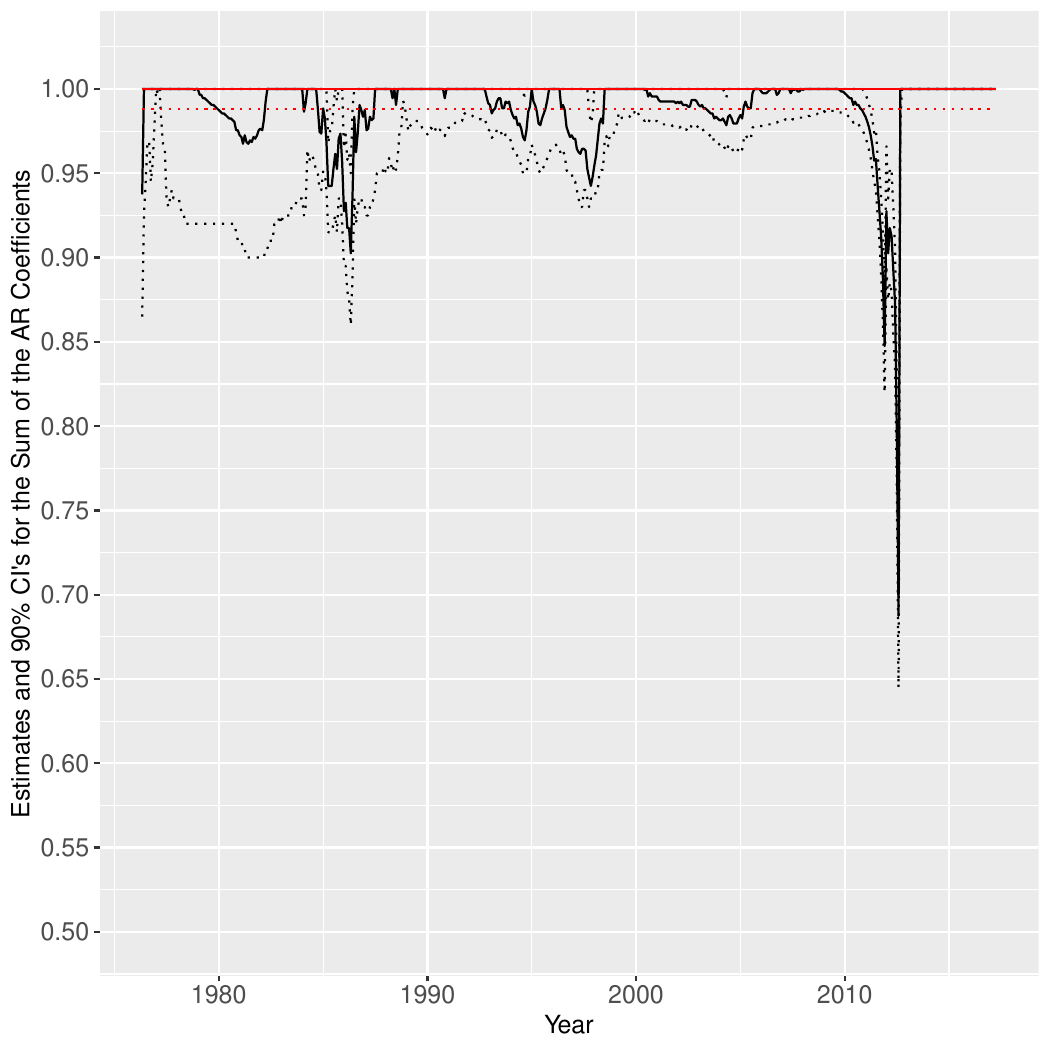}

}\quad{}\subfloat[US Interest Rate, $1.5n\widehat{h}_{us}=$ 155]{\includegraphics[scale=0.36]{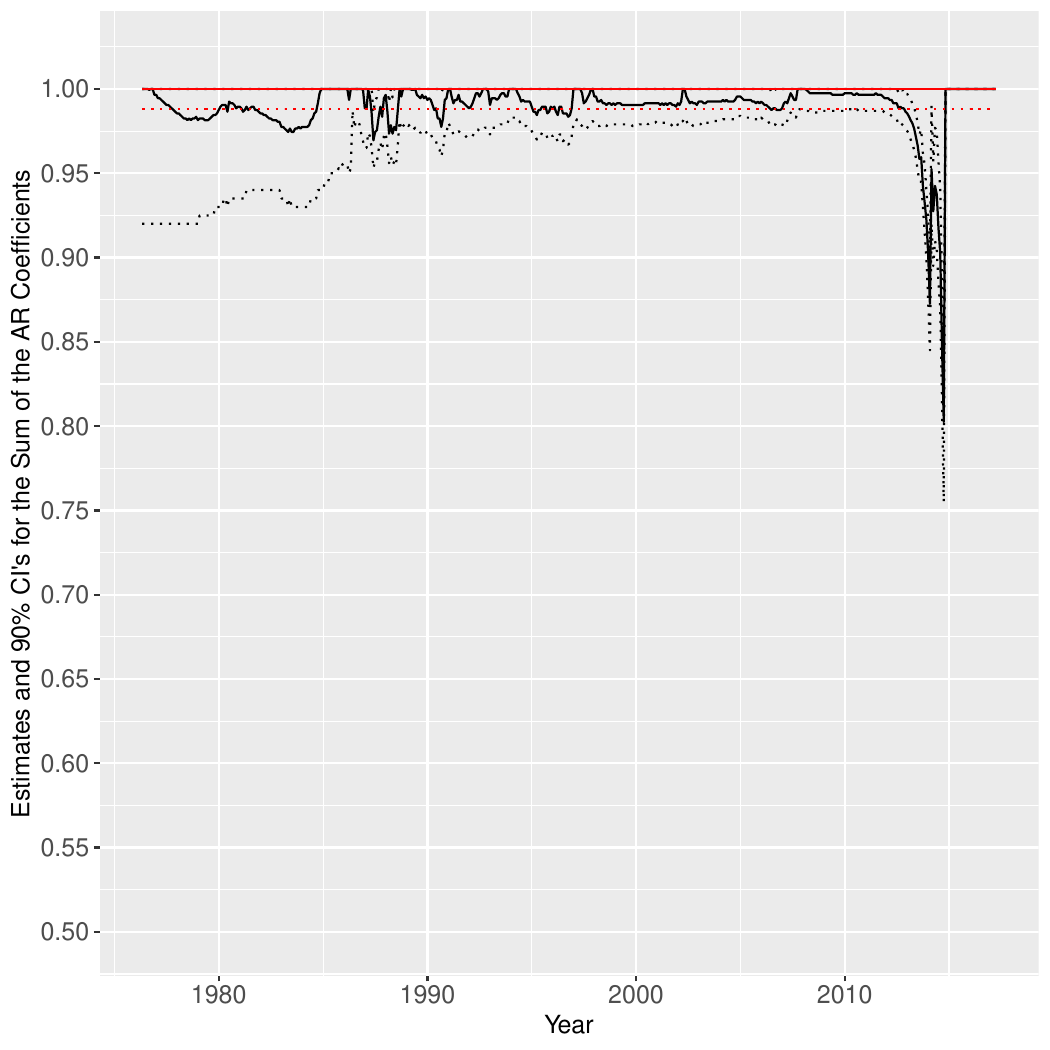}

}
\par\end{centering}
\caption{Estimates and 90\% CI's for the Sum of the AR Coefficients in TVP-AR(6)
Models: Australia, Canada, and the US Interest Rate\protect\label{fig:EMP_AR6_SM2}}
\end{figure}
\begin{figure}[H]
\begin{centering}
\subfloat[US 10yr Bond Yield, $n\widehat{h}_{us}=$ 16,006]{\includegraphics[scale=0.36]{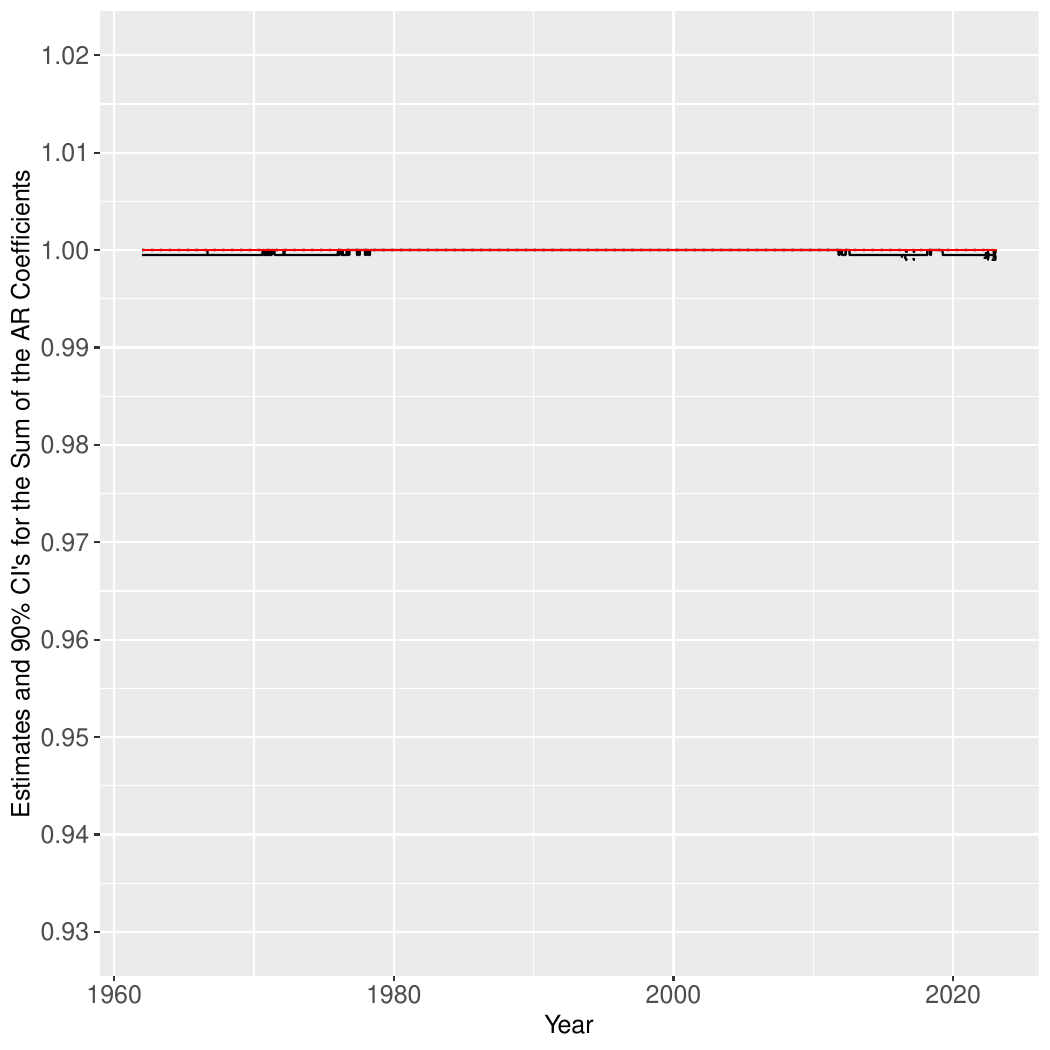}

}\quad{}\subfloat[US 10yr Bond Yield, $1.5n\widehat{h}_{us}=$ 24,009]{\includegraphics[scale=0.36]{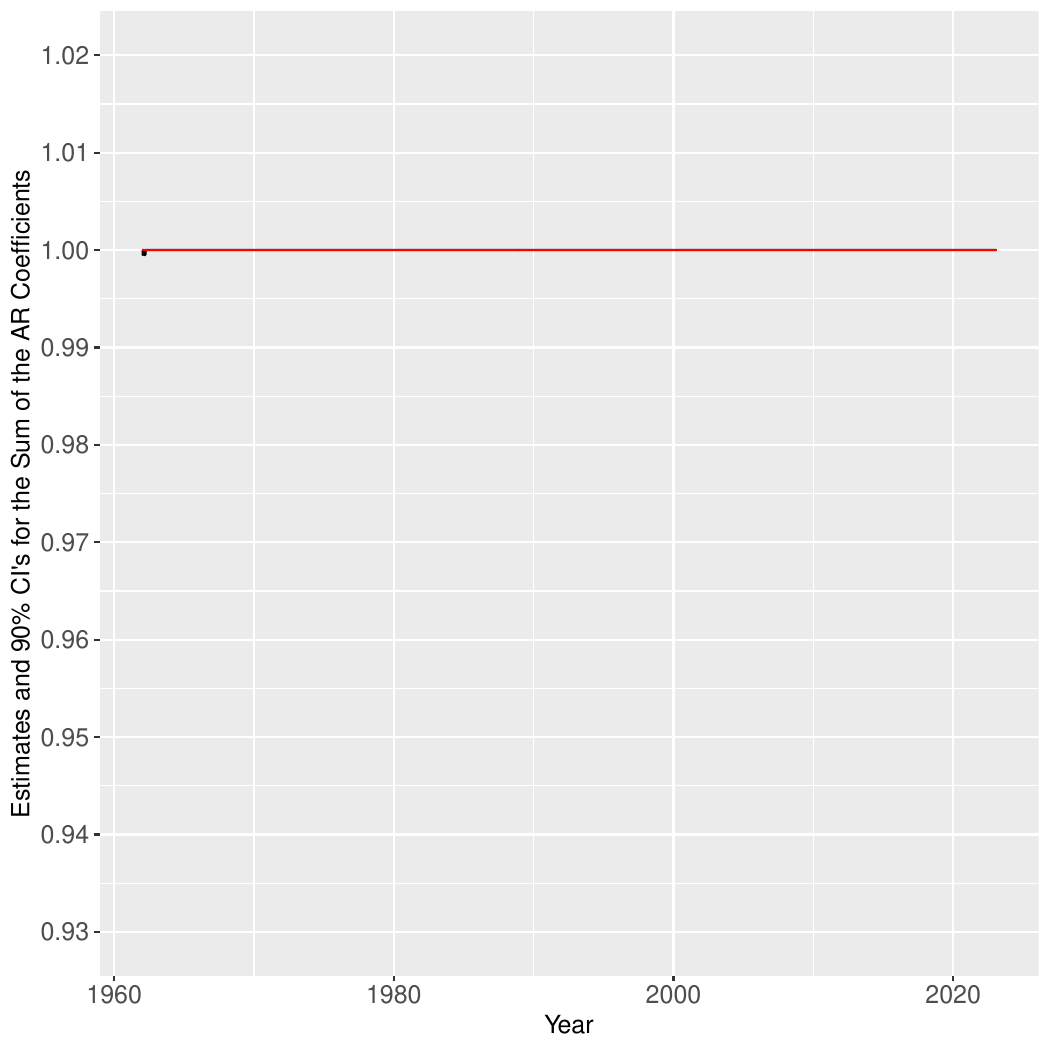}

}
\par\end{centering}
\begin{centering}
\subfloat[US Avg Wages Manufacturing, $n\widehat{h}_{us}=$ 152]{\includegraphics[scale=0.36]{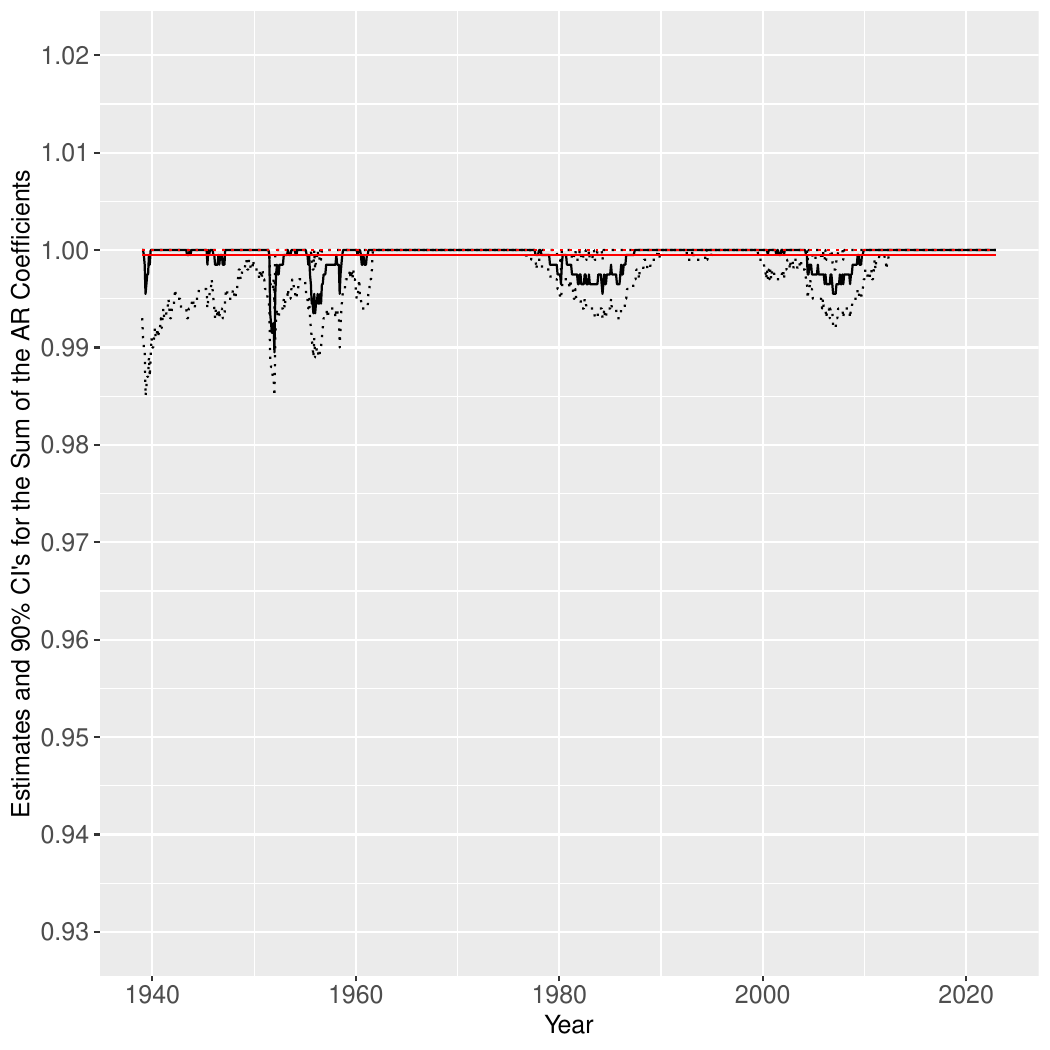}

}\quad{}\subfloat[US Avg Wages Manufacturing, $1.5n\widehat{h}_{us}=$ 228]{\includegraphics[scale=0.36]{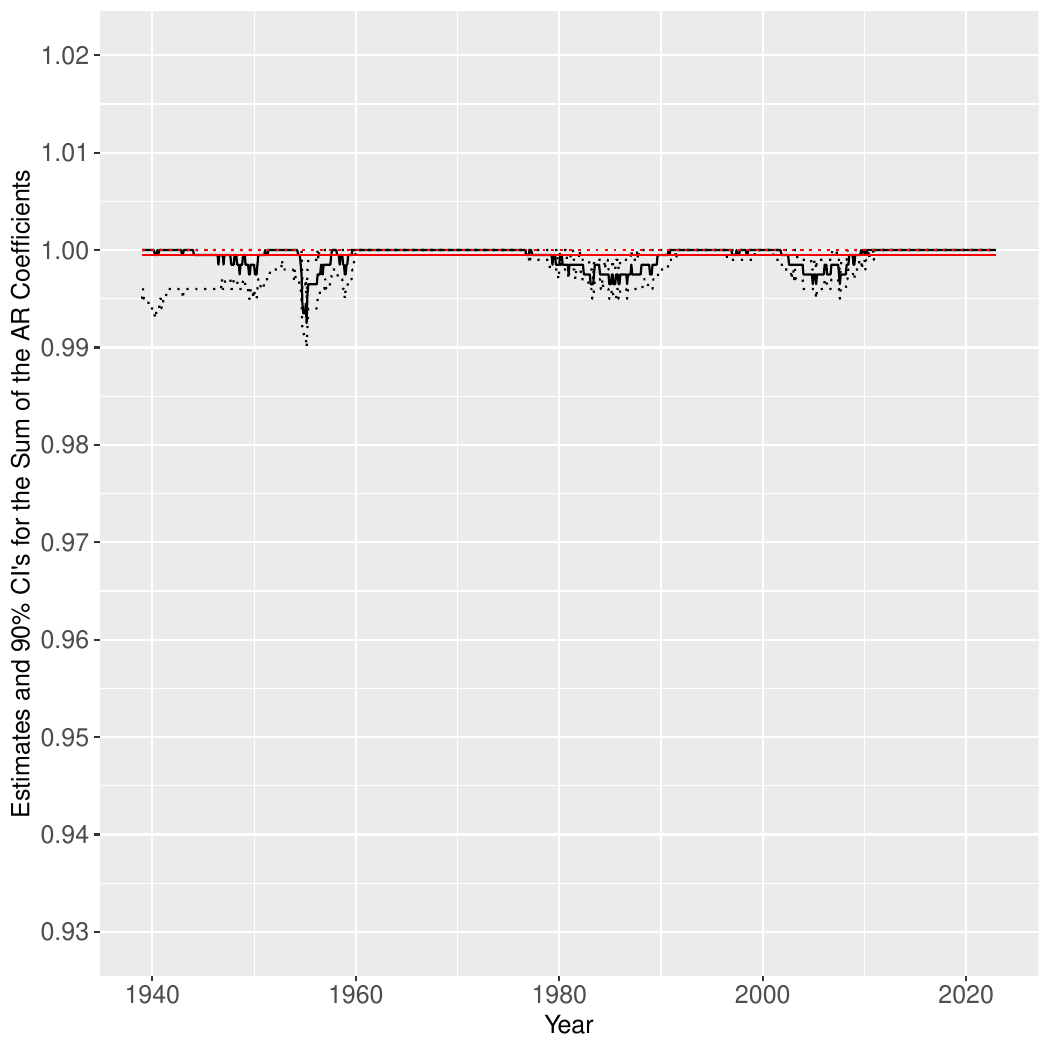}

}
\par\end{centering}
\begin{centering}
\subfloat[US Unemployment Rate, $n\widehat{h}_{us}=$ 900]{\includegraphics[scale=0.36]{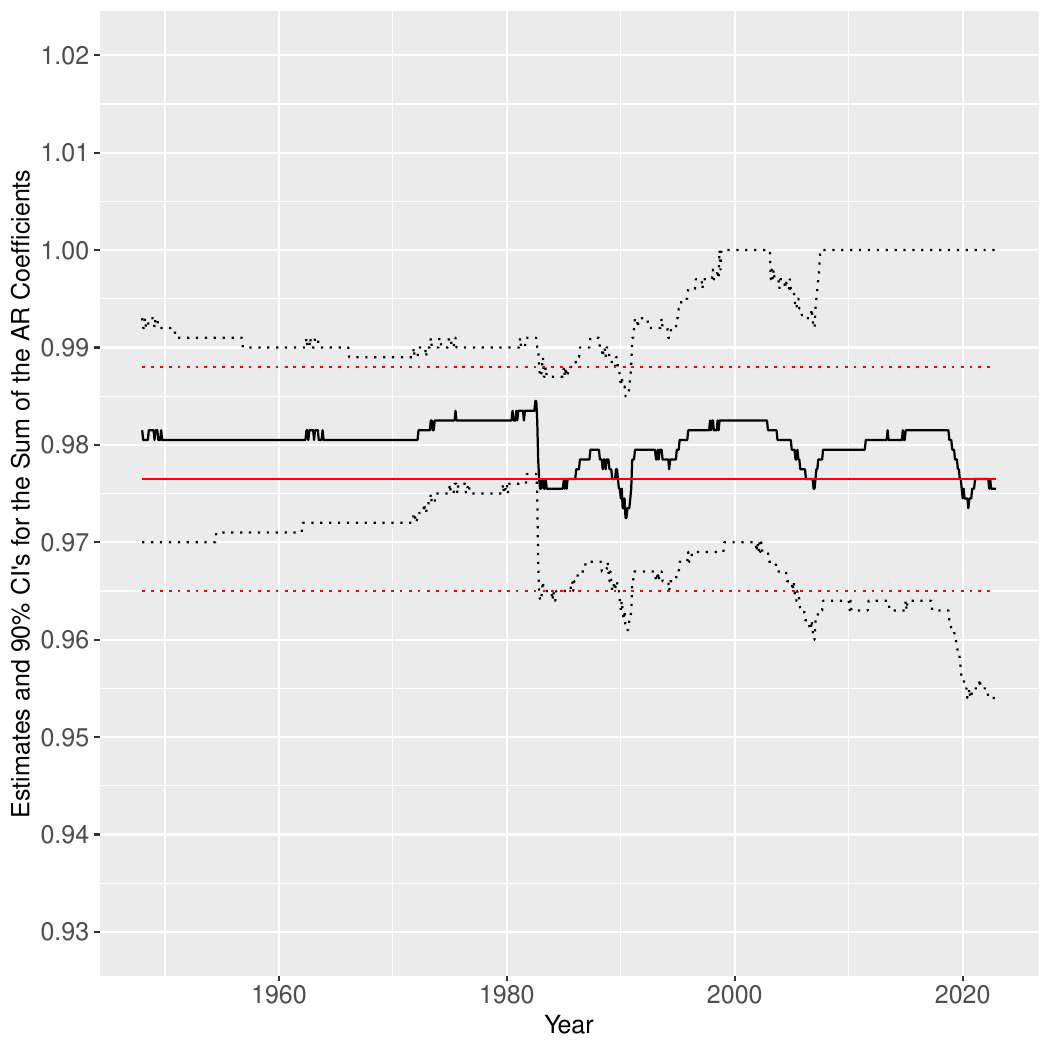}

}\quad{}\subfloat[US Unemployment Rate, $1.5n\widehat{h}_{us}=$ 1,350]{\includegraphics[scale=0.36]{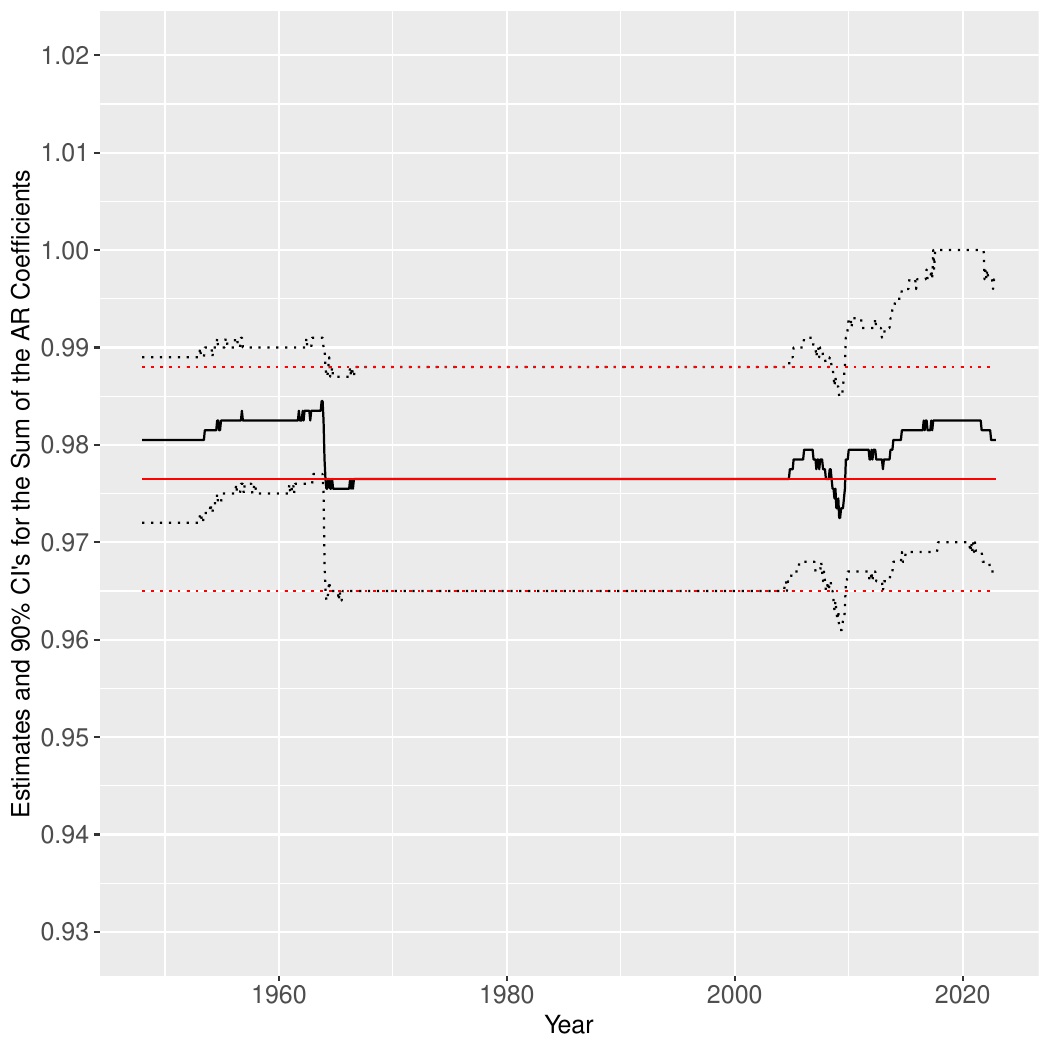}

}
\par\end{centering}
\caption{Estimates and 90\% CI's for the Sum of the AR Coefficients in TVP-AR(6)
Models: US 10yr Bond Yield, US Average Wages Manufacturing, and US
Unemployment Rate\protect\label{fig:EMP_AR6_SM3}}
\end{figure}
\begin{figure}[H]
\begin{centering}
\subfloat[US Real GDP Per Capita, $n\widehat{h}_{us}=$ 317]{\includegraphics[scale=0.36]{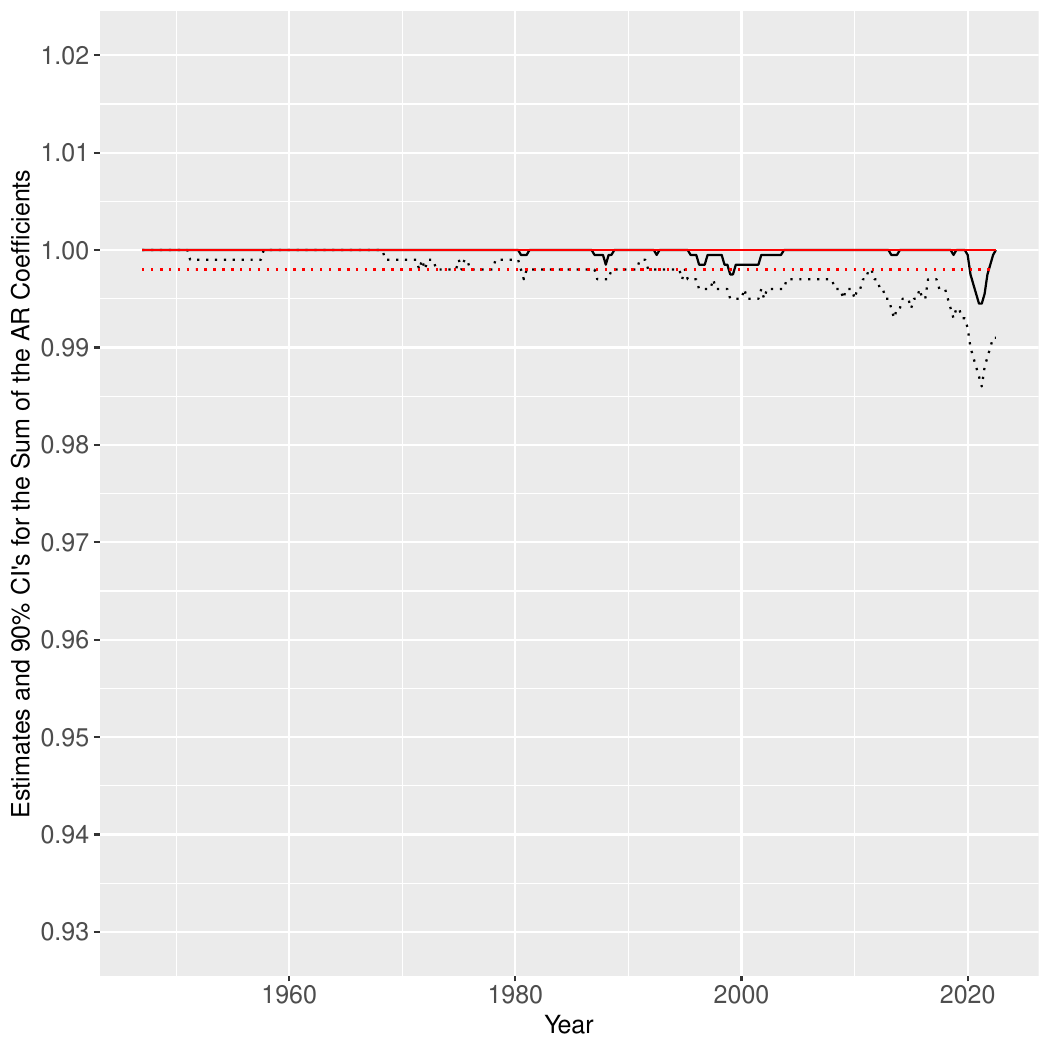}

}\quad{}\subfloat[US Real GDP Per Capita, $1.5n\widehat{h}_{us}=$ 476]{\includegraphics[scale=0.36]{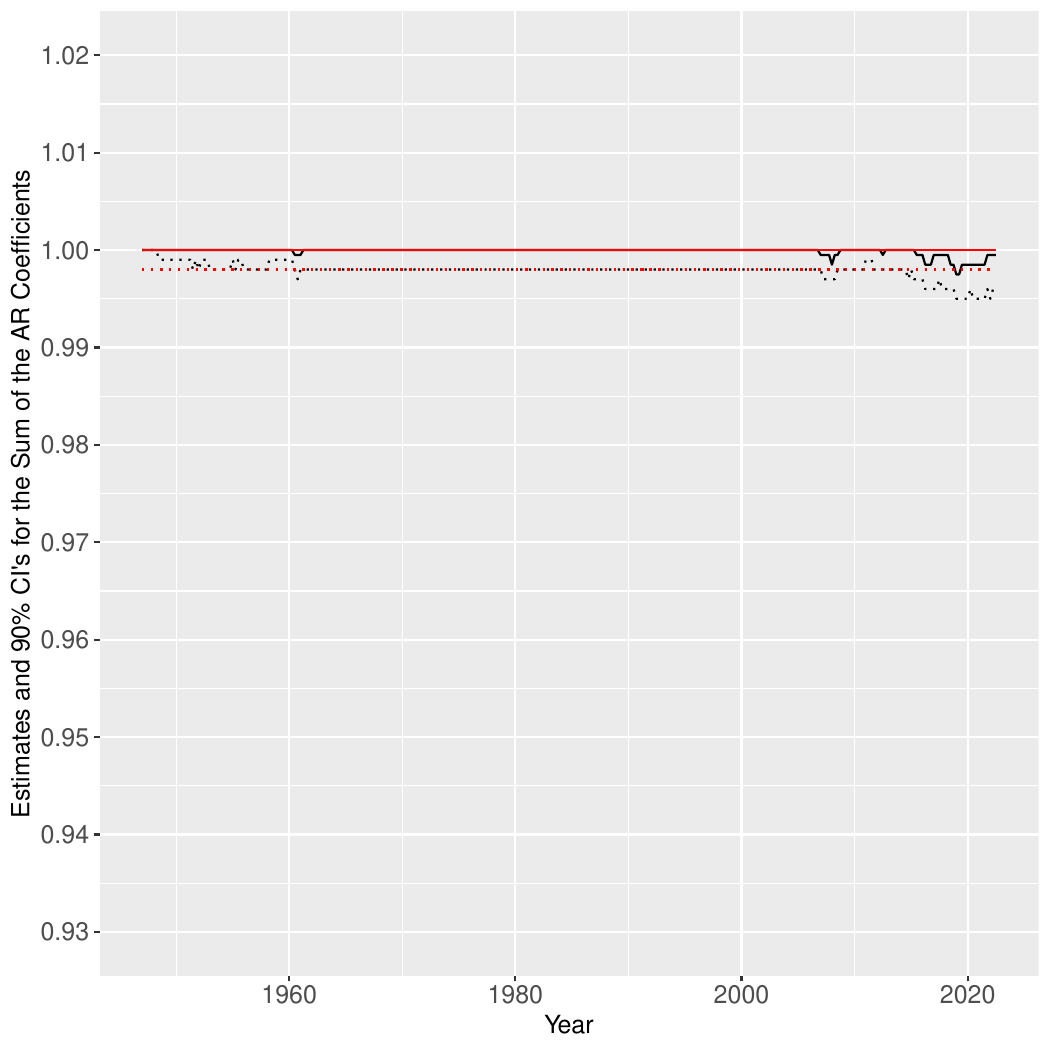}

}
\par\end{centering}
\begin{centering}
\subfloat[US Real GNP, $n\widehat{h}_{us}=$ 317]{\includegraphics[scale=0.36]{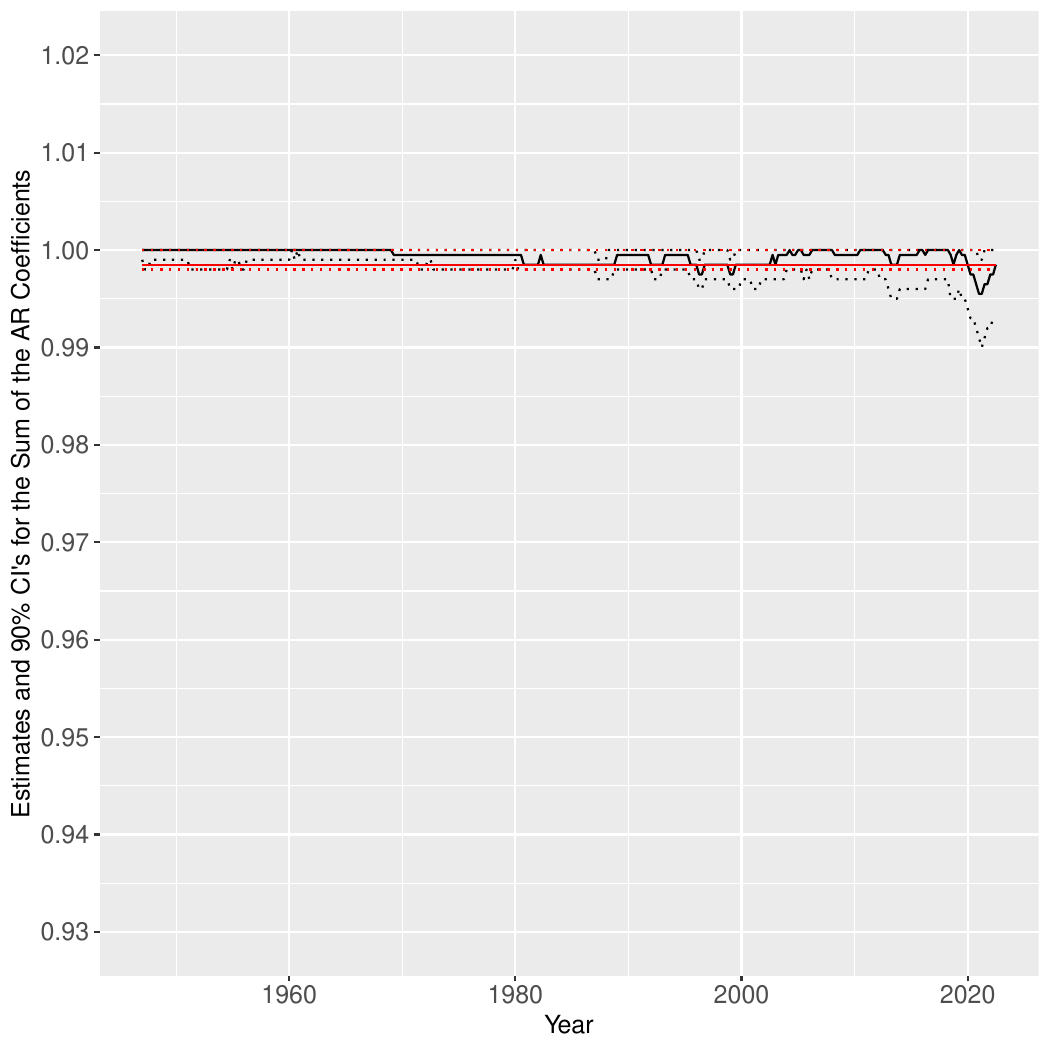}

}\quad{}\subfloat[US Real GNP, $1.5n\widehat{h}_{us}=$ 476]{\includegraphics[scale=0.36]{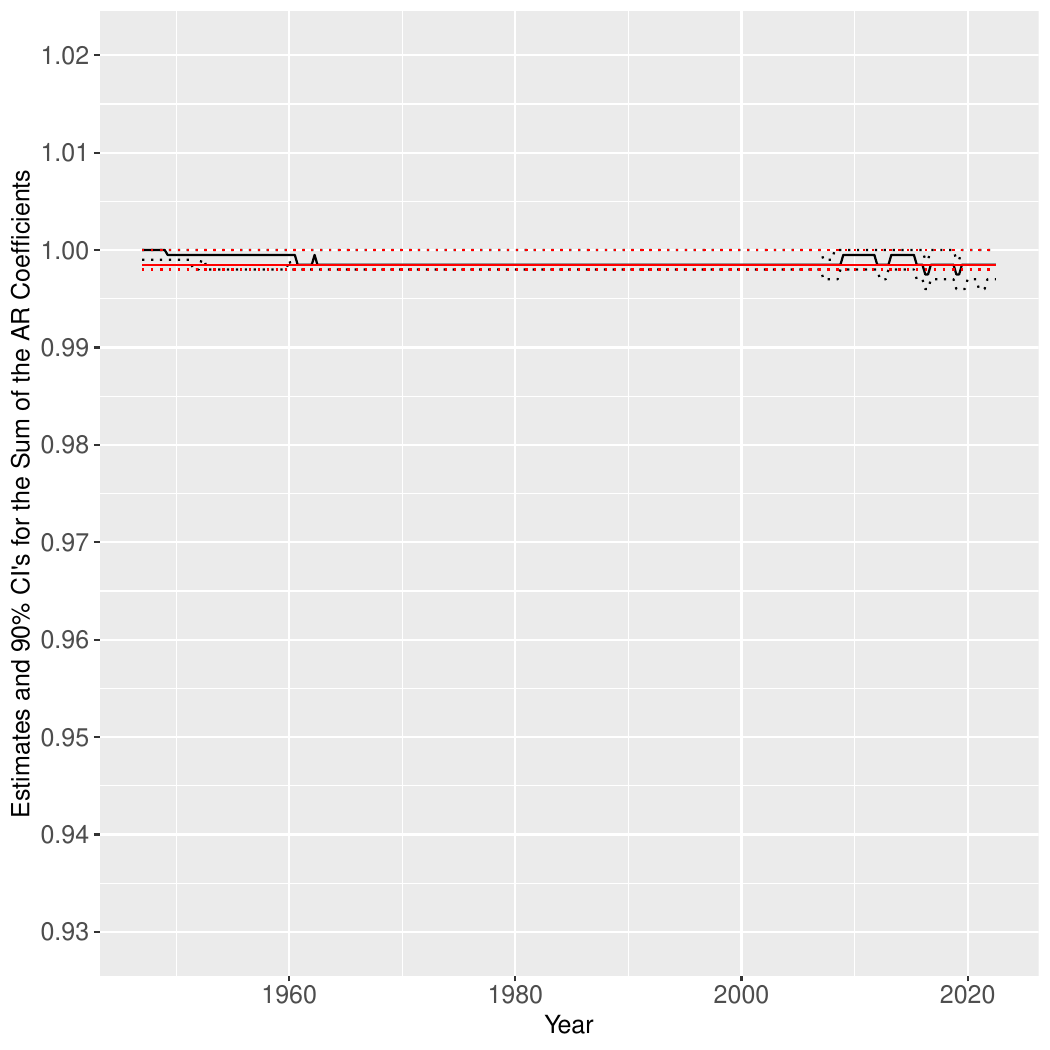}

}
\par\end{centering}
\begin{centering}
\subfloat[US Real GNP Per Capita, $n\widehat{h}_{us}=$ 317]{\includegraphics[scale=0.36]{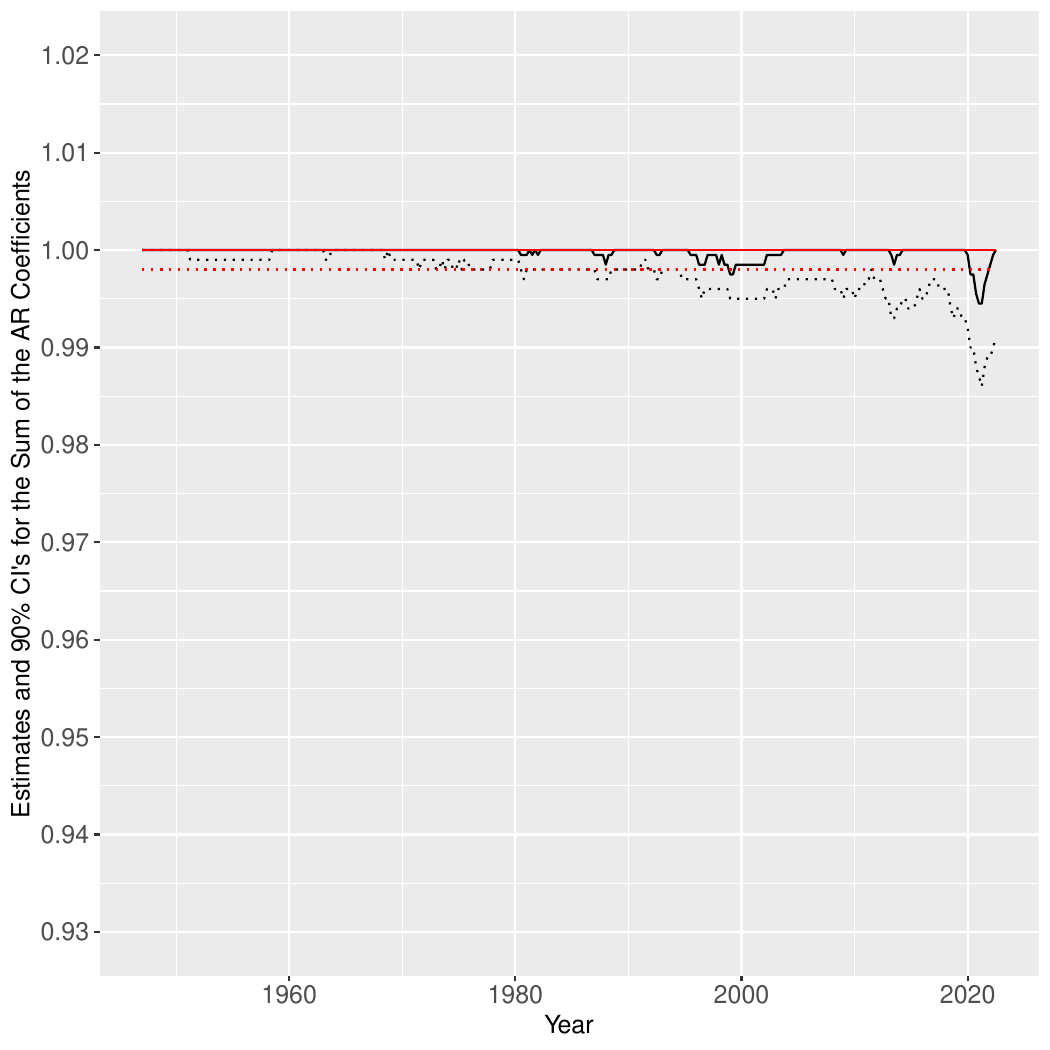}

}\quad{}\subfloat[US Real GNP Per Capita, $1.5n\widehat{h}_{us}=$ 476]{\includegraphics[scale=0.36]{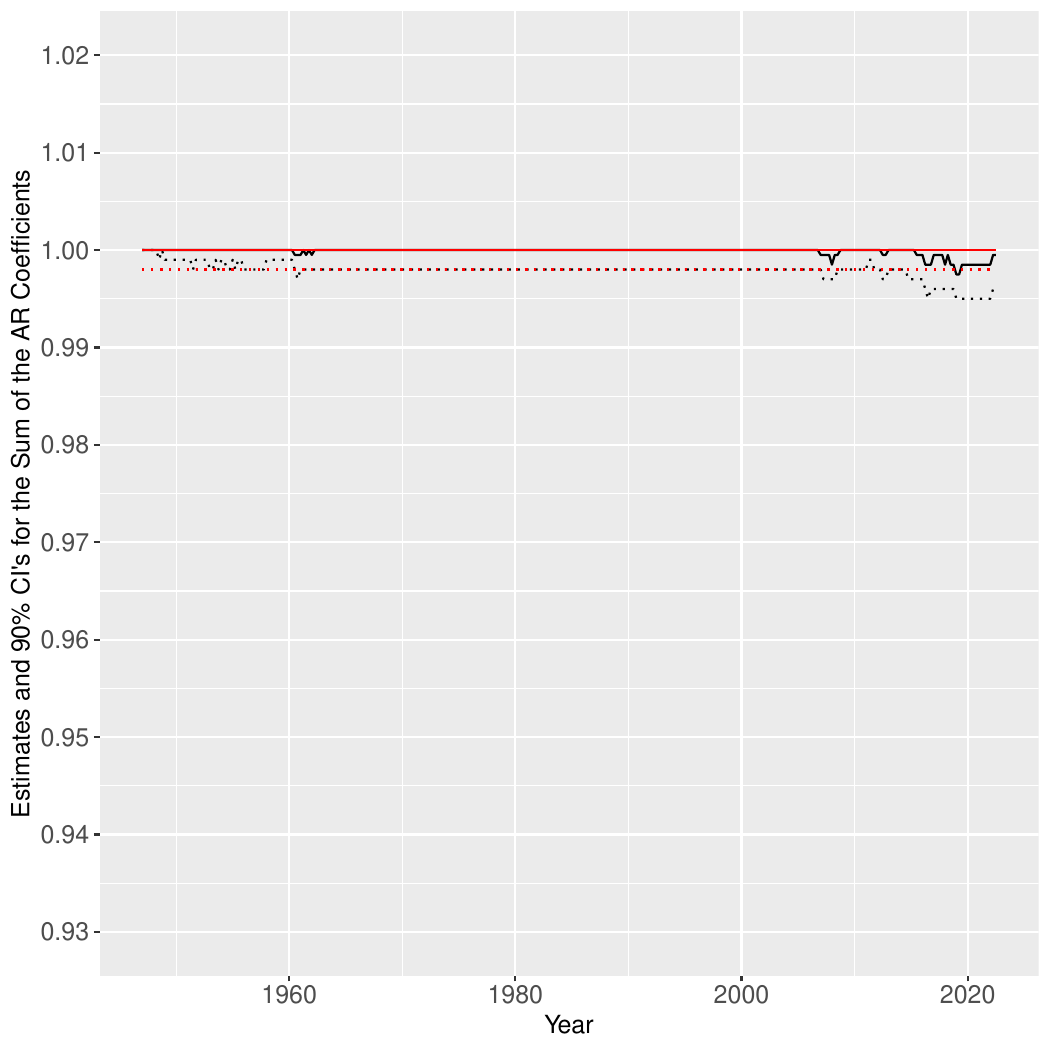}

}
\par\end{centering}
\caption{Estimates and 90\% CI's for the Sum of the AR Coefficients in TVP-AR(6)
Models: US Real GDP Per Capita, US Real GNP, and US Real GNP Per Capita
\protect\label{fig:EMP_AR6_SM4}}
\end{figure}
\begin{figure}[H]
\begin{centering}
\subfloat[Switzerland Inflation, $n\widehat{h}_{us}=$ 125]{\noindent\includegraphics[scale=0.36]{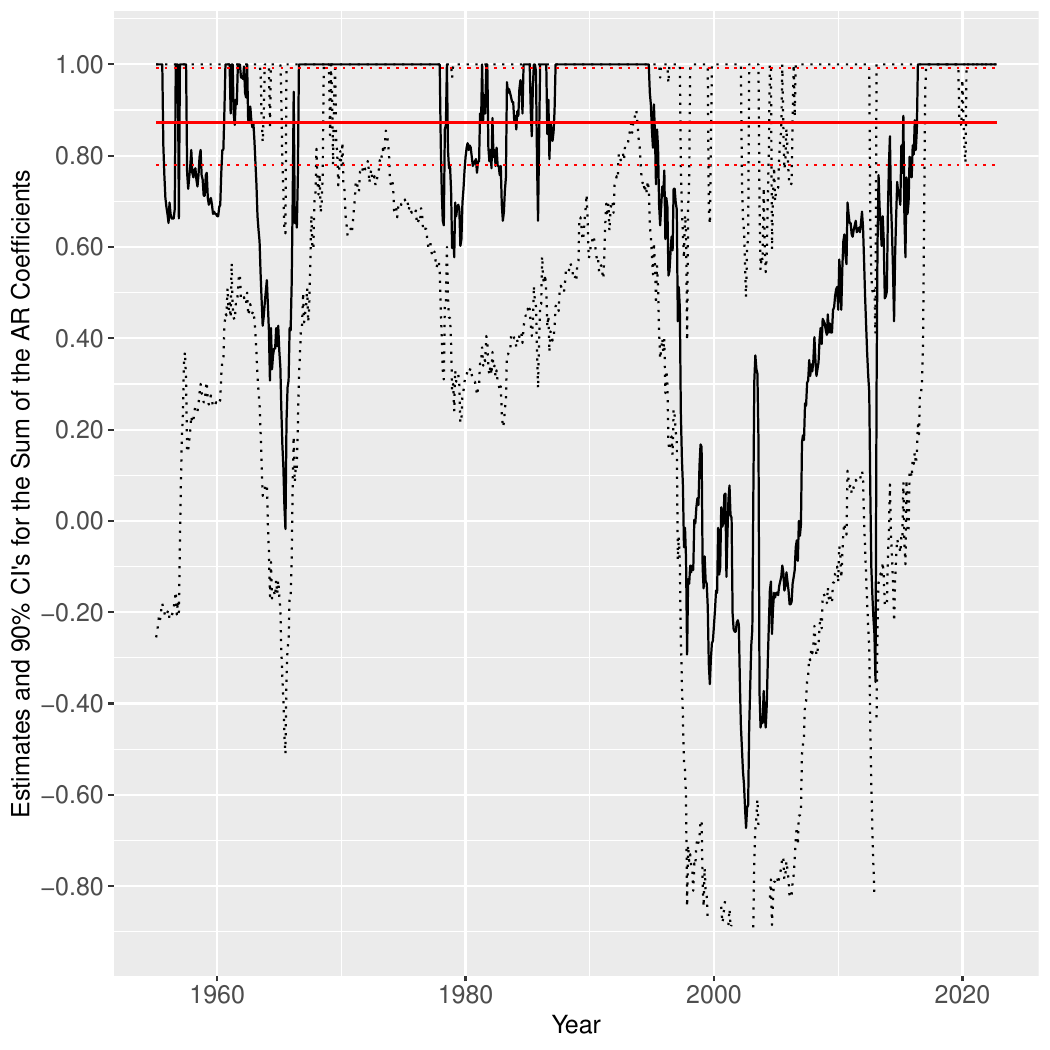}

}\quad{}\subfloat[Switzerland Inflation, $1.5n\widehat{h}_{us}=$ 188]{\includegraphics[scale=0.36]{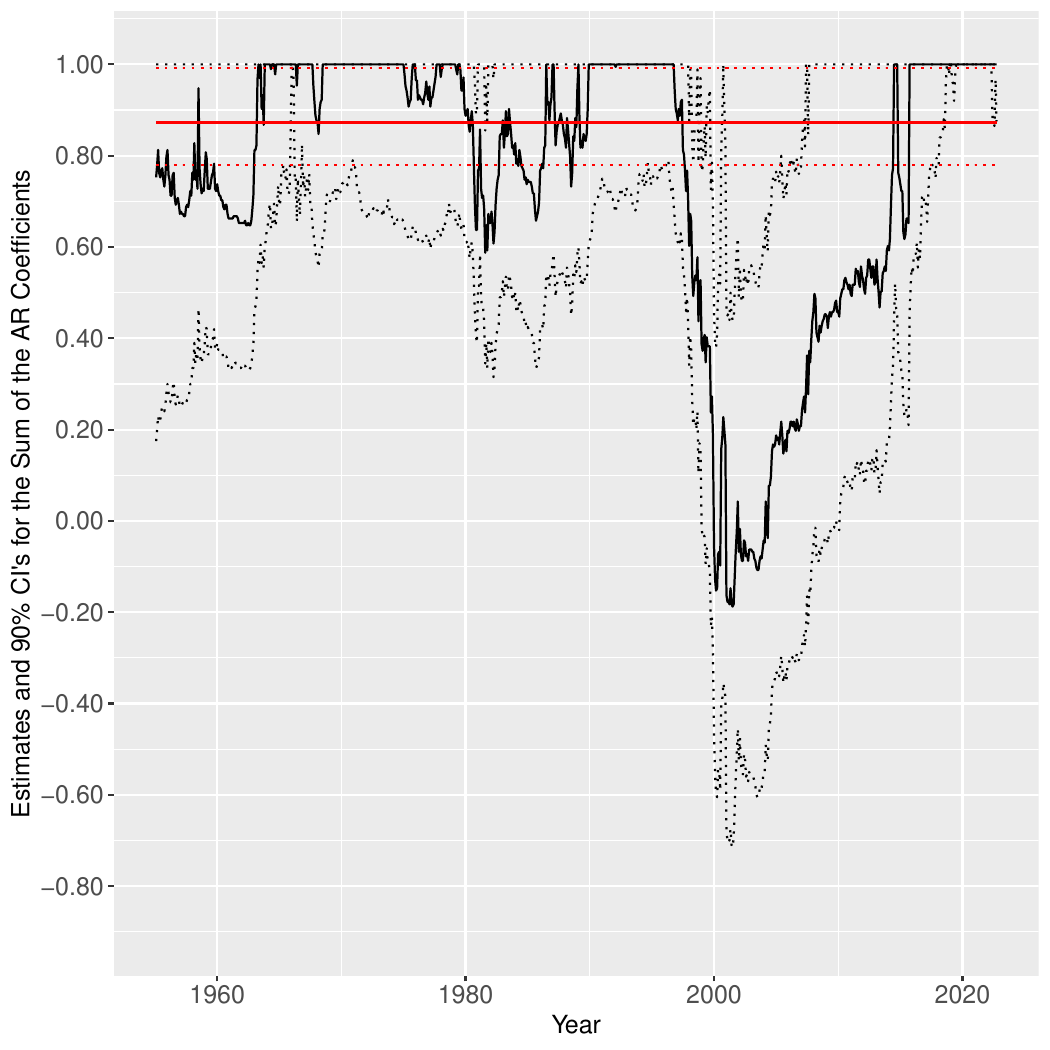}

}
\par\end{centering}
\begin{centering}
\subfloat[S\&P 500 Index, $n\widehat{h}_{us}=$ 2,518]{\includegraphics[scale=0.36]{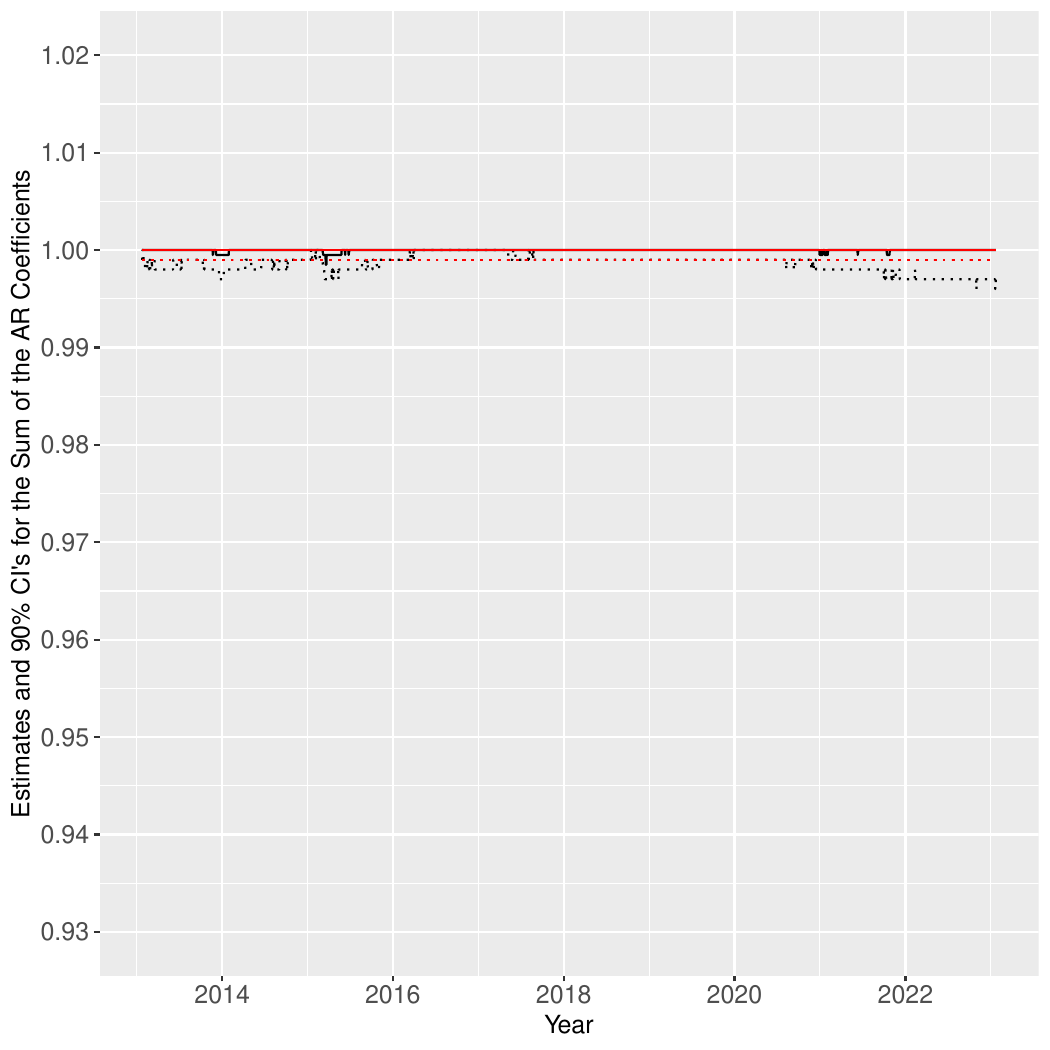}

}\quad{}\subfloat[S\&P 500 Index, $1.5n\widehat{h}_{us}=$ 3,777]{\includegraphics[scale=0.36]{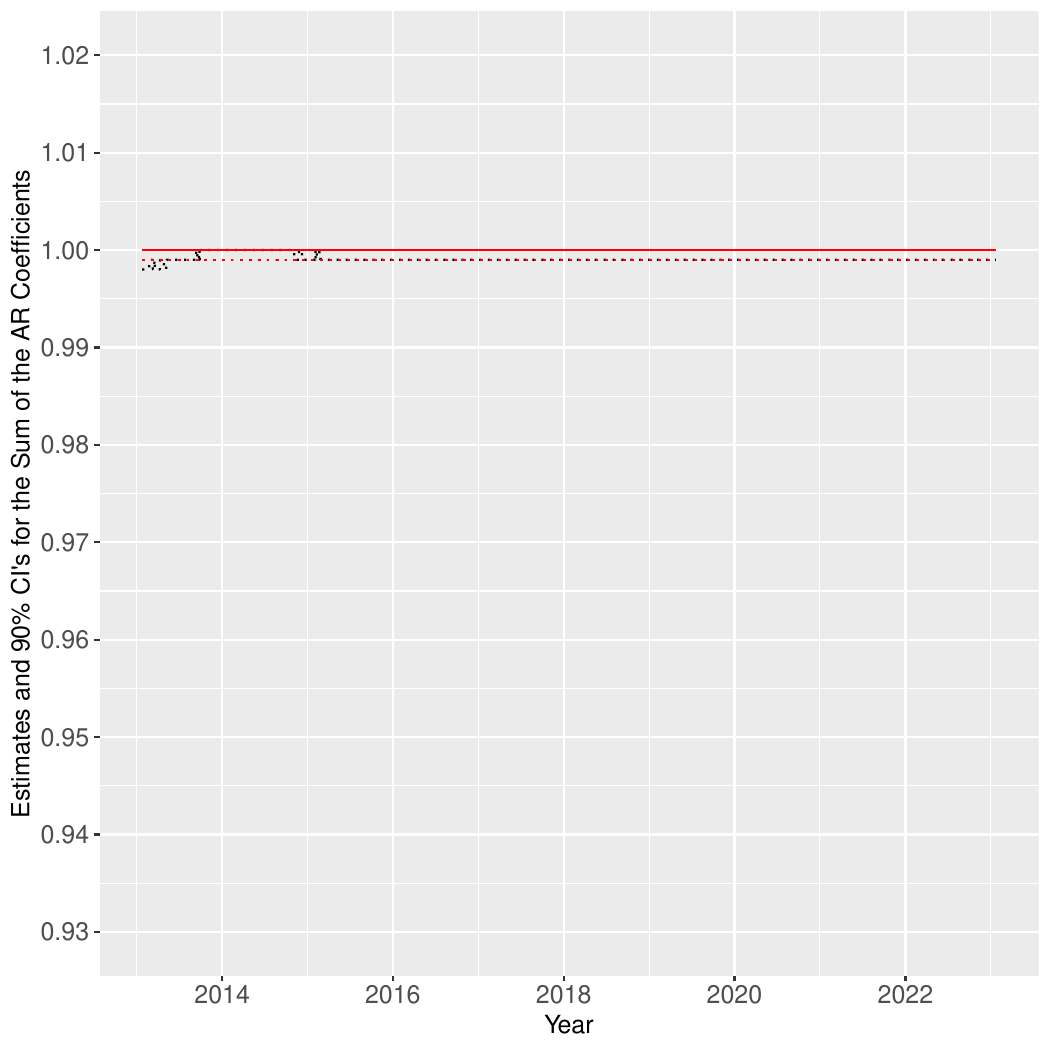}

}
\par\end{centering}
\begin{centering}
\subfloat[US Industrial Production, $n\widehat{h}_{us}=$ 368]{\includegraphics[scale=0.36]{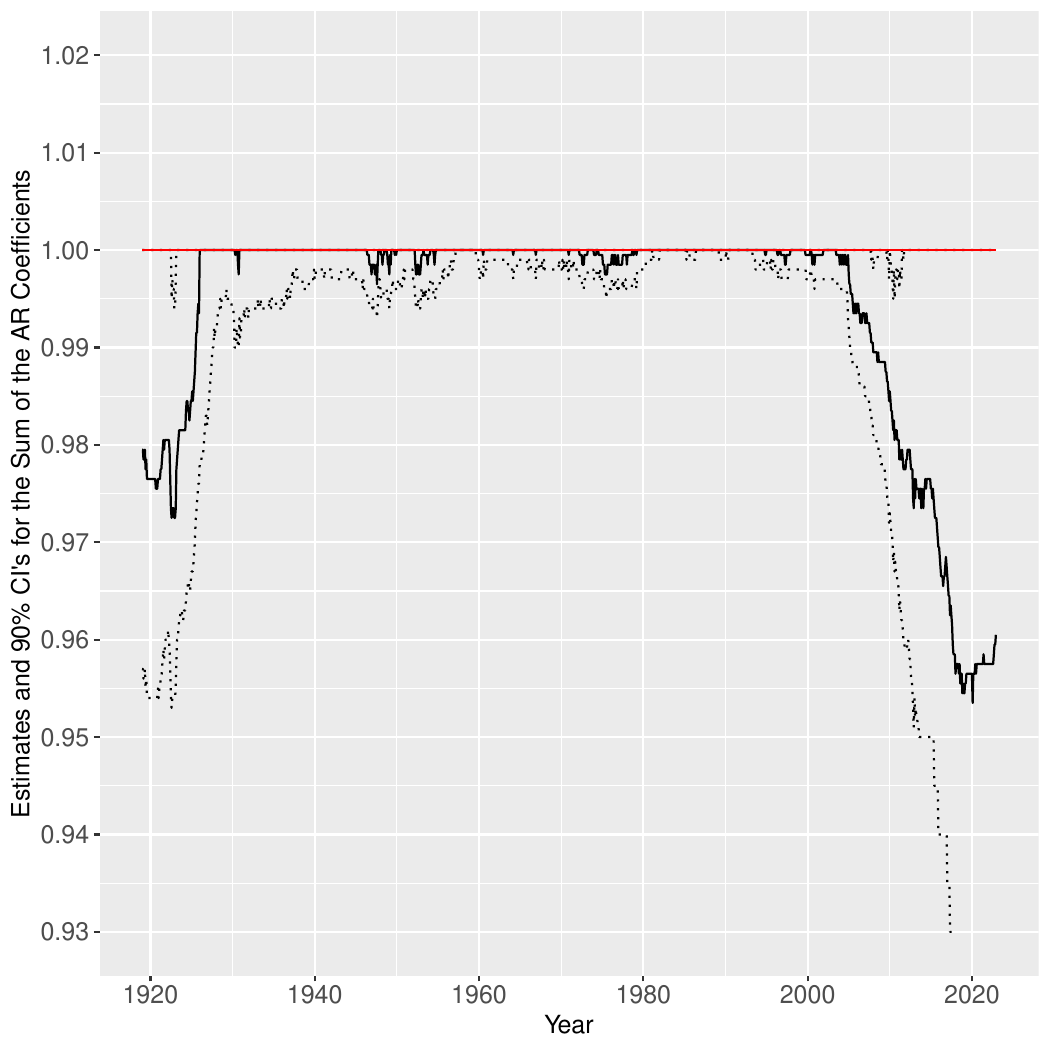}

}\quad{}\subfloat[US Industrial Production, $1.5n\widehat{h}_{us}=$ 551]{\includegraphics[scale=0.36]{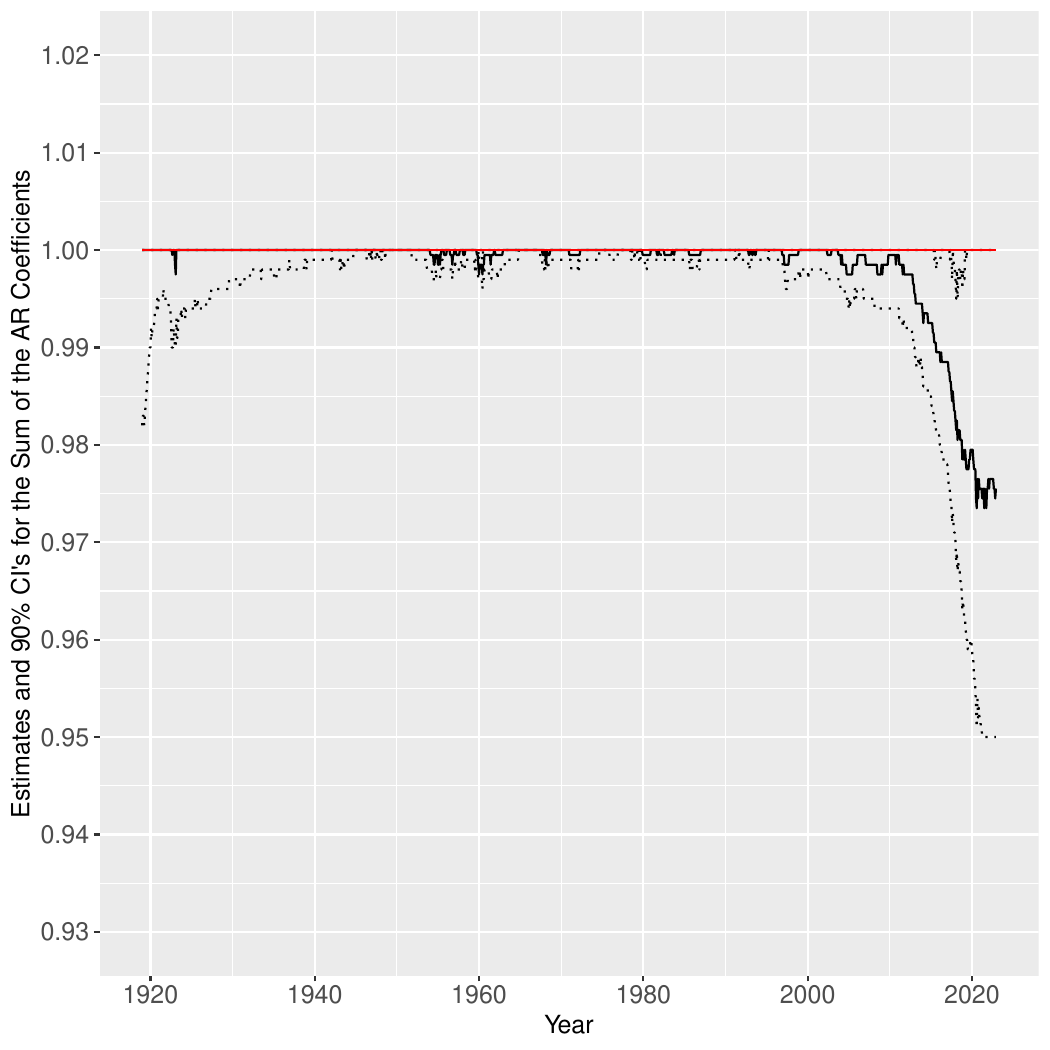}

}
\par\end{centering}
\caption{Estimates and 90\% CI's for the Sum of the AR Coefficients in TVP-AR(12)
Models: Switzerland Inflation, S\&P 500 Index, and US Industrial Production\protect\label{fig:EMP_AR12_SM}}
\end{figure}

\bibliographystyle{qelite}
\bibliography{TVP_AR_bib_Aug28_2023}

\end{document}